\pdfoutput=1
\RequirePackage{ifpdf}
\ifpdf 
\documentclass[pdftex]{sigma}
\else
\documentclass{sigma}
\fi

\numberwithin{equation}{section}

\usepackage{mathtools}

\graphicspath{{images/}}
\usepackage[frame,arc,knot,matrix]{xy}

\usepackage{pgfplots}
\pgfplotsset{compat=1.16}
\usepgfmodule{plot}
\usetikzlibrary{patterns}
\usetikzlibrary{calc}
\usetikzlibrary{decorations}
\usetikzlibrary{decorations.pathreplacing}
\usetikzlibrary{decorations.markings}
\usetikzlibrary{tqft}

\newcommand{\im}{\mathrm{i}}

\newcommand{\R}{\mathbb{R}}
\newcommand{\bC}{\mathbb{C}}
\newcommand{\defeq}{\coloneqq}
\newcommand{\tens}{\otimes}
\DeclareMathOperator{\id}{id}
\newcommand{\one}{\mathbf{1}}
\newcommand{\toi}{\hookrightarrow}
\newcommand{\xd}{\mathrm{d}}
\newcommand{\xD}{\mathcal{D}}
\newcommand{\cH}{\mathcal{H}}
\newcommand{\cS}{\mathbf{S}}
\newcommand{\cA}{\mathbf{A}}
\newcommand{\po}{\mathsf{P}}
\DeclareMathOperator{\tr}{\mathrm{tr}}

\newcommand{\pol}{{\rm P}}
\newcommand{\cpol}{{\overline{\rm P}}}
\newcommand{\ipol}{{\rm int}}
\newcommand{\epol}{{\rm ext}}
\newcommand{\rL}{\mathrm{L}}

\newcommand{\no}[1]{{:}#1{:}}
\newcommand{\cp}{\bullet}
\newcommand{\qp}{\star}
\newcommand{\mpol}{\text{m}}
\newcommand{\tord}{\mathbf{T}}
\newcommand{\rprt}{\mathrm{R}}
\newcommand{\iprt}{\mathrm{I}}
\newcommand{\sst}{*}
\newcommand{\qsoa}{\mathcal{A}}
\newcommand{\qsoi}{\mathcal{I}}
\newcommand{\pols}{\mathcal{P}}
\newcommand{\ds}{\circ}
\newcommand{\corr}{\mathcal{C}}
\newcommand{\qcomp}{\diamond}
\newcommand{\coh}{\mathsf{K}}
\newcommand{\ncoh}{\mathsf{k}}
\newcommand{\dop}{\mathsf{D}}

{\theoremstyle{definition}
\newtheorem{dfn}{Definition}[section]}

\newtheorem{lem}[dfn]{Lemma}
\newtheorem{prop}[dfn]{Proposition}
\newtheorem{thm}[dfn]{Theorem}
\newtheorem{cor}[dfn]{Corollary}

\begin{document}
\allowdisplaybreaks

\newcommand{\arXivNumber}{2009.12342}

\renewcommand{\PaperNumber}{073}

\FirstPageHeading

\ShortArticleName{Locality and General Vacua in Quantum Field Theory}

\ArticleName{Locality and General Vacua in Quantum Field Theory}

\Author{Daniele COLOSI~$^{\rm a}$ and Robert OECKL~$^{\rm b}$}

\AuthorNameForHeading{D.~Colosi and R.~Oeckl}

\Address{$^{\rm a)}$~Escuela Nacional de Estudios Superiores, Unidad Morelia,\\
\hphantom{$^{\rm a)}$}~Universidad Nacional Aut\'onoma de M\'exico, C.P.~58190, Morelia, Michoac\'an, Mexico}
\EmailD{\href{mailto:dcolosi@enesmorelia.unam.mx}{dcolosi@enesmorelia.unam.mx}} 

\Address{$^{\rm b)}$~Centro de Ciencias Matem\'aticas, Universidad Nacional Aut\'onoma de M\'exico,\\
\hphantom{$^{\rm b)}$}~C.P.~58190, Morelia, Michoac\'an, Mexico}
\EmailD{\href{mailto:robert@matmor.unam.mx}{robert@matmor.unam.mx}}

\ArticleDates{Received September 28, 2020, in final form July 13, 2021; Published online July 25, 2021}

\Abstract{We extend the framework of general boundary quantum field theory (GBQFT) to achieve a fully local description of realistic quantum field theories. This requires the quantization of non-K\"ahler polarizations which occur generically on timelike hypersurfaces in Lorentzian spacetimes as has been shown recently. We achieve this in two ways: On the one hand we replace Hilbert space states by observables localized on hypersurfaces, in the spirit of algebraic quantum field theory. On the other hand we apply the GNS construction to twisted star-structures to obtain Hilbert spaces, motivated by the notion of reflection positivity of the Euclidean approach to quantum field theory. As one consequence, the well-known representation of a vacuum state in terms of a sea of particle pairs in the Hilbert space of another vacuum admits a vast generalization to non-K\"ahler vacua, particularly relevant on timelike hypersurfaces.}


\Classification{81P16; 81S10; 81S40; 81T20; 81T70}

{\small \tableofcontents}

\section{Introduction}

\subsection[The $S$-matrix]{The $\boldsymbol S$-matrix}
\label{sec:intro_smatrix}

A central construction on which much of the impressive predictive power of quantum field theory rests is the \emph{S-matrix}. From this, collision cross sections may be directly calculated. The $S$-matrix arises as an asymptotic transition amplitude. We briefly recall this in the following using a convenient language and notation. We refer the reader to standard textbooks such as~\cite{ItZu:qft}.

Consider a quantum process between an initial time $t_1$ and a final time $t_2$. Let $\cH$ be the Hilbert space of states of the system. Denote by $U_{[t_1,t_2]}\colon \cH\to\cH$ the unitary time-evolution operator.
The \emph{transition amplitude} between an initial state $\psi_1\in\cH$ and a final state $\psi_2\in\cH$ may be represented via the Feynman \emph{path integral} as\footnote{For convenience, we use a notation that suggests a scalar field. However, where not explicitly indicated otherwise, our considerations apply to any type of bosonic field.}
\begin{equation}
 \langle \psi_2, U_{[t_1,t_2]}\psi_1\rangle =
 \int_{K_{[t_1,t_2]}} \xD\phi\, \psi_1(\phi|_{t_1}) \overline{\psi_2(\phi|_{t_2})}
 \exp\big(\im S_{[t_1,t_2]}(\phi)\big).
 \label{eq:tamplpi}
\end{equation}
Here, the integral is over field configurations $\phi\in K_{[t_1,t_2]}$ in the spacetime region spanned by the time interval $[t_1,t_2]$. Field configurations restricted to the initial and final time are represented by $\phi|_{t_1}$ and $\phi|_{t_2}$ respectively. As factors in the integrand appear the Schrödinger \emph{wave functions} of the initial and the final state.\footnote{While ubiquitous in non-relativistic quantum mechanics, the Schrödinger representation is not commonly used in quantum field theory. But see for example~\cite{Jac:schroedinger}. For a rigorous definition in the linear case see~\cite{Oe:schroedhol}.} The \emph{action} in the same spacetime region is denoted by $S_{[t_1,t_2]}$. The predictive content of the transition amplitude lies in providing the \emph{probability} $P(\psi_2|\psi_1)$ for measuring the final state $\psi_2$ (as opposed to an orthogonal state), given the initial state $\psi_1$ was prepared. This is the modulus square of the transition amplitude,
\begin{equation}
 P(\psi_2|\psi_1)=\big|\big\langle \psi_2, U_{[t_1,t_2]}\psi_1\big\rangle\big|^2.
 \label{eq:probta}
\end{equation}
We can think of the transition amplitude as encoding the physics in the spacetime region \mbox{$[t_1,t_2]\times\R^3$}, as illustrated in Figure~\ref{fig:tampl}.

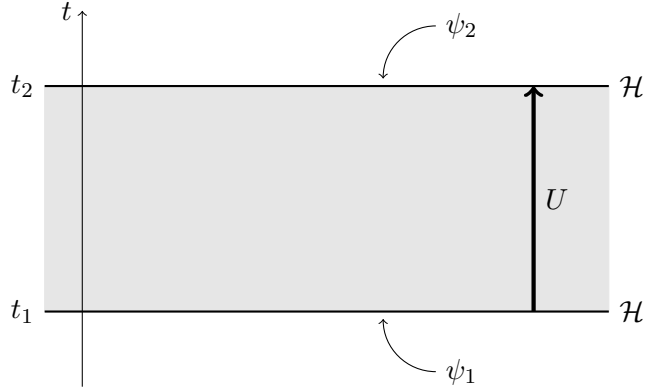
\begin{figure} \centering
 \begin{tikzpicture}
\filldraw[color=gray!20] (-0.5,1) rectangle (7,4);
\draw[->] (0,0) -- (0,5) node [left] {$t$};
\draw[thick] (-0.5,1) node [left] {$t_1$} -- (7,1) node [right] {$\cH$};
\draw[thick] (-0.5,4) node [left] {$t_2$} -- (7,4) node [right] {$\cH$};
\draw[ultra thick,->] (6,1) -- (6,4) node [midway, right] {$U$};
\draw (4.7,0.2) node [right] {$\psi_1$} edge[out=180,in=270,->] (4,0.9);
\draw (4.7,4.8) node [right] {$\psi_2$} edge[out=180,in=90,->] (4,4.1);
\end{tikzpicture}
 \caption{Spacetime picture of a transition amplitude.} \label{fig:tampl}
\end{figure}

If the field theory in question is \emph{free}, i.e., has a \emph{linear} phase space, the Hilbert space $\cH$ is easily constructed as a \emph{Fock space}. Also, the action $S$ is then quadratic and the path integral~(\ref{eq:tamplpi}) can be computed in a straightforward manner.
Physically realistic quantum field theories (such as those of the Standard Model) are non-linear, however. With few exceptions, they can at present only be handled \emph{perturbatively}. To this end, the action $S=S_0+S_{\rm int}$ is split into a free part~$S_0$ that is quadratic, and an \emph{interacting} part $S_{\rm int}$.
We then make the assumption that at very early and very late times in a collision experiment the particles behave as free particles and are in these regimes well described by the linear theory with action $S_0$. In the linear theory the Hilbert space $\cH$ of states is a Fock space, and we have a good understanding of how its states encode asymptotic particle configurations. The $S$-matrix is then a unitary operator from an initial copy of the Hilbert space to a final copy of this Hilbert space, encoding interactions that are idealized to only happen at intermediate times (by switching on $S_{\rm int}$).

In spite of their empirical success, the perturbative methods of quantum field theory have serious limitations. They work only in certain regimes, while in general the perturbation expan\-sion does not converge. For example, the understanding of the proton as a bound state of quarks and gluons is beyond their reach.\footnote{There are alternative methods that work better for bound state systems, such as lattice gauge theory. However, this approach is also approximate and has its own limitations.} At the same time these methods do only partially generalize from Minkowski space to general curved spacetime and have essentially nothing to say in a~regime where quantum properties of gravity would become important.\footnote{Significant perturbative inroads into quantum gravity can be made using effective field theory~\cite{DoHo:qgeffective}.}

\subsection{Topological quantum field theory (TQFT)}

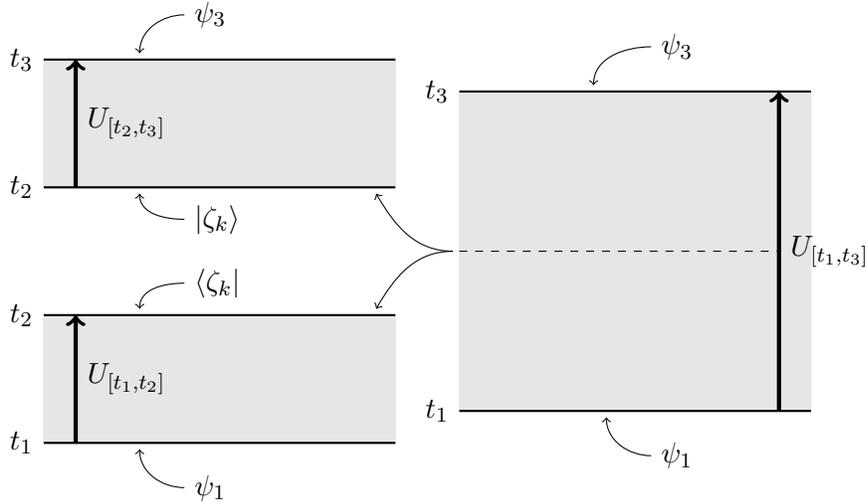
\begin{figure}
 \centering
 \begin{tikzpicture}[scale=0.85]
\filldraw[color=gray!20] (-0.5,0) rectangle (5,2);
\draw[thick] (-0.5,0) node [left] {$t_1$} -- (5,0);
\draw[thick] (-0.5,2) node [left] {$t_2$} -- (5,2);
\draw[ultra thick,->] (0,0) -- (0,2) node [midway, right] {$U_{[t_1,t_2]}$};
\draw (1.7,-0.7) node [right] {$\psi_1$} edge[out=180,in=270,->] (1,-0.1);
\draw (1.7,2.5) node [right] {$\langle \zeta_k |$} edge[out=180,in=90,->] (1,2.1);
\filldraw[color=gray!20] (-0.5,4) rectangle (5,6);
\draw[thick] (-0.5,4) node [left] {$t_2$} -- (5,4);
\draw[thick] (-0.5,6) node [left] {$t_3$} -- (5,6);
\draw[ultra thick,->] (0,4) -- (0,6) node [midway, right] {$U_{[t_2,t_3]}$};
\draw (1.7,3.5) node [right] {$| \zeta_k \rangle$} edge[out=180,in=270,->] (1,3.9);
\draw (1.7,6.7) node [right] {$\psi_3$} edge[out=180,in=90,->] (1,6.1);
\filldraw[color=gray!20] (6,0.5) rectangle (11.5,5.5);
\draw[thick] (6,0.5) node [left] {$t_1$} -- (11.5,0.5);
\draw[thick] (6,5.5) node [left] {$t_3$} -- (11.5,5.5);
\draw[dashed] (6,3) -- (11,3);
\draw[ultra thick,->] (11,0.5) -- (11,5.5) node [midway, right] {$U_{[t_1,t_3]}$};
\draw (9,-0.2) node [right] {$\psi_1$} edge[out=180,in=270,->] (8.3,0.4);
\draw (9,6.2) node [right] {$\psi_3$} edge[out=180,in=90,->] (8.1,5.6);
\draw (5.9,3) edge[out=180,in=300,->] (4.7,3.9);
\draw (5.9,3) edge[out=180,in=60,->] (4.7,2.1);
\end{tikzpicture}
 \caption{Temporal composition of transition amplitudes.}
 \label{fig:t-comp}
\end{figure}

One of the obstacles to making quantum field theory non-perturbatively well-defined lies in the notorious problem of making mathematical sense of the path integral~(\ref{eq:tamplpi}). However, examining the empirically successful methods of perturbative quantum field theory one realizes that what is used is not an actual measure on some measurable space of configurations, but rather certain properties that such a measure, if it existed, would induce in transition amplitudes and related objects. This suggests that instead of postulating the existence of such a measure we should directly axiomatize the relevant properties. One of the most fundamental properties of the path integral is its \emph{composition property}. In Feynman's original non-relativistic setting~\cite{Fey:stnrqm} this is simply the analogue of the evolution operator composition identity $U_{[t_1,t_3]}=U_{[t_2,t_3]}\circ U_{[t_1,t_2]}$ for~times $t_1<t_2<t_3$. In terms of transition amplitudes, let $\{\zeta_k\}_{k\in I}$ be an orthonormal basis of~the Hilbert space $\cH$, so we have
\begin{equation}
 \big\langle \psi_3, U_{[t_1,t_3]} \psi_1\big\rangle
 = \sum_{k\in I} \big\langle \psi_3, U_{[t_2,t_3]} \zeta_k\big\rangle
 \big\langle \zeta_k, U_{[t_1,t_2]} \psi_1\big\rangle.
 \label{eq:tcompax}
\end{equation}
This is illustrated in Figure~\ref{fig:t-comp}.
In terms of the path integral~(\ref{eq:tamplpi}) this means that the integrals on the right-hand side over the configuration spaces $K_{[t_1,t_2]}$ and $K_{[t_2,t_3]}$ are ``glued together'' to an integral over the joint configuration space $K_{[t_1,t_3]}$. While this \emph{temporal composition property} holds in the non-relativistic as in~the relativistic setting, the equal footing of space and time in the latter suggests a vast generalization.

At the end of the 1980s work of Edward Witten on understanding aspects of geometry and topology through quantum field theoretic methods (the path integral in particular) and vice versa was picked up by mathematicians, notably Graeme Segal and Michael Atiyah. The latter extracted from this an axiomatic system known as \emph{topological quantum field theory $($TQFT$)$} \mbox{\cite{Ati:tqft,Seg:cftproc}}. This realizes precisely an axiomatic implementation of the composition property of the path integral, without mention of any actual path integral or measure.
Concretely, fix a~dimen\-sion $n$ of ``spacetime''. The basic objects encoding spacetime are (usually just topological) cobordisms of dimension $n$ and their boundary components. A \emph{cobordism} is a manifold~$M$ of~dimension $n$ with a boundary $\partial M$ presented as the disjoint union of an ``incoming'' and an~``outgoing'' closed manifold of dimension $n-1$, $\partial M=\partial M_{\rm in} \sqcup \partial M_{\rm out}$. These cobordisms are the analogues of the time-interval regions considered previously in the context of transition amplitudes. Now to each closed manifold $\Sigma$ of dimension $n-1$ we associate a complex Hilbert (or perhaps just vector) space $\cH_{\Sigma}$ of ``states''. In the special case that $\Sigma$ is the empty set the associated space is the one-dimensional (Hilbert) space $\bC$. We should think of $\Sigma$ as analogous to an equal-time hypersurface with associated copy of the Hilbert space $\cH$. When a closed $n-1$ manifold $\Sigma$ decomposes into a disjoint union $\Sigma=\Sigma_1 \sqcup \Sigma_2$, the associated state space decomposes into a corresponding tensor product, $\cH_{\Sigma}=\cH_{\Sigma_1}\tens\cH_{\Sigma_2}$.
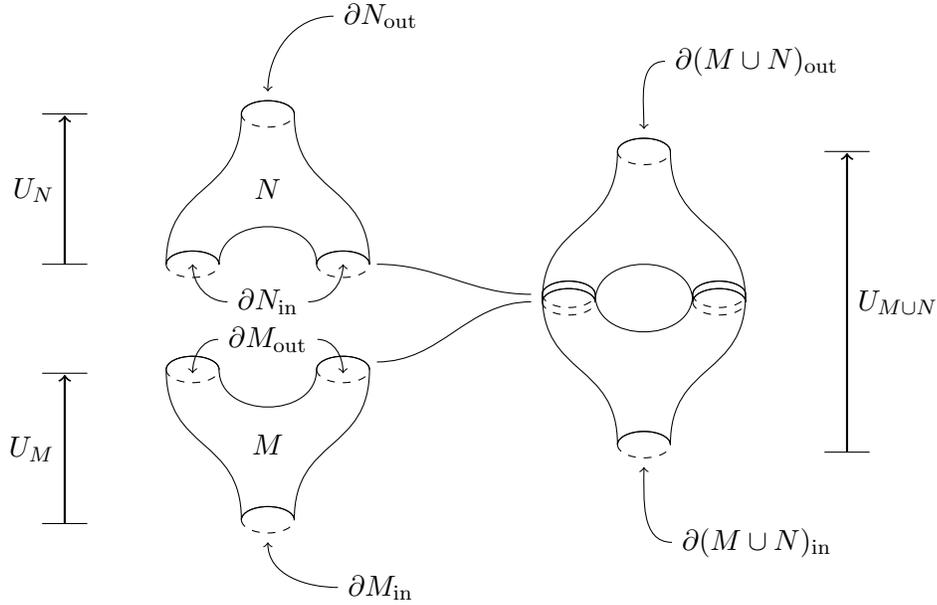
\begin{figure}
 \centering
 \begin{tikzpicture}[
  tqft/.cd,
  cobordism/.style={draw},
  every upper boundary component/.style={draw},
  every lower boundary component/.style={draw,dashed},
]
\pic[tqft/pair of pants,name=a];
\pic[tqft/pair of pants,name=c, at={(5,-0.5)}];
\pic[tqft/reverse pair of pants,upper boundary component/.style={draw,dashed},at ={(4,-2.4)} ] ;
\pic[tqft/reverse pair of pants,at ={(-1,-3.4) } ] ;
\node at (0,-1)  {$N$};
\node at (0,-4.4)  {$M$};
\node at (1.5,1.3) (1) {$\partial N_{\text{out}}$};
\node at (0,-2.5) (2) {$\partial N_{\text{in}}$};
\node at (1.5,-6.3) (3) {$\partial M_{\text{in}}$};
\node at (0,-3) (4) {$\partial M_{\text{out}}$};
\draw [->] (1.west) to [out=180,in=90] (0,0.3);
\draw [->] (2.west) to [out=180,in=270] (-1,-2);
\draw [->] (2.east) to [out=0,in=270] (1,-2);
\draw [->] (3.west) to [out=180,in=270] (0,-5.7);
\draw [->] (4.west) to [out=180,in=90] (-1,-3.45);
\draw [->] (4.east) to [out=0,in=90] (1,-3.45);
\draw[-] (-3,-3.45) -- (-2.4,-3.45);
\draw[-] (-3,-5.45) -- (-2.4,-5.45);
\draw[->,thick] (-2.7,-5.45) -- node [midway,left] {$U_M$} (-2.7,-3.47);
\draw[-] (-3,0) -- (-2.4,0);
\draw[-] (-3,-2) -- (-2.4,-2);
\draw[->,thick] (-2.7,-2) -- node [midway,left] {$U_N$} (-2.7,-0.02);
\node at (6.5,0.7) (4) {$\partial (M \cup N)_{\text{out}}$};
\node at (6.5,-5.7) (5) {$\partial (M \cup N)_{\text{in}}$};
\draw [->] (4.west) to [out=180,in=90] (5,-0.2);
\draw [->] (5.west) to [out=180,in=270] (5,-4.7);
\draw[-] (8,-0.5) -- (7.4,-0.5);
\draw[-] (8,-4.5) -- (7.4,-4.5);
\draw[->,thick] (7.7,-4.5) -- node [midway,right] {$U_{M \cup N}$} (7.7,-0.52);
\draw[-] (1.45,-2) to [out=0,in=180] (3.5,-2.4);
\draw[-] (1.45,-3.3) to [out=0,in=180] (3.5,-2.5);
\end{tikzpicture}
 \caption{Composition of cobordisms and associated morphisms in TQFT.}
 \label{fig:TQFT-comp}
\end{figure}
This is the analogue of the usual rule for combining independent systems in quantum theory via the tensor product. To each cobordism~$M$ with boundary $\partial M=\partial M_{\rm in} \sqcup \partial M_{\rm out}$ we associate a linear map (that is required to be unitary in the Hilbert space setting) $U_M\colon \cH_{\partial M_{\rm in}}\to\cH_{\partial M_{\rm out}}$. Finally, suppose two cobordisms $M$, $N$ can be concatenated, i.e., \emph{glued} to a single cobordism $M\cup N$ by identifying~$\partial M_{\rm out}$ with~$\partial N_{\rm in}$, see Figure~\ref{fig:TQFT-comp}. Then, the associated linear maps compose as $U_{M\cup N}=U_N\circ U_M$. This is the \emph{composition property} motivated from the path integral.\footnote{An attentive reader might complain that this composition axiom is also analogous to the simple temporal composition of evolution of the standard operator picture and does not necessarily require inspiration from the path integral. Indeed, it is only with the generalization to be discussed below that the path integral picture becomes compelling.}

Topological quantum field theory and related developments have been extremely fruitful for mathematics, leading to a revolution of algebraic topology and low dimensional topology, while also involving the areas of knot theory, operator algebras, monoidal category theory and quantum groups to name a few, see, e.g.,~\cite{Tur:qinv}.
However, the theories typically described by TQFTs involve ``spacetimes'' that are topological manifolds with non-trivial topology while admitting only finitely many degrees of freedom. For realistic quantum field theories we need spacetime to~carry a Lorentzian metric, and we crucially need the ability to deal with infinitely many degrees of freedom. On the other hand we are not much interested in non-trivial spacetime topologies, except to a very limited extend in black hole physics and in cosmology. The possibility of TQFT to work with spacetimes that do not carry a metric turns into an attractive feature, however, once we are interested in quantum gravity rather than quantum field theory. This has led to~various suggestions that ``quantum gravity should be a TQFT''~\cite{Bar:qgqft,Cra:tqftqg,Smo:tqftnpgrav}.

\subsection{General boundary quantum field theory (GBQFT)}
\label{sec:intro_gbqft}

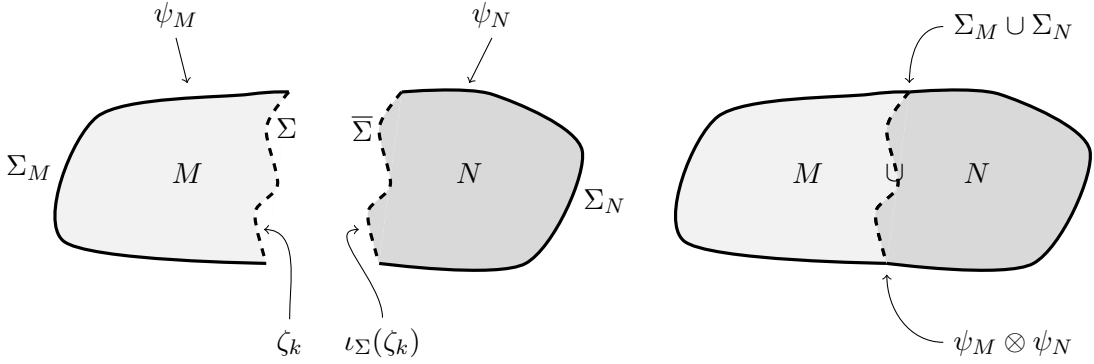
\begin{figure}
 \centering
 \begin{tikzpicture}[scale=1.5]
\begin{scope}
\draw [very thick,xshift=0cm,fill=gray!10] plot [smooth ] coordinates {(2,0)(0.2,0.2) (0.5,1.3) (1.85,1.5) (2.2,1.52)};
\draw [dashed,very thick,xshift=0cm,fill=white] plot [smooth ] coordinates {(2.2,1.52) (2,1.2)(2.1,0.7)(1.9,0.45)(2,0)};
\draw [very thick,xshift=1cm,fill=gray!30] plot [smooth ] coordinates {(2,0)(3.3,0) (3.8,1)(3,1.5) (2.2,1.52)};
\draw [dashed,very thick,xshift=1cm,fill=gray!30] plot [smooth ] coordinates {(2.2,1.52) (2,1.2)(2.1,0.7)(1.9,0.45)(2,0)};
\node at (1.3,0.8) {$M$};
\node at (3.8,0.8) {$N$};
\node at (2.18,1.2) {$\Sigma$};
\node at (2.85,1.2) {$\overline{\Sigma}$};
\node at (-0.1,0.85) {$\Sigma_M$};
\node at (5,0.55) {$\Sigma_N$};
\draw[->] (1.2,2) node [above] {$\psi_M$} -- (1.3,1.6);
\draw[->] (4,2) node [above] {$\psi_N$} -- (3.8,1.6);
\draw (2,-0.7) node [right] {$\zeta_k$} edge[out=90,in=0,->] (1.98,0.3);
\draw (2.6,-0.7) node [right] {$\iota_{\Sigma}(\zeta_k)$} edge[out=90,in=180,->] (2.82,0.3);
\end{scope}
\begin{scope}[xshift=5.5cm]
\draw [very thick,xshift=0cm,fill=gray!10] plot [smooth ] coordinates {(2,0)(0.2,0.2) (0.5,1.3) (1.85,1.5) (2.2,1.52)};
\draw [dashed,very thick,xshift=0cm,fill=white] plot [smooth ] coordinates {(2.2,1.52) (2,1.2)(2.1,0.7)(1.9,0.45)(2,0)};
\draw [very thick,xshift=0cm,fill=gray!30] plot [smooth ] coordinates {(2,0)(3.3,0) (3.8,1)(3,1.5) (2.2,1.52)};
\draw [dashed,very thick,xshift=0cm,fill=gray!30] plot [smooth ] coordinates {(2.2,1.52) (2,1.2)(2.1,0.7)(1.9,0.45)(2,0)};
\node at (1.3,0.8) {$M$};
\node at (2.07,0.8) {$\cup$};
\node at (2.8,0.8) {$N$};
\draw (2.5,-0.7) node [right] {$\psi_M \otimes \psi_N$} edge[out=180,in=270,->] (2,-0.1);
\draw (2.5,2.1) node [right] {$\Sigma_M \cup \Sigma_N$} edge[out=180,in=90,->] (2.2,1.6);
\end{scope}
\end{tikzpicture}
 \caption{Composition of regions and associated amplitudes in GBQFT.}
 \label{fig:GBQFT-comp}
\end{figure}
Taking into account the properties of realistic quantum field theories, the success of textbook methods to extract physical predictions from them~\cite{Oe:reveng}, and motivations from quantum gra\-vity~\cite{Oe:catandclock} lead to \emph{general boundary quantum field theory $($GBQFT$)$}~\cite{Oe:gbqft} as a modern incarnation of~this axiomatic program. (For a more comprehensive perspective, including from the foundations of quantum theory, see~\cite{Oe:posfound}.)
In contrast to (non-extended\footnote{There are also \emph{extended} versions of TQFT that also implement stronger composition axioms~\cite{Wal:tqftnotes}.}) TQFT, a much stronger version of the composition property of the path integral is axiomatized. This reflects the physical principle of \emph{locality} as we shall explain. To implement this we drop the in-out structure of~TQFT. At~the same time we emphasize that all manifolds are \emph{oriented}. While this is usually also required in TQFT, in the interest of simplicity we have omitted to mention this previously. Thus, the oriented $n$-manifolds representing pieces of spacetime, called \emph{regions} in the following, are no longer presented as cobordisms. That is, the boundary $\partial M$ of a region $M$ is no longer equipped with a decomposition into an ``incoming'' and an ``outgoing'' part. As before, we associate to an oriented $n-1$ manifold $\Sigma$, now called a \emph{hypersurface}, a Hilbert space $\cH_{\Sigma}$ of~``states''. However, we allow $\Sigma$ to have a boundary, called \emph{corner}. To a region $M$ we associate a linear \emph{amplitude} map $\rho_{M}\colon \cH_{\partial M}\to\bC$.\footnote{Technically this is as in TQFT if we set $\partial M_{\rm in}=\partial M$ and $\partial M_{\rm out}=\varnothing$.} The \emph{composition axiom} takes the following form.
Let~$M$ and~$N$ be regions with boundaries decomposing as $\partial M=\Sigma_M\cup\Sigma$ and $\partial N=\overline{\Sigma}\cup\Sigma_N$. Here, $\overline{\Sigma}$ deno\-tes a copy of $\Sigma$ with opposite orientation.\footnote{Since the orientations of $\Sigma$ and $\overline{\Sigma}$ are induced from the orientations of $M$ and $N$ their opposite orientations ensure that the orientations of $M$ and $N$ match upon gluing.}
We glue $M$ and $N$ together along $\Sigma$, as illustrated in~Figure~\ref{fig:GBQFT-comp}. Then, the amplitude of the composite region satisfies
\begin{equation}
 \rho_{M\cup N} (\psi_M\tens \psi_N)=\sum_{k\in I} \rho_{M}(\psi_M\tens\zeta_k)
 \rho_{N}(\iota_{\Sigma}(\zeta_k)\tens\psi_N).
 \label{eq:stcompax}
\end{equation}
Here, $\{\zeta_k\}_{k\in I}$ denotes an orthonormal basis of $\cH_\Sigma$ and $\{\iota_{\Sigma}({\zeta_k})\}_{k\in I}$ denotes the dual basis of $\cH_{\overline{\Sigma}}$ (see Section~\ref{sec:stdstate} for notation).
\begin{figure}
 \centering
 \includegraphics[width=0.43\textwidth]{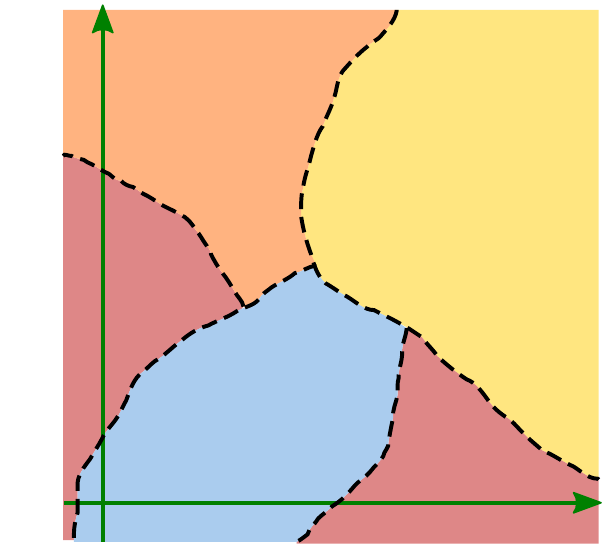}
\put(-190,157){\makebox(0,0)[lb]{time}}
\put(-40,5){\makebox(0,0)[lb]{space}}
 \caption{Illustration of spacetime locality. The physics in spacetime can be entirely reconstructed from the physics in regions, into which it may arbitrarily be decomposed.}
 \label{fig:st-decomp}
\end{figure}
The axiom~(\ref{eq:stcompax}) is the direct generalization of the axiom~(\ref{eq:tcompax}), justified in the same way from the (formal) properties of the path integral. Conceptually, as compared to the evolution picture or to the TQFT picture, we replace ``locality in time'' by the much stronger ``locality in spacetime'' (illustrated in Figure~\ref{fig:st-decomp}). The amplitude for a~region completely encodes the possible physics in that region and any potential interaction with physics in adjacent regions. Moreover, no interaction can take place that is not mediated through adjacency (direct or indirect). This may be seen as a quantum analog of the powerful principle of locality in classical field theory, where any interaction must be mediated by a field traveling through spacetime and thus connecting the interacting systems.

\looseness=1 Besides the restrictive in-out structure there is another crucial obstacle to making TQFT applicable to realistic QFT. That is the restriction to finitely many degrees of freedom. More precisely, the Hilbert (or vector) spaces in TQFT have to be finite-dimensional. To illustrate this, imagine a 2-dimensional TQFT. Associate a vector space $V$ to the circle $S^1$. Consider a~cylinder $C$ as a cobordism from $S^1$ to $S^1$. In a topological setting there is no further structure and the associated linear map $U_C\colon V\to V$ basically has to be the identity $U_C=\id_V$ (or a~projector, but then we may replace $V$ with a quotient) by self-composition. So if we glue the two ends of the cylinder together the associated map is just a complex number (as there is no boundary anymore) and this number is the trace of the identity, i.e., the dimension of~$V$, $\tr(U_C)=\tr(\id_V)=\dim V$. Now, if $V$ was infinite-dimensional this would make no sense, so we have to exclude this possibility.
The way to get around this in GBQFT is to simply exclude certain (large) classes of manifolds and of gluings. From a physical point of view this is no loss. As already emphasized we are not interested in regions or whole spacetimes with non-trivial topologies, except possibly in very special cases. In this way we gain the ability to work with infinite-dimensional Hilbert spaces at the expense of loosing some mathematically interesting (but unphysical) application of the formalism, such as constructing invariants of manifolds.

To model QFT we need manifolds equipped with a metric. In GBQFT manifolds are at least topological, but may carry any additional structure, depending on the theory to be modeled. (For example, for quantum gravity we might want merely topological or differentiable manifolds.) Of course the operations of gluing are required to be compatible with the additional structure. What is more, in general hypersurfaces are not really 1-codimensional manifolds, but rather germs of full-dimensional manifolds around 1-codimensional manifolds. Here, any additional structure on the manifolds also comes into play. Note that the possibility of the inclusion of additional structure is not at all special to GBQFT, but was already considered by Atiyah in his foundational article on TQFT~\cite{Ati:tqft}. However, the restriction to finitely many degrees of freedom in TQFT severely limits the role that this additional structure might play.

\subsection{Predictions in GBQFT}

It is not enough, for doing physics, to just propose some mathematical formalism. Only when the formalism is supplemented by prescriptions of how to extract predictions from it can it potentially serve to encode physics. In QFT the predictive power of the $S$-matrix rests on the simple probability rule~(\ref{eq:probta}) for the transition amplitude. TQFT on its own, being a purely mathematical framework, lacks any such prescription. In GBQFT on the other hand, the transition probability rule~(\ref{eq:probta}) is subject to a vast generalization. We recall here only the most basic version of this rule, originally proposed in~\cite{Oe:gbqft}. Thus, let $M$ be a spacetime region and~$\cH_{\partial M}$ the Hilbert space of states associated to its boundary. The type of prediction we consider here concerns measurements that can be performed in principle at or near the boundary of $M$. This involves the specification of two types of ingredients: On the one hand this is what we ``know'' or ``prepare''. On the other hand this is the ``question'' we want to ask. In the special case of a~transition amplitude such as the $S$-matrix, we usually consider the knowledge to be encoded in~an~initially prepared state ($\psi_1$ in~(\ref{eq:probta})), while the question is associated with the final state ($\psi_2$~in~(\ref{eq:probta})).

In general, the ``knowledge'' or ``preparation'' is encoded in terms of a closed subspace \mbox{$\cS\subseteq \cH_{\partial M}$}. We encode the ``question'' in another closed subspace $\cA\subseteq\cS\subseteq\cH_{\partial M}$. (The subspace relation $\cA\subseteq\cS$ expresses the fact that when asking a question we take into account what we already know.) Let $\po_{\cA}$, $\po_{\cS}$ be the corresponding orthogonal projection operators. Note that they are positive operators satisfying the inequalities $0\le \po_{\cA} \le \po_{\cS}$. The probability $P(\cA|\cS)$ for an affirmative answer is the quotient
\begin{equation}
 P(\cA|\cS)=\frac{\sum_{k\in I} |\rho_M(\po_{\cA}\zeta_k)|^2}{\sum_{k\in I} |\rho_M(\po_{\cS}\zeta_k)|^2}.
 \label{eq:probgbqft}
\end{equation}
Here $\{\zeta_k\}_{k\in I}$ is an orthonormal basis of $\cH_{\partial M}$.
For details, including how the probability rule~(\ref{eq:probta}) arises as a special case, we refer the reader to~\cite{Oe:gbqft}. A deeper understanding of this rule and its derivation from first principles can be found in~\cite{Oe:posfound}. The application of this probability rule in a particle scattering context, with particles coming in from and going out to spatial rather than temporal infinity was discussed for the first time in~\cite{Oe:kgtl}. Consider a ball of radius $R$ in Minkowski space, extended over all of time, see Figure~\ref{fig:hcscatter}. We call this a \emph{hypercylinder}~$M$. A~scattering process is described in terms of incoming and outgoing particles crossing the boundary $\partial M$. In the quantum theory, the Hilbert space of states $\cH_{\partial M}$ contains both incoming and outgoing particles. The rule~(\ref{eq:probgbqft}) then allows to predict for example what the probability for certain particles with certain quantum numbers is to go out given that certain other particles with certain other quantum numbers come in.

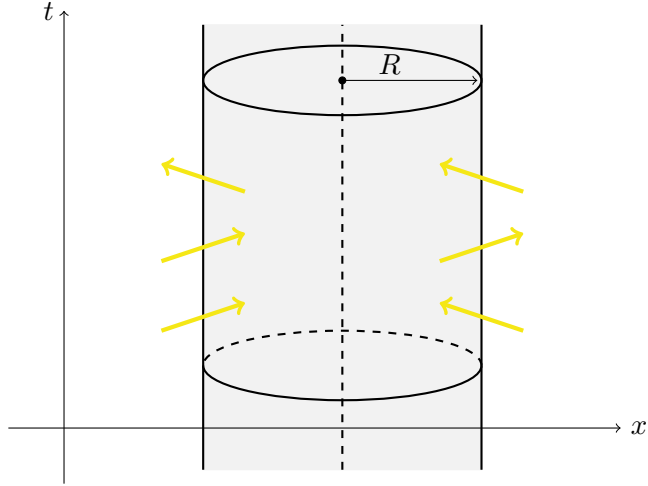
\begin{figure}
 \centering
 \begin{tikzpicture}[scale=1.85]
\filldraw[color=gray!10] (1,-0.3) rectangle (3,2.9);
\draw[->] (-0.4,0) -- (4,0) node [right] {$x$};
\draw[->] (0,-0.4) -- (0,3) node [left] {$t$};
\draw[thick] (1,-0.3) -- (1,2.9);
\draw[dashed,thick] (2,-0.3) -- (2,2.9);
\draw[thick] (3,-0.3) -- (3,2.9);
\draw[thick] (2,2.5) ellipse (1 cm and 0.25 cm);
\draw[fill=black] (2,2.5) circle (0.025 cm);
\draw[->] (2,2.5) -- (2.97,2.5);
\node at (2.34,2.62) {$R$};
\draw[thick,rotate=0] (3,0.45) arc (0:-180:1 and 0.25);
\draw[dashed,thick,rotate=0] (3,0.45) arc (0:180:1 and 0.25);
\draw[ultra thick,color=yellow!95!black,->] (0.7,0.7) -- (1.3,0.9);
\draw[ultra thick,color=yellow!95!black,->] (0.7,1.2) -- (1.3,1.4);
\draw[ultra thick,color=yellow!95!black,->] (1.3,1.7) -- (0.7,1.9);
\draw[ultra thick,color=yellow!95!black,->] (3.3,0.7) -- (2.7,0.9);
\draw[ultra thick,color=yellow!95!black,->] (3.3,1.7) -- (2.7,1.9);
\draw[ultra thick,color=yellow!95!black,->] (2.7,1.2) -- (3.3,1.4);
\end{tikzpicture}
 \caption{Spacetime illustration of scattering in a region given by a hypercylinder, i.e., a ball of ra\-dius~$R$, extended over all of time. Incoming and outgoing particles (in yellow) cross the boundary of~the hypercylinder, a sphere extended over all of time.}
 \label{fig:hcscatter}
\end{figure}

\subsection{GBQFT and quantization}
\label{sec:intro_gbqftquant}

While the axioms of GBQFT concern purely the quantum theory, the path integral~(\ref{eq:tamplpi}) is foremost meant as a \emph{quantization} prescription. That is, it is meant to be used to construct a~quantum theory starting from a classical field theory. The formal adaptation of the quantization formula~(\ref{eq:tamplpi}) from a time-interval to a general region $M$ in GBQFT is straightforward,
\begin{equation*}
 \rho_M(\psi)=\int_{K_M} \xD\phi\, \psi(\phi|_{\partial M}) \exp(\im S(\phi)).
\end{equation*}
Here, the integral is over field configurations $\phi\in K_M$ in the \emph{region} $M$, $S$ is the \emph{action} in $M$, $\psi(\phi|_{\partial M})$ is the Schrödinger \emph{wave function} of the state $\psi\in\cH_{\partial M}$ evaluated on the restriction of the field configuration $\phi$ to the \emph{boundary} of $M$. While a formula like this has heuristic value, it~does not constitute a rigorous quantization scheme. Such a rigorous scheme was developed for~linear bosonic field theory in~\cite{Oe:holomorphic} as we briefly lay out in the following.

To work in a manner independent of any particular choice of bosonic field theory or metric background, it is convenient to capture the relevant data of the classical theory in terms of an axiomatic system~\cite{Oe:holomorphic}, see Appendix~\ref{sec:caxioms} (for a slightly generalized version). The data includes spaces $L_M$ of \emph{solutions of the equations of motion} in spacetime \emph{regions} $M$ and spaces $L_{\Sigma}$ of~\emph{germs} of such solutions on \emph{hypersurfaces} $\Sigma$. The latter carry in addition a \emph{symplectic form} $\omega_{\Sigma}\colon L_{\Sigma}\times L_{\Sigma}\to\R$ arising from a second variation of the Lagrangian. A symplectic form is a non-degenerate anti-symmetric bilinear map. The spaces $L_M$ and $L_{\Sigma}$ are real vector spaces since we work with linear field theory. These space are not sensitive to the orientation of the underlying region or hypersurface. However, the symplectic structure is, and changes sign under orientation reversal, i.e., $\omega_{\overline{\Sigma}}=-\omega_{\Sigma}$.
For a spacetime region $M$, the restriction of solutions to germs on the boundary $\partial M$ gives rise to a map $L_M\to L_{\partial M}$. We denote the image of this map by $L_{\tilde{M}}\subseteq L_{\partial M}$. We also assume this map to be surjective as well as injective.\footnote{The injectivity condition can be relaxed at the cost of restricting allowed observables (see later sections). This does not affect boundary observables. We shall not further elaborate on this possibility.} Consequently, we frequently allow ourselves to not notationally distinguish between an element $\phi\in L_M$ and its image $\phi\in L_{\tilde{M}}$ under this map. A central property of (well-behaved) Lagrangian field theory is that $L_{\tilde{M}}$ is a~\emph{Lagrangian subspace} of $L_{\partial M}$. This means that $L_{\tilde{M}}$ is \emph{isotropic}, i.e.,
\begin{gather*}
 \omega_{\partial M}(\phi,\eta) =0,\qquad\forall \phi,\eta\in L_{\tilde{M}},
 \end{gather*}
as well as \emph{coisotropic},
\begin{gather*}
\omega_{\partial M}(\phi,\eta) =0,\qquad\forall\phi\in L_{\tilde{M}} \Rightarrow \eta\in L_{\tilde{M}}.
\end{gather*}

Quantization of field theory in curved spacetime involves the selection of a set of complex ``positive frequency'' modes~\cite{BiDa:qftcurved}. Viewed on a spacelike hypersurface $\Sigma$, the germs of these modes form a complex subspace $L_{\Sigma}^+$ of the space $L_{\Sigma}^{\bC}=L_{\Sigma}\oplus \im L_{\Sigma}$ of \emph{complexified} germs of solutions on~$\Sigma$. What is more, $L_{\Sigma}^+\subseteq L_{\Sigma}^{\bC}$ is a \emph{positive-definite Lagrangian subspace}. That is, in addition to being a Lagrangian subspace, the \emph{inner product}, given for $\phi,\eta\in L_{\Sigma}^{\bC}$ by
\begin{equation}
 (\phi,\eta)_{\Sigma}\defeq 4\im\omega_{\Sigma}\big(\overline{\phi},\eta\big),
 \label{eq:stdipc}
\end{equation}
is \emph{positive-definite} on $L_{\Sigma}^+$. This makes $L_{\Sigma}^+$ (possibly after completion) into a Hilbert space. The \emph{Fock space} $\cH_{\Sigma}$ over $L_{\Sigma}^+$ is then the Hilbert space of states on $\Sigma$ of the quantum field theory. The choice of the subspace $L_{\Sigma}^+\subseteq L_{\Sigma}^{\bC}$ has the physical interpretation of a choice of \emph{vacuum}.
In this work we will frequently use the term \emph{polarization} in order to refer to the choice of a~Lagrangian subspace. More specifically, we refer to a polarization corresponding to a positive-definite Lagrangian subspace as a \emph{K\"ahler polarization}. Note that a Lagrangian subspace of a~real symplectic vector space, upon complexification, leads to a~complex Lagrangian subspace of a~complex symplectic vector space. We refer to this as a \emph{real polarization}. The main example is that of solution spaces $L_{\tilde{M}}\subseteq L_{\partial M}$ in regions $M$ and their complexification, $L_{\tilde{M}}^\bC\subseteq L_{\partial M}^\bC$. The terminology of polarizations is inspired from geometric quantization, where it generalizes to non-linear theories~\cite{Woo:geomquant}.

In~\cite{Oe:holomorphic} a \emph{rigorous} and (essentially) \emph{functorial} quantization scheme was presented that outputs a GBQFT given a linear classical field theory as input. The input is in axiomatic form (see Appendix~\ref{sec:caxioms}) in the sense just laid out. In addition, the input includes a consistent choice of K\"ahler polarization for each hypersurface. As expected, when restricting to the context of time-intervals in globally hyperbolic spacetimes, the resulting quantization reproduces the well known textbook results.

\subsection{Perturbation theory and LSZ reduction in standard QFT}
\label{sec:intro_pertlsz}

While one might ultimately aspire at a non-perturbative construction of non-linear QFTs based on GBQFT, a more modest goal is to extend the powerful and empirically successful methods of standard QFT to GBQFT. We start by recalling how non-linear theories are handled by \emph{perturbation theory}, going back to the path integral setting of Section~\ref{sec:intro_smatrix}. Consider an \emph{observable} $F\colon K_{[t_1,t_2]}\to\bC$, i.e., a function on the field configuration space $K_{[t_1,t_2]}$. Inserting this into the path integral~(\ref{eq:tamplpi}) we define the \emph{(time-ordered) correlation function},
\begin{equation}
 \langle \psi_2 | \tord F |\psi_1\rangle =
 \int_{K_{[t_1,t_2]}} \xD\phi\, \psi_1(\phi|_{t_1}) \overline{\psi_2(\phi|_{t_2})}
 F(\phi)
 \exp\big(\im S_{[t_1,t_2]}(\phi)\big).
 \label{eq:taobspi}
\end{equation}
In contrast to the notation~(\ref{eq:tamplpi}), here the time-evolution is implicit, suggesting a Heisenberg picture. This notation is similar to textbook notation, where $\tord$ stands for ``time ordering''. (The ``ordering'' aspect makes more sense when $F$ is a product of field evaluations and the left-hand side is viewed as if it was a matrix element of a product of corresponding operators.) An~important special case arises if the observable is a product of the field evaluated at spacetime points $x_1,\ldots,x_n$,\footnote{This explains better the terminology ``correlation function''. Note also that in textbook QFT this terminology is often reserved for vacuum correlation functions, see below.}
\begin{equation}
 F(\phi)=\phi(x_1)\cdots \phi(x_n).
 \label{eq:npointobs}
\end{equation}
Another important case is a \emph{Weyl observable}. Thus, let $D\colon K_{[t_1,t_2]}\to\R$ be a \emph{linear} observable. Then $F(\phi)=\exp\left(\im\, D(\phi)\right)$ is the corresponding Weyl observable. In particular, we might let $D$ be defined by integration of the field in spacetime with a weight function $j\colon [t_1,t_2]\times\R^3\to\R$,
\begin{equation}
 D_j(\phi)=\int_{[t_1,t_2]\times\R^3} \xd^4x\, \phi(x) j(x).
 \label{eq:srcobs}
\end{equation}
The corresponding Weyl observable $F_j(\phi)=\exp\left(\im\, D_j(\phi)\right)$ is said to encode a \emph{source} determined by $j$.
This terminology originates in the free theory as follows. The equations of motions are then homogeneous partial differential equations of the form $\dop \phi=0$, where $\dop$ is the corresponding differential operator. Replacing the action $S$ with $S+D_j$ leads to modified equations of motions which are precisely the inhomogeneous equations $\dop\phi=j$.

To do perturbation theory, recall that the action $S=S_0+S_{\rm int}$ is split into a free part~$S_0$ and an~interacting part $S_{\rm int}$. We can now replace the action $S$ in the path integral~(\ref{eq:taobspi}) by the free action~$S_0$ and encode the interacting part $S_{\rm int}$ instead through the observable $F_{\rm int}(\phi)=\exp\left(\im S_{\rm int}(\phi)\right)$. The term $S_{\rm int}$ depends on one (or several) \emph{coupling constant$($s$)$}, in such a way that the observable $F_{\rm int}$ may be expanded as a power series in the coupling constant(s). The~terms of this expansion are expressible as polynomial observables which are accessible to explicit evaluation. It~is convenient, however, to proceed via Weyl observables encoding sources. We~recall how this works with the simple example of a self-interacting scalar field. Thus, we take the interaction term to be given by
\begin{equation}
 S_{\rm int}(\phi)=\lambda \int\xd^4 x\, V(\phi(x)),
 \label{eq:sintv}
\end{equation}
where $V$ plays the role of a potential.
Then, the transition amplitude of the interacting theory may be written as
\begin{equation*}
 \sum_{n=0}^\infty \frac{1}{n!}\lambda^n \bigg(\im \int\xd^4 x\, V\bigg({-}\im\frac{\delta}{\delta j(x)}\bigg)\bigg)^n \big\langle\psi_2|\tord F_j|\psi_1\big\rangle \bigg|_{j=0}.
\end{equation*}
The terms of this expansion correspond to \emph{Feynman diagrams} with $V$ determining the vertices. Crucially, the correlation function $\langle\psi_2|\tord F_j|\psi_1\rangle$ can be explicitly evaluated. What is more, it remains unchanged when we take the limit $t_1\to-\infty$ and $t_2\to\infty$, as long as we evolve the states~$\psi_1$,~$\psi_2$ as prescribed by the free theory. Recalling the previous discussion of the $S$-matrix, this provides the missing step of extending the interaction over all intermediate times. With this, the full spacetime integration indicated in~(\ref{eq:sintv}) can be performed.
To really make perturbative QFT work one also needs to implement \emph{renormalization}, but this is beyond the scope of the present work, where we are only concerned with the basic structures.

A very important technique for making the $S$-matrix more accessible, and which has had a~profound impact on the development of QFT is \emph{LSZ reduction}~\cite{LSZ:reduction}. What this achieves is a~reformulation of the $S$-matrix in terms of
correlation functions of the type~(\ref{eq:taobspi}), but with initial and final states taken to be the vacuum. Consequently, these are called \emph{$($time-ordered$)$ vacuum correlation functions}. Sometimes they are also called \emph{vacuum expectation values}, even though in general they do not correspond to expectation values of any measurement. We briefly recall the LSZ reduction formula for the case of a real scalar field in Minkowski space. The free theory here is the Klein--Gordon theory. With the differential operator $\dop\defeq\square+m^2$ the equations of~motion are $\dop\eta=0$. We deal with momentum eigenstates labeled by 3-momenta $p\in\R^3$ and satisfying a normalization condition of the form
\begin{equation*}
 \langle p,q\rangle = (2\pi)^3 2 E \delta^3(p-q).
\end{equation*}
We are interested in the $S$-matrix element corresponding to incoming particles with momenta $q_1,\ldots,q_n$ and outgoing particles with momenta $p_1,\ldots,p_m$. The reduction formula then takes the form, compare formula (5-28) in \cite[p.~207]{ItZu:qft} (except for the renormalization constants),
\begin{gather}
 \langle p_1,\ldots,p_m| q_1,\ldots,q_n \rangle= \text{disconnected terms}\nonumber
 \\ \qquad
{} + \im^{n+m}\int \xd^4 x_1\cdots\xd^4 x_n \xd^4 y_1\cdots\xd^4 y_m
 \exp\bigg(\im \sum_{l=1}^m p_l\cdot y_l - \im \sum_{k=1}^n q_k\cdot x_k \bigg)\nonumber
 \\ \qquad\hphantom{+}
{} \times \dop_{x_1}\cdots \dop_{x_n} \dop_{y_1}\cdots \dop_{y_m}
 \langle 0 | \tord \phi(x_1)\cdots\phi(x_n)\phi(y_1)\cdots\phi(y_m)|0\rangle.
 \label{eq:stdlsz}
\end{gather}
Here the notation $\dop_x$ refers to the operator $\dop$ as a differential operator with respect to the $x$ coordinate.
The ``disconnected terms'' encode the contributions where some of the particles do not participate in the scattering process.
The ubiquity of LSZ reduction in QFT is such that it has generally been accepted that the physical content of a QFT in Minkowski space is completely determined by its $n$-point functions, i.e., by its vacuum correlation functions of observables of the type~(\ref{eq:npointobs}). This has even led to an axiomatization of QFT based on $n$-point functions (although the non-time ordered variant) in the form of the famous Wightman axioms~\cite{StWi:pct}, forming the basis of the program of \emph{constructive quantum field theory}.

\subsection{GBQFT with observables}
\label{sec:intro_gbqftobs}

Including sources in GBQFT has allowed constructing amplitudes for regions that are not time-intervals (and thus beyond the means of standard QFT) also in interacting quantum theory, via~perturbation theory. In particular, it was shown that the perturbative $S$-matrix is equivalent to~an~analogous amplitude with asymptotic free states at spatial rather than temporal infinity in~Minkowski space~\cite{CoOe:spsmatrix,CoOe:smatrixgbf}. Concretely, rather than having initial and final states at early and late times (compare Figure~\ref{fig:tampl}) we have a state with incoming and outgoing particles at large spatial distance from the center (compare Figure~\ref{fig:hcscatter}). Instead of taking initial and final times to infinity, the radius determining the distance is taken to infinity. An attractive feature of this setting is that \emph{crossing symmetry} becomes manifest and can thus be seen as an inherent prediction, rather than a property emerging from additional assumptions as in standard QFT.

Subsequently, more general observables were included in GBQFT~\cite{Oe:obsgbf}, motivated by the obvious generalization of the correlation function~(\ref{eq:taobspi}),
\begin{equation}
 \rho_M^F(\psi) =\int_{K_M} \xD\phi\, \psi(\phi|_{\partial M}) F(\phi) \exp(\im S(\phi)).
 \label{eq:obsampl}
\end{equation}
As before, $M$ is a region, $F\colon K_M\to\bC$ is the observable. This was developed into an extended axiomatic system for GBQFT (see Appendix~\ref{sec:qobsaxioms}) and a corresponding rigorous and functorial quantization scheme with observables~\cite{Oe:feynobs}, see Section~\ref{sec:kquant}. With this we can deal in principle with perturbative interacting quantum field theory.

However, while mathematically consistent and convincing, this framework still suffers from serious shortcomings. As mentioned in~Section~\ref{sec:intro_gbqftquant}, the construction of the Hilbert spaces of states on each hypersurface relies on a choice of K\"ahler polarization representing the vacuum. The problem with this is that the standard vacuum of QFT in Minkowski space on a non-spacelike hypersurface generically corresponds to a polarization that is not K\"ahler~\cite{CoOe:vaclag}.\footnote{This problem is already manifest in previous works such as~\cite{Oe:timelike}, where it is addressed by excluding what are there called ``unphysical modes''.} With this, the standard quantization prescription to obtain a Hilbert space of states breaks down. This limits the description of interesting physics on timelike hypersurfaces. More concretely, consider again the example of the hypercylinder $M$ in Minkowski space, see Figure~\ref{fig:hcscatter}. Then, in massive Klein--Gordon theory the space of modes (germs of solutions) on the hypercylinder boundary $\partial M$ splits into two components, $L_{\partial M}=L_{\partial M}^{\rm p}\oplus L_{\partial M}^{\rm e}$. The first, $L_{\partial M}^{\rm p}$ consists of the \emph{propagating} oscillatory solutions. The second component, $L_{\partial M}^{\rm e}$ is formed by the \emph{evanescent} solutions that show an exponential behavior in the radial direction. The standard vacuum yields a K\"ahler polarization on the propagating modes in $L_{\partial M}^{\rm p}$ and there is no problem in constructing the corresponding Hilbert space. However, on the evanescent modes in $L_{\partial M}^{\rm e}$ the polarization is real~\cite{CoOe:vaclag}. In the mentioned work~\cite{CoOe:spsmatrix,CoOe:smatrixgbf} on the $S$-matrix the problem of not being able to construct a Hilbert space for the evanescent modes was noted, but did not affect the result, because the modes are absent asymptotically due to their exponential decay. However, one may very well be interested in situations were measurements take place at finite (or even small) distances. A~quantum theoretical description of evanescent modes then becomes a necessity. So far, this has been beyond the reach of the methods of quantum field theory.

In order to use the GBQFT framework for a truly local description of QFT we absolutely need to be able to decompose spacetime into regions that are ``small'' and certainly compact (recall Figure~\ref{fig:st-decomp}). But a compact region in Minkowski space has a boundary on which the polarization corresponding to the standard vacuum is generically not of K\"ahler type~\cite{CoOe:vaclag}. If we want to describe physical processes in such a region with realistic boundary conditions, we need to be able to deal with non-K\"ahler polarizations.

A related problem arises even for spacelike hypersurfaces when we want to decompose them. Say we want to cut an equal-time hypersurface $\Sigma$ into two pieces $\Sigma_1$ and $\Sigma_2$ along a coordinate axis. The space of germs of solutions and the symplectic structure nicely decompose into a~direct sum, $L_{\Sigma}=L_{\Sigma_1}\oplus L_{\Sigma_2}$ and $\omega=\omega_1+\omega_2$ (as before we assume for simplicity the absence of~gauge symmetries). However, the polarization $L_{\Sigma}^+\subseteq L_{\Sigma}^\bC$ corresponding to the standard vacuum does not. That is, there are no Lagrangian subspaces $L_{\Sigma_1}^+\subseteq L_{\Sigma_1}^\bC$ and $L_{\Sigma_2}^+\subseteq L_{\Sigma_2}^\bC$ such that $L_{\Sigma}^+=L_{\Sigma_1}^+\oplus L_{\Sigma_2}^+$. The reason is that $L_{\Sigma}^+$, being related to global properties of the solution space, is non-local on the hypersurface $\Sigma$, see also remarks at the end of Section~\ref{sec:stdstate}. In~a~dif\-ferent guise this takes the form of the Reeh--Schlieder theorem~\cite{ReSc:unitequiv}. In the fermionic case these problems can be partially solved by going to a mixed state formalism and at the same time selectively dropping the polarization information~\cite{Oe:locqft}. In the bosonic case one could use ``auxiliary'' K\"ahler polarizations at the price of a direct physical interpretation of the respective state spaces. We~return to this example in~Section~\ref{sec:splitmink}.

\subsection{The present work}

With the present work we address the problem of quantization for non-K\"ahler polarizations as well as the problem of dealing with non-decomposable vacua. We do so by learning a lesson from \emph{algebraic quantum field theory $($AQFT$)$}~\cite{Haa:lqp}, which is probably the furthest developed axiomatic approach to QFT to date, based on the axioms of Haag and Kastler~\cite{HaKa:aqft}. The lesson is that Hilbert spaces (of states) are always tied to a specific choice of vacuum, and if we want to work in a way independent of such choices we should consider observables rather than states as primary objects. LSZ reduction can also be seen as pointing to the feasibility of doing so. As~it turns out, we need observables localized on hypersurfaces to do that. These were introduced in~\cite{Oe:feynobs}, and we call them \emph{slice observables} here, see Section~\ref{sec:sliceobs}. In contrast to AQFT, ordinary observables do not form algebras in GBQFT, because the composition of observables cannot be separated from the composition of underlying spacetime regions. But slice observables do form algebras, because the underlying slice regions auto-compose. This turns out to provide a point of contact with AQFT.

We start in~Section~\ref{sec:cgenvac} with elementary considerations of the path integral and correlation functions, highlighting a simple but powerful formula for the path integral that underlies much of the subsequent work. We proceed to review and further develop the framework of GBQFT with observables and K\"ahler polarizations~\cite{Oe:feynobs} in~Section~\ref{sec:kquant}. The notion of slice observable is elaborated in~Section~\ref{sec:sliceobs}, first in a classical and then in a quantum setting. The role of slice observables in the K\"ahler quantization setting of Section~\ref{sec:kquant} is elucidated in~Section~\ref{sec:kobs}. Inparticular, we show how the Hilbert spaces of the K\"ahler quantization scheme are recovered via the application of the GNS construction to the algebra of slice observables, thus deepening the contact with AQFT.

We also establish a correspondence between slice observables and states on the boundary of spacetime regions. This correspondence emboldens us to do away completely with Hilbert space and focus instead on slice observables as primary objects. Consequently, we present in~Section~\ref{sec:genquant} a~quantization scheme analogous to that of Section~\ref{sec:kquant}, but based entirely on observables rather than states. Crucially, it is much more general in that it does away with the restriction of polarizations to be of K\"ahler type. We have emphasized the significance of this in the previous section.
Coherent states play a special role, here in the guise of Weyl slice observables. Composition can be accomplished in a surprising new way via a joint observable (Section~\ref{sec:compobs}), but also (as seen later) in the ``old'' way with a sum over a complete basis (Section~\ref{sec:corrak}). We also go considerably beyond~\cite{Oe:feynobs} in another direction, generalizing further tools of textbook QFT to our setting. Crucially this includes the LSZ reduction formula (Section~\ref{sec:genlsz}), but also ``normal ordering'' (Section~\ref{sec:nordsc}), Wick's theorem (Section~\ref{sec:genwick}), connected amplitudes (Section~\ref{sec:pstates}) etc.

\looseness=1
While our results show that working without Hilbert spaces is fine in many situations, sometimes it is useful to have a concrete Hilbert space of states.
In Section~\ref{sec:iprod} we present a new quantization scheme for constructing Hilbert spaces even for non-K\"ahler pola\-ri\-za\-tions. While for K\"ahler polarizations this just recovers the GNS construction of Section~\ref{sec:ssgns}, for other polarizations this is accomplished by introducing a real structure that manifests as a modified $*$-structure of the corresponding Weyl algebra. This construction is motivated in part from the \emph{reflection positivity} condition arising in Euclidean approach to constructive quantum field theory, as we explain in~Section~\ref{sec:refpos}.
Recall that a change of vacuum (as relevant in particular in curved spacetime QFT) gives rise to a (generally non-normalizable) state consisting of a~``sea'' of particle pairs representing one vacuum in the Hilbert space of the other vacuum~\cite{BiDa:qftcurved}. We show in~Sections~\ref{sec:wfvac} and \ref{sec:bogvac} that this phenomenon generalizes to the novel non-K\"ahler vacua.

With its focus on development of the conceptual and mathematical framework, applications are outside the scope of the present paper. Nevertheless, we make an exception with Section~\ref{sec:minkrind}, where we reexamine the question of the splitting of the Minkowski vacuum along partial hyper\-sur\-fa\-ces as well as the relation between Rindler and Minkowski vacuum in Rindler space. On~the one hand this serves to confirm that well-established standard results are recovered with our methods.~On the other hand a few new insights arise on the problems in question.

A key characteristic of our approach lies in the aim of reducing quantum (field) theory to its structural essence. Naturally, this results in a high degree of abstraction. On the one hand this means that making it work for any particular field theory requires additional effort in dealing with concrete partial differential equations, boundary value problems, topologies on solution spaces etc. On the other hand this means that the potential applicability of the framework is vast, including not only all kinds of scenarios involving curved spacetimes, but possibly even in~contexts with theories living on manifolds without metric backgrounds, as one would expect in~quantum gravity.

\section{Correlation functions in generalized vacua}
\label{sec:cgenvac}

\subsection{Path integral formula}
\label{sec:piwobs}

In this section we introduce a formula for the vacuum correlation functions of observables in GBQFT that is central for the subsequent considerations in this work. We first recall the standard case of QFT in Minkowski space. The correlation function~(\ref{eq:taobspi}) for the case of initial and final vacuum states is customarily written as
\begin{equation}
 \langle 0 | \tord F | 0\rangle =
 \int_{K} \xD\phi\, F(\phi) \exp\left(\im S(\phi)\right).
 \label{eq:stdvev}
\end{equation}
The integration is here formally over field configurations $\phi\in K$ in all of Minkowski space. This notation hides the fact that the path integral is evaluated with specific boundary conditions for the field configuration $\phi$ in the infinite past and future, encoding the vacuum. As mentioned in~Section~\ref{sec:intro_pertlsz}, the path integral can be explicitly evaluated if $F$ is a Weyl observable. Thus, let $D\colon K\to\R$ be a linear observable and set $F=\exp(\im D)$. Now consider $S+D$ as a modified action. Let $\eta$ be the solution of the equations of motion for the modified action satisfying the standard boundary conditions. Recall that this means that $\eta$ is a ``positive energy'' solution in the far future and a ``negative energy'' solution in the distant past. Then,
\begin{equation}
 \langle 0 | \tord F | 0\rangle =\exp\bigg(\frac{\im}{2}D(\eta)\bigg).
 \label{eq:stdpiwobs}
\end{equation}
In the special case that $D=D_j$ is determined by a source $j$ via~(\ref{eq:srcobs}) we have
\begin{equation}
 \langle 0 | \tord F | 0\rangle =\exp\bigg(\frac{\im}{2}\int\xd^4 x\, \eta(x) j(x)\bigg) =\exp\bigg(\frac{\im}{2}\int \xd^4 x\,\xd^4 y\, j(x) G_F(x,y) j(y) \bigg).
 \label{eq:stdwobsfeyn}
\end{equation}
Here, $G_F$ is the \emph{Feynman propagator}.

In GBQFT we consider the analogue of~(\ref{eq:stdvev}) for a general region $M$. We use a notation analogous to~(\ref{eq:obsampl}) to write the vacuum correlation function as
\begin{equation}
 \rho^F_M\big(W^\pol\big)=\int_{K_M} \xD\phi\, F(\phi) \exp(\im S(\phi)).
 \label{eq:amplgvac}
\end{equation}
While this notation deliberately suggests that $W^\pol$ is a state and that $\rho_M^F$ is a function on the corresponding state space we do neither define such a state space nor such a map for the moment.
Here, the boundary conditions are determined by a \emph{polarization} $\pol$, i.e., a \emph{Lagrangian subspace} $L_{\partial M}^\pol\subseteq L_{\partial M}^\bC$ of the complexified space of germs of solutions on the boundary $\partial M$~\cite{CoOe:vaclag}. We make the choice of polarization explicit with our notation. Thus, we take $W^\pol$ to mean the vacuum determined by the polarization $\pol$.
The standard QFT case is recovered by considering Minkowski space with the K\"ahler polarization of the standard vacuum (i.e., positive and negative energy solutions) at positive and negative temporal infinity. Here, in contrast, the polarization need not be K\"ahler and the term \emph{vacuum} is to be understood in the corresponding generalized sense~\cite{CoOe:vaclag}. As~a~technical condition, the polarization has to be \emph{transversal} to the real polarization $L_{\tilde{M}}^\bC\subseteq L_{\partial M}^\bC$ given by the complexified solutions on the boundary that come from interior solutions (recall Section~\ref{sec:intro_gbqftquant}). This means that the two subspaces satisfy $L_{\tilde{M}}^\bC\oplus L_{\partial M}^\pol=L_{\partial M}^\bC$. Transversality is guaranteed if $L_{\partial M}^\pol$ is a K\"ahler polarization \cite[Proposition~B.11]{CoOe:vaclag}.

Consider a real \emph{linear observable} $D$ in $M$, i.e., a linear map $D\colon K_M\to\R$. We may add such a linear observable to the action $S\colon K_M\to\R$ to obtain a modified action $S+D$. We denote the space of solutions in $M$ of the equations of motion for this modified action by $A_M^D$.
For later use we note that the evaluation of the observable $D$ on any $\phi\in L_M$ is given by a very simple formula. To this end let $\eta$ be any element of $A_M^D$. Then, \cite[equation~(52)]{Oe:feynobs}\footnote{Note that we adopt the sign conventions for the symplectic form relative to the action as in~\cite{Oe:feynobs} and not as in the more recent paper~\cite{CoOe:vaclag} which would invert the sign in this equation.}
\begin{equation}
 D(\phi)=2\omega_{\partial M}(\phi,\eta).
 \label{eq:obssol}
\end{equation}
$A_M^D$ is an affine space, admitting translations by elements of $L_M$. Typically, the solutions for $S$ are \emph{homogeneous} partial differential equations, while those for $S+D$ are associated \emph{inhomogeneous} ones.
Let $F=\exp(\im D)$ be the corresponding \emph{Weyl observable}.
Consider the complexification $A_M^{D,\bC}\defeq A_M^{D}\oplus\im L_M$ of the affine space $A_M^D$. The transversality condition ensures that there is exactly one complexified solution $\eta$ in the intersection $A_M^{D,\bC}\cap L_{\partial M}^\pol$.\footnote{To see this consider an arbitrary element $\xi\in A_M^{D,\bC}$. Using transversality on the boundary $\partial M$ decompose this uniquely into $\xi=\tau+\eta$ with $\tau\in L_{\tilde{M}}^{\bC}$ and $\eta\in L_{\partial M}^\pol$. Starting with a different element $\xi'=\tau'+\eta'$, it is easy to see that $\eta=\eta'$.} The path integral~(\ref{eq:amplgvac}) evaluates to~\cite{CoOe:vaclag}
\begin{equation}
 \rho_M^F\big(W^\pol\big)=\exp\bigg(\frac{\im}{2} D(\eta)\bigg),
 \label{eq:piwobs}
\end{equation}
directly generalizing formula~(\ref{eq:stdpiwobs}).
In fact, rather than further reasoning about the path integral, we shall treat formula~(\ref{eq:piwobs}) from here onwards as a \emph{definition} (of the path integral).

Note that while we have originally assumed the observable $D$ to be real in order to use it as a~source to modify the action, the resulting formula~(\ref{eq:piwobs}) extends perfectly well to complex~$D$. Of~course, in this case the space $A_M^D$ does not consist exclusively of real solutions. Nevertheless, we~will continue to use the notation $A_M^{D,\bC}$ for the ``complexification'' which can now be understood as $A_M^{D,\bC}= A_M^{D}+ L_M^\bC$. (Note that the sum is not direct.) Similarly, we recall that the (generalized) space of solutions $L_M$ in a region $M$ is not always naturally a real Lagrangian subspace of the space $L_{\partial M}$ of germs on the boundary, but may similarly contain complex solutions, in particular if $M$ is not compact~\cite{CoOe:vaclag}. We still write $L_{M}^\bC$, even if the space does not arise as the complexification of a real vector space.

Besides having complex values on real solutions, formula~(\ref{eq:piwobs}) also implies that we need to evaluate observables on complex solutions. In the case of a linear observable $D$ as above this is simply done by extending real to complex linearity. The obvious generalization to the non-linear case is achieved by demanding \emph{holomorphicity}. To be precise, we say that a complex function $f\colon V\to\bC$ on a complex vector space $V$ is \emph{holomorphic} iff for any two elements $a,b\in V$ the function $\bC\to\bC$ given by $z\mapsto f(a+z b)$ is everywhere holomorphic (i.e., entire). In particular, the Weyl observable induced by a complex linear observable is holomorphic. From here onwards it is understood that all observables are required to be holomorphic, if not explicitly stated otherwise.

While arguments made so far about solution spaces, observables and their relations were justified by appeal to differential analytic contexts involving partial differential equations in manifolds, such contexts are not actually necessary and may not even be desired for the results that are going to be discussed. Instead, key structures (such as solution spaces) are understood as objects in their own right and their relations, previously thought of as derived, are axiomatized. In particular, we shall take for granted the axioms for classical field theory (Appendix~\ref{sec:caxioms}) as well as those involving observables as well \cite[Section~4.6]{Oe:feynobs}. This will not necessarily be evident in our discourse which is aimed primarily at an intuitive understanding, but will be evident in relevant proofs when appeal is made to these axioms rather than to a differential analytic context.

\subsection{Vacuum correlation functions for general observables}
\label{sec:vcgenobs}

Crucially, evaluating the vacuum correlation function of an observable that is not a Weyl obser\-vable can also be reduced to formula~(\ref{eq:piwobs}), by using Weyl observables as \emph{generators}. Consider the case of an observable $D\colon K_M^\bC\to\bC$ that is a product of linear observables $D=D_1\cdots D_n$. Define the family of linear observables
\begin{equation*}
 D_{\lambda_1,\ldots,\lambda_n}=\lambda_1 D_1+\cdots +\lambda_n D_n,
\end{equation*}
parametrized by real numbers $\lambda_1,\ldots,\lambda_n$. We may then obtain $D$ from the family of Weyl observables
\begin{equation*}
 F_{\lambda_1,\ldots,\lambda_n}=\exp\left(\im D_{\lambda_1,\ldots,\lambda_n}\right),
\end{equation*}
by variation,
\begin{equation}
 D=(-\im)^n\frac{\partial}{\partial \lambda_1}\cdots\frac{\partial}{\partial \lambda_n} F_{\lambda_1,\ldots,\lambda_n} \bigg|_{\lambda_1,\ldots,\lambda_n=0}.
 \label{eq:polydweyl}
\end{equation}
Exploiting linearity of the amplitude we obtain the vacuum correlation function of $D$ from that of the family of Weyl observables,
\begin{equation}
 \rho_M^D\big(W^\pol\big) =(-\im)^n\frac{\partial}{\partial \lambda_1}\cdots\frac{\partial}{\partial \lambda_n} \rho_M^{F_{\lambda_1,\ldots,\lambda_n}}\big(W^\pol\big) \bigg|_{\lambda_1,\ldots,\lambda_n=0}.
 \label{eq:polydwva}
\end{equation}

We proceed to a more explicit evaluation of this correlation function depending on the degree of $D$. We start with the case of a quadratic observable $D= D_1 D_2$. With~(\ref{eq:polydwva}) we get in this case
\begin{align*}
 \rho_M^{D_1 D_2}\big(W^\pol\big)&=- \frac{\partial}{\partial \lambda_1}\frac{\partial}{\partial \lambda_2} \rho_M^{F_{\lambda_1,\ldots,\lambda_n}}\big(W^\pol\big) \bigg|_{\lambda_1,\dots,\lambda_n=0} \\
 &=- \frac{\partial}{\partial \lambda_1}\frac{\partial}{\partial \lambda_2}
 \exp\bigg(\frac{\im}{2} D_{\lambda_1,\lambda_2}(\eta_{D_{\lambda_1,\lambda_2}})\bigg) \bigg|_{\lambda_1,\dots,\lambda_n=0}.
\end{align*}
Here $\eta_{D_{\lambda_1,\lambda_2}}\in A_M^{D_{\lambda_1,\lambda_2},\bC}\cap L_{\partial M}^\pol$. Now notice that by linearity
\begin{equation*}
 \eta_{D_{\lambda_1,\lambda_2}} = \lambda_1 \eta_1 + \lambda_2 \eta_2,
\end{equation*}
where $\eta_k\in A_M^{D_k,\bC}\cap L_{\partial M}^\pol$. Thus,
\begin{gather}
 \rho_M^{D_1 D_2}\big(W^\pol\big) \nonumber
 \\ \qquad
 {}= - \frac{\partial}{\partial \lambda_1}\frac{\partial}{\partial \lambda_2}
 \exp\bigg(\frac{\im}{2} \left(\lambda_1\lambda_2 (D_1(\eta_2)+D_2(\eta_1)) + \lambda_1^2 D_1(\eta_1)+\lambda_2^2 D_2(\eta_2)\right)\bigg) \bigg|_{\lambda_1,\dots,\lambda_n=0} \nonumber
 \\\qquad
 {} = -\frac{\im}{2}\big(D_1(\eta_2)+D_2(\eta_1)\big).
 \label{eq:twopoint}
\end{gather}
In the example of Klein--Gordon theory in Minkowski space let $D_1$, $D_2$ be determined by sources $j_1$, $j_2$ according to~(\ref{eq:srcobs}). Using~(\ref{eq:stdwobsfeyn}) we get
\begin{equation}
 \rho_M^{D_1 D_2}\big(W^\pol\big) =-\im \int \xd^4 x\,\xd^4 y\, j_1(x) G_F(x,y) j_2(y).
 \label{eq:fpropsrc}
\end{equation}
Taking the sources to be delta functions so that $D_1(\phi)=\phi(x_1)$, $D_2(\phi)=\phi(x_2)$ we have (also exhibiting the textbook notation)
\begin{equation*}
 \langle 0 | \tord \phi(x_1) \phi(x_2) | 0\rangle = \rho_M^{D_1 D_2}\big(W^\pol\big) =-\im G_F(x_1,x_2).
\end{equation*}

We return to the general setting and consider the product $D=D_1\cdots D_n$,
\begin{align*}
 \rho_M^{D_1 \cdots D_n}\big(W^\pol\big)
 & = \bigg(\prod_{k=1}^n \bigg({-}\im \frac{\partial}{\partial \lambda_k}\bigg)\bigg) \rho_M^F\big(W^\pol\big) \bigg|_{\lambda_1,\dots,\lambda_n=0} \nonumber
 \\
 & = \bigg(\prod_{k=1}^n \bigg({-}\im \frac{\partial}{\partial \lambda_k}\bigg)\bigg) \exp\bigg(\frac{\im}{2} \sum_{k,l=1}^n \lambda_k \lambda_l D_k(\eta_l)\bigg) \bigg|_{\lambda_1,\dots,\lambda_n=0}.
\end{align*}
Since any term in the expansion of the exponential involves the product of an even number of parameters $\lambda_k$ this expression vanishes if $n$ is odd. If $n$ is even, set $n=2m$. This yields
\begin{align*}
 \rho_M^{D_1 \cdots D_{2m}}\big(W^\pol\big)
 & = \bigg(\prod_{k=1}^{2m} \bigg({-}\im \frac{\partial}{\partial \lambda_k}\bigg)\bigg) \frac{1}{m!} \bigg(\frac{\im}{2} \sum_{k,l=1}^{2m} \lambda_k \lambda_l D_k(\eta_l)\bigg)^{m} \bigg|_{\lambda_1,\dots,\lambda_n=0} \nonumber
 \\
 & = \frac{1}{m!} \sum_{\sigma\in S^{2m}} \prod_{j=1}^m \bigg({-}\frac{\im}{2} D_{\sigma(2j-1)}(\eta_{\sigma(2j)})\bigg) \nonumber
 \\
 & = \frac{1}{ 2^m\, m!} \sum_{\sigma\in S^{2m}} \prod_{j=1}^m\, \rho_M^{D_{\sigma(2j-1)} D_{\sigma(2j)}}\big(W^\pol\big).
\end{align*}
Here, $S^k$ denotes the group of permutations $\sigma$ of $k$ elements. We recover a well known formula expressing the correlation function as a sum over products of correlation functions corresponding to all possible partitions of the monomial observable into pairs of linear observables.

\section{Quantization with K\"ahler vacua}
\label{sec:kquant}

As recalled in the introduction (Section~\ref{sec:intro_gbqftobs}), the framework of linear GBQFT with observables has been fully worked out for the case that a K\"ahler polarization (Section~\ref{sec:intro_gbqftquant}) is chosen on each hypersurface. We give a short review of this framework here. The axioms of the quantum theory to be satisfied by the quantization scheme are included in Appendix~\ref{sec:qobsaxioms} to which we shall refer occasionally. For a full account we refer the reader to~\cite{Oe:feynobs}. Some terminology used is from~\cite{CoOe:vaclag}.

\subsection{State spaces and coherent states}
\label{sec:stdstate}

Let $\Sigma$ be a hypersurface, $L_{\Sigma}$ the associated space of germs of solutions and $\omega_{\Sigma}$ the corresponding symplectic form. We also assume a choice $L_{\Sigma}^+\subseteq L_{\Sigma}^\bC$ of Lagrangian subspace, positive-definite with respect to the inner product~(\ref{eq:stdipc}), i.e., a K\"ahler polarization.
The complex conjugate subspace $L_{\Sigma}^-\defeq \overline{L_{\Sigma}^+}$ is then also a Lagrangian subspace of $L_{\Sigma}^\bC$. Moreover, it is negative-definite with respect to the inner product~(\ref{eq:stdipc}) and it is \emph{transversal} to $L_{\Sigma}^+$, i.e., $L_{\Sigma}^\bC=L_{\Sigma}^+\oplus L_{\Sigma}^-$.
A~more traditional way to specify the subspace $L_{\Sigma}^+\subseteq L_{\Sigma}^\bC$ is by providing an orthonormal basis~$\{u_k\}_{k\in I}$ of it with respect to the inner product~(\ref{eq:stdipc}). This is also referred to as a choice of \emph{negative energy modes}. The complex conjugate elements $\{\bar{u}_k\}_{k\in I}$ then provide an orthonormal basis of $L_{\Sigma}^-\subseteq L_{\Sigma}^\bC$ with respect to the negative of the inner product~(\ref{eq:stdipc}). These are referred to as \emph{positive energy modes}. In total we have
\begin{equation}
 (u_k,u_l)_{\Sigma}=\delta_{k,l},\qquad
 (\bar{u}_k,\bar{u}_l)_{\Sigma}=-\delta_{k,l}, \qquad
 (u_k,\bar{u}_l)_{\Sigma}=0.
 \label{eq:modeip}
\end{equation}

It will be useful to recall an equivalent way to characterize a K\"ahler polarization. Let $J_{\Sigma}\colon L_{\Sigma}\to L_{\Sigma}$ be a \emph{complex structure}, i.e., a real linear operator such that $J_{\Sigma}^2\phi=-\phi$. Also require \emph{compatibility} with the symplectic form, i.e., $\omega_{\Sigma}\big(J_{\Sigma}\phi,J_{\Sigma}\eta\big)=\omega_{\Sigma}(\phi,\eta)$. Finally, consider the complex valued bilinear form
\begin{equation}
 \{\phi,\eta\}_{\Sigma}\defeq 2\omega_{\Sigma}\big(\phi,J_{\Sigma} \eta\big) +2\im\omega_{\Sigma}(\phi,\eta).
 \label{eq:stdjip}
\end{equation}
This is in fact \emph{hermitian} and \emph{sesquilinear} with respect to the complex structure $J_{\Sigma}$, i.e., by taking multiplication with $\im$ to be given by $J_{\Sigma}$. We require it to be \emph{positive-definite} as well. This makes $L_{\Sigma}$ into a complex Hilbert space (possibly upon completion). We call a complex structure~$J_{\Sigma}$ with these properties \emph{positive-definite complex structure}. For later use we also introduce a~notation for the real part of this inner product,
\begin{equation}
 g_{\Sigma}(\phi,\eta)\defeq 2\omega_{\Sigma}\big(\phi,J_{\Sigma} \eta\big).
 \label{eq:gip}
\end{equation}
There is a one-to-one correspondence between positive-definite Lagrangian subspaces of $L_{\Sigma}^\bC$ and positive-definite complex structures $J_{\Sigma}$ on $L_{\Sigma}$ as follows. First note that by complexifica\-tion~$J_{\Sigma}$ extends to a complex linear map $L_{\Sigma}^\bC\to L_{\Sigma}^\bC$ (that we also denote by $J_{\Sigma}$). Then, $L_{\Sigma}^+$ and $L_{\Sigma}^-$ are precisely the eigenspaces of the operator $J_{\Sigma}$, with eigenvalues $\im$ and $-\im$. We note that the projectors on the eigenspaces can be written as
\begin{equation}
 P_{\Sigma}^+(\phi)=\frac{1}{2} \big(\phi-\im J_{\Sigma}\phi\big),\qquad
 P_{\Sigma}^-(\phi)=\frac{1}{2} \big(\phi+\im J_{\Sigma}\phi\big).
 \label{eq:projkp}
\end{equation}
The relation between the inner product~(\ref{eq:stdipc}) on $L_{\Sigma}^+$ and the inner product~(\ref{eq:stdjip}) is given by
\begin{equation}
 \big(P_{\Sigma}^+(\phi),P_{\Sigma}^+(\eta)\big)_{\Sigma}=\{\phi,\eta\}_{\Sigma},
 \label{eq:relkip}
\end{equation}
where $\phi,\eta\in L_{\Sigma}$. We note that $P_{\Sigma}^\pm$ act as a real vector space isomorphisms between $L_{\Sigma}$ and~$L^\pm_{\Sigma}$. These isomorphisms also serve to bring into one-to-one correspondence orthonormal bases $\{w_k\}_{k\in I}$ of~$L_{\Sigma}$ with respect to the inner product~(\ref{eq:stdjip}) with orthonormal bases $\{u_k\}_{k\in I}$ of~$L^+_{\Sigma}$ with respect to the inner product~(\ref{eq:stdipc}) via
\begin{equation}
 w_k=u_k+\bar{u}_k.
 \label{eq:kbaserel}
\end{equation}

The Hilbert space $\cH_{\Sigma}$ of states can be constructed as the \emph{Fock space} over $L_{\Sigma}$ considered as a~Hilbert space with the inner product~(\ref{eq:stdjip}). Equivalently, we may view $\cH_{\Sigma}$ as a space of square-integrable holomorphic functions on $L_{\Sigma}$ with respect to a Gaussian measure $\nu_{\Sigma}$ determined by the inner product~(\ref{eq:stdjip})~\cite{Oe:holomorphic}.\footnote{Strictly speaking the measure $\nu_{\Sigma}$ and corresponding integrals are not over the space $L_{\Sigma}$, but over an extension~$\hat{L}_{\Sigma}$ of this space~\cite{Oe:holomorphic}. However, this detail is inconsequential for the present considerations. While we make it apparent in the notation for integrals, we will not provide any further discussion here.}
 Here, \emph{holomorphic} is understood with respect to the complex structure $J_{\Sigma}$. This is called the \emph{holomorphic representation}.
We denote the inner product in~$\cH_{\Sigma}$ by~$\langle\cdot,\cdot\rangle_{\Sigma}$.
Creation and annihilation operators are labeled by elements of the phase space~$L_{\Sigma}$, which at the same time can be identified with the subspace of $\cH_{\Sigma}$ of one-particle states. Given $\xi\in L_{\Sigma}$ we denote the associated creation operator by $a^\dagger_\xi$ and the associated annihilation operator by $a_\xi$. These satisfy the commutation relations
\begin{equation}
 \big[a_\xi,a^\dagger_\eta\big]=\{\eta,\xi\}_{\Sigma}.
 \label{eq:comrelkip}
\end{equation}
Their actions on holomorphic wave functions are given by~\cite{Oe:obsgbf}
\begin{gather}
 \big(a^\dagger_\xi \psi\big)(\phi) =\frac{1}{\sqrt{2}} \,\{\xi,\phi\}_\Sigma \psi(\phi), \label{eq:coact}
 \\
 (a_\xi \psi)(\phi) = \sqrt{2}\,\frac{\partial}{\partial \lambda}
 \psi(\phi+\lambda\xi) \bigg|_{\lambda=0}.
\end{gather}

A particularly important class of states are the \emph{coherent states}, which generate a dense subspace of $\cH_{\Sigma}$. Thus, associated to each element of $\xi\in L_{\Sigma}$, there is a coherent state $\coh_{\xi}\in \cH_{\Sigma}$. In the holomorphic representation its wave function is given by
\begin{equation}
 \coh_\xi(\phi)=\exp\bigg(\frac{1}{2}\{\xi,\phi\}_{\Sigma}\bigg).
 \label{eq:cohwf}
\end{equation}
The coherent states are eigenstates of the annihilation operators,
\begin{equation}
 a_\eta \coh_\xi = \frac{1}{\sqrt{2}}\,\{\xi,\eta\}_{\Sigma} \coh_\xi.
 \label{eq:anncoh}
\end{equation}
The wave function of a state $\psi\in \cH_{\Sigma}$ may be obtained via the \emph{reproducing property},
\begin{equation}
 \psi(\phi)=\langle \coh_\phi,\psi\rangle_{\Sigma}.
 \label{eq:reprod}
\end{equation}
Consequently, the inner product between coherent states is given by
\begin{equation}
 \langle \coh_\xi,\coh_\phi\rangle_{\Sigma}=\exp\bigg(\frac{1}{2} \{\phi,\xi\}_{\Sigma}\bigg).
 \label{eq:ipcohk}
\end{equation}
Moreover, they satisfy a \emph{completeness relation},
\begin{equation}
 \langle\eta,\psi\rangle_{\Sigma} = \int_{\hat{L}_{\Sigma}} \langle \eta, \coh_{\phi}\rangle_{\Sigma} \langle \coh_{\phi}, \psi\rangle_{\Sigma}\,\xd\nu_{\Sigma}(\phi).
 \label{eq:cohcompl}
\end{equation}
For later use we note that expanding coherent states to first order in their parameters in this formula leads to the following completeness relation for $L_{\Sigma}$,
\begin{equation}
 \{\eta,\xi\}_{\Sigma}=\frac12\int_{\hat{L}_{\Sigma}} \{\eta,\phi\}_{\Sigma} \{\phi,\xi\}_{\Sigma}\,
 \xd\nu_{\Sigma}(\phi).
 \label{eq:lcomplint}
\end{equation}
In the following we also consider the \emph{normalized coherent states} defined as
\[
\ncoh_\xi\defeq\exp\bigg({-}\frac{1}{4}\{\xi,\xi\}_{\Sigma}\bigg) \coh_\xi.
\]
The \emph{vacuum state} is the coherent state $\ncoh_0=\coh_0$. To emphasize that a coherent state lives on a~hypersurface $\Sigma$ we sometimes use the more explicit notation $\coh_{\Sigma,\xi}$ instead of $\coh_{\xi}$.

Concerning states with definite particle number we remark that an $n$-particle state $\psi$ is represented by an $n$-linear wave function, i.e., $\psi(\lambda\phi)=\lambda^n\psi(\phi)$. More specifically, a state $\psi$ encoding $n$ particles characterized by elements $\xi_1,\ldots,\xi_n\in L_{\Sigma}$ has (up to normalization) a wave function (see, e.g.,~\cite{Oe:freefermi})
\begin{equation}
 \psi(\phi)=\{\xi_1,\phi\}_{\Sigma}\cdots\{\xi_n,\phi\}_{\Sigma}.
 \label{eq:mswf}
\end{equation}
Note that this wave function may be obtained from that of a coherent state by applying suitable derivatives as follows
\begin{equation}
 \psi=\bigg(\prod_{k=1}^n 2 \frac{\partial}{\partial \lambda_k}\bigg) \coh_{\lambda_1\xi_1+\cdots+\lambda_n\xi_n} \bigg|_{\lambda_1,\dots,\lambda_n=0}.
 \label{eq:msderivcohs}
\end{equation}

In the presented quantization, the K\"ahler vacua on the two differently oriented versions of a hypersurface $\Sigma$ are required to be related as follows. Concretely, taking the positive-definite Lagrangian subspace $L_{\Sigma}^+\subseteq L_{\Sigma}^\bC$ associated to $\Sigma$, its complex conjugate $\overline{L_{\Sigma}^+}$ will be a positive-definite Lagrangian subspace of $L_{\overline{\Sigma}}^\bC$. This is because the notion of Lagrangian subspace is the same on $\Sigma$ and $\overline{\Sigma}$ while the relevant inner product~(\ref{eq:stdipc}) behaves as
\begin{equation*}
 \big(\overline{\phi},\overline{\eta}\big)_{\overline{\Sigma}}=\overline{(\phi,\eta)_{\Sigma}}.
\end{equation*}
Both stem from the fact that $\omega_{\overline{\Sigma}}=-\omega_{\Sigma}$ (see Axiom~(C2) of Appendix~\ref{sec:caxioms}). It is thus natural to set $L_{\overline{\Sigma}}^+\defeq\overline{L_{\Sigma}^+}$. This also agrees with the standard notion of vacuum in curved spacetime~\cite{CoOe:vaclag}. In~terms of complex structures this means, $J_{\overline{\Sigma}}=-J_{\Sigma}$. As a consequence, the state spaces on the same hypersurface, but with opposite orientation are related by a specific complex conjugate linear involution $\iota_{\Sigma}\colon \cH_{\Sigma}\to\cH_{\overline{\Sigma}}$. In terms of wave functions this takes the form
\begin{equation*}
 \left(\iota_{\Sigma}(\psi)\right)(\phi)=\overline{\psi(\phi)}.
\end{equation*}
For coherent states this map is simply $\iota_{\Sigma}(\coh_{\Sigma,\xi})=\coh_{\overline{\Sigma},\xi}$. With the structures of the Hilbert spaces $\cH_{\Sigma}$ and involutions $\iota_{\Sigma}$ we satisfy Axioms~(T1) and (T1b) of Appendix~\ref{sec:qobsaxioms}.

Finally, suppose that a hypersurface $\Sigma$ is decomposed into two pieces, $\Sigma_1$ and $\Sigma_2$, i.e., $\Sigma=\Sigma_1\cup\Sigma_2$, either disjointly or along edges. The associated space of germs of solutions then decomposes as a direct sum $L_{\Sigma}=L_{\Sigma_1}\oplus L_{\Sigma_2}$ and the symplectic form decomposes accordingly as $\omega_{\Sigma}=\omega_{\Sigma_1}+\omega_{\Sigma_2}$ (Axiom~(C3) of Appendix~\ref{sec:caxioms}). If this decomposition is to be admissible in the present quantization scheme, we require matching K\"ahler polarizations on each of the hypersurfaces. That is, $L_{\Sigma}^+=L_{\Sigma_1}^+\oplus L_{\Sigma_2}^+$. Then quantization yields an isometric isomorphism of Hilbert spaces $\tau_{\Sigma_1,\Sigma_2;\Sigma}\colon \cH_{\Sigma_1}\tens\cH_{\Sigma_2}\to\cH_{\Sigma}$ as required by Axiom~(T2) and satisfying Axiom~(T2b) of Appendix~\ref{sec:qobsaxioms}. In terms of wave functions we have
\begin{equation*}
 \left(\tau_{\Sigma_1,\Sigma_2;\Sigma}(\psi_1\tens\psi_2)\right)((\phi_1,\phi_2))
 =\psi_1(\phi_1)\psi_2(\phi_2).
\end{equation*}
For coherent states this is\footnote{We use the notation $(\xi_1,\xi_2)$ for the element $\xi_1+\xi_2 \in L_{\Sigma_1}\oplus L_{\Sigma_2}$ with $\xi_1\in L_{\Sigma_1}$ and $\xi_2\in L_{\Sigma_2}$.}
\begin{equation*}
 \tau_{\Sigma_1,\Sigma_2;\Sigma}\big(\coh_{\xi_1}\tens \coh_{\xi_2}\big)=\coh_{(\xi_1,\xi_2)}.
\end{equation*}

In a standard QFT in globally hyperbolic spacetime if $\Sigma$ consists of the disjoint union of spacelike hypersurfaces $\Sigma_1$ and $\Sigma_2$ this notion of decomposition works perfectly well with the usual K\"ahler vacua. However, if in the same context, say $\Sigma$ is a single spacelike hypersurface and~$\Sigma_1$ and $\Sigma_2$ are pieces of it (glued along a boundary) then a serious problem arises, as already mentioned at the end of Section~\ref{sec:intro_gbqftobs} of the introduction. Namely, the Lagrangian subspace $L_{\Sigma}^+\subseteq L_{\Sigma}^\bC$ encoding a reasonable vacuum (e.g., the standard one in Minkowski space) will be non-local along the hypersurface. That is, there are no subspaces $L_{\Sigma_1}^+\subseteq L_{\Sigma_1}^\bC$ and $L_{\Sigma_2}^+\subseteq L_{\Sigma_2}$ such that $L_{\Sigma}^+=L_{\Sigma_1}^+\oplus L_{\Sigma_2}^+$. In terms of the complex structure, the operator $J_{\Sigma}$ is not a differential operator on $L_{\Sigma}$, but only a pseudo-differential operator. We shall see in~Section~\ref{sec:splitmink} how some of the methods developed in the following allow us to deal with this situation.

\subsection{Standard example: states in Klein--Gordon theory}
\label{sec:kgstates}

To connect our notions and notations to those familiar from textbook QFT, we consider the example of massive Klein--Gordon theory in Minkowski space on equal-time hypersurfaces. That is, we return to the context of the beginning of Section~\ref{sec:piwobs}. We shall adopt the conventions of~\cite{Oe:holomorphic}. Thus, we denote by $L_t$ the space of germs of solutions of the Klein--Gordon equation on the spacelike hypersurface $\Sigma_t$, determined by fixing the time $t$. Due to the Cauchy property, this space can be identified with the space of global solutions. We parametrize solutions $\phi$ as usual in terms of plane waves, using coefficient functions $\phi^{\rm a}$ and $\phi^{{\rm b}}$ on momentum space, with $E=\sqrt{k^2+m^2}$,
\begin{equation}
 \phi(t,x)=\int\frac{\xd^3 k}{(2\pi)^3 2E}
 \Big(\phi^{\rm a}(k) {\rm e}^{-\im(E t-k x)}+\overline{\phi^{\rm b}(k)} {\rm e}^{\im(E t-k x)}\Big).
 \label{eq:kgmodes}
\end{equation}
This describes elements of the complexified solution space $L_t^\bC$. Real solutions, i.e., elements of~$L_t$ are characterized by the property $\phi^{\rm b}(k)=\phi^{\rm a}(k)$. The symplectic form on $L_t^\bC$ is\footnote{The conventions for the sign of the symplectic structure are as in~\cite{Oe:holomorphic,Oe:feynobs} and opposite to those in~\cite{CoOe:vaclag}.}
\begin{align}
 \omega_t(\phi_1,\phi_2) & =\frac{1}{2}\int\xd^3 x\,
 \left(\phi_2(t,x) (\partial_0 \phi_1)(t,x) - \phi_1(t,x)(\partial_0\phi_2)(t,x)\right) \label{eq:sfsft} \\
 & =\frac{\im}{2}\int\frac{\xd^3 k}{(2\pi)^3 2E}
 \Big(\phi_2^{\rm a}(k)\overline{\phi_1^{\rm b}(k)}-\phi_1^{\rm a}(k)\overline{\phi_2^{\rm b}(k)}\Big).
 \label{eq:sfkgm}
\end{align}
Here the orientation of the hypersurface $\Sigma_t$ is chosen to correspond to that of the past boundary of a region in the future of $t$. This choice is in accordance with the fact that we want to construct the Hilbert space $\cH_t$ of \emph{initial} states on $\Sigma_t$. In the traditional bra-ket notation these are the ket-states.

It is easy to verify with~(\ref{eq:stdipc}) that the subspace $L_t^+\subseteq L_t^\bC$ of polarized solutions is indeed a~positive-definite Lagrangian subspace,
\begin{equation}
 L^+_t=\big\{\phi\in L_t^\bC \colon \phi^{\rm a}(k)=0\,\forall k\big\}.
 \label{eq:kgpvac}
\end{equation}
This is the subspace of \emph{negative energy} solutions, which recovers the standard past boundary condition for Klein--Gordon theory in Minkowski space, compare Section~\ref{sec:piwobs}. (For \emph{final} states on~the hypersurface $\overline{\Sigma}_t$ we correspondingly get the \emph{positive energy} solutions recovering the standard future boundary condition.) The corresponding complex structure $J_t$ and inner pro\-duct~(\ref{eq:stdjip})~are
\begin{gather*}
 (J_t(\phi))^{\rm a/b}(k) =-\im \phi^{\rm a/b}(k),
 \\
 \{\phi_1,\phi_2\}_t =2\int \frac{\xd^3 k}{(2\pi)^3 2E}\, \phi_1^{\rm a}(k) \overline{\phi_2^{{\rm b}}(k)}.
\end{gather*}

It is common to consider the (singular) momentum modes, with normalization,
\begin{gather*}
 \phi_p(t,x) =\frac{1}{\sqrt{2}} \big({\rm e}^{-\im (E_pt-p x)}+{\rm e}^{\im (E_pt-p x)}\big),
 \\
 \{\phi_p,\phi_{p'}\}_t =(2\pi)^3 2 E_p \delta^3(p-p').
\end{gather*}
In turn we use the momentum modes to define corresponding (singular) single particle momentum eigenstates, with normalization,
\begin{gather*}
 \psi_p(\phi) =\frac{1}{\sqrt{2}}\{\phi_p,\phi\}_t=\overline{\phi^{{\rm b}}(p)},
 \\
 \langle\psi_p,\psi_{p'}\rangle_t =(2\pi)^3 2 E_p \delta^3(p-p').
\end{gather*}
These are the usual textbook momentum eigenstates, often written as $|p\rangle$ instead of $\psi_p$. The corresponding creation and annihilation operators are
\begin{equation*}
 a^\dagger_p=a^\dagger_{\phi_p},\qquad
 a_p=a_{\phi_p}.
\end{equation*}
By multiple application of the creation operator~(\ref{eq:coact}) we see that the wave function for a~nor\-malized $n$-particle state with momenta $p_1,\ldots,p_n$ is
\begin{equation}
 \psi_{p_1,\ldots,p_n}(\phi)=\frac{1}{\sqrt{2^n}}\{\phi_{p_1},\phi\}_t \cdots \{\phi_{p_n},\phi\}_t.
 \label{eq:mms}
\end{equation}

\subsection{Correlation functions for coherent states}\label{sec:corrfack}

Let $M$ be a region with a choice of K\"ahler polarization on the boundary, i.e., a positive-definite Lagrangian subspace $L_{\partial M}^+\subseteq L_{\partial M}^{\bC}$. We recall that the space $L_{\partial M}$ decomposes as a direct sum of real vector spaces as follows, $L_{\partial M}=L_M \oplus J_{\partial M} L_M$ \cite[Lemma~4.1]{Oe:holomorphic}. Here, $J_{\partial M}$ denotes the complex structure corresponding to the K\"ahler polarization. For elements of $L_{\partial M}^\bC$, we~write this decomposition as $\xi=\xi^\rprt+J_{\partial M}\xi^\iprt$, with $\xi^\rprt,\xi^\iprt\in L_M$. Consider the correlation function for an~observable $F$ in~$M$ with a normalized coherent state $\ncoh_\xi$ for $\xi\in L_{\partial M}$ on the boundary~$\partial M$, compare Figure~\ref{fig:boundary-cohs-corr}. If $F$ is a Weyl observable the Feynman path integral~(\ref{eq:obsampl}) leads to the following factorization theorem, see~\cite[Proposition~4.2]{Oe:holomorphic} and~\cite[Proposition~4.3 and equation~(85)]{Oe:feynobs}.

\begin{figure}
 \centering
 \begin{tikzpicture}[scale=1.85]
\draw [red,thick,yshift=0.5cm,fill=white,scale=1.05] plot [smooth cycle] coordinates {(-1,-0.3) (0,-0.7) (1.6,0.22) (2,1.2)(1.1,2)(0,1.6)(-1.1,0.3)};
\draw [thick,yshift=0.5cm,fill=gray!10] plot [smooth cycle] coordinates {(-1,-0.3) (0,-0.7) (1.6,0.22) (2,1.2)(1.1,2)(0,1.6)(-1.1,0.3)};
\draw [thick,xshift=0cm,fill=gray!30] plot [smooth cycle] coordinates {(0,0.3) (1,0.5) (0.8,1.2) (0.2,1.4)  (-0.2,0.8) (-0.1,0.55)};
\node at (0.4,0.8) {$F$};
\node at (1.2,1.8) {$M$};
\draw (2.3,2.2) node [right] {$\partial M$} edge[out=180,in=30,->] (1.85,2);
\draw[red] (-1,2.7) node [right] {${\color{red}k_{\xi}}$} edge[out=0,in=130,->] (0.45,2.5);
\end{tikzpicture}
 \caption{Spacetime region $M$ containing an observable $F$ with normalized coherent state $\ncoh_\xi$ on the boundary.} \label{fig:boundary-cohs-corr}
\end{figure}
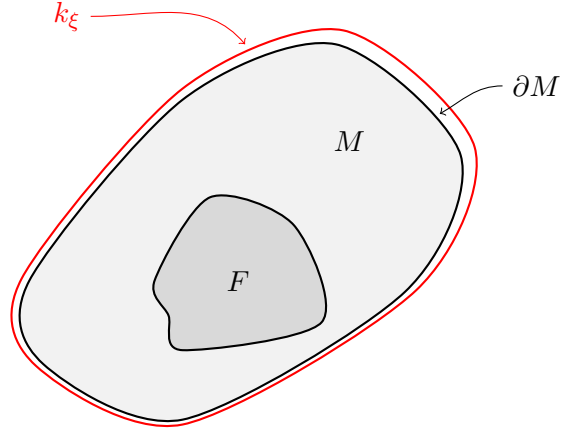

\begin{thm}
 \label{thm:stdcorrfact}
Let $M$ be a region, $D\colon K_M^\bC\to\bC$ a linear observable, and $\xi\in L_{\partial M}$. Set $F\defeq\exp(\im D)$, $\hat{\xi}\defeq\xi^{\rprt}-\im \xi^{\iprt}$. Then,
\begin{gather}
 \rho^F_M(\ncoh_\xi) =\rho_M(\ncoh_\xi) F\big(\hat{\xi}\big) \rho^F_M(\coh_0), \label{eq:wobsfact}
 \qquad\text{with}
 \\
 \rho_M(\ncoh_\xi) = \exp\bigg({-}\frac{\im}{2} g_{\partial M}\big(\xi^{\rprt},\xi^{\iprt}\big)
 -\frac{1}{2} g_{\partial M}\big(\xi^{\iprt},\xi^{\iprt}\big) \bigg), \label{eq:amplk}
 \\
 \rho_M^F(\coh_0) = \exp\bigg( \frac{\im}{2} D(\eta)\bigg). \label{eq:sfqvev}
\end{gather}
Here $\eta\in A_M^{D,\bC}\cap L^+_{\partial M}$ is unique.
\end{thm}
Viewing this as a definition, we satisfy Axiom~(TO4) of Appendix~\ref{sec:qobsaxioms}.

Note that in particular, the vacuum correlation function~(\ref{eq:sfqvev}) recovers formula~(\ref{eq:piwobs}) of~Sec\-tion~\ref{sec:piwobs}. In our notation we have
\begin{equation}
 \rho_M^F\big(W^\pol\big)=\rho_M^F(\coh_0). \label{eq:vaccorrnot}
\end{equation}
This may be seen as the justification for the definition~(\ref{eq:piwobs}) in the same way as the unconstrained path integral formula~(\ref{eq:stdvev}) in standard QFT may be justified by reproducing the results of formula~(\ref{eq:taobspi}) for the case of initial and final vacuum states. In the special case of Klein--Gordon theory in Minkowski space we thus indeed recover formulas~(\ref{eq:stdpiwobs}) and~(\ref{eq:stdwobsfeyn}) with the usual Feynman propagator as consequences of the quantization laid out in~Section~\ref{sec:kgstates} together with Theorem~\ref{thm:stdcorrfact}.

The notion of \emph{normal ordered quantization} was generalized from the instant-time setting of~standard QFT to GBQFT with K\"ahler polarizations in~\cite{Oe:obsgbf,Oe:feynobs}. In the present setting the normal ordered quantization of the observable $F$ is given by
\begin{equation}
 \rho_M^{\no{F}}(\coh_\xi)\defeq \rho_M(\coh_\xi) F\big(\hat{\xi}\big).
 \label{eq:defnoqk}
\end{equation}
We shall take this formula as a definition. In the case of a Weyl observable we can compare this formula directly with the factorization formula~(\ref{eq:wobsfact}). We observe that it then corresponds precisely to removing the factor representing the vacuum correlation function. For the understanding in terms of ordering of operators, see Section~\ref{sec:qsobsopk}.

Note that using the completeness relation~(\ref{eq:cohcompl}) and reproducing property~(\ref{eq:reprod}), correlation functions for arbitrary states $\psi\in\cH_{\partial M}$ can be expressed as integrals over correlation functions for coherent states as follows
\begin{equation}
 \rho_M^F(\psi) =\int_{\hat{L}_{\partial M}} \langle \coh_\xi,\psi\rangle_{\partial M}\,\rho_M^F(\coh_\xi)\,\xd\nu(\xi)
 =\int_{\hat{L}_{\partial M}} \psi(\xi)\,\rho_M^F(\coh_\xi)\,\xd\nu(\xi).
 \label{eq:kcorrgs}
\end{equation}

\subsection{Semiclassical interpretation}\label{sec:sck}

A remarkable aspect of the Structure Theorem~\ref{thm:stdcorrfact} for correlation functions is its bearing on the semiclassical interpretation of linear bosonic field theory. In particular, it can be used to show that a (normalized) coherent state $\ncoh_{\xi}$, for $\xi\in L_{\partial M}$ a solution of the classical field equation near the boundary $\partial M$, indeed behaves as a semiclassical approximation to that solution.

\subsubsection{Amplitudes}\label{sec:scka}

Firstly, we consider this in the context of pure amplitudes, i.e., in the absence of observables. The amplitude is then given by formula~(\ref{eq:amplk}). This admits an extremely simple and compelling semiclassical interpretation~\cite{Oe:holomorphic}. Crucial is the decomposition $L_{\partial M}=L_M\oplus J_{\partial M} L_M$ of the boundary solution space, in terms of elements written as, $\xi=\xi^{\rprt}+ J_{\partial M}\xi^{\iprt}$. Here, $\xi^{\rprt}$ is the component of the solution $\xi$ that comes from a solution in the interior of $M$. In other words, it is a solution that admits a \emph{classical continuation to the interior} of $M$ in terms of the classical equations of motion. In contrast, $J_{\partial M}\xi^{\iprt}$ is the component that does not admit such a continuation. That is, its occurrence on the boundary $\partial M$ is \emph{classically forbidden}. Now, if the classically forbidden component $J_{\partial M}\xi^{\iprt}$ is not present, i.e., $\xi$ is classically allowed, the amplitude~(\ref{eq:amplk}) is simply unity. On the other hand, if $J_{\partial M}\xi^{\iprt}$ is turned on, a phase factor appears, and more importantly, an exponential suppression term appears (recall that $g_{\partial M}$ is positive-definite), depending on the magnitude of $J_{\partial M}\xi^{\iprt}$ (or equivalently of $\xi^{\iprt}$). This is precisely the tunneling behavior expected in the quantum theory.

We remark that the inner product of coherent states can be seen to arise as a limiting case of the amplitude when the region $M$ is being squeezed to a slice region $\hat{\Sigma}$, compare Axiom~(T3x) of Appendix~\ref{sec:qobsaxioms}. In that case the classically allowed solutions $L_{\hat{\Sigma}}\subseteq L_{\partial \hat{\Sigma}}$ are the pairs $(\xi,\xi)$ with~$\xi\in L_{\Sigma}$. Then, the amplitude is the inner product of two identical normalized coherent states~$\ncoh_{\xi}$ and thus unity. On the other hand, switching on a classically forbidden component means making the two solutions $\xi_1$ and $\xi_2$ on the two sides of $\hat{\Sigma}$ different so that the inner product between the normalized coherent states becomes suppressed. Indeed, from expression~(\ref{eq:ipcohk}) we~obtain
\begin{equation*}
 \big\langle \ncoh_{\xi_1},\ncoh_{\xi_2}\big\rangle_{\Sigma}
 = \exp\bigg(\im\omega_{\Sigma}(\xi_2,\xi_1)-\frac{1}{4}g_{\Sigma}(\xi_2-\xi_1,\xi_2-\xi_1)\bigg).
\end{equation*}

\subsubsection{Observables}\label{sec:scko}

Secondly, we consider correlation functions with observables. Let $\xi$ be a
classically allowed solution in the sense that $\xi\in L_M$. For a Weyl observable $F$ we have formula~(\ref{eq:wobsfact}), where the first factor is unity and the last depends on $F$ only (the vacuum correlation). The remaining factor is simply $F(\xi)$ (observe $\hat{\xi}=\xi$ here). That is, it is precisely the value of the classical observable $F$ evaluated on the classically continued solution $\xi$ in the interior of $M$. The semiclassical interpretation becomes even more clear when we switch to normal ordered quantization, formula~(\ref{eq:defnoqk}). The vacuum correlation factor is removed and, more importantly, the classical value $F(\xi)$ is obtained for arbitrary observables $F$. When the non-classical component $\xi^{\iprt}$ is switched on the boundary, we no longer expect to obtain the value of the observable $F$ on some classical solution. Indeed, in this case the solution $\hat{\xi}$ in the interior becomes complex, a behavior well known in a~tunneling context.

As for pure amplitudes, we can also in the case of observables consider the limit of squeezing a region to a slice region. This leads to the slice observables to be considered in~Section~\ref{sec:sliceobs}. The~semiclassical interpretation in this limit will be considered in~Section~\ref{sec:qsobsopk}.

\subsection{Composition}
\label{sec:compk}

As laid out in~Section~\ref{sec:intro_gbqft}, the notion of composition is central to GBQFT, as manifest in the composition rule for amplitudes~(\ref{eq:stcompax}). It turns out that to make this actually work with path integral quantization, a \emph{gluing anomaly} factor has to be included that only depends on the geometry~\cite{Oe:holomorphic}. At the same time, it is convenient to formulate the composition rule in terms of the self-gluing of a region along matching pieces of its boundary. The composition of two distinct regions is then achieved by first performing a disjoint union, followed by a self-gluing.

We first describe the disjoint composition. Let $M$ and $N$ be regions and $\psi_M\in \cH_{\partial M}$, \mbox{$\psi_N\in \cH_{\partial N}$} states on the boundaries. Let $F,G$ be observables in $M$ and $N$ respectively. Then,
\begin{equation*}
 \rho_{M\sqcup N}^{F\cp G}(\psi_M\tens\psi_N)=\rho_M^F(\psi_M)\rho_N^G(\psi_N).
\end{equation*}
Here $F\cp G$ denotes the spacetime composition of the classical observables. This means the following. $F$ and $G$ are functions on the configuration spaces $K_M$ and $K_N$ respectively. We~extend $F$ to a function on $K_{M\sqcup N}=K_M\times K_N$ trivially, i.e., without dependence on $K_N$, and proceed with $G$ analogously. Then, we multiply the extended observables $F$ and $G$ as functions on $K_{M\sqcup N}$ in the ordinary way. With this we satisfy Axiom~(TO5a) of Appendix~\ref{sec:qobsaxioms}.

\begin{figure}
 \centering
 \input{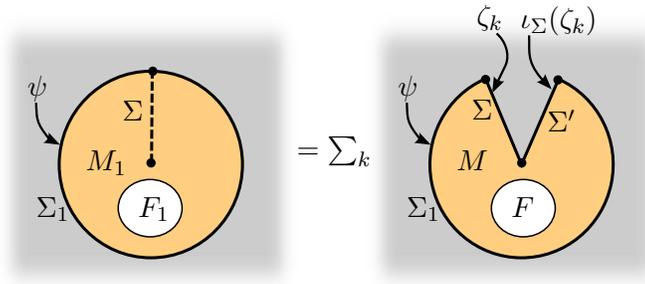}
 \caption{Self-composition of spacetime region with observable.}
 \label{fig:compobs}
\end{figure}

We turn to the self-composition. Let $M$ be a region with its boundary decomposing as $\partial M=\Sigma_1\cup \Sigma\cup \overline{\Sigma'}$. Suppose that $\Sigma$ and $\Sigma'$ are identical so that $M$ can be glued to itself along~$\Sigma$,~$\Sigma'$. Call the resulting manifold $M_1$. Then, $\partial M_1=\Sigma_1$, see Figure~\ref{fig:compobs}. In order for the composition rule to be valid, the gluing anomaly has to be well-defined.

\begin{dfn}
 We call the gluing \emph{admissible} if the \emph{gluing anomaly} $c_{M;\Sigma,\Sigma'}$ given by the following formula is well-defined due to absolute convergence of the sum,
 \begin{equation*}
 c_{M;\Sigma,\Sigma'}=\sum_{k\in I} \rho_M\big(\ncoh_0\tens\zeta_k\tens\iota_{\Sigma'}(\zeta_k)\big).
 \end{equation*}
 Here, $\{\zeta_k\}_{k\in I}$ is an arbitrary orthonormal basis of $\cH_{\Sigma}$.
\end{dfn}

Let $F\colon K_M^\bC\to\bC$ be an observable in $M$. Since field configurations in $M_1$ yield field configurations in $M$, we have a map $K_{M_1}\to K_{M}$. By composition, we obtain from the observable $F$ in~$M$ an observable $F_1\colon K_{M_1}^\bC\to\bC$ in $M_1$. The observables are also illustrated in Figure~\ref{fig:compobs}. If the gluing is admissible, it can be shown that the correlation functions characterized by Theorem~\ref{thm:stdcorrfact} satisfy the following composition rule~\cite[Proposition~4.2]{Oe:feynobs}.
\begin{thm}
 \label{thm:stdcompo}
 Let $\psi\in \cH^\circ_{\partial M_1}$. Then,
 \begin{equation*}
 \rho_{M_1}^{F_1}(\psi) \cdot c_{M;\Sigma,\Sigma'} =\sum_{k\in I} \rho_M^F\big(\psi\tens\zeta_k\tens\iota_{\Sigma'}(\zeta_k)\big).
 \end{equation*}
 Here, $\{\zeta_k\}_{k\in I}$ is an arbitrary orthonormal basis of $\cH_{\Sigma}$.
\end{thm}
$\cH^\circ_{\partial M_1}$ denotes a dense subspace of the Hilbert space $\cH_{\partial M_1}$, where the amplitude map is well-defined. Here we take $\cH^\circ_{\partial M_1}$ to be the subspace spanned by the coherent states. With this we satisfy Axiom~(TO5b) of Appendix~\ref{sec:qobsaxioms}.

For the convenience of the reader we spell out the resulting composition rule for regions~$M$ and~$N$, recall Figure~\ref{fig:GBQFT-comp}. Thus, $M$ has boundary $\partial M=\Sigma_M\cup \Sigma$, and~$N$ has boundary $\partial N=\Sigma_N\cup \overline{\Sigma'}$. Here $\Sigma'$ is a copy of $\Sigma$. As before, $F$ is an observable in $M$ and $G$ in $N$ (not depicted in the figure). Then,
\begin{equation}
 \rho_{M\cup N}^{F\cp G}(\psi_M\tens\psi_N) \cdot c_{M,N;\Sigma,\Sigma'} =\sum_{k\in I} \rho_M^F(\psi_M\tens\zeta_k) \,\rho_N^G(\psi_N\tens\iota_{\Sigma'}(\zeta_k)).
 \label{eq:dualcompo}
\end{equation}
This is the anomaly-corrected version of the composition rule~(\ref{eq:stcompax}).

\section{Slice observables}
\label{sec:sliceobs}

\subsection{Classical slice observables}
\label{sec:csobs}

In the present field theoretic setting, observables are functions on field configurations in spacetime regions. The idea of slice observables is simple: These are observables that only depend on the field in a \emph{slice region} $\hat{\Sigma}$, i.e., an infinitesimal neighborhood of a hypersurface $\Sigma$. In~the case of equal-time hypersurfaces these observables are the analogues of the observables of~non-relativistic quantum mechanics. As the latter, their quantized versions form algebras, in~contrast to other observables in GBQFT.

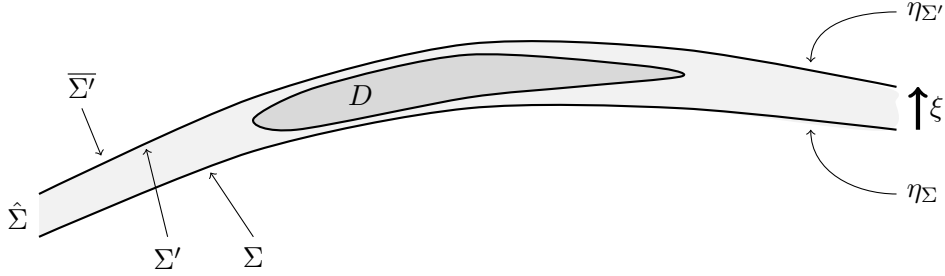
\begin{figure}[h!]
 \centering
 \begin{tikzpicture}[scale=2.85]
\filldraw [thick,yshift=0cm,color=gray!10] plot [smooth] coordinates {(0,0.4) (1,0.85)(2,1.1)(3,1.075) (3.92,0.9) (3.99,0.83)(3.92,0.7)(3,0.8)(2,0.8) (1,0.6)  (0,0.2)  };
\draw [thick,yshift=0cm] plot [smooth] coordinates {(0,0.4) (1,0.85)(2,1.1)(3,1.075) (4,0.9)};
\draw [thick,yshift=0cm] plot [smooth] coordinates {(0,0.2) (1,0.6)(2,0.8)(3,0.8) (4,0.7)};
\draw [thick,xshift=0cm,fill=gray!30] plot [smooth cycle] coordinates {(1,0.75) (1.2,0.7) (2,0.85) (2.5,0.9) (3,0.95) (2.8,1) (2,1.05)(1.3,0.9)};
\node at (1.5,0.86) {$D$};
\node at (-0.1,0.3) {$\hat{\Sigma}$};
\draw[->] (1,0.2) node [below] {$\Sigma$} -- (0.8,0.5);
\draw[->] (0.6,0.2) node [below] {$\Sigma'$} -- (0.5,0.61);
\draw[->] (0.2,0.8) node [above] {$\overline{\Sigma'}$} -- (0.29,0.58);
\draw (4,0.4) node [right] {$\eta_{\Sigma}$} edge[out=180,in=270,->] (3.6,0.7);
\draw (4,1.25) node [right] {$\eta_{\Sigma'}$} edge[out=180,in=90,->] (3.6,1.015);
\draw[ultra thick, ->] (4.1,0.7)  -- (4.1,0.9) node [midway,right] {$\xi$};
\end{tikzpicture}
 \caption{A slice observable $D$ in the slice region $\hat{\Sigma}$.}
 \label{fig:slice-obs}
\end{figure}

Slice observables were introduced to GBQFT in \cite[Section~4.10]{Oe:feynobs}. We recall the setup, with minor adjustments. Consider the slice region $\hat{\Sigma}$ associated to the hypersurface $\Sigma$. We~have the boundary decomposing into two copies of $\Sigma$, called $\Sigma$ and $\Sigma'$, with the second copy oppositely oriented, $\partial\hat{\Sigma}=\Sigma \cup \overline{\Sigma'}$, see Figure~\ref{fig:slice-obs}. Note that $\Sigma$ and $\Sigma'$ are really ``the same'' hypersurface, i.e., ``at the same place''. The separation drawn in Figure~\ref{fig:slice-obs} is purely for purposes of illustration and intuition.
Recall that the space of solution in the interior of the slice region, $L_{\hat{\Sigma}}\subseteq L_{\partial\hat{\Sigma}}=L_{\Sigma}\oplus L_{\overline{\Sigma'}}$ consists of those elements that take the form $(\phi,\phi)\in L_{\Sigma}\oplus L_{\overline{\Sigma'}}$. As for the space $K_{\hat{\Sigma}}$ of field configurations, it turns out that the definition yielding the right properties is $K_{\hat{\Sigma}}\defeq L_{\partial \hat{\Sigma}}=L_{\Sigma}\oplus L_{\overline{\Sigma'}}$. That is, the field configurations are given by no more than the germs on the (two-component) boundary. However, we allow a slice observable $D\colon K_{\hat{\Sigma}}\to\bC$ only to depend on one copy of $L_{\Sigma}$. This ensures that it encodes the right degrees of freedom, achieving the desired correspondence to the non-relativistic setting. Concretely, the \emph{slice observable} $D$ is determined by a map $D'\colon L_{\Sigma}\to\bC$. $D$ takes on $(\phi_{\Sigma},\phi_{\Sigma'})\in K_{\hat{\Sigma}}$ by definition the value
\begin{equation}
 D((\phi_{\Sigma},\phi_{\Sigma'}))\defeq D'\bigg(\frac{1}{2}(\phi_{\Sigma}+\phi_{\Sigma'})\bigg).
 \label{eq:sobseval}
\end{equation}
One can think of this as an ``averaging of the boundary values''. Another way to think about this is in terms of a decomposition $L_{\partial\hat{\Sigma}}=L_{\hat{\Sigma}}\oplus L_{\hat{\Sigma}}^{\rm c}$ into solutions in the interior, $L_{\hat{\Sigma}}$, and a~complementary space $L_{\hat{\Sigma}}^{\rm c}$ of elements of the form $(\phi,-\phi)$. The ``averaging'' in formula~(\ref{eq:sobseval}) is then really a projection onto the subspace $L_{\hat{\Sigma}}\subseteq L_{\partial\hat{\Sigma}}$ with the effect that $D$ only depends on this subspace.

If the slice observable $D$ is real \emph{linear} we can encode it through an element $\xi\in L_{\Sigma}$ such that, for all $\phi\in L_{\Sigma}$,\footnote{There is a difference in convention compared to equation (104) in~\cite{Oe:feynobs} corresponding to a relative minus sign.}
\begin{equation}
 D'(\phi)=2\omega_{\Sigma}(\xi,\phi).
 \label{eq:lsobsdef}
\end{equation}
Let this slice observable now act as a source, localized in $\hat{\Sigma}$. That is, we modify the action $S$ by adding $D$ to obtain $S+D$. As before, we call the affine space of solutions of the modified equations of motion $A_{\hat{\Sigma}}^D$. It is useful to think of these solutions as extending to the ambient spacetime around $\hat{\Sigma}$. Now, since the source is singular in the sense of being concentrated on a~hypersurface, so will be the solutions. In fact, using the relation~(\ref{eq:obssol}) we can calculate exactly how a solution ``jumps'' from one side of $\hat{\Sigma}$ to the other. On each side of $\hat{\Sigma}$ a solution $\eta\in A_{\hat{\Sigma}}^D$ behaves as a solution of the homogeneous equations of motion corresponding to the free action~$S$. Call these partial solutions $\eta_{\Sigma}$ and $\eta_{\Sigma'}$ respectively. Formally we treat these as elements of~$L_{\Sigma}$.

Let $\phi\in L_{\Sigma}$ be arbitrary. Then $(\phi,\phi)\in L_{\hat{\Sigma}}\subseteq K_{\hat{\Sigma}}$. Evaluating the observable $D$ on $(\phi,\phi)$ according to~(\ref{eq:sobseval}) and~(\ref{eq:lsobsdef}) yields
\begin{equation*}
 D((\phi,\phi))=D'(\phi)=2\omega_{\Sigma}(\xi,\phi).
\end{equation*}
On the other hand, formula~(\ref{eq:obssol}) is equally applicable in this case and leads to
\begin{equation*}
 D((\phi,\phi))=2\omega_{\partial\hat{\Sigma}}\left((\phi,\phi),(\eta_{\Sigma},\eta_{\Sigma'})\right) =2\omega_{\Sigma}(\phi,\eta_{\Sigma})+2\omega_{\overline{\Sigma'}}(\phi,\eta_{\Sigma'})=2\omega_{\Sigma}(\phi,\eta_{\Sigma}-\eta_{\Sigma'}).
\end{equation*}
Since $\phi$ was arbitrary we can conclude that $\eta\in A_{\hat{\Sigma}}^D$, i.e., $\eta$ is a solution of the inhomogeneous equations of motion iff
\begin{equation}
 \eta_{\Sigma'}-\eta_{\Sigma}=\xi.
 \label{eq:sobsdiff}
\end{equation}
The ``jump'' is precisely given by $\xi$. This is also illustrated in Figure~\ref{fig:slice-obs}. A particular choice of solution is given by
\begin{equation*}
 \eta=(\eta_{\Sigma},\eta_{\Sigma'})=\bigg({-}\frac{1}{2} \xi,\frac{1}{2} \xi\bigg).
\end{equation*}

Analogous to the previous discussion in~Section~\ref{sec:piwobs}, once we have the result~(\ref{eq:sobsdiff}) we drop the restriction to the real case. That is, we allow $D$ to be complex valued (on real solutions) and correspondingly $\xi$ to be complex, i.e., $\xi\in L_{\Sigma}^\bC$, extending the applicability of formula~(\ref{eq:sobsdiff}). Moreover, we require $D$ to be defined and holomorphic on the complexified configuration space by requiring $D'$ to be defined and holomorphic on $L_{\Sigma}^\bC$.

As an example, it turns out that the field value at a point, besides defining an ordinary (spacetime) observable (as in~Section~\ref{sec:intro_pertlsz}), can also be implemented as a slice observable. We~make this concrete in the previously considered case of Klein--Gordon theory on an equal-time hyper\-sur\-face $\Sigma_t$ in Minkowski space. Thus fix $x\in\Sigma_t$ and in terms of the mode expansion~(\ref{eq:kgmodes}) let\footnote{Strictly speaking the solution $\xi$ is not normalizable in the space of germs of solutions viewed as a Hilbert space. To remedy this one would typically smear it with source functions in spacetime.}
\begin{gather}
 \xi^{\rm a}(k)=\xi^{{\rm b}}(k)=\im \exp(\im(E t-k x)).
 \label{eq:xipoint}
\end{gather}
Then, define the corresponding linear observable $D\colon K_{\hat{\Sigma}_t}\to\R$ by equation~(\ref{eq:lsobsdef}). With the explicit form of the symplectic structure~(\ref{eq:sfkgm}) we find
\begin{equation}
 D'(\phi)=2\omega_{\Sigma_t}(\xi,\phi)=\phi(t,x).
 \label{eq:dpoint}
\end{equation}

\subsection{Quantum algebra of slice observables}
\label{sec:qsobs}

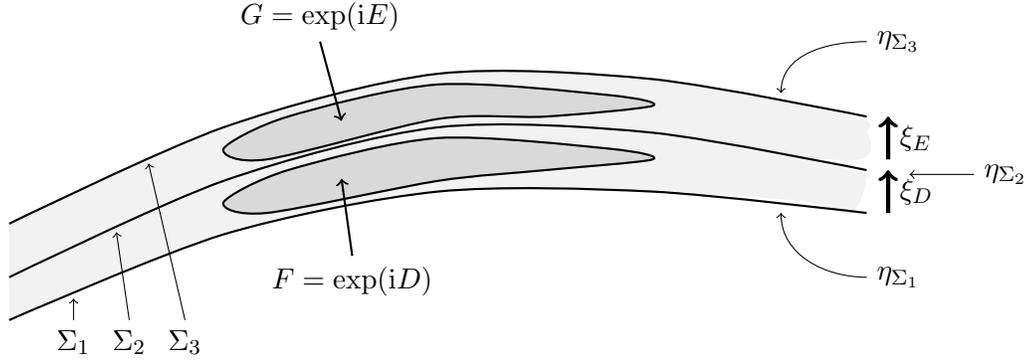
\begin{figure}
 \centering
 \begin{tikzpicture}[scale=2.85]
\filldraw [thick,yshift=0.25cm,color=gray!10] plot [smooth] coordinates {(0,0.4) (1,0.85)(2,1.1)(3,1.075) (3.92,0.9) (4,0.83)(3.925,0.7)(3,0.8)(2,0.8) (1,0.6)  (0,0.1)  };
\draw [thick,yshift=0.25cm] plot [smooth] coordinates {(0,0.4) (1,0.85)(2,1.1)(3,1.075) (4,0.9)};
\draw [thick,yshift=0.25cm,fill=gray!30] plot [smooth cycle] coordinates {(1,0.75) (1.2,0.7) (2,0.889) (2.5,0.9) (3,0.95) (2.8,1) (2,1.05)(1.3,0.9)};
\filldraw [thick,yshift=0cm,color=gray!10] plot [smooth] coordinates {(0,0.4) (1,0.85)(2,1.1)(3,1.078)  (3.92,0.9) (3.99,0.85)(3.92,0.7)(3.7,0.76)(3,0.8)(2,0.8) (1,0.6)  (0,0.2)  };
\draw [thick,yshift=0cm] plot [smooth] coordinates {(0,0.4) (1,0.85)(2,1.1)(3,1.075) (4,0.9)};
\draw [thick,yshift=0cm] plot [smooth] coordinates {(0,0.2) (1,0.6)(2,0.8)(3,0.8) (4,0.7)};
\draw [thick,xshift=0cm,fill=gray!30] plot [smooth cycle] coordinates {(1,0.75) (1.2,0.7) (2,0.85) (2.5,0.9) (3,0.95) (2.8,1) (2,1.05)(1.3,0.9)};
\draw[->] (0.3,0.2) node [below] {$\Sigma_1$} -- (0.3,0.3);
\draw[->] (0.56,0.2) node [below] {$\Sigma_2$} -- (0.5,0.61);
\draw[->] (0.82,0.2) node [below] {$\Sigma_3$} -- (0.65,0.93);
\draw (4,0.4) node [right] {$\eta_{\Sigma_1}$} edge[out=180,in=270,->] (3.6,0.7);
\draw (4,1.5) node [right] {$\eta_{\Sigma_3}$} edge[out=180,in=90,->] (3.6,1.26);
\draw[ultra thick, ->] (4.1,0.7)  -- (4.1,0.9) node [midway,right] {$\xi_D$};
\draw[yshift=0.25cm,ultra thick, ->] (4.1,0.7)  -- (4.1,0.9) node [midway,right] {$\xi_E$};
\draw[thick,<-]  (1.55,1.13) -- (1.45,1.5) node [above] {$G=\exp(\im E)$};
\draw[thick,<-]  (1.55,0.86) -- (1.6,0.5) node [below] {$F=\exp(\im D)$};
\draw (4.5,0.88) node [right] {$\eta_{\Sigma_2}$} edge[out=180,in=0,->] (4.2,0.88);
\end{tikzpicture}
 \caption{Composition of Weyl slice observables $F=\exp(\im D)$ and $G=\exp(\im E)$.}
 \label{fig:slice-obs-comp}
\end{figure}

Classically, slice observables obviously form a commutative algebra, simply by viewing them as holomorphic maps $L_{\Sigma}^\bC\to\bC$ $\big($or equivalently induced maps $K_{\hat{\Sigma}}^\bC\to\bC\big)$. In the quantum theory the only notion of composition is mediated by spacetime composition. A slice region $\hat{\Sigma}$ can be glued to a copy of itself, resulting in the very same slice region. Thus, a slice observable can be composed with another slice observable defined on the same slice region. What is more, we~shall show that the resulting object, which is again associated to the same slice region, is in fact again a slice observable. That is, upon quantization, the slice observables form an algebra. The~product is non-commutative with the order reflecting the spatio-temporal arrangement of the gluing process, see Figure~\ref{fig:slice-obs-comp}. In the particular case of equal-time hypersurfaces, we recover the product reflecting the temporal order of observables, as in non-relativistic quantum mechanics.

The quantization of observables is ultimately governed by formula~(\ref{eq:piwobs}) and its implicit gene\-ra\-li\-za\-tion for non-Weyl observables. Thus, to show that the composition of two slice observables is a specific slice observable, we have to show that inserting the former observables into this formula yields the exact same results as inserting the latter observable. This has to be true not only in the presence of any vacuum (i.e., any polarization on the boundary) but also in the presence of arbitrary other observables, as long as they are localized outside the slice region $\hat{\Sigma}$. Due to their generating nature, it is sufficient to consider this for the case that all observables are Weyl observables.

Thus, consider two copies of the slice region $\hat{\Sigma}$, glued to form another copy, see Figure~\ref{fig:slice-obs-comp}. To keep track of the involved hypersurfaces (all copies of $\Sigma$, but with infinitesimal transversal displacement) we label them $\Sigma_1$, $\Sigma_2$, $\Sigma_3$. Thus, the first component slice region has boundary $\Sigma_1\cup \overline{\Sigma_2}$ and the second one $\Sigma_2\cup\overline{\Sigma_3}$. Now consider linear maps $D'\colon L_{\Sigma}^\bC\to\bC$ and \mbox{$E'\colon L_{\Sigma}^\bC\to\bC$} that determine linear slice observables $D\colon L_{\partial\hat{\Sigma}}^\bC\to\bC$ and $E\colon L_{\partial \hat{\Sigma}}^\bC\to\bC$ via formula~(\ref{eq:sobseval}). We~denote the corresponding elements of $L_{\Sigma}^\bC$ via formula~(\ref{eq:lsobsdef}) by $\xi_D$ and $\xi_E$ respectively. Also, let $F\defeq\exp(\im D)$ and $G\defeq\exp(\im E)$ denote the corresponding Weyl observables. We locate $F$ between $\Sigma_1$ and $\Sigma_2$, and $G$ between~$\Sigma_2$ and~$\Sigma_3$. In line with our previous use of this notation we write~$G\cp F$ for the spacetime composite of the Weyl observables on $K_{\hat{\Sigma}}^\bC\cup K_{\hat{\Sigma}}^\bC$. However, we will wish to identify $K_{\hat{\Sigma}}^\bC\cup K_{\hat{\Sigma}}^\bC$ with $K_{\hat{\Sigma}}^\bC$ itself, in line with the gluing identity $\hat{\Sigma}\cup\hat{\Sigma}=\hat{\Sigma}$ for the underlying slice regions.

Consider now a solution $\eta$ in a neighborhood of $\hat{\Sigma}$, of the equations of motion modified by adding $D$ and $E$ to the action. The analogue of equation~(\ref{eq:sobsdiff}) is now given by two corresponding equations, with obvious notation,
\begin{equation}
 \eta_{\Sigma_2}-\eta_{\Sigma_1}=\xi_D,\qquad\text{and}\qquad
 \eta_{\Sigma_3}-\eta_{\Sigma_2}=\xi_E.
 \label{eq:eta_2}
\end{equation}
An immediate consequence is
\begin{equation}
 \eta_{\Sigma_3}-\eta_{\Sigma_1}=\xi_D+\xi_E=\xi_{D+E}.
 \label{eq:compssol}
\end{equation}

We proceed to evaluate the analogue of the right-hand side of expression~(\ref{eq:piwobs}) using formula~(\ref{eq:sobseval}) for each of $D$ and $E$. Crucially, the following computation arises as a factor in any evaluation of~(\ref{eq:piwobs}), irrespective of the boundary polarization and any additional Weyl obser\-vables present outside $\hat{\Sigma}$,
\begin{align*}
 \exp\bigg(\frac{\im}{2} (E + D)(\eta)\bigg)
 & = \exp\bigg(\frac{\im}{2}\bigg(E'\bigg(\frac{1}{2}(\eta_{\Sigma_2}+\eta_{\Sigma_3})\bigg) +D'\bigg(\frac{1}{2}(\eta_{\Sigma_1}+\eta_{\Sigma_2})\bigg)\bigg)\bigg)
 \\
 & = \exp\big(\im\omega_{\Sigma}(\xi_E,\eta_{\Sigma_2})+\im\omega_{\Sigma}(\xi_D,\eta_{\Sigma_2})\big)
 \\
 &= \exp\big(\im\omega_{\Sigma}(\xi_E+\xi_D,\eta_{\Sigma_2})\big) .
\end{align*}
Substituting in the last expression $\eta_{\Sigma_2}$ with $\frac12 \big( \xi_D + \eta_{\Sigma_1}+ \eta_{\Sigma_3} - \xi_E\big)$, a consequence of \eqref{eq:eta_2}, we~obtain
\begin{align*}
 \exp\bigg(\frac{\im}{2} (E + D)(\eta)\bigg)
 & = \exp(\im\omega_{\Sigma}(\xi_E,\xi_D))
 \exp\bigg(\frac{\im}{2}\omega_{\Sigma}\big(\xi_E+\xi_D,\eta_{\Sigma_1}+\eta_{\Sigma_3}\big)\bigg)
 \\
 & = \exp(\im\omega_{\Sigma}(\xi_E,\xi_D))
 \exp\bigg(\frac{\im}{2} (E'+D')\bigg(\frac{1}{2}\big(\eta_{\Sigma_1}+\eta_{\Sigma_3}\big)\bigg)\bigg).
\end{align*}
We notice that the resulting expression does not contain any explicit reference to the ``intermediate'' value $\eta_{\Sigma_2}$. In particular, we may consistently set $K_{\hat{\Sigma}}^\bC\cup K_{\hat{\Sigma}}^\bC=K_{\hat{\Sigma}}=L_{\partial \hat{\Sigma}}$, as desired. Furthermore, the resulting dependence on the values $\eta_{\Sigma_1}$ and $\eta_{\Sigma_3}$ on the boundary of $\hat{\Sigma}$ is exac\-tly as in formula~(\ref{eq:sobseval}). That is, the composite $G\cp F$ of the Weyl slice observables, upon quantization, i.e., insertion into formula~(\ref{eq:piwobs}) behaves as a single slice observable that we shall call $G\qp F$. In fact, it behaves even as a Weyl slice observable, up to a phase factor. Taking also into account relation~(\ref{eq:compssol}) we note that this resulting Weyl slice observable coincides with the quantization of the classical product of $G$ and $F$, i.e., with $G\cdot F$, up to the phase factor. In~formal notation,
\begin{equation}
 G\qp F=\exp\left(\im\omega_{\Sigma}(\xi_E,\xi_D)\right) G\cdot F.
 \label{eq:weylrel}
\end{equation}
It should not come as a surprise that the relations we obtain are precisely the \emph{Weyl relations}. In particular, the complex vector space spanned by the Weyl observables becomes a~non-commutative algebra in this way. This algebra structure extends to all other slice observables that we can generate by differentiating with respect to linear observables.
The new product on the elements of this algebra viewed as functions on the phase space $L_{\Sigma}$ is also precisely the \emph{Groenewold--Moyal product}~\cite{Gro:principlesqm,Moy:quantstat}.
We denote the quantum algebra of slice observables associated to a hypersurface $\Sigma$ by $\qsoa_{\Sigma}$.

\subsection{Vacuum correlation function of a Weyl slice observable}
\label{sec:vevweylslice}

The vacuum correlation function for a single Weyl slice observable in a slice region is easily obtained. Due to its importance for later considerations, we exhibit it explicitly. We assume the same setting as in~Section~\ref{sec:csobs} in terms of a hypersurface $\Sigma$, the corresponding slice region~$\hat{\Sigma}$ (recall Figure~\ref{fig:slice-obs}), as well as a linear slice observable $D$ determined by a complex linear map $D'\colon L_{\Sigma}^{\bC}\to\bC$ in terms of equation~(\ref{eq:sobseval}), with the latter corresponding to $\xi\in L_{\Sigma}^{\bC}$ via equation~(\ref{eq:lsobsdef}). Let $F=\exp(\im D)$ be the corresponding Weyl slice observable. We assume the vacua on the two sides ($\Sigma$ and $\overline{\Sigma'}$) of the slice region $\hat{\Sigma}$ given in terms of transversal polarizations~$L_{\Sigma}^\pol$ and $L_{\Sigma}^\cpol$. That is, the space $L_{\Sigma}^{\bC}$ of complexified germs of solutions on $\Sigma$ decomposes as a direct sum of Lagrangian subspaces $L_{\Sigma}^{\bC}=L_{\Sigma}^{\pol} \oplus L_{\Sigma}^{\cpol}$. In terms of elements we write this decomposition as $\xi=\xi^\pol+\xi^\cpol$.

To evaluate formula~(\ref{eq:piwobs}) for the vacuum correlation function we need to find the solution $\eta\in A_{\hat{\Sigma}}^{D,\bC}$ of the inhomogeneous equations of motion determined by $D$ satisfying the boundary conditions, recall Section~\ref{sec:csobs}. The latter are given by the transversal polarizations $L_{\Sigma}^\pol$ and $L_{\Sigma}^\cpol$, i.e., $\eta_{\Sigma}\in L^\pol_{\Sigma}$ and $\eta_{\Sigma'}\in L^\cpol_{\Sigma}$. With relation~(\ref{eq:sobsdiff}) we obtain, $\eta_{\Sigma}=-\xi^\pol$ and $\eta_{\Sigma'}=\xi^\cpol$. We may now evaluate~(\ref{eq:piwobs}) to get
\begin{align}
 \rho_{\hat{\Sigma}}^F\big(W^\pol\big)&=\exp\bigg(\frac{\im}{2} D(\eta_{\Sigma},\eta_{\Sigma'})\bigg)
 =\exp\bigg(\frac{\im}{2} D'\bigg(\frac{1}{2} (\eta_{\Sigma}+\eta_{\Sigma'})\bigg)\bigg)\nonumber
 \\
 &=\exp\bigg(\frac{\im}{2} \omega_{\Sigma}\big(\xi,\xi^\cpol -\xi^\pol\big)\bigg)
 =\exp\big({-}\im\, \omega_{\Sigma}\big(\xi^\cpol,\xi^\pol\big)\big).
 \label{eq:vevweyl}
\end{align}
In the case that $L_{\Sigma}^\pol$ is a K\"ahler polarization $L_{\Sigma}^+$ with conjugate $L_{\Sigma}^\cpol=L_{\Sigma}^-$ we can rewrite this in terms of the inner product~(\ref{eq:stdjip}) by using the explicit form~(\ref{eq:projkp}) of the projectors onto the Lagrangian subspaces,
\begin{equation}
 \rho_{\hat{\Sigma}}^F\big(W^\pol\big)=\exp\bigg({-}\frac{1}{4} \{\xi,\xi\}_{\Sigma}\bigg).
 \label{eq:vevweylk}
\end{equation}

\subsection{Vacuum correlation functions of quadratic observables}
\label{sec:propslice}

In order to illustrate the versatility of slice observables we show in the following how they can be used to derive vacuum correlation functions of quadratic observables. We also consider the concrete case of the Feynman propagator of Klein--Gordon quantum field theory.

We are interested in the product observable $D=D_1 \cdot D_2$ of the linear observables $D_1$ and~$D_2$. Its vacuum correlation function is given by formula~(\ref{eq:twopoint}). To evaluate it, we need to find the inhomogeneous solutions $\eta_1$ and~$\eta_2$ of the equations of motion modified by $D_1$ and~$D_2$ respectively and satisfying the boundary conditions corresponding to the vacuum. To this end we shall assume that $D_1$ and~$D_2$ are supported on disjoint spacetime regions, separated by a~hypersurface $\Sigma$. For ease of language we shall refer to the side of $\Sigma$, where $D_1$ is supported as the ``past'' and the side where $D_2$ is supported as the ``future''. This suggests that $\Sigma$ is a~spacelike hypersurface, although this might not be the case. We shall take the orientation of~$\Sigma$ to correspond to its ``past'' side while that of its ``future'' side is given by $\overline{\Sigma}$. Let the pair~$L_{\Sigma}^\pol$ and~$L_{\Sigma}^\cpol$ of transversal Lagrangian subspaces of~$L_{\Sigma}^{\bC}$ determine the vacuum in the ``past'' and ``future'' of $\Sigma$ respectively.

For the moment we assume furthermore that $D_1$ is given by a slice observable on a hypersurface $\Sigma_1$ to the past of $\Sigma$ and such that the spaces $L_{\Sigma}$ and $L_{\Sigma_1}$ of germs of solutions are in one-to-one correspondence as are the vacua. We also make the corresponding assumption for $D_2$. Thus, there are elements $\xi_i\in L_{\Sigma_i}^\bC$ determining the observables $D_1$ and $D_2$ via equations~(\ref{eq:sobseval}) and~(\ref{eq:lsobsdef}). The solution $\eta_1$ then consists of an ``early'' and ``late time'' part $\eta_{1,<}$ and $\eta_{1,>}$, ``before'' and ``after'' hypersurface $\Sigma_1$ respectively. Their relation as homogeneous solutions is given by $\eta_{1,>}=\eta_{1,<}+\xi_1$, compare Section~\ref{sec:csobs}. Moreover, the boundary conditions are given by, $\eta_{1,<}\in L_{\Sigma_1}^{\pol}$ and $\eta_{1,>}\in L_{\Sigma_1}^\cpol$. Thus, $\eta_{1,<}=-\xi_1^\pol$ and $\eta_{1,>}=\xi_1^\cpol$. The considerations for solution $\eta_2$ are analogous, and we obtain, $\eta_{2,<}=-\xi_2^\pol$ and $\eta_{2,>}=\xi_2^\cpol$. At this point we may notice that we can identify all the elements $\eta_1$, $\eta_2$, $\xi_1$, $\xi_2$ as well as their polarized components with cor\-res\-pon\-ding elements of $L_{\Sigma}^\bC$, and shall do so without modifying notation. In this way the locations of the hypersurfaces $\Sigma_1$ and $\Sigma_2$ become irrelevant. Due to linearity we may then replace $D_1$ and~$D_2$ with linear combinations of slice observables located on different hypersurfaces and thus with arbitrary linear observables (still located respectively in the ``past'' and ``future'' of $\Sigma$, however). We obtain for the vacuum correlation function
\begin{align}
 \rho_M^{D_1 D_2}\big(W^\pol\big) &= -\frac{\im}{2}\big(D_1(\eta_{2,<})+D_2(\eta_{1,>})\big) = -\im \Big({-}\omega_{\Sigma}\big(\xi_1,\xi_2^\pol\big)+\omega_{\Sigma}\big(\xi_2,\xi_1^\cpol\big)\Big)\nonumber
 \\
& = 2 \im \omega_{\Sigma}\big(\xi_1^\cpol,\xi_2^\pol\big).
 \label{eq:twopointord}
\end{align}

{\sloppy
We now assume in addition that the polarizations are conjugate K\"ahler polarizations, \mbox{$L_{\Sigma}^\pol=L_{\Sigma}^+$} and $L_{\Sigma}^\cpol=L_{\Sigma}^-$.
Then, in order to re-express $\xi_1$ and $\xi_2$ in terms of $D_1$ and $D_2$ we may use an~orthonormal basis $\{u_k\}_{k\in I}$ of $L_\Sigma^+$ with respect to the inner product~(\ref{eq:stdipc}). Thus,
\begin{align*}
 \rho_M^{D_1 D_2}\big(W^\pol\big)
 &= \sum_k 2 \im \omega_{\Sigma}(\xi_1^-,u_k)\, \big(u_k,\xi_2^+\big)_{\Sigma}
 = \sum_k 2 \im \omega_{\Sigma}(\xi_1^-,u_k)\, 4\im\omega_{\Sigma}\big(\overline{u}_k,\xi_2^+\big)
 \\
& = \sum_k 2 \im \omega_{\Sigma}(\xi_1,u_k)\, 4\im\omega_{\Sigma}(\overline{u}_k,\xi_2)
 = 2 \sum_k D_1(u_k) D_2(\overline{u}_k).
\end{align*}}
This recovers and generalizes well known formulas from the literature, such as DeWitt's vacuum correlation function for the energy-momentum tensor~\cite{Dew:qftcurved}.
Alternatively, we may use other completeness relations. Using~(\ref{eq:cohcompl}) we obtain
\begin{align*}
 \rho_M^{D_1 D_2}\big(W^\pol\big)
 &= \frac{1}{2}\{\xi_1^-,\xi_2^+\}_{\Sigma}
 = \frac{1}{4}\int_{\hat{L}_{\Sigma}} \{\xi_1^-,\phi\}_{\Sigma} \{\phi,\xi_2^+\}_{\Sigma}\,\xd\nu_{\Sigma}(\phi)
 \\
 &= -\!\int_{\hat{L}_{\Sigma}}\!\! 2\omega_{\Sigma}(\xi_1^-,\phi^+)\, 2\omega_{\Sigma}\big(\phi^-,\xi_2^+\big)\,\xd\nu_{\Sigma}(\phi)
 =\! \int_{\hat{L}_{\Sigma}}\!\! 2\omega_{\Sigma}(\xi_1,\phi^+)\, 2\omega_{\Sigma}(\xi_2,\phi^-)\,\xd\nu_{\Sigma}(\phi)
 \\
 &= \int_{\hat{L}_{\Sigma}} D_1(\phi^+)\, D_2(\phi^-)\,\xd\nu_{\Sigma}(\phi).
\end{align*}

We proceed to focus on the particular quadratic observable given by the two-point function of Klein--Gordon theory in Minkowski spacetime. Consider spacetime points $(t_1,x_1)$ and $(t_2,x_2)$ and suppose $t_2> t_1$. We set up linear observables $D_1$ and $D_2$ such that
\begin{equation*}
 D_i(\phi)=\phi(t_i,x_i).
\end{equation*}
For the hypersurface $\Sigma$ we may take any equal-time hypersurface at time $t$ with $t_1<t<t_2$. By~comparison with equation~(\ref{eq:dpoint}) we find that the elements $\xi_1,\xi_2\in L_{\Sigma}$ determining the observables~$D_1$, $D_2$ take the form~(\ref{eq:xipoint}).
Inserting this into expression~(\ref{eq:twopointord}) we obtain the familiar 2-point function and Feynman propagator (for $t_1< t_2$),
\begin{align*}
 \rho_M^{D_1 D_2}\big(W^\pol\big)& = -\im G_F((t_1,x_1),(t_2,x_2)) = \int\frac{\xd^3 k}{(2\pi)^3 2E}\, \xi_1(k) \overline{\xi_2(k)}
 \\
 &= \int\frac{\xd^3 k}{(2\pi)^3 2E}\, {\rm e}^{\im (E(t_1-t_2)-k(x_1-x_2))}.
\end{align*}
Removing the constraint on the time ordering of arguments, we have the usual formula,
\begin{gather*}
 G_F((t_1,x_1),(t_2,x_2))
 \\ \qquad
 {}= \im \int\frac{\xd^3 k}{(2\pi)^3 2E}\, \big(\theta(t_2-t_1) {\rm e}^{\im (E(t_1-t_2)-k(x_1-x_2))} + \theta(t_1-t_2) {\rm e}^{\im (E(t_2-t_1)-k(x_2-x_1))}\big).
\end{gather*}

\section{Slice observables in K\"ahler quantization}
\label{sec:kobs}

\subsection{Quantized slice observables as operators}
\label{sec:qsobsopk}

In GBQFT, observables give rise to correlation functions, but in general not to operators on some particular space. In the K\"ahler quantization setting (Section~\ref{sec:kquant}), the slice observables are an exception. Consider a hypersurface $\Sigma$ with associated space of germs $L_{\Sigma}$, symplectic form $\omega_{\Sigma}$, and K\"ahler polarization $L_{\Sigma}^+\subseteq L_{\Sigma}^\bC$. Denote by $\hat{\Sigma}$ the slice region associated to $\Sigma$, with boundary $\partial\hat{\Sigma}=\Sigma\cup \overline{\Sigma}$. We recall the relation between the inner product on $\cH_{\Sigma}$ and the amplitude map~$\rho_{\hat{\Sigma}}$ (Axiom~(T3x) in Appendix~\ref{sec:qobsaxioms}),
\begin{equation}
\langle\eta, \psi\rangle_{\Sigma}=\rho_{\hat{\Sigma}}(\psi\tens\iota_{\Sigma}(\eta)).
 \label{eq:ipampl}
\end{equation}
An observable $F\colon K_{\hat{\Sigma}}^\bC=L_{\partial\hat{\Sigma}}^\bC\to\bC$ determines an operator $\hat{F}$ on the state space $\cH_{\Sigma}$ via
\begin{equation}
 \big\langle\eta,\hat{F} \psi\big\rangle_{\Sigma}=\rho_{\hat{\Sigma}}^F(\psi\tens\iota_{\Sigma}(\eta)).
 \label{eq:opsobs}
\end{equation}
Recall Figure~\ref{fig:slice-obs} for an illustration.

Crucially, this definition brings into exact correspondence the composition of so defined ope\-ra\-tors with the spacetime composition of the underlying slice observables, compare Figure~\ref{fig:slice-obs-comp}. That is
\begin{align}
 \big\langle\eta,\hat{G}\hat{F} \psi\big\rangle_{\Sigma}
 &=\sum_{k\in I}\big\langle\eta, \hat{G}\zeta_k\big\rangle_{\Sigma} \big\langle\zeta_k,\hat{F}\psi\big\rangle_{\Sigma}
 =\sum_{k\in I} \rho^G_{\hat{\Sigma}}(\zeta_k\tens\iota_{\Sigma}(\eta))
 \rho^F_{\hat{\Sigma}}(\psi\tens\iota_{\Sigma}(\zeta_k))\nonumber
 \\
 &=\rho^{G\cp F}_{\hat{\Sigma}}(\psi\tens\iota_{\Sigma}(\eta)).
 \label{eq:opcomp}
\end{align}
The last equality is due to Theorem~\ref{thm:stdcompo}, in the form of equation~(\ref{eq:dualcompo}), with no anomaly present in this case, i.e., $c=1$. As before, we have adopted the notation $G\cp F$ for the spacetime composite of the slice observables $G$ and $F$, arranged as in Figure~\ref{fig:slice-obs-comp}. Again, $\{\zeta_k\}_{k\in I}$ denotes an orthonormal basis of $\cH_{\Sigma}$.

Suppose that $F$ is a Weyl slice observable $F=\exp(\im D)$ with $D$ linear and thus determined by~equation~(\ref{eq:sobseval}) from a linear map $D'\colon L_{\Sigma}^\bC\to\bC$. We let $D'$ be determined by an element $\xi\in L_{\Sigma}^\bC$ according to equation~(\ref{eq:lsobsdef}).
\begin{prop}
 \label{prop:weylactk}
 The operator $\hat{F}$ on $\cH_{\Sigma}$ acts on coherent states as
 \begin{equation}
 \big(\hat{F} \coh_\tau\big)(\phi)
 =\exp\bigg(\frac12\{\tau,\phi\}_{\Sigma}+\frac12\{\xi,\phi\}_{\Sigma} -\frac12\{\tau,\xi\}_{\Sigma}-\frac14\{\xi,\xi\}_{\Sigma}\bigg).
 \label{eq:weylactk}
 \end{equation}
\end{prop}
\begin{proof}
 By equations~(\ref{eq:reprod}),~(\ref{eq:ipampl}) and~(\ref{eq:opsobs}) as well as equations~(\ref{eq:wobsfact}) and~(\ref{eq:sfqvev}) of Theorem~\ref{thm:stdcorrfact} we have
 \begin{equation*}
 \big(\hat{F} \coh_\tau\big)(\phi)
 =\big\langle \coh_\phi,\hat{F}\coh_\tau\big\rangle_{\Sigma}
 =\rho_{\hat{\Sigma}}^F(\coh_\tau\tens\coh_\phi)
 =\langle\coh_\phi,\coh_\tau\rangle_{\Sigma} F\big(\hat{(\tau,\phi)}\big) \exp\bigg(\frac{\im}{2} D(\eta)\bigg).
 \end{equation*}
 (Note that in our implicit notation $\iota_{\Sigma}(\coh_{\phi})=\coh_{\phi}$.)
 Here, $\eta\in A_{\hat{\Sigma}}^{D,\bC}\cap L_{\partial\hat{\Sigma}}^+$ is unique and the last factor is given by equation~(\ref{eq:vevweylk}), compare Section~\ref{sec:vevweylslice},
 \begin{equation*}
 \exp\bigg(\frac{\im}{2} D(\eta)\bigg)=\exp\bigg({-}\frac14\{\xi,\xi\}_{\Sigma}\bigg).
 \end{equation*}
 The first factor is the inner product~(\ref{eq:ipcohk}), here,
 \begin{equation}
 \langle\coh_\phi,\coh_\tau\rangle_{\Sigma}=\exp\bigg(\frac{1}{2}\{\tau,\phi\}_{\Sigma}\bigg).
 \label{eq:ppw1}
 \end{equation}
 As for the remaining factor, we note
 \begin{equation*}
 \hat{(\tau,\phi)}=\frac12 \left(\tau+\phi -\im J_{\Sigma}(\phi-\tau)\right).
 \end{equation*}
 With equation~(\ref{eq:lsobsdef}) we get
 \begin{equation*}
 F\big(\hat{(\tau,\phi)}\big)=\exp\left(\im\,\omega_{\Sigma}\left(\xi,\tau+\phi -\im J_{\Sigma}(\phi-\tau)\right)\right).
 \end{equation*}
 With the definition~(\ref{eq:stdjip}) of the inner product we obtain
 \begin{equation}
 F\big(\hat{(\tau,\phi)}\big)=\exp\bigg(\frac12\{\xi,\phi\}_{\Sigma}-\frac12\{\tau,\xi\}_{\Sigma}\bigg).
 \label{eq:ppw2}
 \end{equation}
 Joining the factors yields the claimed identity.
\end{proof}
As is easily seen, this satisfies the Weyl relation~(\ref{eq:weylrel}) as follows already from~(\ref{eq:opcomp}) when comparing with the results of Section~\ref{sec:qsobs}. This fact was also shown previously in~\cite[Proposition~4.5]{Oe:feynobs}.\footnote{There, more general observables where considered which results in extra terms that are absent here. Also note differences in ordering conventions.}
Since Weyl observables are generators for general observables and coherent states are generators for general states, Proposition~\ref{prop:weylactk} completely characterizes the action of slice observables on~states.

The classical algebra of slice observables (see Section~\ref{sec:csobs}) carries a natural $\sst$-structure. If~we consider such an observable $F$ as merely living on the real configuration space $K_{\Sigma}$, i.e., as a~map $F\colon K_{\Sigma}\to\bC$, then the $\sst$-operation is simply complex conjugation. However, as explained in~Section~\ref{sec:piwobs} we generally consider observables as holomorphic functions on the complexification of the configuration space, here $F\colon K_{\Sigma}^\bC\to\bC$. The natural extension of the notion of complex conjugation to this setting is given by
\begin{equation}
 F^\sst(\phi)\defeq \overline{F\big(\overline{\phi}\big)}.
 \label{eq:stdss}
\end{equation}
In this way, $F^\sst$ is again holomorphic. It is now easy to verify from~(\ref{eq:weylrel}) (and a well-known fact) that this $\sst$-structure is also compatible with the product of the quantum algebra of slice observables (Section~\ref{sec:qsobs}). That is, $(G\qp F)^\sst=F^\sst\qp G^\sst$. We may verify moreover, that it translates precisely to taking the adjoint of the corresponding operator in the sense of relation~(\ref{eq:opsobs}).
\begin{prop} For an arbitrary slice observable $F$, $\hat{F^*}=\hat{F}^\dagger$.
\end{prop}
\begin{proof}
 As coherent states span a dense subspace, it is sufficient to show this for matrix elements of coherent states. That is, it is sufficient to show
 \begin{equation*}
 \big\langle \coh_\phi,\hat{F}^\dagger \coh_\tau\big\rangle_{\Sigma}
 =\big\langle \coh_\phi,\hat{F^*} \coh_\tau\big\rangle_{\Sigma}.
 \end{equation*}
 Moreover, since Weyl observables are generating observables it is sufficient to show this for $F$ a~Weyl obser\-vable. Defining $F$ as above, using the reproducing property~(\ref{eq:reprod}) and Proposition~\ref{prop:weylactk} we~get for the left-hand side
 \begin{align*}
 \big\langle \coh_\phi,\hat{F}^\dagger \coh_\tau\big\rangle_{\Sigma}
 &=\big\langle \hat{F}\coh_\phi,\coh_\tau\big\rangle_{\Sigma}
 =\overline{\big\langle \coh_\tau,\hat{F}\coh_\phi\big\rangle_{\Sigma}}
 \\
& =\exp\bigg(\frac12\overline{\{\phi,\tau\}_{\Sigma}}+\frac12\overline{\{\xi,\tau\}_{\Sigma}} -\frac12\overline{\{\phi,\xi\}_{\Sigma}}-\frac14\overline{\{\xi,\xi\}_{\Sigma}}\bigg)
\\
 &=\exp\bigg(\frac12\{\tau,\phi\}_{\Sigma}+\frac12\{\tau,\overline{\xi}\}_{\Sigma} -\frac12\{\overline{\xi},\phi\}_{\Sigma}-\frac14\big\{\overline{\xi},\overline{\xi}\big\}_{\Sigma}\bigg).
 \end{align*}
 On the other hand, as is easy to see, replacing $F$ with $F^*$ amounts to replacing $D$ by $-D^*$ and thus $\xi$ by $-\overline{\xi}$. So, again using~(\ref{eq:reprod}) and Proposition~\ref{prop:weylactk} we get for the right-hand side
 \begin{equation*}
 \big\langle \coh_\phi,\hat{F^*} \coh_\tau\big\rangle_{\Sigma}
 =\exp\bigg(\frac12\{\tau,\phi\}_{\Sigma}+\frac12\big\{{-}\overline{\xi},\phi\big\}_{\Sigma} -\frac12\big\{\tau,-\overline{\xi}\big\}_{\Sigma} -\frac14\big\{{-}\overline{\xi},-\overline{\xi}\big\}_{\Sigma}\bigg).
 \end{equation*}
 It remains to read off the coincidence between the two sides.
\end{proof}

We return to the case that $F$ is a Weyl observable determined by the linear observable $D$, in~turn determined by $\xi\in L_{\Sigma}^\bC$. Important special cases for $\xi$ in certain subspaces are given by the following Corollary of Proposition~\ref{prop:weylactk}.
\begin{cor}
 \label{cor:weylactsk}
 If $\xi$ is real, i.e., $\xi\in L_{\Sigma}$ we obtain the unitary action
 \begin{equation}
 \hat{F} \ncoh_\tau=\exp\left(\im\omega_{\Sigma}(\xi,\tau)\right) \ncoh_{\tau+\xi}.
 \label{eq:weylactku}
 \end{equation}
 If $\xi\in L_{\Sigma}^-$ we get
 \begin{equation}
 \hat{F} \coh_\tau=\coh_{\tau+\xi+\overline{\xi}}.
 \label{eq:weylactkc}
 \end{equation}
\end{cor}
\begin{proof}
 If $\xi\in L_{\Sigma}$, we recognize the first two terms in the exponential on the right-hand side of~(\ref{eq:weylactk}) as the wave function~(\ref{eq:cohwf}) of the coherent state $\coh_{\xi+\tau}$. This yields,
 \begin{equation*}
 \hat{F} \coh_\tau
 =\exp\bigg({-}\frac12\{\tau,\xi\}_{\Sigma}-\frac14\{\xi,\xi\}_{\Sigma}\bigg) \coh_{\tau+\xi}.
 \end{equation*}
 Substituting ordinary coherent states with normalized coherent states with the respective normalization factors, recall~(\ref{eq:ipcohk}), leads to~(\ref{eq:weylactku}).

 If $\xi\in L_{\Sigma}^{-}$, then the third and fourth term in the exponential on the right-hand side of~(\ref{eq:weylactk}) vanish. This yields,
 \begin{equation*}
 \big(\hat{F} \coh_\tau\big)(\phi)
 =\exp\bigg(\frac12\{\tau,\phi\}_{\Sigma}+\frac12\{\xi,\phi\}_{\Sigma}\bigg)
 =\exp\bigg(\frac12\big\{\tau+\xi+\overline{\xi},\phi\big\}_{\Sigma}\bigg).
 \end{equation*}
 This is the desired expression~(\ref{eq:weylactkc}) in terms of wave functions.
\end{proof}
The last part of the Corollary, manifest in relation~(\ref{eq:weylactkc}), can be generalized considerably as follows. This follows from the fact that observables holomorphic with respect to $J_{\Sigma}$ are generated by Weyl observables holomorphic with respect to $J_{\Sigma}$.
\begin{prop}
\label{prop:holomactk}
 Suppose the slice observable $F$ is holomorphic with respect to $J_{\Sigma}$
$($invariant under translation in $L_{\Sigma}^-$, i.e., $F(\phi)=F(\phi+\xi)$ for all $\phi\in L_{\Sigma}^\bC$ and $\xi\in L_{\Sigma}^-)$. Then, $\hat{F}$ acts by multiplication of wave functions,
 \begin{equation}
 \big(\hat{F}\psi\big)(\phi)=F(\phi) \psi(\phi).
 \label{eq:holomactk}
 \end{equation}
\end{prop}
\begin{proof}
 We assume at first that $F$ is a Weyl observable with $\xi\in L_{\Sigma}^{-}$. Then, we can rewrite expression~(\ref{eq:weylactkc}) as,
 \begin{equation*}
 \big(\hat{F}\coh_{\tau}\big)(\phi)=\exp\bigg(\frac12\{\xi,\phi\}_{\Sigma}\bigg)\coh_{\tau}(\phi) = \exp\left(2\im\omega_{\Sigma}(\xi,\phi)\right)\coh_{\tau}(\phi) = F(\phi)\coh_{\tau}(\phi).
 \end{equation*}
 Since the coherent states span a dense subspace, the corresponding equality extends to all states. We may then extend the obtained equality~(\ref{eq:holomactk}) to more general $J_{\Sigma}$-holomorphic observables by derivative methods analogous to those exhibited in~Section~\ref{sec:vcgenobs}. On the other hand, we may take advantage of $J_{\Sigma}$-holomorphicity to treat the observable $F$ as if it was a wave function. (This requires square-integrability of $F$ with respect to the measure $\nu_\Sigma$, recall Section~\ref{sec:stdstate}.) In particular, using the completeness relation~(\ref{eq:cohcompl}) for coherent states and the reproducing property~(\ref{eq:reprod}) we obtain an integral representation of $F$ in terms of Weyl observables. To this end we shall denote the Weyl observable determined by $\xi\in L_{\Sigma}^{\bC}$ by $G_\xi$. We have for $\phi\in L_{\Sigma}$,
 \begin{align*}
 F(\phi)&=\langle \coh_\phi, F\rangle_{\Sigma}
 =\int_{\hat{L}_{\Sigma}} \langle \coh_\phi, \coh_\xi\rangle_{\Sigma}
 \langle\coh_\xi,F\rangle_{\Sigma}\,\xd\nu_{\Sigma}(\xi)
 =\int_{\hat{L}_{\Sigma}} \exp\bigg(\frac12\{\xi,\phi\}_{\Sigma}\bigg)
 F(\xi)\,\xd\nu_{\Sigma}(\xi)
 \\
 &=\int_{\hat{L}_{\Sigma}} \exp\big(2\im\omega_{\Sigma}(P^{-}(\xi),\phi)\big)
 F(\xi)\,\xd\nu_{\Sigma}(\xi)
 =\int_{\hat{L}_{\Sigma}} G_{P^{-}(\xi)}(\phi) F(\xi)\,\xd\nu_{\Sigma}(\xi).
 \end{align*}
 Inserting this integral representation in the quantization map, yields, in terms of wave functions
 \begin{align*}
 \big(\hat{F}\psi\big)(\phi)&=\int_{\hat{L}_{\Sigma}} \left(\hat{G}_{P^{-}(\xi)}\psi\right)(\phi)
 F(\xi)\,\xd\nu_{\Sigma}(\xi)
 =\int_{\hat{L}_{\Sigma}} G_{P^{-}(\xi)}(\phi)\psi(\phi)
 F(\xi)\,\xd\nu_{\Sigma}(\xi)
 \\
 &=F(\phi)\psi(\phi). \tag*{\qed}
\end{align*}
\renewcommand{\qed}{}
\end{proof}

On the other hand, we may specialize to the case that we act on the vacuum, yielding another Corollary of Proposition~\ref{eq:weylactk}.
\begin{cor}
 Let $\nu\defeq P^-(\xi)+\overline{P^-(\xi)}$. Then,
\label{cor:weylactvk}
 \begin{equation}
 \hat{F} \coh_0
 =\exp\left(\im\omega_{\Sigma}(\xi,\nu)\right) \ncoh_{\nu}.
 \label{eq:weylactvk}
 \end{equation}
\end{cor}
\begin{proof}
 If $\tau=0$ the first and third term on the right-hand side of~(\ref{eq:weylactk}) vanish. On the other hand, we notice, $\{\xi,\phi\}_{\Sigma}=\{P^{-}(\xi),\phi\}_{\Sigma}=\big\{P^{-}(\xi) +\overline{P^{-}(\xi)},\phi\big\}_{\Sigma}=\{\nu,\phi\}_{\Sigma}$. This yields
 \begin{align*}
 \hat{F} \coh_0&=\exp\bigg({-}\frac14\{\xi,\xi\}_{\Sigma}\bigg) \coh_{\nu}
 =\exp\bigg({-}\frac14\{\xi,\xi\}_{\Sigma}+\frac14\{\nu,\nu\}_{\Sigma}\bigg) \ncoh_{\nu}
 \\
 &=\exp\bigg({-}\frac14\{\nu,\xi\}_{\Sigma}+\frac14\{\xi,\nu\}_{\Sigma}\bigg) \ncoh_{\nu}.
 \end{align*}
 The exponential term can be rewritten to yield the desired equality~(\ref{eq:weylactvk}).
\end{proof}

We also consider the action of linear slice observables, here $D$.
\begin{prop}
\label{prop:linactk}
 Let $\xi^{\rm c}\defeq P^-(\xi) + \overline{P^-(\xi)}$ and $\xi^{\rm a}\defeq P^+(\xi) + \overline{P^+(\xi)}$. The operator $\hat{D}$ on $\cH_{\Sigma}$ acts as
 \begin{equation}
 \hat{D}=\frac{\im}{\sqrt{2}} \big(a_{\xi^{\rm a}} - a^\dagger_{\xi^{\rm c}}\big).
 \label{eq:linactk}
 \end{equation}
\end{prop}
\begin{proof}
 Taking a derivative of expression~(\ref{eq:weylactk}) as in relation~(\ref{eq:polydweyl}), we obtain
 \begin{align*}
 \big(\hat{D} \coh_\tau\big)(\phi)
 & =(-\im)\frac{\partial}{\partial\lambda}\exp\bigg(\frac12\{\tau,\phi\}_{\Sigma} +\frac12\lambda\{\xi,\phi\}_{\Sigma}-\frac12\lambda\{\tau,\xi\}_{\Sigma} -\frac14\lambda^2\{\xi,\xi\}_{\Sigma}\bigg) \bigg|_{\lambda=0} \nonumber
 \\
 & = -\im \exp\bigg(\frac12\{\tau,\phi\}_{\Sigma}\bigg)
 \bigg(\frac12\{\xi,\phi\}_{\Sigma}-\frac12\{\tau,\xi\}_{\Sigma}\bigg)
 \nonumber
 \\
 & = \frac{\im}{\sqrt2}
 \bigg(\frac{1}{\sqrt2}\{\tau,\xi\}_{\Sigma}-\frac{1}{\sqrt2}\{\xi,\phi\}_{\Sigma}\bigg)
 \exp\bigg(\frac12\{\tau,\phi\}_{\Sigma}\bigg) \nonumber
 \\
 & = \frac{\im}{\sqrt2}
 \bigg(\frac{1}{\sqrt2}\{\tau,\xi^{\rm a}\}_{\Sigma} -\frac{1}{\sqrt2}\{\xi^{\rm c},\phi\}_{\Sigma}\bigg)
 \exp\bigg(\frac12\{\tau,\phi\}_{\Sigma}\bigg) \nonumber
 \\
 & = \frac{\im}{\sqrt2}
 \big(a_{\xi^{\rm a}}-a^\dagger_{\xi^{\rm c}}\big) \coh_\tau(\phi).
 \end{align*}
 Here we have used expressions~(\ref{eq:coact}) and~(\ref{eq:anncoh}) for the action of creation and annihilation operators. It remains to observe that by linearity and by denseness of the coherent states this equality applies to all states, yielding~(\ref{eq:linactk}).
\end{proof}

It is useful to define linear slice observables that upon quantization turn into a given creation or annihilation operator. For $\xi\in L_{\Sigma}$ we define the slice observables $A^\dagger_\xi$ and $A_\xi$ on $\Sigma$ via
\begin{equation*}
 {A_\xi^\dagger}'(\phi)\defeq \frac{1}{\sqrt{2}} \{\xi,\phi\}_{\Sigma},\qquad
 A_\xi'(\phi)\defeq \frac{1}{\sqrt{2}} \{\phi,\xi\}_{\Sigma}.
\end{equation*}
With Proposition~\ref{prop:linactk} this yields
\begin{equation*}
 \hat{A}^\dagger_\xi = a^\dagger_\xi,\qquad \hat{A}_\xi = a_\xi.
\end{equation*}
In the example of Klein--Gordon theory in Minkowski space on an equal-time hypersurface $\Sigma_t$ (recall Section~\ref{sec:kgstates}), the observables $A_p^\dagger$ and $A_p$ yielding the creation and annihilation operators~$a_p^\dagger$ and $a_p$ for the momentum modes thus take the form
\begin{equation*}
 {A_p^\dagger}'(\phi)= \overline{\phi^{{\rm b}}(p)},\qquad
 A_p'(\phi)= \phi^{\rm a}(p).
\end{equation*}

We proceed to consider normal ordered quantization as defined by formula~(\ref{eq:defnoqk}) for the case of slice observables. In particular, this gives justification to the name ``normal ordered'' as we recover the usual normal ordering of operators.
\begin{prop}
\label{prop:weylnactk}
 The operator $\no{\hat{F}}$ acts on coherent states as
 \begin{equation}
 \big(\no{\hat{F}} \coh_\tau\big)(\phi)
 =\exp\bigg(\frac12\{\tau,\phi\}_{\Sigma}+\frac12\{\xi,\phi\}_{\Sigma}-\frac12\{\tau,\xi\}_{\Sigma}\bigg).
 \label{eq:weylnactk}
 \end{equation}
\end{prop}
\begin{proof}
 With the reproducing property~(\ref{eq:reprod}), the normal ordered version of relation~(\ref{eq:opsobs}), and the definition~(\ref{eq:defnoqk}) we have
 \begin{equation*}
 \big(\no{\hat{F}} \coh_\tau\big)(\phi)
 =\big\langle \coh_\phi, \no{\hat{F}}\coh_\tau\big\rangle_{\Sigma}
 =\rho_{\hat{\Sigma}}^{\no{F}}(\coh_\tau\tens\coh_\phi)
 =\langle \coh_\phi, \coh_\tau\rangle_{\Sigma} F\big(\hat{(\phi,\tau)}\big).
 \end{equation*}
 We now recall that we have evaluated the two factors on the right-hand side already in the proof of Proposition~\ref{prop:weylactk}, see equations~(\ref{eq:ppw1}) and~(\ref{eq:ppw2}). This recovers the result~(\ref{eq:weylnactk}).
\end{proof}
\begin{prop}
 Let $\xi_1,\dots,\xi_n,\eta_1,\dots,\eta_m\in L_{\Sigma}$. Let $E=A_{\xi_1}\cdots A_{\xi_n}\cdot A_{\eta_1}^\dagger\cdots A_{\eta_m}^\dagger$. Then,
 \begin{equation*}
 \no{\hat{E}}= a_{\eta_1}^\dagger\cdots a_{\eta_m}^\dagger a_{\xi_1}\cdots a_{\xi_n}.
 \end{equation*}
\end{prop}
\begin{proof}
 We proceed as at the beginning of Section~\ref{sec:vcgenobs} to deal with a product of linear obser\-vables. Thus, define
 \begin{gather*}
 E_{\lambda_1,\ldots,\lambda_n,\mu_1,\ldots,\mu_m}
 \defeq \lambda_1 A_{\xi_1}+\cdots+\lambda_n A_{\xi_n}+ \mu_1 A_{\eta_1}^\dagger\cdots \mu_m A_{\eta_m}, \\
 F_{\lambda_1,\ldots,\lambda_n,\mu_1,\ldots,\mu_m}\defeq
 \exp\big(\im\, E_{\lambda_1,\ldots,\lambda_n,\mu_1,\ldots,\mu_m}\big).
 \end{gather*}
 Then, as in relation~(\ref{eq:polydweyl}),
 \begin{equation}
 E=(-\im)^{n+m}\frac{\partial}{\partial \lambda_1}\cdots\frac{\partial}{\partial \lambda_n}\frac{\partial}{\partial \mu_1}\cdots\frac{\partial}{\partial \mu_m} F_{\lambda_1,\ldots,\lambda_n,\mu_1,\ldots,\mu_m} \bigg|_{\lambda_1,\ldots,\lambda_n,\mu_1,\ldots,\mu_m=0}.
 \label{eq:pwn1}
 \end{equation}
 If $\xi_{\lambda_1,\ldots,\lambda_n,\mu_1,\ldots,\mu_m}\in L_{\Sigma}^\bC$ is related to $E_{\lambda_1,\ldots,\lambda_n,\mu_1,\ldots,\mu_m}$ as in equation~(\ref{eq:lsobsdef}), then,
 \begin{equation*}
 \xi_{\lambda_1,\ldots,\lambda_n,\mu_1,\ldots,\mu_m}=\sqrt{2}\im \big({-}\lambda_1\xi_1^+-\cdots-\lambda_n\xi_n^++\mu_1\eta_1^-+\cdots+\mu_m\eta_m^-\big).
 \end{equation*}
 With Proposition~{\ref{prop:weylnactk}} we obtain for the action on the coherent state $\coh_\tau$,
 \begin{gather*}
 \big(\no{\hat{F}_{\lambda_1,\ldots,\lambda_n,\mu_1,\ldots,\mu_m}} \coh_\tau\big)(\phi)
 \\ \qquad
 {}=\exp\bigg(\frac12\{\tau,\phi\}_{\Sigma}+\frac{\im}{\sqrt{2}} \{\mu_1\eta_1+\cdots+\mu_m\eta_m,\phi\}_{\Sigma}+\frac{\im}{\sqrt{2}} \{\tau,\lambda_1\xi_1+\cdots+\lambda_n\xi_n\}_{\Sigma}\bigg).
 \end{gather*}
 With the relation~(\ref{eq:pwn1}) we obtain
 \begin{align*}
 \big(\no{\hat{E}} \coh_\tau\big)(\phi)
 &=\frac{1}{\sqrt{2}^{n+m}}\{\eta_1,\phi\}_{\Sigma}\cdots \{\eta_m,\phi\}_{\Sigma}
 \{\tau,\xi_1\}_{\Sigma}\cdots \{\tau,\xi_n\}_{\Sigma} \exp\bigg(\frac12\{\tau,\phi\}_{\Sigma}\bigg)
 \\
 &=\big(a_{\eta_1}^\dagger\cdots a_{\eta_m}^\dagger a_{\xi_1}\cdots a_{\xi_n}\coh_\tau\big)(\phi).
 \end{align*}
 In the last equality we have used the formulas~(\ref{eq:coact}) and~(\ref{eq:anncoh}) for the action of creation and annihilation operators.
\end{proof}

{\samepage
Finally, we consider the semiclassical interpretation of coherent states and correlation functions for slice observable. This can be seen as a limiting case of the case of general regions treated in~Section~\ref{sec:scko}. Thus, we fix a normalized coherent state $\ncoh_\tau$ and consider the expectation value of the Weyl slice observable $F$ in this state. With relation~(\ref{eq:opsobs}) and the Structure Theorem~\ref{thm:stdcorrfact} we get
\begin{align*}
 \big\langle \ncoh_\tau, \hat{F} \ncoh_\tau\big\rangle_{\Sigma}
 & =\rho_{\hat{\Sigma}}^F\big(\ncoh_\tau\tens\ncoh_\tau\big)
 =\rho_{\hat{\Sigma}}(\ncoh_\tau\tens\ncoh_\tau) F(\hat{(\tau,\tau)})
 \rho_{\hat{\Sigma}}^F(\coh_0\tens \coh_0)
 \\
 & =\langle \ncoh_\tau, \ncoh_\tau\rangle_{\Sigma} F((\tau,\tau))
 \exp\bigg({-}\frac{1}{4}\{\xi,\xi\}_{\Sigma}\bigg)
 =F'(\tau) \exp\bigg({-}\frac{1}{4}\{\xi,\xi\}_{\Sigma}\bigg).
\end{align*}

}
So, up to the vacuum expectation value, this is precisely the value of the classical slice observable~$F$ on the classical solution $\tau$. What is more, moving to normal ordered quantization the vacuum expectation value is removed, compare equation~(\ref{eq:defnoqk}). Thus, we get for arbitrary observables~$F$
\begin{equation*}
 \big\langle \ncoh_\tau, \no{\hat{F}} \ncoh_\tau\big\rangle_{\Sigma} =F'(\tau).
\end{equation*}

\subsection[$*$-structure and GNS construction]{$\boldsymbol\sst$-structure and GNS construction}
\label{sec:ssgns}

It turns out that the same representation of the slice observables on a Hilbert space can be obtained in a rather different manner. Namely, we may start with the quantum algebra of slice observables and construct a Hilbert space representation via the GNS construction. To this end we need the $\sst$-structure on the algebra and a positive $\sst$-functional corresponding to the vacuum correlation function. We detail this in the present section. Fix a hypersurface $\Sigma$.

The vacuum correlation functions of slice observables in the sense of Section~\ref{sec:cgenvac} can be viewed as defining a linear functional $v_{\Sigma}\colon \qsoa_{\Sigma}\to\bC$ on the algebra of slice observables,
\begin{equation}
 v_{\Sigma}(F)\defeq \rho_{\hat{\Sigma}}^F\big(W^\pol\big).
 \label{eq:vacf}
\end{equation}
For Weyl slice observables we have evaluated this in~Section~\ref{sec:vevweylslice} leading to equation~(\ref{eq:vevweylk}). This is in the context where $D'\colon L_{\Sigma}^\bC\to\bC$ defines a linear slice observable determined by an element $\xi\in L_{\Sigma}^\bC$ via formula~(\ref{eq:lsobsdef}), compare Figure~\ref{fig:slice-obs}. $F=\exp(\im D)$ denotes the induced Weyl slice observable. Thus, in our present notation,
\begin{equation}
 v_{\Sigma}(F) =\exp\bigg({-}\frac{1}{4} \{\xi,\xi\}_{\Sigma}\bigg).
 \label{eq:vevweylf}
\end{equation}
Note that the expression $\{\xi,\xi\}_{\Sigma}$ is not necessarily positive or even real for general complex $\xi\in L_{\Sigma}^\bC$. It is strictly positive, however, for non-vanishing real $\xi\in L_{\Sigma}$. The coincidence of the present formula with the results of Section~\ref{sec:qsobsopk} follows here from the fact that the relevant formula~(\ref{eq:sfqvev}) of Theorem~\ref{thm:stdcorrfact} is simply identical to formula~(\ref{eq:piwobs}) underlying equation~(\ref{eq:vevweylf}).

It is easy to see that $v_{\Sigma}$ is compatible with the $\sst$-structure of $\qsoa_{\Sigma}$, i.e., $v_{\Sigma}(F^\sst)=\overline{v_{\Sigma}(F)}$. Similarly, $v_{\Sigma}$ defines in fact a positive functional. We limit ourselves here to remark that the linear observable $D$ is self-adjoint with respect to the $\sst$-structure, $D^\sst=D$, precisely if the corresponding element $\xi\in L_{\Sigma}^\bC$ is real, i.e., $\xi\in L_{\Sigma}$.\footnote{We do not include an independent proof of the positivity of $v_{\Sigma}$ here, because we know already from Section~\ref{sec:qsobsopk} that it is positive, since $v_{\Sigma}(F^\sst \qp F)=\rho_{\hat{\Sigma}}^{F^{\sst}\cp F}(\coh_0\tens\coh_0)=\langle \hat{F}\coh_0,\hat{F}\coh_0\rangle_{\Sigma}\ge 0$ for arbitrary $F$.}
Following the GNS construction we define a hermitian sesquilinear form $[\cdot,\cdot]_{\Sigma}$ on $\qsoa_{\Sigma}$ via
\begin{equation}
 [G,F]_{\Sigma}\defeq v_{\Sigma}(G^\sst \qp F).
 \label{eq:sokip}
\end{equation}
For Weyl slice observables $F=\exp(\im D)$, $G=\exp(\im E)$ determined by linear observables $D$ and~$E$ which in turn are determined by elements $\xi_D,\xi_E\in L_{\Sigma}^\bC$ via equation~(\ref{eq:lsobsdef}) as in~Section~\ref{sec:qsobs}, we~obtain
\begin{align}
 [G,F]_{\Sigma}&=v_{\Sigma}\big(\exp\big(\im\omega_{\Sigma}\big({-}\overline{\xi_E},\xi_D\big)\big) G^{\sst}\cdot F\big) \nonumber
 \\
 &=\exp\big(\im\omega_{\Sigma}\big({-}\overline{\xi_E},\xi_D\big)\big)
 \exp\bigg({-}\frac{1}{4}\big\{\xi_D-\overline{\xi_E},\xi_D-\overline{\xi_E}\big\}_{\Sigma}\bigg) \nonumber
 \\
 &=\exp\bigg({-}\frac{1}{4}\{\xi_D,\xi_D\}_{\Sigma} -\frac{1}{4}\big\{\overline{\xi_E},\overline{\xi_E}\big\}_{\Sigma} +\frac{1}{2}\big\{\xi_D,\overline{\xi_E}\big\}_{\Sigma}\bigg).
 \label{eq:ipgns}
\end{align}
Note that $\{\phi,\eta\}=0$ for any $\eta\in L_{\Sigma}^\bC$ iff $\phi\in L_{\Sigma}^+$.
From this we may deduce that the left ideal $\qsoi_{\Sigma}\subseteq\qsoa_{\Sigma}$ on which the sesquilinear form vanishes is generated by the relation $F\sim \one$ for the Weyl observables $F=\exp(\im D)$ with $\xi_D\in L_{\Sigma}^+$. According to the GNS construction the Hilbert space on which $\qsoa_{\Sigma}$ will be represented, may now be obtained as the completion of the quo\-tient~$\qsoa_{\Sigma}/\qsoi_{\Sigma}$. This Hilbert space turns out to be naturally isomorphic to $\cH_{\Sigma}$. Indeed, for real $\xi\in L_{\Sigma}$ we may identify the normalized coherent state $\ncoh_\xi \in\cH_{\Sigma}$ with the (equivalence class of the) Weyl observable $F=\exp(\im D)$, where the linear observable $D$ is determined by $\xi$ via equation~(\ref{eq:lsobsdef}). Doing so, we realize that~(\ref{eq:ipgns}) precisely recovers the (normalized version) of the inner product of coherent states~(\ref{eq:ipcohk}). What is more, by comparing relation~(\ref{eq:weylrel}) with equation~(\ref{eq:weylactku}) one may appreciate that the induced action of $\qsoa_{\Sigma}$ on $\cH_{\Sigma}$ coincides precisely with the one given by~Proposition~\ref{prop:weylactk}.

\subsection{Boundary observables as states}
\label{sec:bobss}

Consider again the correlation function of a Weyl observable $F$ in a region $M$ with a normalized coherent state $\ncoh_\xi$ on the boundary, compare Section~\ref{sec:corrfack} and Figure~\ref{fig:boundary-cohs-corr}. As we have learned in~Section~\ref{sec:qsobsopk}, the state $\ncoh_\xi$ can be obtained from the vacuum by acting with a Weyl slice observable. In the present case this slice observable will be associated with the slice region~$\hat{\partial M}$ determined by the boundary $\partial M$ of $M$. We call slice observables with the slice region determined by a~boundary also \emph{boundary observables}. Here, let $\xi\in L_{\partial M}$ define the linear boundary observable~$E$ via $E'(\phi)=2\omega_{\partial M}(\xi,\phi)$ and $G=\exp(\im E)$ be the corresponding Weyl observable. Then, by~equation~(\ref{eq:weylactku}), we have $\hat{G} \coh_0= \ncoh_\xi$. This implies
\begin{align*}
 \rho_M^F(\ncoh_\xi)&=
 \rho_M^F\big(\hat{G} \coh_0\big)
 =\sum_{k\in I} \big\langle\zeta_k,\hat{G} \coh_0\big\rangle_{\partial M}
 \rho_M^F(\zeta_k)
 =\sum_{k\in I} \rho_{\hat{\partial M}}^G(\coh_0\tens\iota_{\partial M}(\zeta_k))
 \rho_M^F(\zeta_k)
 \\
 &=\rho_M^{G\cp F}(\coh_0).
\end{align*}
Here, the third equality arises from equation~(\ref{eq:opsobs}) while the forth is the composition rule in~the form~(\ref{eq:dualcompo}), where $c=1$. In other words, the correlation function for a Weyl observable~$F$ on a normalized coherent state $\ncoh_\xi$ is the same as the correlation function of the product Weyl observable $G\cp F$ on the vacuum state. Compare Figure~\ref{fig:boundary-cohs-corr} to Figure~\ref{fig:boundary-obs-corr}. With relation~(\ref{eq:vaccorrnot}), we may write this as
\begin{equation}
 \rho_M^F(\ncoh_\xi)=\rho_M^{G\cp F}\big(W^\pol\big).
 \label{eq:cohsbobs}
\end{equation}
Since $F$ may be an arbitrary Weyl observable, this is valid for any observable $F$.
\begin{figure}
 \centering
 \begin{tikzpicture}[scale=2.15]
\draw [red,thick,yshift=0.5cm,fill=white,scale=1.04] plot [smooth cycle] coordinates {(-1,-0.3) (0,-0.7) (1.6,0.22) (2,1.2)(1.1,2)(0,1.6)(-1.1,0.3)};
\draw [thick,yshift=0.5cm,fill=red!30] plot [smooth cycle] coordinates {(-1,-0.3) (0,-0.7) (1.6,0.22) (2,1.2)(1.1,2)(0,1.6)(-1.1,0.3)};
\draw [dashed, red,thick,fill=gray!10,yshift=0.5cm,scale=0.93] plot [smooth cycle] coordinates {(-1,-0.3) (0,-0.7) (1.6,0.22) (2,1.2)(1.1,2)(0,1.6)(-1.1,0.3)};
\draw [thick,xshift=0cm,fill=gray!30] plot [smooth cycle] coordinates {(0,0.3) (1,0.5) (0.8,1.2) (0.2,1.4)  (-0.2,0.8) (-0.1,0.55)};
\node at (0.4,0.8) {$F$};
\node at (1,1.8) {$M$};
\draw (2.3,2.2) node [right] {$\partial M$} edge[out=180,in=30,->] (1.8,2.02);
\draw[red] (-1,2.7) node [right] {${\color{red}K_{0}}$} edge[out=0,in=130,->] (0.5,2.48);
\draw (2,2.9) node [right] {$G=\exp(\im E)$} edge[out=180,in=80,->] (1.3,2.3);
\draw[very thick,->] (-1.1,0) node [below left] {$\xi$}-- (-0.8,0.3);
\end{tikzpicture}
 \caption{Spacetime region $M$ containing an observable $F$ with Weyl slice observable $G$ determined by $\xi$ on the boundary.}
 \label{fig:boundary-obs-corr}
\end{figure}
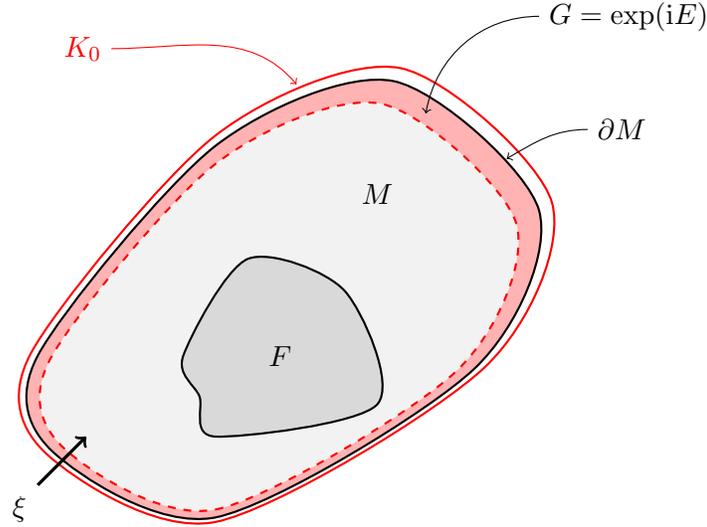

What formula~(\ref{eq:cohsbobs}) implies is that the role played by a normalized coherent state (here~$\ncoh_\xi$) can be completely captured through a corresponding boundary Weyl observable (here~$G$). Moreover, since general states can be constructed from coherent states via derivatives, recall equation~(\ref{eq:msderivcohs}), the role played by any state can be captured through a suitable boundary observable in this way. Remarkably, the right-hand side of equation~(\ref{eq:cohsbobs}) does not make any use of a~Hilbert space of states. In contrast to the left-hand side, it is perfectly well-defined even if the polarization on the boundary $\partial M$ is not a K\"ahler polarization.
This suggests an~approach to quantization where the role of states is played by slice observables and no Hilbert space of states need to be constructed. What is more, this notion of quantization would work for non-K\"ahler polarizations as well. This is the subject of the following section.

\section{Quantization without Hilbert spaces}
\label{sec:genquant}

In this section we develop an approach to quantization based on slice observables instead of~states and applicable to any polarization, as suggested in~Section~\ref{sec:bobss}.

\subsection{Correlation functions with boundary observables}
\label{sec:sobsampl}

Consider a region $M$ and a polarization $L_{\partial M}^\pol \subseteq L_{\partial M}^\bC$ transverse to the polarization $L_{\tilde{M}}^\bC\subseteq L_{\partial M}^\bC$, that is $L_{\partial M}^\bC=L_{\tilde{M}}^\bC\oplus L_{\partial M}^\pol$. For an element $\tau\in L_{\partial M}^\bC$ we denote its unique decomposition by~$\tau=\tau^\ipol+\tau^\epol$, where $\tau^\ipol \in L_{\tilde{M}}^{\bC}$ and $\tau^\epol\in L_{\partial M}^\pol$.
Let $F\colon K_M^\bC\to \bC$ denote an observable in the interior of $M$. Additionally, consider a slice observable $G\colon K_{\hat{\partial M}}^\bC\to\bC$ on the boundary~$\partial M$ of~$M$. In this context we introduce the notation
\begin{equation*}
 \rho_M^F\big(W^\pol_G\big)\defeq \rho_M^{G\cp F}\big(W^\pol\big),
\end{equation*}
where the right-hand side is defined by equation~(\ref{eq:piwobs}).
This suggests that $W^\pol_G$ represents a~``state'' determined by the boundary observable $G$. For the moment, however, this is merely a notation, and we do not take $W^\pol_G$ (or similarly the ``vacuum state'' $W^\pol$) itself to stand for any mathematical object.\footnote{One might of course simply think of $W^\pol_G$ as the pair $\big(L_{\partial M}^\pol, G\big)$, i.e., as an element of the Cartesian product $\pols_{\partial M}\times\qsoa_{\partial M}$, where $\pols_{\partial M}$ is the space of polarizations on $\partial M$, i.e., Lagrangian subspaces of $L_{\partial M}^\bC$.} However, since choosing a polarization and a slice observable also makes sense on an oriented hypersurface that is not necessarily the boundary of a region, we~occa\-sionally use the notation in this more general context.

Suppose $G$ is a Weyl slice observable $G=\exp(\im\, E)$ with the linear slice observable $E$ deter\-mined by $\xi\in L_{\partial M}^\bC$ according to equations~(\ref{eq:sobseval}) and~(\ref{eq:lsobsdef}). This is illustrated in Figure~\ref{fig:boundary-obs-corr} (disregard $\coh_0$ there). We also use the more specific notation
\begin{equation}
 \rho_M^F\big(K^\pol_\xi\big)\defeq \rho_M^{G\cp F}\big(W^\pol\big).
 \label{eq:defcohsg}
\end{equation}
In terms of our symbolic notation we may write $K^\pol_\xi\defeq W^\pol_G$. Note also $K^\pol_0=W^\pol$. This of course suggests a coherent state. Indeed, this is motivated by the identity~(\ref{eq:cohsbobs}), valid in the case of a K\"ahler polarization and when $\xi\in L_{\partial M}$ is real. That is,
\begin{equation}
 \rho_M^F\big(K^\pol_\xi\big)=\rho_M^F(\ncoh_\xi).
 \label{eq:rencohs}
\end{equation}
The coherent state here is normalized. If on the other hand (still in the K\"ahler case) we replace~$\xi$ with the projection $P^-(\xi)$ we obtain instead the ordinary (non-normalized) coherent state, compare Corollary~\ref{cor:weylactsk},
\begin{equation}
 \rho_M^F\big(K^\pol_{P^-(\xi)}\big)=\rho_M^F(\coh_\xi).
 \label{eq:recohs}
\end{equation}
Since this is valid for arbitrary observables $F$ and even arbitrary regions $M$, we write equations~(\ref{eq:rencohs}) and~(\ref{eq:recohs}) informally as $K^\pol_\xi=\ncoh_\xi$ and $K^\pol_{P^-(\xi)}=\coh_\xi$, and extend them thus to hypersurfaces that are not necessarily boundaries of regions. Note, however, that this makes sense only with a fixed choice of hypersurface orientation and K\"ahler polarization.

Finally, formula~(\ref{eq:mswf}) for the wave function of an $n$-particle state in K\"ahler quantization suggests the following definition. Let $\xi_1,\ldots,\xi_n\in L_{\partial M}^\bC$. Define the boundary observable $G$ by
\begin{equation}
 G'(\phi)=\prod_{k=1}^n 4\im\omega_{\partial M}(\xi_k,\phi).
 \label{eq:prodobsxi}
\end{equation}
Then set
\begin{equation}
 Q^\pol_{\xi_1,\dots,\xi_n}\defeq W^\pol_G.
 \label{eq:defmpso}
\end{equation}
In the case of a K\"ahler polarization we may take $\xi_1,\ldots,\xi_n\in L_{\partial M}$ and note $4\im\omega_{\partial M}(P^-(\xi_k),\phi)=\{\xi_k,\phi\}_{\partial M}$ so that by Proposition~\ref{prop:holomactk} we obtain (in our informal notation)
\begin{equation}
 Q^\pol_{P^-(\xi_1),\dots,P^-(\xi_n)}=\psi,
 \label{eq:remps}
\end{equation}
where $\psi$ is the $n$-particle state with wave function~(\ref{eq:mswf}). We also note the general relation
\begin{equation}
 Q^\pol_{\xi_1,\dots,\xi_n}= \bigg(\prod_{k=1}^n 2 \frac{\partial}{\partial \lambda_k}\bigg) K^\pol_{\lambda_1 \xi_1+\cdots+\lambda_n \xi_n} \bigg|_{\lambda_1,\dots,\lambda_n=0}.
 \label{eq:mpsweyl}
\end{equation}
If we replace $\xi_1,\ldots,\xi_n$ by $P^-(\xi_1),\ldots,P^-(\xi_n)$ in the K\"ahler case we recover expression~(\ref{eq:msderivcohs}).

We briefly consider momentum states in the standard example of Klein--Gordon theory on an equal-time hypersurface in Minkowski space. Comparing~(\ref{eq:remps}) with~(\ref{eq:mms}) we have
\begin{equation}
 Q^\pol_{P^-(\phi_{p_1}),\ldots,P^-(\phi_{p_n})}=\sqrt{2^n} \psi_{p_1,\ldots,p_n}.
 \label{eq:mmso}
\end{equation}

Since Weyl observables generate all other observables, and similarly coherent states generate all other states we focus primarily on ``states'' of the form $K^\pol_{\xi}$, i.e., Weyl boundary observables.
If $F$ is also a Weyl observable, we can evaluate $\rho_M^F\big(K^\pol_\xi\big)$ explicitly by using formula~(\ref{eq:piwobs}). Thus, let $F=\exp(\im\, D)$ with $D\colon K_M^\bC\to\bC$ linear. Let $\phi\in A_M^{D+E,\bC}\cap L^\pol_{\partial M}$ be the unique solution of the modified equations of motion satisfying the boundary condition on $\partial M$. To obtain a more explicit form of the amplitude expression we need to know $\phi$ in different pieces of spacetime. In the interior of $M$ we write $\phi_{M}$ while in the exterior we write $\phi_{X}$, see Figure~\ref{fig:boundary-obs-eval}. 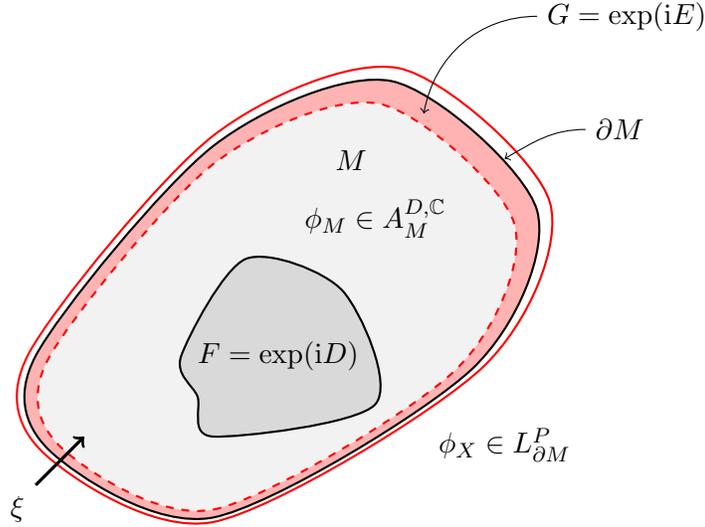
\begin{figure}
 \centering
 \begin{tikzpicture}[scale=2.15]
\draw [red,thick,yshift=0.5cm,fill=white,scale=1.04] plot [smooth cycle] coordinates {(-1,-0.3) (0,-0.7) (1.6,0.22) (2,1.2)(1.1,2)(0,1.6)(-1.1,0.3)};
\draw [thick,yshift=0.5cm,fill=red!30] plot [smooth cycle] coordinates {(-1,-0.3) (0,-0.7) (1.6,0.22) (2,1.2)(1.1,2)(0,1.6)(-1.1,0.3)};
\draw [dashed, red,thick,fill=gray!10,yshift=0.5cm,scale=0.93] plot [smooth cycle] coordinates {(-1,-0.3) (0,-0.7) (1.6,0.22) (2,1.2)(1.1,2)(0,1.6)(-1.1,0.3)};
\draw [thick,xshift=0cm,fill=gray!30] plot [smooth cycle] coordinates {(0,0.3) (1,0.5) (0.8,1.2) (0.2,1.4)  (-0.2,0.8) (-0.1,0.55)};
\node at (0.4,0.8) {$F=\exp (\im D)$};
\node at (0.85,2.01) {$M$};
\node at (1,1.65) {$\phi_M \in A_M^{D,\bC}$};
\node at (1.8,0.25) {$\phi_X \in L_{\partial M}^{P}$};
\draw (2.3,2.2) node [right] {$\partial M$} edge[out=180,in=30,->] (1.8,2.02);
\draw (2,2.9) node [right] {$G=\exp(\im E)$} edge[out=180,in=80,->] (1.3,2.3);
\draw[very thick,->] (-1.1,0) node [below left] {$\xi$}-- (-0.8,0.3);
\end{tikzpicture}
 \caption{Spacetime region $M$ containing a Weyl observable $F$ with Weyl slice observable $G$ determined by $\xi$ on the boundary. The interior solution is marked $\phi_M$ and the exterior one $\phi_X$.}
 \label{fig:boundary-obs-eval}
\end{figure}
From~(\ref{eq:sobsdiff}) we get, $\phi_{X}=\phi_{M}-\xi$. On the other hand we require $\phi_{X}\in L_{\partial M}^\pol$ and $\phi_{M}\in A_M^{D,\bC}$. To satisfy the latter requirement we can set $\phi_{M}=\eta+\Delta$, where $\Delta\in L_{M}^\bC$, while $\eta\in A_M^{D,\bC}\cap L^\pol_{\partial M}$ unique. We get $\phi_{X}=\eta+\Delta-\xi$. $\Delta$ is thus determined by the requirement $\Delta-\xi\in L^\pol_{\partial M}$. Decomposing $\xi=\xi^\ipol+\xi^\epol$ we get $\Delta=\xi^\ipol$.
We obtain
\begin{align}
 \rho_M^F\big(K^\pol_\xi\big) & = \rho_M^{G\cp F}\big(W^\pol\big)
 = \exp\bigg(\frac{\im}{2} (D+E)(\phi)\bigg)
 = \exp\bigg(\frac{\im}{2}D(\phi_{M})
 +\frac{\im}{2}E'\bigg(\frac{1}{2}(\phi_{M}+\phi_{X})\bigg)\bigg) \nonumber
\\
 & = \exp\bigg(\frac{\im}{2}D(\eta+\xi^\ipol)
 +\im\,\omega_{\partial M}\bigg(\xi,\eta+\xi^\ipol-\frac{1}{2}\xi\bigg)\bigg) \nonumber
 \\
 & = \exp\bigg(\frac{\im}{2}D(\eta)+\frac{\im}{2}D\big(\xi^\ipol\big)
 +\im\,\omega_{\partial M}\big(\xi,\eta+\xi^\ipol\big)\bigg) \nonumber
 \\
 & = \exp\bigg(\frac{\im}{2}D(\eta)+\frac{\im}{2}D\big(\xi^\ipol\big)+\frac{\im}{2}D\big(\xi^\ipol\big)
 +\im\,\omega_{\partial M}\big(\xi,\xi^\ipol\big)\bigg) \nonumber
 \\
 & = \exp\bigg(\frac{\im}{2}D(\eta)\bigg) F\big(\xi^\ipol\big)
 \exp\big(\im\,\omega_{\partial M}\big(\xi,\xi^\ipol\big)\big).
 \label{eq:sobsampl}
\end{align}

We can formulate the result analogous to the Factorization Theorem~\ref{thm:stdcorrfact} as follows.
\begin{thm}
 \label{thm:wfac}
 If $F$ is the Weyl observable $\exp(\im D)$, then,
\begin{gather}
 \rho_M^F\big(K^\pol_{\xi}\big) = \rho_M\big(K^\pol_{\xi}\big) F\big(\xi^\ipol\big) \rho_M^F\big(W^\pol\big),\quad\text{with} \label{eq:wobsfactg}
 \\[.5ex]
 \rho_M\big(K^\pol_{\xi}\big) =\exp\left(\im\,\omega_{\partial M}\left(\xi,\xi^\ipol\right)\right), \label{eq:amplg}
 \\
 \rho_M^F\big(W^\pol\big) =\exp\bigg(\frac{\im}{2}D(\eta)\bigg).
\end{gather}
Here $\eta\in A_M^{D,\bC}\cap L_{\partial M}^\pol$.
\end{thm}

We already know that in the K\"ahler case and for $\xi\in L_{\partial M}$ we must recover Theorem~\ref{thm:stdcorrfact}. Nevertheless, it is instructive to see this explicitly. Thus, we have to bring into exact correspondence the three factors on the right-hand side of~(\ref{eq:wobsfact}) with the three factors on the right-hand side of~(\ref{eq:wobsfactg}). We have already commented extensively on the equality of the third factor, see equations~(\ref{eq:cohsbobs}) and~(\ref{eq:defcohsg}). To see the equality for the two remaining factors decompose $\xi=\xi^{\rm R}+J_{\partial M}\xi^{\rm I}$ with $\xi^{\rm R},\xi^{\rm I}\in L_M^\bC$ as in~Section~\ref{sec:corrfack}. Then, it is easily verified that $\xi^\ipol=\xi^{\rm R}-\im \xi^{\rm I}=\hat{\xi}$. In particular, $F\big(\xi^\ipol\big)= F\big(\hat{\xi}\big)$, which shows equality for the second factor. For the remaining factor we have
\begin{align*}
 \rho_M\big(K^\pol_\xi\big) & =\exp\big(\im\,\omega_{\partial M}\big(\xi,\xi^\ipol\big)\big)
 = \exp\big(\im\,\omega_{\partial M}\big(J_{\partial M} \xi^{\rm I},\xi^\ipol\big)\big) \nonumber\\
 & = \exp\bigg({-}\frac{\im}{2}\,g_{\partial M}\big(\xi^{\rm I},\xi^\ipol\big)\bigg)
 = \exp\bigg({-}\frac{\im}{2}\,g_{\partial M}\big(\xi^{\rm I},\xi^{\rm R}-\im \xi^{\rm I}\big)\bigg)
 = \rho_M(\ncoh_{\xi}).
\end{align*}
As expected, we confirm the equality~(\ref{eq:rencohs}) for any Weyl observable $F$ and thus for any obser\-vable~$F$ (since the Weyl observables generate all observables).

Back to the general polarization setting, for later use we consider the special case that the slice observable $G$ (and thus also $E$) satisfies the following \emph{interior translation invariance} condition: $G'(\phi+\delta)=G'(\phi)$ for any $\phi\in L_{\partial M}^\bC$ and any $\delta\in L_{\tilde{M}}^\bC$. Here this is equivalent to $\xi\in L_{\tilde{M}}^\bC$ and thus $\xi=\xi^{\ipol}$. We can then rewrite~(\ref{eq:sobsampl}) as
\begin{equation*}
 \exp\bigg(\frac{\im}{2}D(\eta)\bigg) F\big(\xi^\ipol\big)
 =\exp\bigg(\frac{\im}{2}D(\eta)\bigg) \exp(2\im\,\omega(\xi,\eta))
 =\exp\bigg(\frac{\im}{2}D(\eta)\bigg) G'(\eta).
\end{equation*}
Linearity in $G$ of the right-hand expression shows that this generalizes to arbitrary slice observables $G$ that satisfy the interior translation invariance condition.
\begin{lem}
 \label{lem:sobsit}
 If $F$ is the Weyl observable $\exp(\im D)$, $\eta\in A_M^{D,\bC}\cap L_{\partial M}^\pol$, and $G$ is an interior translation invariant slice observable, then
\begin{equation}
 \rho_M^{G\cp F}\big(W^{\pol}\big)=\rho_M^F\big(W^\pol\big)\, G'(\eta).
 \label{eq:sobsit}
\end{equation}
\end{lem}

\subsection{Normal ordering and semiclassical interpretation}
\label{sec:nordsc}

We recall the generalized notion of normal ordering of the K\"ahler setting, given by the defining relation~(\ref{eq:defnoqk}) in~Section~\ref{sec:stdstate}. There is no difficulty in further generalizing this to the present setting with general polarizations. Thus, for $F\colon K_M^\bC\to\bC$ an observable in $M$ and $\xi\in L_{\partial M}^\bC$ we~define
\begin{equation}
 \rho_M^{\no{F}}\big(K_\xi^\pol\big)\defeq \rho_M\big(K_\xi^\pol\big) F\big(\xi^\ipol\big).
 \label{eq:defnoq}
\end{equation}
In the K\"ahler case take $\xi\in L_{\partial M}$ and note $(P^-(\xi))^\ipol=\xi^\ipol=\hat{\xi}$ to recover~(\ref{eq:defnoqk}).

For the special case of $F$ being a Weyl observable we obtain by comparison of definition~(\ref{eq:defnoq}) with the factorization identity~(\ref{eq:wobsfactg}) the simple relation
\begin{equation*}
 \rho_M^F\big(K_\xi^\pol\big)=\rho_M^{\no{F}}\big(K_\xi^\pol\big)\, \rho_M^F\big(W^\pol\big).
\end{equation*}
By linearity this relation extends to arbitrary states, (i.e., arbitrary boundary observables). To~emphasize this we chose a notation $X^\pol$ for the state, which really stands for $W^\pol_G$ with $G$ an arbitrary boundary observable,
\begin{equation}
 \rho_M^F\big(X^\pol\big)=\rho_M^{\no{F}}\big(X^\pol\big)\, \rho_M^F\big(W^\pol\big).
 \label{eq:Wickgen}
\end{equation}

In analogy to the coherent states in the K\"ahler quantization setting we can ask for a semiclassical interpretation of the Weyl slice observables represented by the notation $K^\pol_\xi$, compare Section~\ref{sec:sck}. As for the semiclassical interpretation of amplitudes (compare Section~\ref{sec:scka}), we do not have here an analogue of the tunneling interpretation afforded by formula~(\ref{eq:amplk}). Instead, with the corresponding formula~(\ref{eq:amplg}) we can merely affirm that the amplitude is unity if $\xi$ is an interior solution, i.e., if $\xi=\xi^\ipol$. If this is not the case, i.e., we consider a solution ``not allowed'' in the interior, there need not in general be a suppression. Since in the absence of an inner product we also do not have a notion of normalization of the ``state'' $K^\pol_\xi$ this need not be a cause for concern.

As for the semiclassical interpretation of correlation functions for $K^\pol_\xi$, the considerations of the K\"ahler case (Section~\ref{sec:scko}) largely still hold. In particular for $\xi=\xi^\ipol$ an interior solution, we obtain for Weyl observables the classical expectation value $F(\xi)$, see formula~(\ref{eq:wobsfactg}). Moreover, for normal ordered quantization this extends to arbitrary observables, compare formula~(\ref{eq:defnoq}). This is quite a strong indication that $K^\pol_\xi$ (or something like $K^\pol_{P^-(\xi)}$, see above) provides a good semiclassical description of the classical solution $\xi$ near the boundary $\partial M$, up to normalization.

\subsection{Changing the vacuum}
\label{sec:chvacso}

It is a natural question to consider what happens when the vacuum, i.e., the polarization deter\-mining the vacuum changes. More precisely, consider a region $M$ and
two polarizations $L_{\partial M}^\pol$ and $L_{\partial M}^{\pol'}$, both transverse to $L_{\tilde{M}}^\bC$. Is there then a ``state'' $Y^{\pol}_{\pol'}$ with respect to the unprimed vacuum that mimics the primed vacuum, i.e., such that all correlation functions coincide? That is, we want
\begin{equation*}
 \rho_M^F\big(W^{\pol'}\big)=\rho_M^F\big(Y^{\pol}_{\pol'}\big)
\end{equation*}
for all observables $F$ in $M$. As usual, it is sufficient to consider this for the special case of Weyl observables $F=\exp(\im D)$ with $D$ linear.

Let us write
\begin{equation}
 \rho_M^F\big(W^{\pol'}\big)=q\cdot \rho_M^F\big(W^{\pol}\big).
 \label{eq:relvacobs}
\end{equation}
Inserting the explicit formula~(\ref{eq:piwobs}) for the vacuum correlation functions yields
\begin{equation*}
 q=\exp\bigg(\frac{\im}{2} D(\eta'-\eta)\bigg).
\end{equation*}
Here, $\eta\in A_M^{D,\bC}\cap L_{\partial M}^\pol$ and $\eta'\in A_M^{D,\bC}\cap L_{\partial M}^{\pol'}$. Note that $\eta'-\eta\in L_M^\bC$, so we can use equation~(\ref{eq:obssol}) to get
\begin{equation*}
 q=\exp\big(\im\,\omega_{\partial M}(\eta'-\eta,\eta)\big)=\exp\big(\im\,\omega_{\partial M}(\eta',\eta)\big).
\end{equation*}
With the obvious notation we notice, $\eta=\eta^\epol$ and $\eta'=(\eta')^{\epol'}=(\eta+(\eta'-\eta))^{\epol'}=\eta^{\epol'}$. Define $G'\colon L_{\partial M}^\bC\to\bC$ by
\begin{equation}
 G'(\phi)=\exp\big(\im\,\omega_{\partial M}\big(\phi^{\epol'},\phi^{\epol}\big)\big).
\label{eq:chvacobs}
\end{equation}
Then, $q=G'(\eta)$. Consider the slice observable $G$ determined by $G'$. As is easily seen, $G$ is interior translation invariant and Lemma~\ref{lem:sobsit} applies. Comparison of the right-hand side of equation~(\ref{eq:relvacobs}) with equation~(\ref{eq:sobsit}) thus yields
\begin{equation*}
 Y^{\pol}_{\pol'}=W^{\pol}_G.
\end{equation*}

\begin{prop}
 For any observable $F$ in $M$ we have
 \begin{equation*}
 \rho_M^F\big(W^{\pol'}\big)=\rho_M^F\big(W^\pol_G\big).
 \end{equation*}
\end{prop}

By allowing the observable $F$ itself to be a product of a slice observable and a general observable we obtain an apparent generalization.

\begin{cor}
 For any slice observable $H$ in $\partial M$ and any observable $F$ in $M$ we have
 \begin{equation*}
 \rho_M^F\big(W^{\pol'}_H\big)=\rho_M^F\big(W^\pol_{G\qp H}\big).
 \end{equation*}
\end{cor}

\subsection{Composition via observables}
\label{sec:compobs}

In the setting of Hilbert spaces on hypersurfaces constructed from K\"ahler polarizations, the composition of correlation functions for regions can be accomplished via a sum over a complete basis of the Hilbert space associated to the gluing hypersurface. This was reviewed in~Section~\ref{sec:compk}. Equivalently (and often more conveniently), the sum is replaced by a complete integral over coherent states. However, the present emphasis on slice observables over Hilbert space states suggests using the former also to accomplish composition. This turns out to be not only possible, but natural, induced by thinking of the composition as effected by a change of polarization. What is more, no gluing anomaly occurs for this manner of composing.

\begin{figure}
 \centering
 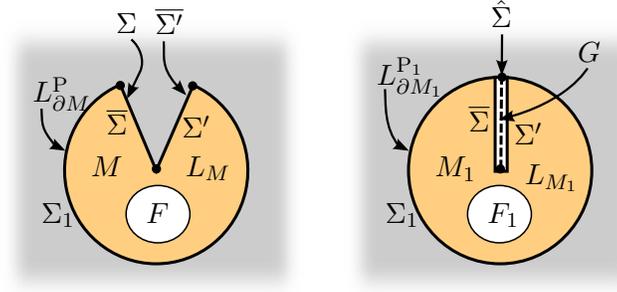
 \caption{Composition via slice observables: Setup for composition (left) with boundary polarization~$L_{\partial M}^\pol$, region composed along slice observable $\hat{\Sigma}$ (right) with boundary polarization $L_{\partial M_1}^{\pol_1}$, and location of boundary observable $G$ within region.}
 \label{fig:compsobs}
\end{figure}

Thus, let $M$ be a region with boundary $\partial M=\Sigma_1\cup\Sigma\cup\overline{\Sigma'}$, where $\Sigma'$ is a copy of $\Sigma$. Let~$M_1$ denote the gluing of $M$ to itself along $\Sigma$ with $\Sigma'$. We can think of the gluing as effected by the formation of a slice region $\hat{\Sigma}$ with boundary $\partial\hat{\Sigma}=\overline{\Sigma}\cup\Sigma'$, see Figure~\ref{fig:compsobs}.
We consider a polarization $L_{\partial M}^\pol\subseteq L_{\partial M}^\bC$ transversal to $L_{\tilde{M}}^\bC\subseteq L_{\partial M}^\bC$. Similarly, we consider a polarization $L_{\partial M_1}^{\pol_1}\subseteq L_{\partial M_1}^\bC$ transversal to $L_{\tilde{M}_1}^\bC\subseteq L_{\partial M_1}^\bC$. It is then clear that the Lagrangian subspace $L_{\partial M}^{(\pol_1,\mpol)}\defeq \big(L_{\partial M_1}^{\pol_1}\oplus L_{\overline{\hat{\Sigma}}}^\bC\big)\subseteq L_{\partial M}^\bC$ intersects $L_{\tilde{M}}^\bC$ only in $0$, but we require full transversality, $L_{\partial M}^\bC=L_{\partial M}^{(\pol_1,\mpol)}\oplus L_{\tilde{M}}^\bC$. In the following we write the decomposition $L_{\partial M}^\bC=L_{\partial M}^\pol\oplus L_{\tilde{M}}^\bC$ for~ele\-ments as $\phi=\phi^\epol + \phi^\ipol$. Similarly, we write the decomposition $L_{\partial M}^\bC= L_{\partial M}^{(\pol_1,\mpol)}\oplus L_{\tilde{M}}^\bC$ for~ele\-ments as $\phi=\phi^{\epol'}+ \phi^{\ipol'}$.

\begin{thm}
 \label{thm:scompo}
 Let $D\colon K_M^\bC\to \bC$ be a linear observable in $M$ and $D_1\colon K_{M_1}^\bC\to\bC$ the induced linear observable in $M_1$. Set $F=\exp(\im\, D)$ and $F_1=\exp(\im\, D_1)$. Then,
 \begin{equation}
 \rho_{M_1}^{F_1}\big(W^{\pol_1}\big)=\rho_{M}^F\big(Y^{\pol}_{(\pol_1,\mpol)}\big).
 \label{eq:omapcomp}
 \end{equation}
\end{thm}
\begin{proof}
 Define $G'\colon L_{\partial M}^\bC\to\bC$ as
\begin{equation*}
 G'(\phi)=\exp\big(\im\,\omega_{\partial_M}\big(\phi^{\epol'},\phi^{\epol}\big)\big).
\end{equation*}
This defines a slice observable $G$ in $M$ and a state $Y^{\pol}_{(\pol_1,\mpol)}=W^\pol_G$ that encodes a change of polarization from $\pol$ to $(\pol_1,\mpol)$. Now let $\eta'\in A_M^{D,\bC}\cap L_{\partial M_1}^{\pol_1,\mpol}$. Then, we have
\begin{equation*}
 \rho_{M}^F\big(Y^{\pol}_{(\pol_1,\mpol)}\big)=\rho_{M}^F\big(W^{(\pol_1,\mpol)}\big) =\exp\bigg(\frac{\im}{2} D(\eta')\bigg).
\end{equation*}
On the other hand let $\eta_1\in A_{M_1}^{D_1,\bC}\cap L_{\partial M_1}^{\pol_1}$. By Axiom~(C7) of~\cite[Section~4.6]{Oe:feynobs} we have \mbox{$D_1(\eta_1)=D(\eta')$}. But,
\begin{equation*}
 \rho_{M_1}^{F_1}\big(W^{\pol_1}\big)=\exp\bigg(\frac{\im}{2} D_1(\eta_1)\bigg).
\end{equation*}
Thus, we obtain relation~(\ref{eq:omapcomp}).
\end{proof}

One aspect of this composition theorem that is not quite satisfactory is the fact that the slice observable $G$ acts not only on $L_{\Sigma}^\bC$ and $L_{\overline{\Sigma'}}^\bC$, but also on $L_{\partial M_1}^\bC$, where no gluing takes place. Of course with polarizations before and after composition completely independent, this cannot be otherwise. However, in case that the polarization is unchanged on $\partial M_1$, one might hope for the composition to be effectable by a slice observable $G$ not dependent on $L_{\partial M_1}^\bC$. This hope is justified as we show in the following. Thus, suppose that $L_{\partial M}^\pol=L_{\partial M_1}^{\pol_1}\oplus L_{\overline{\partial\hat{\Sigma}}}^{\pol_{\Sigma}}$.

For $\phi\in L_{\partial M}^\bC=L_{\partial M_1}^\bC \oplus L_{\Sigma}^\bC\oplus L_{\overline{\Sigma'}}^\bC$ we use the notation $\phi=(\phi_1,\phi_{\Sigma},\phi_{\Sigma'})$. With this we define a slice observable $G$ on $\partial M$ via
\begin{equation}
 G'(\phi)=\exp\big(\im\,\omega_{\partial M}\big((0,\phi_{\Sigma},\phi_{\Sigma'})^{\epol'},(0,\phi_{\Sigma},\phi_{\Sigma'})^{\epol}\big)\big).
 \label{eq:sobslcomp}
\end{equation}
Note that this slice observable is not interior translation invariant. However, it has two other invariance properties.
\begin{lem}
 $G$ as defined above has the following properties:
 \begin{enumerate}\itemsep=0pt
 \item[$(1)$]
 $G'(\phi+(\delta_1,0,0))=G'(\phi)$ for any $\phi\in L_{\partial M}^\bC$ and any $\delta_1\in L_{\partial M_1}^\bC$.
 \item[$(2)$]
 $G'\big(\phi+\big(0,\delta_{\Sigma},\delta_{\Sigma'}\big)^\ipol\big)=G'(\phi)$ for any $\phi\in L_{\partial M}^\bC$ and any $\big(\delta_\Sigma,\delta_{\Sigma'}\big)\in L_{\Sigma}^\bC \oplus L_{\Sigma'}^\bC$.
 \end{enumerate}
\end{lem}

\begin{proof}
 The first property follows immediately since the definition~(\ref{eq:sobslcomp}) makes no reference to the first component $\phi_1$ of the argument $\phi$. As for the second property, it suffices to show this for the argument of the exponential in~(\ref{eq:sobslcomp}). We evaluate the left-hand side of the property~(2) to be demonstrated order by order in $\delta$. That is, we have to show that both the first and the second order in $\delta$ vanish. We introduce the notation $(\cdot)_{\rm x00}$ to mean projection onto the first component. Similarly, $(\cdot)_{\rm 0xx}$ means projection onto the second and third components. We start with the first order expression in $\delta$:
 \begin{gather*}
 \omega_{\partial M}\big((0,\phi_{\Sigma},\phi_{\Sigma'})^{\epol'}, (((0,\delta_{\Sigma},\delta_{\Sigma'})^\ipol)_{\rm 0xx})^{\epol}\big)
 +\omega_{\partial M}\big((((0,\delta_{\Sigma},\delta_{\Sigma'})^\ipol)_{\rm 0xx})^{\epol'}, (0,\phi_{\Sigma},\phi_{\Sigma'})^{\epol}\big)
 \\ \quad\
 {} = \omega_{\partial M}\big((0,\phi_{\Sigma},\phi_{\Sigma'}) -(0,\phi_{\Sigma},\phi_{\Sigma'})^{\ipol'},((0,\delta_{\Sigma},\delta_{\Sigma'})^\ipol)_{\rm 0xx} -(((0,\delta_{\Sigma},\delta_{\Sigma'})^\ipol)_{\rm 0xx})^{\ipol}\big)
 \\ \quad\ \phantom{=}
 {}-\omega_{\partial M}\big((0,\phi_{\Sigma},\phi_{\Sigma'})-(0,\phi_{\Sigma}, \phi_{\Sigma'})^{\ipol},((0,\delta_{\Sigma},\delta_{\Sigma'})^\ipol)_{\rm 0xx} -(((0,\delta_{\Sigma},\delta_{\Sigma'})^\ipol)_{\rm 0xx})^{\ipol'}\big)
 \\ \quad\
 {} = \omega_{\partial M}\big((0,\phi_{\Sigma},\phi_{\Sigma'}), (0,\delta_{\Sigma},\delta_{\Sigma'})^\ipol\big)
 -\omega_{\partial M}\big((0,\phi_{\Sigma},\phi_{\Sigma'}), (((0,\delta_{\Sigma},\delta_{\Sigma'})^\ipol)_{\rm 0xx})^{\ipol}\big)
 \\ \quad\ \phantom{=}
 {}-\omega_{\partial M}\big((0,\phi_{\Sigma},\phi_{\Sigma'})^{\ipol'},((0,\delta_{\Sigma},\delta_{\Sigma'})^\ipol)_{\rm 0xx}\big) -\omega_{\partial M}\big((0,\phi_{\Sigma},\phi_{\Sigma'}),(0,\delta_{\Sigma},\delta_{\Sigma'})^\ipol\big)
 \\ \quad\ \phantom{=}
 {}+\omega_{\partial M}\big((0,\phi_{\Sigma},\phi_{\Sigma'}),(((0,\delta_{\Sigma}, \delta_{\Sigma'})^\ipol)_{\rm 0xx})^{\ipol'}\big)
 +\omega_{\partial M}\big((0,\phi_{\Sigma},\phi_{\Sigma'})^{\ipol}, ((0,\delta_{\Sigma},\delta_{\Sigma'})^\ipol)_{\rm 0xx}\big)
 \\ \quad\
 {}= -\omega_{\partial M}\big((0,\phi_{\Sigma},\phi_{\Sigma'})^{\epol},((0,\delta_{\Sigma}, \delta_{\Sigma'})^\ipol)_{\rm 0xx}\big)
 -\omega_{\partial M}\big((0,\phi_{\Sigma},\phi_{\Sigma'})^{\ipol'}, ((0,\delta_{\Sigma},\delta_{\Sigma'})^\ipol)_{\rm 0xx}\big)
 \\ \quad\ \phantom{=}
{} +\omega_{\partial M}\big((0,\phi_{\Sigma},\phi_{\Sigma'})^{\epol'}, ((0,\delta_{\Sigma},\delta_{\Sigma'})^\ipol)_{\rm 0xx}\big)
 +\omega_{\partial M}\big((0,\phi_{\Sigma},\phi_{\Sigma'})^{\ipol}, ((0,\delta_{\Sigma},\delta_{\Sigma'})^\ipol)_{\rm 0xx}\big)
 \\ \quad\
 {}= 2\omega_{\partial M}\big((0,\phi_{\Sigma},\phi_{\Sigma'})^{\ipol}, ((0,\delta_{\Sigma},\delta_{\Sigma'})^\ipol)_{\rm 0xx}\big)
 -2\omega_{\partial M}\big((0,\phi_{\Sigma},\phi_{\Sigma'})^{\ipol'}, ((0,\delta_{\Sigma},\delta_{\Sigma'})^\ipol)_{\rm 0xx}\big)
 \\ \quad\
 {} = 2\omega_{\partial M}\big((0,\phi_{\Sigma},\phi_{\Sigma'})^{\ipol}-(0,\phi_{\Sigma}, \phi_{\Sigma'})^{\ipol'},((0,\delta_{\Sigma},\delta_{\Sigma'})^\ipol)_{\rm 0xx}\big)
 \\ \quad\
 {} = 2\omega_{\partial M}\big((0,\phi_{\Sigma},\phi_{\Sigma'})^{\ipol} -(0,\phi_{\Sigma},\phi_{\Sigma'})^{\ipol'}, (0,\delta_{\Sigma},\delta_{\Sigma'})^\ipol-((0,\delta_{\Sigma},\delta_{\Sigma'})^\ipol)_{\rm x00}\big) \\ \quad\
 {} = -2\omega_{\partial M}\big((0,\phi_{\Sigma},\phi_{\Sigma'})^{\ipol} -(0,\phi_{\Sigma},\phi_{\Sigma'})^{\ipol'},((0,\delta_{\Sigma},\delta_{\Sigma'})^\ipol)_{\rm x00}\big) \\ \quad\
{} = 2\omega_{\partial M}\big((0,\phi_{\Sigma},\phi_{\Sigma'})^{\epol} -(0,\phi_{\Sigma},\phi_{\Sigma'})^{\epol'},((0,\delta_{\Sigma},\delta_{\Sigma'})^\ipol)_{\rm x00}\big) \\ \quad\
 {} = 2\omega_{\partial M}\big((0,\phi_{\Sigma},\phi_{\Sigma'})^{\epol} -(0,\phi_{\Sigma},\phi_{\Sigma'})^{\epol'},(-(0,\delta_{\Sigma},\delta_{\Sigma'})^\epol)_{\rm x00}\big) = 0.
 \end{gather*}
 The last equality follows, because both arguments of the symplectic form have first components in the Lagrangian subspace $L_{M_1}^{\pol_1}$.
 It remains to show that the term of second order in $\delta$ also vanishes. This is
 \begin{equation*}
 \omega_{\partial M}\big(\big(\big((0,\delta_{\Sigma},\delta_{\Sigma'})^\ipol\big)_{\rm 0xx}\big)^{\epol'}, \big(\big(\big(0,\delta_{\Sigma},\delta_{\Sigma'}\big)^\ipol\big)_{\rm 0xx}\big)^{\epol}\big).
 \end{equation*}
 However, if we set $(0,\phi_{\Sigma},\phi_{\Sigma'})\defeq \big((0,\delta_{\Sigma},\delta_{\Sigma'})^\ipol\big)_{\rm 0xx}$ this reduces to the previous demonstration. This completes the proof.
\end{proof}
It turns out that this gives us an analogue of Lemma~\ref{lem:sobsit}.

\begin{lem}
 \label{lem:sobst2}
 If $F$ is the Weyl observable $\exp(\im D)$, $\eta\in A_M^{D,\bC}\cap L_{\partial M}^\pol$, and $G$ is a slice observable with the stated invariance properties, then
\begin{equation}
 \rho_M^{G\cp F}\big(W^{\pol}\big)=\rho_M^F\big(W^\pol\big)\, G'(\eta).
 \label{eq:sobst2}
\end{equation}
\end{lem}
\begin{proof}
 As in the proof of Lemma~\ref{lem:sobsit}, we may at first assume that $G$ is a Weyl observable determined by $G'(\phi)=\exp\left(\im\,\omega_{\partial M}(\xi,\phi)\right)$ with $\xi\in L_{\partial M}^{\bC}$. By the first invariance property $\xi_1=0$ and by the second invariance property $\omega_{\Sigma}\left(\xi,(0,\delta_{\Sigma},\delta_{\Sigma'})^\ipol\right)=0$ for any $\big(\delta_\Sigma,\delta_{\Sigma'}\big)\in L_{\Sigma}^\bC \oplus L_{\Sigma'}^\bC$. But this implies $\omega_{\partial M}(\xi,\xi^\ipol)=0$ and thus equation~(\ref{eq:sobst2}) as in the proof of Lemma~\ref{lem:sobsit}.
\end{proof}

With this we are ready to show that the slice observable $G$ given by expression~(\ref{eq:sobslcomp}) effects the desired composition.

\begin{thm}
 \label{thm:sfcompo}
 Let $D\colon K_M^\bC\to \bC$ be a linear observable in $M$ and $D_1\colon K_{M_1}^\bC\to\bC$ the induced linear observable in $M_1$. Set $F=\exp(\im\, D)$ and $F_1=\exp(\im\, D_1)$. Then,
 \begin{equation}
 \rho_{M_1}^{F_1}\big(W^{\pol_1}\big)=\rho_{M}^F\big(W^{\pol}_G\big).
 \label{eq:omaplcomp}
\end{equation}
\end{thm}
\begin{proof}
Let $\eta\in A_M^{D,\bC}\cap \Big(L_{\partial M_1}^{\pol_1}\oplus L_{\overline{\partial\hat{\Sigma}}}^{\pol_{\Sigma}}\Big)$. By Lemma~\ref{lem:sobst2} we have
\begin{equation*}
 \rho_{M}^F\big(W^{\pol}_G\big)=\exp\bigg(\frac{\im}{2} D(\eta)\bigg) G'(\eta).
\end{equation*}
Now let $\eta'\in A_M^{D,\bC}\cap L_{\partial M}^{(\pol_1,\mpol)}$. Then,
\begin{align*}
 G'(\eta) & =\exp\big(\im\,\omega_{\partial M}\big((0,\eta_{\Sigma},\eta_{\Sigma'})^{\epol'},(0,\eta_{\Sigma},\eta_{\Sigma'})^{\epol}\big)\big)
 \nonumber \\
 & =\exp\big(\im\,\omega_{\partial M}\big((0,\eta_{\Sigma},\eta_{\Sigma'})^{\epol'}+ (\eta_1,0,0)^{\epol'},(0,\eta_{\Sigma},\eta_{\Sigma'})^{\epol}+(\eta_1,0,0)^{\epol}\big)\big) \nonumber \\
 & =\exp\big(\im\,\omega_{\partial M}\big(\eta^{\epol'},\eta^{\epol}\big)\big)
 = \exp\big(\im\,\omega_{\partial M}\big(\eta',\eta\big)\big).
\end{align*}
With this we get
\begin{equation*}
 \rho_{M}^F\big(W^{\pol}_G\big)=\exp\bigg(\frac{\im}{2} D(\eta)\bigg) \exp(\im\,\omega_{\partial M}(\eta',\eta)) = \exp\bigg(\frac{\im}{2} D(\eta')\bigg).
\end{equation*}
On the other hand let $\eta_1\in A_{M_1}^{D_1,\bC}\cap L_{\partial M_1}^{\pol_1}$. By Axiom~(C7) of \cite[Section~4.6]{Oe:feynobs} we have \mbox{$D_1(\eta_1)=D(\eta')$}. But,
\begin{equation*}
 \rho_{M_1}^{F_1}\big(W^{\pol_1}\big)=\exp\bigg(\frac{\im}{2} D_1(\eta_1)\bigg).
\end{equation*}
Thus, we obtain relation~(\ref{eq:omaplcomp}).
\end{proof}

\subsection{Wick's theorem}
\label{sec:genwick}

 {\samepage In the remainder of this Section~\ref{sec:genquant} we explore how various basic tools of textbook QFT generalize to the present framework. We start with Wick's theorem. That is, we consider correlation functions of monomial observables as in~Section~\ref{sec:vcgenobs}, but for general states rather than the vacuum state. To this end, we resort, as usual, to normal ordering, here as defined in~Section~\ref{sec:nordsc}. In particular, we recall equation~(\ref{eq:Wickgen}) for Weyl observables. Its right hand side serves as the generating function of Wick's theorem. Thus, for the product $D_1\cdots D_n$ we find on the arbitrary state $X^\pol$
\begin{align*}
 \rho_M^{D_1 \cdots D_n}\big(X^\pol\big) & = \bigg(\prod_{k=1}^n \bigg({-}\im \frac{\partial}{\partial \lambda_k}\bigg)\bigg) \rho_M^{\no{F}}\big(X^\pol\big)\, \rho_M^F\big(W^\pol\big) \bigg|_{\lambda_1,\dots,\lambda_n=0} \nonumber
 \\
 & = \sum_{m=0}^{\lfloor n/2 \rfloor} \sum_{\sigma\in S^n} \frac{1}{(2m)!\, (n-2m)!} \bigg(\bigg(\prod_{k=2m+1}^n \bigg({-}\im \frac{\partial}{\partial \lambda_{\sigma(k)}}\bigg)\bigg) \rho_M^{\no{F}}\big(X^\pol\big)\bigg) \nonumber
 \\
 &\phantom{=} \times \bigg(\bigg(\prod_{k=1}^{2m} \bigg({-}\im \frac{\partial}{\partial \lambda_{\sigma(k)}}\bigg)\bigg) \rho_M^F\big(W^\pol\big)\bigg) \bigg|_{\lambda_1,\dots,\lambda_n=0} \nonumber
 \\
 & = \sum_{m=0}^{\lfloor n/2 \rfloor} \sum_{\sigma\in S^n} \frac{1}{(2m)!\, (n-2m)!} \rho_M^{\no{D_{\sigma(2m+1)}\cdots D_{\sigma(n)}}}\big(X^\pol\big)\, \rho_M^{D_{\sigma(1)}\cdots D_{\sigma(2m)}}\big(W^\pol\big) \nonumber
\\
 & = \sum_{m=0}^{\lfloor n/2 \rfloor} \sum_{\sigma\in S^n} \frac{1}{2^m\, m!\, (n-2m)!} \rho_M^{\no{D_{\sigma(2m+1)}\cdots D_{\sigma(n)}}}\big(X^\pol\big)\, \prod_{j=1}^m\, \rho_M^{D_{\sigma(2j-1)} D_{\sigma(2j)}}\big(W^\pol\big).
\end{align*}
Here $\lfloor a \rfloor$ denotes the ``floor'' of $a$, i.e., the largest integer smaller or equal to $a$.

}

We remark the particular cases of a linear observable $D$ and a quadratic observable $E$,
\begin{align*}
 \rho_M^{D}\big(X^\pol\big) & = \rho_M^{\no{D}}\big(X^\pol\big),
 \\[.5ex]
 \rho_M^{E}\big(X^\pol\big) & = \rho_M^{\no{E}}\big(X^\pol\big) + \rho_M^{E}\big(W^\pol\big).
\end{align*}

\subsection{Connected amplitude and multi-particle states}
\label{sec:pstates}

We proceed to have a closer look at multi-particle states. We first consider the free amplitude for a multi-particle state. Let $\xi_1,\ldots,\xi_n\in L_{\partial M}^\bC$. Then, recalling~(\ref{eq:mpsweyl}) and~(\ref{eq:amplg}),
\begin{align*}
 \rho_M\big(Q^\pol_{\xi_1,\dots,\xi_n}\big) & = \bigg(\prod_{k=1}^n 2 \frac{\partial}{\partial \lambda_k}\bigg) \rho_M\big(K^\pol_{\lambda_1 \xi_1+\cdots+\lambda_n \xi_n}\big) \bigg|_{\lambda_1,\dots,\lambda_n=0} \nonumber
 \\
 & = \bigg(\prod_{k=1}^n 2 \frac{\partial}{\partial \lambda_k}\bigg) \exp\bigg(\im \sum_{k,l=0}^n \lambda_k\lambda_l \omega_{\partial M}\big(\xi_k,\xi_l^\ipol\big)\bigg) \bigg|_{\lambda_1,\dots,\lambda_n=0}.
\end{align*}
This vanishes if $n$ is odd. Otherwise, set $n=2m$ and we get
\begin{align*}
 \rho_M\big(Q^\pol_{\xi_1,\dots,\xi_n}\big)
 & = \bigg(\prod_{k=1}^{2m} 2 \frac{\partial}{\partial \lambda_k}\bigg) \frac{1}{m!} \bigg(\im \sum_{k,l=0}^n \lambda_k\lambda_l \omega_{\partial M}\big(\xi_k,\xi_l^\ipol\big)\bigg)^m \bigg|_{\lambda_1,\dots,\lambda_n=0} \nonumber\\
 & = \frac{1}{m!} \sum_{\sigma\in S^{2m}} \prod_{j=1}^m 4\im\omega_{\partial M}\big(\xi_{\sigma(2j-1)},\xi_{\sigma(2j)}^\ipol\big).
\end{align*}
This can be interpreted in the usual way, namely that there is no interaction and particles are merely pairwise identified.
We proceed to consider the same amplitude with a Weyl observable $F=\exp(\im D)$ present. This may serve as a generating function for the corresponding amplitude of an interacting theory. Thus, with Theorem~\ref{thm:wfac},
\begin{align*}
 \rho_M^F\big(Q^\pol_{\xi_1,\dots,\xi_n}\big)
 & = \bigg(\prod_{k=1}^n 2 \frac{\partial}{\partial \lambda_k}\bigg) \rho_M^F\big(K^\pol_{\lambda_1 \xi_1+\cdots+\lambda_n \xi_n}\big) \bigg|_{\lambda_1,\dots,\lambda_n=0} \nonumber
 \\
 & = \bigg(\prod_{k=1}^n\! 2 \frac{\partial}{\partial \lambda_k}\bigg) \rho_M\big(K^\pol_{\lambda_1 \xi_1+\cdots+\lambda_n \xi_n}\big) F\big(\lambda_1\xi_1^\ipol\!+\cdots\! +\lambda_n \xi_n^\ipol\big) \rho_M^F\big(W^\pol\big) \bigg|_{\lambda_1,\dots,\lambda_n=0} \nonumber
 \\
 & = \sum_{m=0}^{\lfloor n/2 \rfloor} \sum_{\sigma\in S^n} \frac{1}{(2m)!\, (n-2m)!} \bigg(\bigg(\prod_{k=1}^{2m} 2 \frac{\partial}{\partial \lambda_{\sigma(k)}}\bigg) \rho_M\big(K^\pol_{\lambda_1 \xi_1+\cdots+\lambda_n \xi_n}\big)\bigg)
 \nonumber
 \\
 & \phantom{=}\times \bigg(\bigg(\prod_{k=2m+1}^n 2 \frac{\partial}{\partial \lambda_{\sigma(k)}}\bigg) F\big(\lambda_1\xi_1^\ipol+\cdots +\lambda_n \xi_n^\ipol\big) \rho_M^F\big(W^\pol\big)\bigg) \bigg|_{\lambda_1,\dots,\lambda_n=0} \nonumber\\
 & = \sum_{m=0}^{\lfloor n/2 \rfloor} \sum_{\sigma\in S^n} \frac{1}{(2m)! (n\!-\!2m)!} \rho_M\big(Q^\pol_{\xi_{\sigma(1)},\dots,\xi_{\sigma(2m)}}\big)
 \bigg(\prod_{k=2m+1}^n\!\!\!\! 2\im D\big(\xi^\ipol_{\sigma(k)}\big)\bigg) \rho_M^F\big(W^\pol\big).
\end{align*}
The parameter $m$ in the outermost sum can be interpreted as follows. The contribution with given value $m$ corresponds to $m$ pairs of particles being identified and not interacting, while the other $n-2m$ particles participate in the interaction. In terms of diagrams this may be expressed as $m$ lines that connect $2m$ particles into pairs while the remaining particles are connected to~$F$. This motivates the definition of the \emph{connected} amplitude for coherent states as
\begin{equation*}
 \rho_{{\rm c},M}^F\big(K^\pol_\xi\big)\defeq \frac{\rho_M^F\big(K^\pol_\xi\big)}{\rho_M\big(K^\pol_\xi\big)} =\rho_M^F\big(K^\pol_{\xi^\ipol}\big).
\end{equation*}
We observe that the denominator never vanishes. The second equality is obtained by explicit inspection, using the results of Theorem~\ref{thm:wfac}. Note that by linearity this definition and the equality are valid for arbitrary (not necessarily Weyl) observables $F$ and thus for $F$ possibly defining an interacting theory.

For multi-particle states this yields
\begin{equation*}
 \rho_M^F\big(Q^\pol_{\xi_1,\dots,\xi_n}\big)
 = \sum_{m=0}^{\lfloor n/2 \rfloor} \sum_{\sigma\in S^n} \frac{1}{(2m)!\, (n-2m)!}\, \rho_M\big(Q^\pol_{\xi_{\sigma(1)},\dots,\xi_{\sigma(2m)}}\big)
 \rho_{{\rm c},M}^F\big(Q^\pol_{\xi_{\sigma(2m+1)},\dots,\xi_{\sigma(n)}}\big),
\end{equation*}
as well as its ``inverse'',
\begin{equation*}
 \rho_{{\rm c},M}^F\big(Q^\pol_{\xi_1,\dots,\xi_n}\big)
 = \sum_{m=0}^{\lfloor n/2 \rfloor} \sum_{\sigma\in S^n} \frac{(-1)^m}{(2m)!\, (n-2m)!}\, \rho_M\big(Q^\pol_{\xi_{\sigma(1)},\dots,\xi_{\sigma(2m)}}\big)
 \rho_M^F\big(Q^\pol_{\xi_{\sigma(2m+1)},\dots,\xi_{\sigma(n)}}\big).
\end{equation*}
Moreover, we have
\begin{equation}
 \rho_{{\rm c},M}^F\big(Q^\pol_{\xi_1,\dots,\xi_n}\big) =\rho_{M}^F\big(Q^\pol_{\xi_1^\ipol,\dots,\xi_n^\ipol}\big).
 \label{eq:acamps}
\end{equation}
Again, these expressions are valid for $F$ an arbitrary observable.

It is straightforward to carry out the same steps for \emph{normal ordered quantization}. In particular, we can make the analogous definition of connected amplitudes. In that case an additional simplification occurs, and we get
\begin{equation*}
 \rho_{{\rm c},M}^{\no{F}}\big(K^{\pol}_\xi\big)
 =\rho_{M}^{\no{F}}\big(K^{\pol}_{\xi^\ipol}\big)=F\big(\xi^\ipol\big).
\end{equation*}
For multi-particle states this translates to
\begin{equation*}
 \rho_{{\rm c},M}^{\no{F}}\big(Q^\pol_{\xi_1,\dots,\xi_n}\big)
 =\rho_{M}^{\no{F}}\big(Q^\pol_{\xi_1^\ipol,\dots,\xi_n^\ipol}\big)
 =\bigg(\prod_{k=1}^n 2 \frac{\partial}{\partial \lambda_k}\bigg) F\big(\lambda_1 \xi_1^\ipol+\cdots+\lambda_n \xi_n^\ipol\big) \bigg|_{\lambda_1,\dots,\lambda_n=0}.
\end{equation*}
In the case of a Weyl observable $F=\exp(\im D)$ this specializes to
\begin{equation*}
\rho_{{\rm c},M}^{\no{F}}\big(Q^\pol_{\xi_1,\dots,\xi_n}\big)
=\rho_{M}^{\no{F}}\big(Q^\pol_{\xi_1^\ipol,\dots,\xi_n^\ipol}\big)
=\prod_{k=1}^n 2\im D\big(\xi^\ipol_{k}\big).
\end{equation*}
On the other hand, for Weyl observables $F$ on general states we have the analogue of equation~(\ref{eq:Wickgen}),
\begin{equation*}
 \rho_{{\rm c},M}^{F}\big(X^{\pol}\big)=\rho_{{\rm c},M}^{\no{F}}\big(X^{\pol}\big)\rho_M^F\big(W^\pol\big).
\end{equation*}

\subsection{LSZ reduction}
\label{sec:genlsz}

The LSZ reduction procedure in QFT allows re-expressing transition amplitudes between initial and final particle states in terms of vacuum $n$-point functions, compare Section~\ref{sec:intro_pertlsz}. As previously mentioned, this provides one motivation for the proposal of the present section where we actually define ``states'' in terms of observables in the vacuum. This indeed turns out to make the LSZ reduction formula~(\ref{eq:stdlsz}) essentially tautological. Nevertheless, we have to correctly identify the observables involved.

We begin by recalling from Section~\ref{sec:sobsampl} the way multi-particle states are represented in our setting. Let $M$ be a region with exterior polarization $L_{\partial M}^\pol$ transversal to $L_M^\bC$. Let $\xi_1,\ldots,\xi_n\allowbreak\in L_{\partial M}^\bC$. Let $E_k$ define the linear boundary observable defined by $E_k'(\xi)=4\im\omega_{\partial M}(\xi_k,\phi)$. Let~$G$ be the product boundary observable~(\ref{eq:prodobsxi}), i.e., $G=E_1\cdots E_n$. By definition, we have~(\ref{eq:defmpso}). That is, with an interior observable $F$ representing an interaction we have
\begin{equation*}
 \rho_M^F\big(Q^\pol_{\xi_1,\dots,\xi_n}\big) = \rho_M^{E_1\cdots E_n\cp F}\big(W^\pol\big).
\end{equation*}
In the reduction formula~(\ref{eq:stdlsz}) we care about the connected amplitude, however, signaled by separating off on the right-hand side the ``disconnected terms''. In our setting we obtain the corresponding result by using equation~(\ref{eq:acamps}). To this end we define the boundary observables $\tilde{E}'_k(\phi)=4\im\omega_{\partial M}\big(\xi_k^\ipol,\phi\big)$. Then,
\begin{equation}
 \rho_{{\rm c},M}^F\big(Q^\pol_{\xi_1,\dots,\xi_n}\big)
 =\rho_{M}^F\big(Q^\pol_{\xi_1^\ipol,\dots,\xi_n^\ipol}\big)
 =\rho_M^{\tilde{E}_1\cdots\tilde{E}_n\cp F}\big(W^\pol\big).
 \label{eq:abslsz}
\end{equation}
We claim that this is precisely the generalization of the LSZ reduction formula~(\ref{eq:stdlsz}) to our setting.

In the remainder of this section we show how formula~(\ref{eq:stdlsz}) arises as a special case of~for\-mula~(\ref{eq:abslsz}). As a first step, we take a closer look at the linear slice observables $\tilde{E}_k$. Suppose for the moment that $F$ is a Weyl observable $F=\exp(\im D)$. Take $\eta$ to be a complexified solution of the equations of motion modified by $D$, i.e., $\eta\in A_M^{D,\bC}$. Then,
\begin{equation*}
 \tilde{E}'_k(\eta)=4\im\omega_{\partial M}\big(\xi_k^\ipol,\eta\big)= 2\im D\big(\xi_k^\ipol\big),
\end{equation*}
where the second equality arises from relation~(\ref{eq:obssol}). We assume further that the equations of motion are given in the form
\begin{equation*}
 (\dop \phi)(x)=0,
\end{equation*}
where $\dop$ is a differential operator in spacetime. We let $D$ be given by a source $j\colon M\to\bC$,
\begin{equation*}
 D(\phi)=\int_M \xd x\, \phi(x) j(x).
\end{equation*}
Then, $\eta\in A_M^{D,\bC}$ means that $\eta$ is a solution of the inhomogeneous equations of motion with source $j$. In particular,
\begin{equation*}
 (\dop\eta)(x)=j(x).
\end{equation*}
With this we can write
\begin{equation}
 \tilde{E}'_k(\eta)=2\im\int_M\xd x\, \xi_k^\ipol(x) (\dop\eta)(x).
 \label{eq:bdyobsint}
\end{equation}
Note that this equation makes no longer reference to the source $j$ or the observable $D$. So this is valid for $\eta\in A_M^{D,\bC}$ for all possible $D$. We conclude that it is valid for arbitrary configurations $\eta\in K_M^\bC$.\footnote{One might object that the configuration space $K_M$ might include configurations too irregular (e.g., non-differentiable) for this to make sense. However, we (intentionally) have not provided a formal definition of $K_M$ anywhere. It turns out that a good definition of $K_M$ appears to be precisely to take it to be generated by the spaces $A_M^D$ for all $D$. We will not provide any evidence for this claim in the present paper, however.} This also means that we can drop the requirement for $F$ to be a Weyl observable.

In order to recover expression~(\ref{eq:stdlsz}) we now specialize to Klein--Gordon theory in Minkowski spacetime. Moreover, we take the region $M$ to be determined by a time interval, $M=[t_1,t_2]\times\R^3$, which we extend to infinity, $t_1\to-\infty$, $t_2\to\infty$. Thus, consider a momentum state on the boundary $\partial M=\Sigma_{t_1}\sqcup \overline{\Sigma}_{t_2}$ with incoming momenta $q_1,\ldots,q_n$ and outgoing momenta $p_1,\ldots,p_m$. Set
\begin{equation*}
 \xi_k\defeq \frac{1}{\sqrt{2}} P^-((\phi_{q_k},0)),\qquad\text{and}\qquad
 \tau_l\defeq \frac{1}{\sqrt{2}} P^-((0,\phi_{p_l})).
\end{equation*}
Then, $Q^\pol_{\xi_1,\ldots,\xi_n,\tau_1,\ldots,\tau_m}$ describes the state in question, compare~(\ref{eq:mmso}). With this, the left-hand side of~(\ref{eq:stdlsz}) translates to our notation as
\begin{equation*}
 \langle p_1,\ldots,p_m| q_1,\ldots,q_n \rangle=\rho_M^F\big(Q^\pol_{\xi_1,\ldots,\xi_n,\tau_1,\ldots,\tau_m}\big),
\end{equation*}
where, again, $F$ encodes the interactions. As for the ``disconnected terms'' in~(\ref{eq:stdlsz}),
\begin{equation*}
 \rho_M^F\big(Q^\pol_{\xi_1,\ldots,\xi_n,\tau_1,\ldots,\tau_m}\big)= \text{disconnected terms}
 +\rho_{{\rm c},M}^F\big(Q^\pol_{\xi_1,\ldots,\xi_n,\tau_1,\ldots,\tau_m}\big).
\end{equation*}
It remains to rewrite the connected amplitude in terms of boundary observables. To this end define boundary observables by $B'_k(\phi)\defeq 4\im\omega_{\partial M}\big(\xi_k^\ipol,\phi\big)$ and $C'_l(\phi)\defeq 4\im\omega_{\partial M}\big(\tau_l^\ipol,\phi\big)$. Thus, by the generalized LSZ formula~(\ref{eq:abslsz}) we have
\begin{equation}
 \rho_{{\rm c},M}^F\big(Q^\pol_{\xi_1,\ldots,\xi_n,\tau_1,\ldots,\tau_m}\big)
 =\rho_M^{B_1\cdots B_n\cdot C_1\cdots C_m\cp F}\big(W^\pol\big)
 = \langle 0 | \tord B_1 \cdots B_n C_1\cdots C_m|0\rangle.
 \label{eq:mlsz}
\end{equation}
Here we have used on the right-hand side standard textbook notation. To understand the observables $B_k$ and $C_l$ we note that $(P^-(\xi))^\ipol=\xi^\ipol$ for arbitrary $\xi\in L_{\partial M}^\bC$. Also, as is easily checked,
\begin{equation*}
 (\phi_{q_k},0)^\ipol(t,x)= \frac{1}{\sqrt{2}} {\rm e}^{-\im (E_{q_k} t - q_k x)}\qquad\text{and}\qquad
 (0,\phi_{p_l})^\ipol(t,x)= \frac{1}{\sqrt{2}} {\rm e}^{\im (E_{p_l} t - p_l x)}.
\end{equation*}
Combining this with equation~(\ref{eq:bdyobsint}) we get for the boundary observables
\begin{equation*}
 B'_k(\phi)=\im\int_M\xd t\,\xd^3 x\, {\rm e}^{-\im (E_{q_k} t - q_k x)} (\dop\phi)(x),\qquad
 C'_l(\phi)=\im\int_M\xd t\,\xd^3 x\, {\rm e}^{\im (E_{p_l} t - p_l x)} (\dop\phi)(x).
\end{equation*}
Inserting this into~(\ref{eq:mlsz}) yields by linearity
\begin{gather*}
 \rho_M^{B_1\cdots B_n\cdot C_1\cdots C_m\cp F}\big(W^\pol\big)=
 \im^{n+m}\!\int\! \xd^4 x_1\cdots\xd^4 x_n \xd^4 y_1\cdots\xd^4 y_m
 \exp\bigg(\im \sum_{l=1}^m p_l\cdot y_l - \im \sum_{k=1}^n q_k\cdot x_k \bigg)
 \\ \hphantom{\rho_M^{B_1\cdots B_n\cdot C_1\cdots C_m\cp F}\big(W^\pol\big)=}
{} \times \dop_{x_1}\cdots \dop_{x_n} \dop_{y_1}\cdots \dop_{y_m}
 \langle 0 | \tord \phi(x_1)\cdots\phi(x_n)\phi(y_1)\cdots\phi(y_m)|0\rangle.
\end{gather*}
Here, 4-dimensional notation is used that translates as, $q_k\cdot x=E_{q_k} t - q_k x$ and $p_l\cdot y=E_{p_l} t - p_l y$. Thus, we recover also the right-hand side of the LSZ formula~(\ref{eq:stdlsz}).

\section{Inner products and Hilbert spaces of states}
\label{sec:iprod}

\subsection{General polarizations, real structure and inner product}
\label{sec:modreal}

We recall that a \emph{real structure} on a complex vector space $V$ is a complex conjugate linear involution $\alpha\colon V\to V$. If $V$ arises as the complexification $V=L^\bC$
of a real vector space $L$, then complex conjugation given by
$a+\im b \mapsto a-\im b$ for $a,b\in L$ defines a real structure.
Conversely, given a real structure on $V$, define $V^\alpha$ as the real subspace invariant under $\alpha$. Then $V$ arises as the complexification of $V^\alpha$ with complex conjugation given by $\alpha$. Suppose that $V$ is equipped additionally with a complex bilinear symplectic form $\omega\colon V\times V\to\bC$. We say that the real structure $\alpha$ is \emph{compatible} with the symplectic form $\omega$ if for all $x,y\in V$
\begin{equation}
 \omega(\alpha(x),\alpha(y))=\overline{\omega(x,y)}.
 \label{eq:realcomp}
\end{equation}
In that case, $\omega$ arises as the complexification of a real valued symplectic form on $V^\alpha$.

Let $(V,\omega)$ be a complex symplectic vector space with compatible real structure $\alpha$. Define the sesquilinear form $V\times V\to\bC$ given for $\phi,\eta\in V$ by
\begin{equation}
 (\phi,\eta)^\alpha\defeq 4\im\omega\left(\alpha(\phi),\eta\right).
 \label{eq:alphaip}
\end{equation}
As is easy to verify, this sesquilinear form is non-degenerate and hermitian. If $V=L^\bC$ and~$\alpha$ is given by complex conjugation this is just the sesquilinear form~(\ref{eq:stdipc}). Consider a complex Lagrangian subspace $V^\pol\subseteq V$ on which this sesquilinear form is positive-definite. If $V=L^\bC$ and $\alpha$ is complex conjugation, this defines precisely a K\"ahler polarization. In the general case, we call this an \emph{$\alpha$-K\"ahler polarization}. The conjugate $V^\cpol\defeq \alpha(V^\pol)$ of $V^\pol$ with respect to $\alpha$ is a~Lagrangian subspace of $V$ transversal to $V^\pol$.

Consider a complex symplectic vector space $(V,\omega)$ with a pair of transversal Lagrangian subspaces, $V^\pol$ and $V^\cpol$. That is, we have $V=V^\pol\oplus V^\cpol$. We write $x=x^\pol+x^\cpol$ for the corresponding decomposition of elements of $V$. Let $J\colon V\to V$ be the associated complex structure. That is, $J$ is the complex linear map equal to $\im\id$ on $V^\pol$ and equal to $-\im\id$ on $V^\cpol$. Note that we automatically have compatibility with the symplectic form, $\omega(J x,J y)=\omega(x,y)$. We may then define a complex bilinear form, compare equation~(\ref{eq:stdjip}),
\begin{equation}
 \{x,y\}\defeq 2 \omega(x,J y)+2\im\omega(x,y)=
 4\im\omega\big(x^\cpol,y^\pol\big).
 \label{eq:cipspol}
\end{equation}
This form is sesquilinear with respect to $J$, but it is clearly degenerate.

Suppose the complex symplectic vector space $(V,\omega)$ is equipped with a compatible real structure $\alpha$ and an $\alpha$-K\"ahler polarization $V^\pol\subseteq V$. Note that $J\circ\alpha=\alpha\circ J$. Restricting the bilinear form~(\ref{eq:cipspol}) to the real subspace $V^\alpha$ makes it hermitian and positive-definite. What is more, the positive-definite inner products~(\ref{eq:alphaip}) on $V^\pol$ and~(\ref{eq:cipspol}) are related for $\phi,\eta\in V^\alpha$ as
\begin{equation}
 \big(\phi^\pol,\eta^\pol\big)^\alpha=\{\phi,\eta\}.
 \label{eq:relaip}
\end{equation}
In the case where $V=L^\bC$ is the complexification of $L$ and $\alpha$ is given by complex conjugation we recover precisely the usual K\"ahler polarization setting and the relation between the positive-definite inner products on $L^+=L^\pol$ and on $L=(L^\bC)^\alpha$ given by~(\ref{eq:relaip}) becomes~(\ref{eq:relkip}). This is then the basis for the usual K\"ahler quantization of the phase space $L$.

In general the situation of interest is the following: We are given a real phase space $L$ with symplectic form $\omega$. Moreover, we are given a pair $L^\pol$, $L^\cpol$ of transversal complex Lagrangian subspaces of $L^\bC$, (encoding vacua on the two sides of a hypersurface). The task is then to find a~suitable compatible real structure $\alpha$ on $L^\bC$ so that $L^\pol$ encodes an $\alpha$-K\"ahler polarization. That is, $(i)$ $\alpha$ has to be a compatible real structure, $(ii)$ $\alpha$ has to interchange $L^\pol$ and $L^\cpol$, and~$(iii)$, the bilinear form~(\ref{eq:alphaip}) has to be positive-definite on $L^\pol$.

\subsection[Modified $\sst$-structure and GNS construction]{Modified $\boldsymbol{\sst}$-structure and GNS construction}
\label{sec:modstar}

We suppose in this section that we are given on the complexified symplectic phases space $(L_{\Sigma},\omega_{\Sigma})$ a compatible real structure $\alpha_{\Sigma}\colon L_{\Sigma}\to L_{\Sigma}$ as well as an $\alpha$-K\"ahler polarization $L_{\Sigma}^\pol\subseteq L_{\Sigma}^\bC$. We then show that this induces a modified $\sst$-structure on the algebra $\qsoa_{\Sigma}$ of slice observables. Moreover, the GNS construction, carried out in analogy to Section~\ref{sec:ssgns}, but with respect to the modified $\sst$-algebra, yields a Hilbert space of states.

If the real structure on $L_{\Sigma}^\bC$ is changed, then this affects the $\sst$-structure on the classical algebra of slice observables in a canonical way. Namely, the standard structure~(\ref{eq:stdss}) is replaced by,
\begin{equation}
 F^\alpha(\phi)\defeq \overline{F\left(\alpha(\phi)\right)}.
 \label{eq:alphass}
\end{equation}
Note that we write $F^\alpha$ instead of $F^\sst$ to distinguish this $\sst$-structure from the standard one. This $\sst$-structure also carries over to the quantum algebra $\qsoa_{\Sigma}$ because the defining relation~(\ref{eq:weylrel}) of the latter is compatible with it. This holds true for the present $\sst$-structure in the same way as it does for the usual complex conjugation, since the relation~(\ref{eq:weylrel}) does not refer to any specific compatible real structure. We denote by $\qsoa^\alpha_{\Sigma}$ the quantum algebra of slice observables with $\sst$-structure induced from the compatible real structure $\alpha_{\Sigma}$.

We recall the linear functional $v_{\Sigma}\colon \qsoa^\alpha_{\Sigma}\to\bC$ encoding the vacuum and given by the vacuum correlation function~(\ref{eq:vacf}). Its evaluation on a Weyl observable $F=\exp(\im D)$ was discussed in~Section~\ref{sec:vevweylslice}. The result is given by equation~(\ref{eq:vevweyl}), where the linear observable $D$ is determined by $\xi\in L_{\Sigma}^\bC$ via relation~(\ref{eq:lsobsdef}). In terms of the bilinear form~(\ref{eq:cipspol}) this is
\begin{equation*}
 v_{\Sigma}(F)=\exp\bigg({-}\frac{1}{4}\{\xi,\xi\}_{\Sigma}\bigg),
\end{equation*}
directly generalizing the notation of the K\"ahler case, see equation~(\ref{eq:vevweylk}). If $L_{\Sigma}^\pol$ is not a K\"ahler polarization this is not a positive functional on $\qsoa_{\Sigma}$, but since $L_{\Sigma}^\pol$ is an $\alpha$-K\"ahler polarization it is a positive functional on $\qsoa_{\Sigma}^\alpha$.

The sesquilinear form on $\qsoa_{\Sigma}^\alpha$ induced by $v_{\Sigma}$ is now given by
\begin{equation*}
 [G,F]_{\Sigma}^\alpha\defeq v_{\Sigma}(G^\alpha \qp F),
\end{equation*}
generalizing expression~(\ref{eq:sokip}). Moreover, for Weyl slice observables $F$ and $G$ determined by $\xi_D,\xi_E\in L_{\Sigma}^\bC$ as in~Section~\ref{sec:ssgns}, we get, generalizing~(\ref{eq:ipgns}),
\begin{equation*}
 [G,F]_{\Sigma}^\alpha
 =\exp\bigg({-}\frac{1}{4}\{\xi_D,\xi_D\}_{\Sigma} -\frac{1}{4}\{\alpha_{\Sigma}(\xi_E),\alpha_{\Sigma}(\xi_E)\}_{\Sigma} +\frac{1}{2}\{\xi_D,\alpha_{\Sigma}(\xi_E)\}_{\Sigma}\bigg).
\end{equation*}
In complete analogy to the K\"ahler case (Section~\ref{sec:ssgns}) the left ideal $\qsoi_{\Sigma}^\alpha\subseteq\qsoa_{\Sigma}^\alpha$ on which the sesquilinear form vanishes is generated by the relation $F\sim \one$ for the Weyl observables $F=\exp(\im D)$ with $\xi_D\in L_{\Sigma}^\pol$. The Hilbert space of the GNS construction on which $\qsoa_{\Sigma}^\alpha$ will be represented is then obtained as the completion of the quotient $\qsoa_{\Sigma}^\alpha/\qsoi_{\Sigma}^\alpha$. We denote this Hilbert space by $\cH_{\Sigma}^\alpha$.

As in the K\"ahler case, the Hilbert space $\cH_{\Sigma}^\alpha$ may alternatively be constructed as a Fock space or as a space of square-integrable $J_{\Sigma}$-holomorphic functions on $L_{\Sigma}^\alpha$. All its properties (including particle states and coherent states) are exactly as outlined in~Section~\ref{sec:stdstate}, except that we have to replace $L_{\Sigma}$ everywhere by $L_{\Sigma}^\alpha$ and complex conjugation on $L_{\Sigma}^\bC=(L_{\Sigma}^\alpha)^\bC$ by the map $\alpha_{\Sigma}$. What is more, the whole quantization scheme outlined in~Section~\ref{sec:stdstate} including hypersurface orientation reversal and hypersurface decomposition carries over to the $\alpha$-K\"ahler setting. Note that this requires $\alpha_{\Sigma}$ to be the same for both hypersurface orientations. That is, $\alpha_{\overline{\Sigma}}=\alpha_{\Sigma}$, which we assume from now onwards. It also raises the issue of locality for $\alpha_{\Sigma}$. That is, for a~hyper\-sur\-face decomposition, $\alpha_{\Sigma}$ needs to decompose accordingly, compare the similar discussion for K\"ahler polarizations at the end of Section~\ref{sec:stdstate}. With a consistent assignment of compatible real structures $\alpha_{\Sigma}$ to hypersurfaces $\Sigma$ (in addition to polarizations) we again satisfy Axioms~(T1), (T1b), (T2), (T2b) of Appendix~\ref{sec:qobsaxioms} upon quantization. What is more, since the functional $v_{\Sigma}$ is determined by the correlation function~(\ref{eq:vacf}), the GNS construction automatically ensures that Axiom~(T3x) of Appendix~\ref{sec:qobsaxioms} is satisfied, as in the K\"ahler case.

\subsection{Correlation functions and composition}
\label{sec:corrak}

We proceed to explore the interplay between the quantization based on observables presented in~Section~\ref{sec:genquant} and the Hilbert spaces constructed in this section, at the level of amplitudes and correlation functions. Firstly, it is clear that correlation functions continue to be well-defined by Theorem~\ref{thm:wfac}. We merely have to reinterpret the correlation functions in terms of boun\-dary observables as correlation functions in terms of boun\-dary states, as prescribed by equations~(\ref{eq:rencohs}),~(\ref{eq:recohs}) and~(\ref{eq:remps}). One might then expect to precisely recover Theorem~\ref{thm:stdcorrfact} describing the structure of correlation functions in the K\"ahler polarization case, as long as we replace everywhere $L_{\partial M}$ by $L_{\partial M}^\alpha$. There is a subtlety, however. Namely, there is an implicit assumption of compatibility of complex conjugation with the complexified interior space of solutions $L_M^\bC\subseteq L_{\partial M}^\bC$. More precisely, the assumption is that $L_M^\bC$ is invariant under complex conjugation. Since $L_M^\bC$ is the complexification of $L_M$, this is trivially the case. However, when we generalize to the $\alpha$-K\"ahler polarization, this condition takes the form
\begin{equation}
 \alpha_{\partial M}\big(L_M^\bC\big)=L_M^\bC.
 \label{eq:aintcomp}
\end{equation}
If this holds we may replace everywhere $L_M$ by $L_M^\alpha\defeq \big\{\phi\in L_M^\bC\colon \alpha_{\partial M}(\phi)=\phi\big\}$. Except for the substitutions, Theorem~\ref{thm:stdcorrfact} then holds as stated. Otherwise, even the formulation of~Theo\-rem~\ref{thm:stdcorrfact} as given in~Section~\ref{sec:corrfack} does not necessarily make sense. This is because it relies on the real orthogonal decomposition (with respect to $g_{\partial M}^\alpha$),
\begin{equation*}
 L_{\partial M}^\alpha=L_{M}^\alpha\oplus J_{\partial M} L_M^\alpha.
\end{equation*}
In the K\"ahler case this is \cite[Lemma~4.1]{Oe:holomorphic}. We shall refer to the compatibility condition~(\ref{eq:aintcomp}) as \emph{interior compatibility}.

Let us emphasize again that even without interior compatibility, the correlation functions are perfectly well-defined on the states of $\cH_{\partial M}^\alpha$ with their structure described by Theorem~\ref{thm:wfac}. To~more explicitly relate the boundary observables and states induced by their action, we note that by the $\alpha$-K\"ahler version of Corollary~\ref{cor:weylactvk} we have for $\xi\in L_{\Sigma}^\bC$
\begin{equation}
 K^\pol_\xi \cong \exp\left(\im\omega_{\Sigma}(\xi,\tau)\right) \ncoh^\alpha_{\tau},
 \label{eq:relsoca}
\end{equation}
where $\tau\defeq P^-_{\Sigma}(\xi)+\alpha_{\Sigma}(P^-_{\Sigma}(\xi))$.
This informal notation really stands for an equality of correlation functions,
\begin{equation*}
 \rho_M^F\big(K^\pol_\xi\big) = \exp\left(\im\omega_{\partial M}(\xi,\tau)\right) \rho_M^F(\ncoh^\alpha_{\tau}).
\end{equation*}
Here, $F$ is an arbitrary observable. We write a superscript $\alpha$ for the coherent states in $\cH^\alpha_{\Sigma}$. Recalling the discussion of Section~\ref{sec:nordsc}, relation~(\ref{eq:relsoca}) also clarifies the semiclassical interpretation of the coherent states in the $\alpha$-K\"ahler setting. That is, even though the latter are not labeled by real solutions they can be given a semiclassical interpretation through the mapping
\begin{equation*}
 \xi\mapsto P^-_{\Sigma}(\xi)+\alpha_{\Sigma}(P^-_{\Sigma}(\xi)),
\end{equation*}
when $\xi\in L_{\Sigma}$ is real.

As for the axiomatic system of GBQFT of Appendix~\ref{sec:qobsaxioms} with respect to amplitudes, we note that we satisfy Axiom~(TO4). Also, there is no difficulty in satisfying the disjoint composition Axiom~(TO5a). For the self-composition Axiom~(TO5b), however, we might expect that the corresponding Theorem~\ref{thm:stdcompo} in the K\"ahler case does depend on the interior compatibility condition~(\ref{eq:aintcomp}). There is certainly no difficulty if this is satisfied for all involved regions. Note that in the special case of a slice region, interior compatibility is satisfied as long as $\alpha$ is the same on both boundary hypersurfaces. This implies in particular that Proposition~\ref{prop:weylactk} on the action of slice observables on the Hilbert space of states carries over straightforwardly to the $\alpha$-K\"ahler setting, providing agreement with the GNS construction as in the K\"ahler case.

It turns out that even when the interior compatibility condition is not satisfied, a suitably generalized version of the composition Theorem~\ref{thm:stdcompo} holds, satisfying Axiom~(TO5b). To recall the context (see Section~\ref{sec:corrfack} for the K\"ahler case), consider a region $M$ with boundary $\partial M=\Sigma_1\cup \Sigma\cup \overline{\Sigma'}$, where $\Sigma'$ is a copy of $\Sigma$. Denote by $M_1$ the region obtained by gluing $M$ to itself along $\Sigma$ with $\Sigma'$, see Figure~\ref{fig:compobs}. We have polarizations $L_{\partial M}^\pol\subseteq L_{\partial M}^\bC$ and $L_{\partial M_1}^{\pol_1}\subseteq L_{\partial M_1}^\bC$, transversal to the interior polarizations $L_M^\bC\subseteq L_{\partial M}^\bC$ and $L_{M_1}^\bC\subseteq L_{\partial M_1}^\bC$ respectively.

We shall make the additional assumption that an element in $L_{\partial M_1}^{\pol_1}$ ``extends'' to an element in~$L_{\partial M}^{\pol}$ exactly in the same way as a solution in $L_{M_1}^\bC$ extends to a solution in $L_M^\bC$ by Axiom~(C7), equation~(\ref{eq:solext}) of Appendix~\ref{sec:caxioms}. This is justified by thinking of elements in $L_{\partial M_1}^{\pol_1}$, $L_{\partial M}^{\pol}$ as ``exterior solutions'', compare~\cite{CoOe:vaclag}. Note that this assumption is stronger than in Theorem~\ref{thm:scompo}, where the polarizations $L_{\partial M}^\pol$ and $L_{\partial M_1}^{\pol_1}$ are completely independent. On the other hand it is weaker than in Theorems~\ref{thm:stdcompo} and \ref{thm:sfcompo}, where $L_{\partial M}^\pol$ factorizes in terms of the decomposition of $\partial M$.

We shall not need to assume $\alpha$-K\"ahler polarizations. We shall merely assume that we are given a real structure $\alpha_{\Sigma}\colon L_{\Sigma}^\bC\to L_{\Sigma}^\bC$ on $L_{\Sigma}^\bC$ (not necessarily compatible with $\omega_\Sigma$). Moreover, we~assume we are given a real positive definite inner product $g_{\Sigma}^\alpha\colon L_{\Sigma}^\alpha\times L_{\Sigma}^\alpha\to\R$ making $L_{\Sigma}^\alpha$ into a~real separable Hilbert space. We do not assume that this is related to the symplectic form or a~complex structure as e.g., in equation~(\ref{eq:gip}). Recall that the inner product gives rise to a Gaussian measure $\nu_{\Sigma}^\alpha$, or simply denoted $\nu$, on an extension $\hat{L}_{\Sigma}^\alpha$ of the vector space~$L_{\Sigma}^\alpha$~\cite{Oe:holomorphic}, compare also Section~\ref{sec:stdstate}. We formulate the composition rule in the following not in terms of a sum over a~complete basis (as in Axiom~(TO5b) of Appendix~\ref{sec:qobsaxioms} or as in~Section~\ref{sec:corrfack}), but equivalently in terms of an integral over $\hat{L}_{\Sigma}^\alpha$, using a completeness relation analogous to formula~(\ref{eq:cohcompl}). For~relevant discussion of this point, see~\cite{Oe:holomorphic,Oe:feynobs}. What is more, we formulate the gluing rule only for Weyl observables, analogous to coherent states. As usual this is sufficient.

\begin{dfn}
\label{dfn:aadmiss}
We say the gluing is \emph{admissible} if the function
 \begin{equation*}
 \xi\mapsto \rho_M\big(K^\pol_{(0,\xi,\xi)}\big)\,\exp\bigg(\frac{1}{2} g_{\Sigma}^\alpha(\xi,\xi)\bigg)
 \end{equation*}
 is \emph{integrable}, i.e., is an element of $\rL^1\big(\hat{L_{\Sigma}^\alpha}\big)$. In this case the \emph{gluing anomaly factor} is defined as
 \begin{equation}
 \label{eq:defglanom}
 c\big(M;\Sigma,\overline{\Sigma'}\big)\defeq \int_{\hat{L}_{\Sigma}^\alpha} \rho_M\big(K^\pol_{(0,\xi,\xi)}\big)\,\exp\bigg(\frac{1}{2} g_{\Sigma}^\alpha(\xi,\xi)\bigg)\, \xd\nu(\xi).
 \end{equation}
\end{dfn}

\begin{thm}
\label{thm:acompo}
 In the above context, suppose the gluing is admissible. Let $D\colon K_M^\bC\to \bC$ be a~linear observable in $M$ and $D_1\colon K_{M_1}^\bC\to\bC$ the induced linear observable in $M_1$. Set $F=\exp(\im\, D)$ and $F_1=\exp(\im\, D_1)$. Then, for any $\phi\in L_{\partial M_1}^\bC$,
\begin{equation*}
 \rho_{M_1}^{F_1}\big(K^{\pol_1}_{\phi}\big)\, c\big(M;\Sigma,\overline{\Sigma'}\big) =\int_{\hat{L}_{\Sigma}^\alpha} \rho_M^F\big(K^\pol_{(\phi,\xi,\xi)}\big)\,\exp\bigg(\frac{1}{2} g_{\Sigma}^\alpha(\xi,\xi)\bigg)\, \xd\nu(\xi).
 \end{equation*}
\end{thm}
The proof is provided in Appendix~\ref{sec:compip}.

\subsection{Wave function of vacuum change}
\label{sec:wfvac}

If there is a Hilbert space of states as in the traditional K\"ahler quantization scheme (Section~\ref{sec:kquant}) or~the present generalized $\alpha$-K\"ahler quantization scheme, it is natural to ask what state encodes a~change of vacuum. The corresponding question for slice observables was answered in~Sec\-tion~\ref{sec:chvacso}. Thus, to get this state we would apparently just have to let this observable, determined by equation~(\ref{eq:chvacobs}), act on the vacuum state $\coh_0$ of the Hilbert space in the sense of Section~\ref{sec:qsobsopk} or in the generalized sense of the $\alpha$-K\"ahler setting of the present section. Since the observable is the exponential of a quadratic form this is not quite as straightforward as considering a Weyl observable for example. One difficulty arises from the fact that we cannot in general expect a~vacuum to live as a state in the Hilbert space build on a different vacuum. This manifests itself in that the ``state'' we are looking for is not normalizable in general. However, it turns out that it is well-defined as a (generally non-square-integrable) holomorphic wave function. We~shall call it a \emph{pseudo-state}.

The basis for our considerations will be the relation between vacuum states and amplitudes~\cite{CoOe:vaclag}. Consider a hypersurface $\Sigma$ and choose transversal polarizations $L_{\Sigma}^\pol$ and $L_{\Sigma}^\cpol$. We take $L_{\Sigma}^\pol$ to encode the vacuum on the side $\Sigma$ of the hypersurface while $L_{\Sigma}^\cpol$ encodes it on the opposite side, given by $\overline{\Sigma}$. Suppose we have a compatible positive-definite real structure $\alpha$ and have performed the corresponding $\alpha$-K\"ahler quantization yielding a Hilbert space $\cH_{\Sigma}^\alpha$. We now consider a~pola\-ri\-za\-tion $L_{\Sigma}^{\pol'}\subseteq L_{\Sigma}^\bC$ on $\Sigma$, different from $L_{\Sigma}^\pol$, but such that it is still transversal to $L_{\Sigma}^\cpol$.
Denote the pseudo-state encoding the polarization $L_{\Sigma}^{\pol'}$ on $\Sigma$ by $Y^\pol_{\pol'}$, in analogy to the notation in~Section~\ref{sec:chvacso} for the corresponding slice observable. We now imagine that $\overline{\Sigma}$ is the boundary~$\partial X$ of a region~$X$ whose space $L_X^\bC$ of interior solutions is precisely given by the polarization~$L_{\partial X}^{\pol'}$ on~$\overline{\Sigma}$, i.e., $L_X^\bC\defeq L_{\partial X}^{\pol'}$.\footnote{We recall that although the notation seems to suggest it, it is not meant to imply that $L_X^\bC$ is the complexification of some real subspace $L_X$.} For any state $\psi\in \cH_{\overline{\Sigma}}^\alpha$ we then have\footnote{Strictly speaking the inner product might not exist for all states $\psi$. It does exist, however, for coherent states~$\psi$, which is sufficient.}
\begin{equation*}
 \rho_X(\psi)=\big\langle \iota_{\overline{\Sigma}}(Y^{\pol}_{\pol'}),\psi\big\rangle_{\overline{\Sigma}}=\big\langle \iota_{\Sigma}(\psi), Y^{\pol}_{\pol'}\big\rangle_{\Sigma}.
\end{equation*}
This relation follows from the duality between amplitudes and vacuum wave functions~\cite{CoOe:vaclag}. Alter\-na\-ti\-vely, we can also see this as an implication of the composition rule (Theorem~\ref{thm:acompo} together with its disjoint version) if we imagine $\psi$ replaced by a region on the other side of $X$ with an~arbitrary observable in it.

To obtain the wave function of $Y^\pol_{\pol'}$ we use the reproducing property~(\ref{eq:reprod}) of coherent states,
\begin{equation*}
 Y^\pol_{\pol'}(\xi)=\big\langle \coh^\alpha_{\xi}, Y^\pol_{\pol'}\big\rangle_{\Sigma}
 = \big\langle \iota_{\Sigma}\big(\coh^\alpha_{\xi}\big), Y^\pol_{\pol'}\big\rangle_{\Sigma}
 = \rho_X\big(\coh_\xi^\alpha\big).
\end{equation*}
Note that we use the same notation $\coh^\alpha_{\xi}$ for the corresponding coherent state on either side (i.e., orientation) of the hypersurface $\Sigma$. We proceed to evaluate the amplitude $\rho_X$ with equation~(\ref{eq:amplg}) of Theorem~\ref{thm:wfac}, taking into account relation~(\ref{eq:recohs}). Denote the decomposition $L_{\Sigma}^{\bC}=L_{\Sigma}^\pol\oplus L_{\Sigma}^{\cpol}$ by $\xi=\xi^++\xi^-$ and $L_{\Sigma}^{\bC}=L_{\Sigma}^{\pol'}\oplus L_{\Sigma}^{\cpol}$ by $\xi=\xi^{\mathrm{X}}+\xi^{\mathrm{M}}$. Then, for $\xi\in L_{\Sigma}^\alpha$,
\begin{align*}
 \rho_X\big(\coh_\xi^\alpha\big)&
 =\rho_X\big(K^\cpol_{\xi^+}\big)=\exp\big(\im\,\omega_{\partial X}\big(\xi^+,(\xi^+)^X\big)\big) =\exp\big(\im\,\omega_{\partial X}\big(\xi^+,\xi^X\big)\big)
 \\
 & = \exp\big(\im\,\omega_{\Sigma}\big(\xi^X,\xi^+\big)\big).
\end{align*}
We note that the right-hand side does not explicitly depend on the region $X$. That is, it makes sense for any hypersurface $\Sigma$, depending only on the polarizations $L_{\Sigma}^{\pol}$, $L_{\Sigma}^{\pol'}$ and $L_{\Sigma}^{\cpol}$, with the condition that each of the first two is transversal to the third. We take this to mean that the pseudo-state $Y^\pol_{\pol'}$ and its wave function make sense on hypersurfaces that are not necessarily boundaries of a region. We obtain the wave function
\begin{equation}
 Y^\pol_{\pol'}(\xi)
 = \exp\big(\im\,\omega_{\Sigma}\big(\xi^X,\xi^+\big)\big)
 = \exp\bigg(\frac{1}{4}\big\{\xi^X,\xi\big\}_{\Sigma}\bigg).
 \label{eq:chvacwf}
\end{equation}
This formula is even independent of the choice of $\alpha_{\Sigma}$. However, given a compatible choice of $\alpha_{\Sigma}$, it is easily seen to define a holomorphic function on $L_{\Sigma}^\alpha$ with respect to the complex structure~$J_{\Sigma}$.
By construction, the wave function~(\ref{eq:chvacwf}) is normalized to yield unity when taking its inner product with the standard vacuum. Whether it can be normalized to correspond to an~ordinary state in the Hilbert space depends on the polarizations. The standard criterion in the case of~K\"ahler polarizations is that the difference between the corresponding complex structures has to be a Hilbert--Schmidt operator~\cite{Sha:linsymboson}.

To facilitate the interpretation of formula~(\ref{eq:chvacwf}) in terms of a superposition of multi-particle states we use the completeness relation~(\ref{eq:lcomplint}) for $L_{\Sigma}^\alpha$. With this we can rewrite the wave function~(\ref{eq:chvacwf}) as
\begin{gather}
 Y^\pol_{\pol'}(\xi)
 =\exp\bigg(\frac{1}{8}\int_{\hat{L}_{\Sigma}^\alpha}\big\{\xi^X,\phi\big\}_{\Sigma} \{\phi,\xi\}_{\Sigma}\xd\nu_{\Sigma}(\phi)\bigg)\nonumber\\
 \hphantom{Y^\pol_{\pol'}(\xi)}{}
 =\exp\bigg({-}\frac{1}{8}\int_{\hat{L}_{\Sigma}^\alpha}\big\{\tilde{\phi},\xi\big\}_{\Sigma} \{\phi,\xi\}_{\Sigma}\xd\nu_{\Sigma}(\phi)\bigg).
 \label{eq:wvfintexp}
\end{gather}
Here, $\tilde{\phi}\defeq (\phi^+)^M+\alpha\big((\phi^+)^M\big)$.
The significance of the integrand of the right-hand expression is that we can interpret it plainly as a 2-particle wave function with quantum numbers $\phi$ and $\tilde{\phi}$ respectively, compare expression~(\ref{eq:mswf}). It is also instructive to rewrite this in terms of creation operators, compare expression~(\ref{eq:coact}),
\begin{equation*}
 Y^\pol_{\pol'} =\exp\bigg({-}\frac{1}{4}\int_{\hat{L}_{\Sigma}^\alpha}
 a^\dagger_{\tilde{\phi}}\, a^\dagger_{\phi}\, \xd\nu_{\Sigma}(\phi)\bigg) \coh^\alpha_0.
\end{equation*}

\subsection{Vacuum change and Bogoliubov coefficients}
\label{sec:bogvac}

The traditional approach to quantization in curved spacetime relies on choices of basis~\cite{BiDa:qftcurved}. Accordingly, a change of vacuum is parametrized in terms of \emph{Bogoliubov coefficients} that relate the relevant basis. We examine in the following how the traditional language and results are recovered from the present setting.
At the same time we generalize these results to the $\alpha$-K\"ahler setting.

We suppose the same setting as in the previous Section~\ref{sec:wfvac}. That is, we have a hyper\-sur\-face~$\Sigma$ and transversal polarizations $L_{\Sigma}^\pol$ and $L_{\Sigma}^\cpol$. We also assume that we have a compatible positive-definite real structure $\alpha_{\Sigma}$. Let $\{u_k\}_{k\in I}$ be an orthonormal basis of $L_{\Sigma}^\pol$ with respect to the inner product~(\ref{eq:alphaip}). Then, $\{\alpha_{\Sigma}(u_k)\}_{k\in I}$ is an orthonormal basis of~$L_{\Sigma}^\cpol$ with respect to the negative of~the inner product. We have
\begin{equation*}
 (u_k,u_l)_{\Sigma}^\alpha=\delta_{k,l},\qquad (\alpha_{\Sigma}(u_k),\alpha_{\Sigma}(u_l))_{\Sigma}^\alpha=-\delta_{k,l}, \qquad (u_k,\alpha_{\Sigma}(u_l))_{\Sigma}^\alpha=0.
\end{equation*}
This is in complete analogy to the K\"ahler case~(\ref{eq:modeip}), which we recover when $\alpha_{\Sigma}$ is ordinary complex conjugation. Also, as in the K\"ahler case~(\ref{eq:kbaserel}), the orthonormal basis $\{u_k\}_{k\in I}$ is in~one-to-one correspondence to an orthonormal basis $\{w_k\}_{k\in I}$ of $L_{\Sigma}^\alpha$ with respect to~(\ref{eq:cipspol}) via $w_k=u_k+\alpha_{\Sigma}(u_k)$.
We suppose that we have another polarization $L_{\Sigma}^{\pol'}$ on $\Sigma$. We require that the inner product~(\ref{eq:alphaip}) is also positive-definite on $L_{\Sigma}^{\pol'}$. This implies that $L_{\Sigma}^{\pol'}$ is transversal to~$L_{\Sigma}^{\cpol}$, an assumption made in~Section~\ref{sec:wfvac}. Let $\{u_k'\}_{\in I}$ be an orthonormal basis of $L_{\Sigma}^{\pol'}$.

We may use the basis $\{w_k\}_{k\in I}$ to expand the wave function of $Y^\pol_{\pol'}$ similar to the continuous expansion~(\ref{eq:wvfintexp}),
\begin{equation*}
 Y^\pol_{\pol'}(\xi)
 =\exp\bigg(\frac{1}{4}\sum_{i\in I}\big\{\xi^X,w_i\big\}_{\Sigma}\{w_i,\xi\}_{\Sigma}\bigg)
 =\exp\bigg({-}\frac{1}{4}\sum_{i\in I}\big\{u_i^M,\xi\big\}_{\Sigma}\{w_i,\xi\}_{\Sigma}\bigg).
\end{equation*}
Consider an $n$-particle state $\psi_{i_1,\dots,i_n}\in \cH_{\Sigma}^\alpha$ parametrized in terms of the same basis via
\begin{equation*}
 \psi_{i_1,\dots,i_n}(\xi)\defeq\{w_{i_1},\xi\}_{\Sigma}\cdots \{w_{i_n},\xi\}_{\Sigma}.
\end{equation*}
In our conventions (compare Section~\ref{sec:stdstate}) the inner product between such states is given by
\begin{equation*}
 \langle\psi_{j_1,\dots,j_n},\psi_{i_1,\dots,i_n}\rangle_{\Sigma}
 =2^n \sum_{\sigma\in S^n} \delta_{{j_1},i_{\sigma(1)}}\cdots\delta_{{j_n},i_{\sigma(n)}}.
\end{equation*}
We now consider the inner product between such a generic $n$-particle state and $Y^\pol_{\pol'}$. If $n$ is odd this vanishes, so suppose $n=2m$. As is easy to see we get
\begin{equation*}
 \big\langle Y^\pol_{\pol'},\psi_{i_1,\dots,i_{2m}}\big\rangle_{\Sigma}
 =\frac{(-1)^m}{m!} \sum_{\sigma\in S^{2m}} \Lambda_{i_{\sigma(1)},i_{\sigma(2)}}\cdots \Lambda_{i_{\sigma(2m-1)},i_{\sigma(2m)}}.
\end{equation*}
Here we have defined
\begin{equation*}
 \Lambda_{k,l}\defeq \sum_{i\in I} \big\{w_k,\alpha\big(u_i^M\big)\big\}_{\Sigma} \{w_l, w_i\}_{\Sigma}.
\end{equation*}
However, using orthogonality and further properties of the objects involved we get
\begin{align*}
 \Lambda_{k,l}
 &= \big\{\alpha_{\Sigma}(u_k),\alpha_{\Sigma}\big(u_{l}^M\big)\big\}_{\Sigma}
 = 4\im\omega_{\Sigma}\big(\alpha_{\Sigma}(u_k),\alpha_{\Sigma}\big(u_{l}^M\big)\big)
 = 4\im\overline{\omega_{\Sigma}\big(u_k,u_l^M\big)}
 = 4\im\overline{\omega_{\Sigma}\big(u_k^X,u_l\big)}
 \\
 &= 4\im\omega_{\Sigma}\big(\alpha_{\Sigma}\big(u_k^X\big),\alpha_{\Sigma}(u_l)\big)
 = \big(u_k^X,\alpha_{\Sigma}(u_l)\big)^\alpha_{\Sigma}.
\end{align*}
We expand $u_k^X$ in terms of the basis $\{u'_i\}_{i\in I}$ via
\begin{equation*}
 u_k^X=\sum_j c_{k,j} u_j'.
\end{equation*}
The coefficients $c_{k,j}$ can be related to the Bogoliubov coefficients in the expansion
\begin{equation*}
 u_i'=\sum_{j\in I} a_{i,j} u_j + \sum_j b_{i,j} \alpha_{\Sigma}(u_j),
\end{equation*}
where $a_{i,j}=(u_j,u_i')^\alpha_{\Sigma}$ and $b_{i,j}=-(\alpha_{\Sigma}(u_j),u_i')^\alpha_{\Sigma}$.
To this end, consider
\begin{equation*}
 \delta_{i,k}=(u_i,u_k)^\alpha_{\Sigma}=\big(u_i,u_k^X\big)^\alpha_{\Sigma}=\sum_{j\in I} c_{k,j} (u_i,u_j')^\alpha_{\Sigma}
 =\sum_{j\in I} c_{k,j} a_{j,i}.
\end{equation*}
That is, $c$ as a matrix is inverse to $a$, $c_{k,j}=\big(a^{-1}\big)_{k,j}$. With this we may rewrite $\Lambda_{k,l}$ as follows,
\begin{equation*}
 \Lambda_{k,j}=\sum_{j\in I} \overline{c_{k,j}} (u_j',\alpha_{\Sigma}(u_l))^\alpha_{\Sigma}
 = -\sum_{j\in I} \overline{(a^{-1})_{k,j} b_{j,l}}.
\end{equation*}
Up to minor changes of conventions, this recovers in the K\"ahler case the well known formulas of DeWitt~\cite{Dew:qftcurved}.

\subsection{Return to reality}
\label{sec:retrel}

Working with the space $L_{\Sigma}^\alpha$ instead of with $L_{\Sigma}$ might appear as an additional complication. This motivates looking for an identification of the two spaces in order to ``pull back'' all the structures from $L_{\Sigma}^\alpha$ to $L_{\Sigma}$. This would allow expressing everything again in terms of $L_{\Sigma}$, as if we were working with a K\"ahler polarization. We shall assume here that we are given such an~identification, i.e., a linear real bijection $I_{\Sigma}\colon L_{\Sigma}\to L_{\Sigma}^\alpha$.
Define
\begin{gather}
 \tilde{\omega}_{\Sigma}\colon \qquad L_{\Sigma}\times L_{\Sigma}\to\R,\quad
 \text{by}
 \quad\tilde{\omega}_{\Sigma}(\phi,\eta)\defeq \omega_{\Sigma}(I_{\Sigma}(\phi),I_{\Sigma}(\eta)), \label{eq:pbsym}
 \\
 \tilde{J}_{\Sigma}\colon \qquad L_{\Sigma}\to L_{\Sigma},\qquad\quad
 \text{by}
 \quad \tilde{J}\phi \defeq I_{\Sigma}^{-1}(J_{\Sigma}I_{\Sigma}(\phi)), \label{eq:pbcs}
 \\
 \tilde{g}_{\Sigma}\colon \qquad L_{\Sigma}\times L_{\Sigma}\to\R,\quad
 \,\text{by}
 \quad\tilde{g}_{\Sigma}(\phi,\eta)\defeq 2\tilde{\omega}_{\Sigma}(\phi,\tilde{J}_{\Sigma}\eta)=g_{\Sigma}(I_{\Sigma}(\phi),I_{\Sigma}(\eta)),
 \\
 \{\cdot,\cdot\}^{\tilde{}}_{\Sigma}\colon \ \ L_{\Sigma}\times L_{\Sigma}\to\bC,\quad
\text{by}
 \quad \{\phi,\eta\}^{\tilde{}}_{\Sigma}\defeq \tilde{g}_\Sigma(\phi,\eta) +2\im\tilde{\omega}_{\Sigma}(\phi,\eta)=\{I_{\Sigma}(\phi),I_{\Sigma}(\eta)\}_{\Sigma}.
\end{gather}
Then, $\tilde{\omega}_{\Sigma}$ is a symplectic form on $L_{\Sigma}$, $\tilde{J}_{\Sigma}$ is a compatible complex structure, $\tilde{g}_{\Sigma}$ is a real inner product and $\{\cdot,\cdot\}^{\tilde{}}_{\Sigma}$ a complex inner product with respect to $\tilde{J}_{\Sigma}$. (We shall assume completeness for the inner products.) In the following we shall use the same notation $I_{\Sigma}$ to denote the complexification of the map, i.e., its complex linear extension $L_{\Sigma}^\bC\to L_{\Sigma}^\bC$. Pulling back the polarization with $I_{\Sigma}$ leads to an ordinary K\"ahler polarization $L_{\Sigma}^{\tilde{\pol}}\defeq I_{\Sigma}^{-1}\big(L_{\Sigma}^\pol\big)$.

With these ingredients we may construct a Hilbert space $\tilde{\cH}_{\Sigma}$ based on $L_{\Sigma}$, and the modified structures $\tilde{\omega}_{\Sigma}$, etc., exactly in the usual way as outlined in~Section~\ref{sec:stdstate}. This Hilbert space is then equivalent to $\cH_{\Sigma}^\alpha$ with the unitary equivalence given by the map $\tilde{U}_{\Sigma}\colon \tilde{\cH}_{\Sigma}\to\cH_{\Sigma}^\alpha$ which on~coherent states takes the simple form
\begin{equation*}
 \tilde{U}_{\Sigma}\big(\tilde{\coh}_{\xi}\big)= \coh_{I_{\Sigma}(\xi)}^\alpha.
\end{equation*}
Here we denote the coherent states in $\tilde{\cH}_{\Sigma}$ with a tilde and those in $\cH_{\Sigma}^\alpha$ with a superscript $\alpha$.
This equivalence extends to the creation and annihilation operators and their actions in the obvious way. It is also straightforward to express $\tilde{U}_{\Sigma}$ in terms of holomorphic wave functions, mapping those on $L_{\Sigma}$ (holomorphic with respect to $\tilde{J}_{\Sigma}$) to those on $L_{\Sigma}^\alpha$ (holomorphic with respect to $J_{\Sigma}$). For $\tilde{\psi}\in\tilde{\cH}_{\Sigma}$ and $\xi\in L_{\Sigma}$ we have
\begin{equation*}
 \tilde{U}_{\Sigma}\big(\tilde{\psi}\big)(I_{\Sigma}(\xi))=\tilde{\psi}(\xi).
\end{equation*}

We proceed to consider quantization in regions, i.e., the construction of amplitudes. Specifically, we are looking to reproduce an analogue of relation~(\ref{eq:amplk}) of Theorem~\ref{thm:stdcorrfact}, expressing the amplitude for coherent states in a region $M$ in terms of the inner product $\tilde{g}_{\partial M}$. As already discussed in~Section~\ref{sec:corrak} to work at the level of the spaces $L_{\Sigma}^\alpha$ this requires the interior compatibility condition for $\alpha_{\Sigma}$,~(\ref{eq:aintcomp}). To additionally work at the level of the pull-back considered here, we also need $I_{\partial M}$ to be compatible with the inclusion $L_M^\bC\subseteq L_{\partial M}^\bC$. That is, we need the \emph{interior compatibility condition} for $I_{\partial M}$,
\begin{equation*}
 I_{\partial M}\big(L_{M}^\bC\big)=L_M^\bC.
\end{equation*}
The induced amplitudes are, for $\tilde{\psi}\in \tilde{\cH}_{\partial M}$,
\begin{equation*}
 \tilde{\rho}_M\big(\tilde{\psi}\big)=\rho_M\big(\tilde{U}_{\partial M}\big(\tilde{\psi}\big)\big).
\end{equation*}
With this, we have the analogue of relation~(\ref{eq:amplk}), for $\xi\in L_{\partial M}$,
\begin{equation}
 \tilde{\rho}_M\big(\tilde{\ncoh}_\xi\big)
 = \exp\bigg({-}\frac{\im}{2} \tilde{g}_{\partial M}\big(\xi^{\tilde{\mathrm{R}}},\xi^{\tilde{\mathrm{I}}}\big)
 -\frac{1}{2} \tilde{g}_{\partial M}\big(\xi^{\tilde{\mathrm{I}}},\xi^{\tilde{\mathrm{I}}}\big)\bigg).
 \label{eq:ampli}
\end{equation}
Here $\xi=\xi^{\tilde{\mathrm{R}}}+ \tilde{J}_{\Sigma}\xi^{\tilde{\mathrm{I}}}$ with $\xi^{\tilde{\mathrm{R}}},\xi^{\tilde{\mathrm{I}}}\in L_{M}$. Equation~(\ref{eq:ampli}) arises from~(\ref{eq:amplk}) because we have $I_{\partial M}(\xi)^{\textrm{R}}=I_{\partial M}\big(\xi^{\tilde{\textrm{R}}}\big)$ and $I_{\partial M}(\xi)^{\textrm{I}}=I_{\partial M}\big(\xi^{\tilde{\textrm{I}}}\big)$. In particular, equation~(\ref{eq:ampli}) recovers precisely equation~(\ref{eq:amplk}), but with the original structures replaced by the pulled back ones.
In essence, we obtain an equivalence between the original quantum theory with $\alpha$-K\"ahler polarizations and a quantum theory with genuine K\"ahler polarizations $L_{\Sigma}^{\tilde{\pol}}$ and complex structures $\tilde{J}_{\Sigma}$ not only at the level of Hilbert spaces, but also of amplitudes. What is more, the theory with Hilbert spaces~$\tilde{\cH}_{\Sigma}$ and amplitudes $\tilde{\rho}_M$ arises by applying precisely the original K\"ahler quantization prescription of Section~\ref{sec:kquant}, but to the pulled back data.

We proceed to consider observables and their correlation functions. Pulling back correlation functions for observables $F$ simply amounts to
\begin{equation*}
 \tilde{\rho}_M^F\big(\tilde{\psi}\big)=\rho_M^F\big(\tilde{U}_{\partial M}\big(\tilde{\psi}\big)\big).
\end{equation*}
We turn to the structure Theorem~\ref{thm:stdcorrfact} for correlation functions in its entirety. In the factorization identity~(\ref{eq:wobsfact}) of correlation functions of Weyl observables, we have already discussed the first factor, the amplitudes~(\ref{eq:amplk}) and its replacement by relation~(\ref{eq:ampli}). As for the third factor, the vacuum correlation function~(\ref{eq:sfqvev}), there is no change as this does not depend on the choice of coherent state. However, it remains determined by the original polarization $L_{\partial M}^\pol$ and not by the pulled back one $L_{\partial M}^{\tilde{\pol}}$.
The second factor is of most interest. It takes the form
$F\big(I_{\partial M}\big(\hat{\tilde{\xi}}\big)\big)$ with $\hat{\tilde{\xi}}\defeq\xi^{\tilde{\mathrm{R}}}-\im \xi^{\tilde{\mathrm{I}}}$. In total we get
\begin{equation}
 \tilde{\rho}^F_M\big(\tilde{\ncoh}_\xi\big) =\tilde{\rho}_M\big(\tilde{\ncoh}_\xi\big) F\big(I_{\partial M}\big(\hat{\tilde{\xi}}\big)\big) \rho^F_M(\coh_0).
 \label{eq:wobsfacti}
\end{equation}
Unsurprisingly, this is structurally different from equation~(\ref{eq:wobsfact}). That is, at the level of obser\-vab\-les the actual theory (satisfying equation~(\ref{eq:wobsfacti})) is not equivalent to the theory obtained by directly quantizing, according to Section~\ref{sec:kquant}, the pulled back data (which would satisfy equation~(\ref{eq:wobsfact})).

The factorization equation~(\ref{eq:wobsfacti}), and in particular the second factor, also make it clear that a pull back map $I_{\Sigma}$ is in general not suitable to achieve a semiclassical interpretation of coherent states. That is, the state $\coh_{I_{\Sigma}(\xi)}^\alpha$ is not in general a reasonable candidate for a~semi\-clas\-si\-cal counterpart of the solution $\xi$. However, in~Section~\ref{sec:corrak} we already saw what a~suit\-able counterpart might look like, compare equation~(\ref{eq:relsoca}).

A natural candidate for the identification map $I_{\Sigma}$ is given by the projection operator $P_{\Sigma}^\alpha$: $L_{\Sigma}^\bC\to L_{\Sigma}^\alpha$, that is,
\begin{equation}
 P_{\Sigma}^\alpha(\phi)=\frac{1}{2}(\phi+\alpha_{\Sigma}(\phi)).
 \label{eq:proja}
\end{equation}
Suppose that the restriction of $P_{\Sigma}^\alpha$ to the real subspace $L_{\Sigma}\subseteq L_{\Sigma}^\bC$ is a bijection $I_{\Sigma}\colon L_{\Sigma}\to L_{\Sigma}^\alpha$ as~required. Then, if the interior compatibility condition is satisfied for the real structure $\alpha$, it is automatically also satisfied for $I_{\Sigma}$.

\subsection{A case of real polarizations}
\label{sec:realpol}

As shown in~\cite{CoOe:vaclag}, polarizations encoding vacua on different hypersurfaces and in different contexts in quantum field theory often decompose into a direct sum of a K\"ahler polarization and a real polarization. It is thus of particular interest to understand the latter in terms of $\alpha$-K\"ahler polarizations. To this end we shall assume that we are given what we will call a \emph{positive-definite reflection map}.

Suppose that we have on the hypersurface $\Sigma$ a real polarization, i.e., a pair of transversal complexified real Lagrangian subspaces $L^\pol_{\Sigma}\subseteq L_{\Sigma}^\bC$ and $L^\cpol_{\Sigma}\subseteq L_{\Sigma}^\bC$. Moreover, say $\gamma_{\Sigma}\colon L_{\Sigma}^\bC\to L_{\Sigma}^\bC$ is the complexification of a real-linear map $L_{\Sigma}\to L_{\Sigma}$ with the following properties:
\begin{itemize}\itemsep=0pt
\item $\gamma_{\Sigma}^2=\id$.
\item $\omega_{\Sigma}(\gamma_{\Sigma}(\phi),\gamma_{\Sigma}(\eta))=-\omega_{\Sigma}(\phi,\eta)$.
\item $\gamma_{\Sigma}$ exchanges polarizations, i.e., $\gamma_{\Sigma}\big(L^\pol_{\Sigma}\big)=L^\cpol_{\Sigma}$ and $\gamma_{\Sigma}\big(L^\cpol_{\Sigma}\big)=L^\pol_{\Sigma}$.
\end{itemize}
We then say that $\gamma_{\Sigma}$ is a \emph{reflection map} on $\Sigma$.
By construction $\gamma_{\Sigma}(\im \phi)=\im\gamma_{\Sigma}(\phi)$ and $\gamma_{\Sigma}\circ J_{\Sigma}=-J_{\Sigma}\circ\gamma_{\Sigma}$. Now define
\begin{equation}
 \alpha_{\Sigma}(\phi)\defeq -\im \overline{\gamma(\phi)}\qquad\text{for} \quad
 \phi\in L_{\Sigma}^\bC.
 \label{eq:alpharef}
\end{equation}
Then $\alpha_{\Sigma}$ is a compatible real structure on $L_{\Sigma}^\bC$. Moreover, we assume the inner product~(\ref{eq:alphaip}) is positive-definite on $L_{\Sigma}^\pol$. More directly, we can write this condition as
\begin{equation}
 \omega_{\Sigma}\big(\overline{\gamma_{\Sigma}(\phi)},\phi\big)>0\qquad
 \forall \phi\in L_{\Sigma}^\pol\setminus\{0\}.
 \label{eq:pdref}
\end{equation}
We then call $\gamma_{\Sigma}$ a \emph{positive-definite reflection map} and $L_{\Sigma}^\pol$ is an $\alpha$-K\"ahler polarization with conjugate $L_{\Sigma}^\cpol$. Moreover, in this situation the projector $P^\alpha_{\Sigma}\colon L_{\Sigma}^\bC\to L_{\Sigma}^\alpha$, see equation~(\ref{eq:proja}), restricted to $L_{\Sigma}$ becomes an isomorphism, as is easily verified. What is more, if we define the identification $I_{\Sigma}\colon L_{\Sigma}\to L_{\Sigma}^\alpha$ of Section~\ref{sec:retrel} to be
\begin{equation*}
 I_{\Sigma}=\sqrt{2} P_\Sigma^\alpha,
\end{equation*}
we get, compare~(\ref{eq:pbsym}),
\begin{equation*}
 \tilde{\omega}_{\Sigma}(\phi,\eta)=\omega_{\Sigma}(\phi,\eta).
\end{equation*}
That is, the pulled back symplectic form is identical to the original symplectic form. In particular, as long as no observables are considered, this means that the pulled back theory of~Section~\ref{sec:retrel} defined in this way differs from the original theory merely by replacing the real polarizations $L_{\Sigma}^\pol$ with K\"ahler polarizations given in terms of the complex structures $\tilde{J}_\Sigma$ given by~(\ref{eq:pbcs}).

We proceed to consider a more concrete way to construct positive-definite reflection maps. To this end we suppose that we have a decomposition of the space $L_\Sigma$ of germs of solutions on $\Sigma$ in terms of a direct sum of (real) Lagrangian subspaces $L_{\Sigma}=N_{\Sigma} \oplus M_{\Sigma}$. These need not coincide with polarizations determining the vacuum. For intuition, we shall think of $N_{\Sigma}$ as encoding ``positions'' (field values) and of $M_{\Sigma}$ as encoding ``momenta'' (normal derivatives). There is an important class of examples where this designation is indeed meaningful. We now define
\begin{equation}
 \gamma_{\Sigma}(n+m)\defeq n-m, \qquad\forall n\in N_{\Sigma},\quad \forall m\in M_{\Sigma}.
 \label{eq:nmreflect}
\end{equation}
Intuitively speaking, what this map does is it inverts field derivatives on the hypersurface. That is, it ``reflects'' initial data. We shall assume that it also interchanges the polarizations determining the vacuum on the two sides of $\Sigma$, i.e., $L^\pol_{\Sigma}$ and $L^\cpol_{\Sigma}$. It is then easy to check that $\gamma_{\Sigma}$ is a reflection map according to our previous definition.

Even though Klein--Gordon theory does not require a reflection map as it admits a perfectly fine K\"ahler quantization we shall consider it here to illustrate the notion of reflection. On an equal-time hypersurface at time $t$ the subspaces $N_t$ and $M_t$ of $L_t$ are given by
\begin{gather}
 M_t=\big\{\phi\in L_t^\bC\colon \phi^{\rm a}(k)=-{\rm e}^{2\im Et} \overline{\phi^{\rm b}(-k)}\,\forall k\big\}, \nonumber
 \\
 N_t=\big\{\phi\in L_t^\bC\colon \phi^{\rm a}(k)={\rm e}^{2\im Et} \overline{\phi^{\rm b}(-k)}\,\forall k\big\}.
 \label{eq:kgmndec}
\end{gather}
With this we obtain
\begin{equation}
 (\gamma_t(\phi))^{\rm a}(k)={\rm e}^{2\im Et} \overline{\phi^{\rm b}(-k)},\qquad
 (\gamma_t(\phi))^{\rm b}(k)={\rm e}^{2\im Et} \overline{\phi^{\rm a}(-k)}.
 \label{eq:kggamma}
\end{equation}
Viewing elements of $L_t^\bC$ as global solutions this is in spacetime terms,
\begin{equation}
 (\gamma_t(\phi))(t',x)=\phi(2t-t',x).
 \label{eq:kgreflect}
\end{equation}
That is, we obtain precisely a reflection of the solution on the hyperplane at time $t$. It is also straightforward to check that $\gamma_t$ interchanges the subspaces $L_t^+$ and $L_t^-$ of negative and positive energy solutions determining the vacuum, compare~(\ref{eq:kgpvac}). That is, $\gamma_t$ is a reflection map in the sense of our definition. On the other hand, $\gamma_t$ is not positive-definite in the sense of~(\ref{eq:pdref}). In fact, the would-be inner product defined by $\gamma_t$ is indefinite. This is not surprising given that $\gamma$ is designed for vacua given by real polarizations rather than K\"ahler polarizations. In~the following Section~\ref{sec:refpos} we exhibit a genuine example of a positive-definite reflection map in a~context motivating the developments of the present section.

\section[Reflection positivity and $\alpha$-K\"ahler quantization]
{Reflection positivity and $\boldsymbol\alpha$-K\"ahler quantization}\label{sec:refpos}

The Euclidean approach to quantum field theory~\cite{Nel:constqft,Sym:euclqft1} is based on relating quantum field theoretic correlation functions in Minkowski space to corresponding correlation functions in~Eucli\-dean space via a version of \emph{Wick rotation}. The \emph{Euclidean correlation functions}, also known as \emph{Schwinger functions}, have better analytic properties than their Lorentzian counterparts and can be naturally related to a statistical path integral, which is mathematically more tractable than the Feynman path integral. Osterwalder and Schrader~\cite{OsSc:axeucl2} have given conditions for a~family of Schwinger functions to correspond to a quantum field theory in the sense of the \emph{Wightman axi\-oms}~\cite{StWi:pct}. One of these conditions is \emph{reflection positivity} and it can be understood to correspond to a realness condition for Lorentzian correlators or a Hermiticity condition for operators. It is also used to construct a Hilbert space directly from the Schwinger functions~\cite{OsSc:axeucl}. This construction is very different from the standard quantization constructions of quantum (field) theory which are always directly or indirectly tied to a K\"ahler polarization. In this case there is no K\"ahler polarization. This provided crucial inspiration for the construction presented in~Section~\ref{sec:iprod} based on the notion of $\alpha$-K\"ahler polarizations. The purpose of the present section is to elucidate this by showing (in the simplest possible setting) how the Hilbert space construction based on reflection positivity is recovered as a special case of what we shall call $\beta$-K\"ahler quantization, a~Euclidean counterpart of $\alpha$-K\"ahler quantization.

\subsection{Statistical path integral}
\label{sec:statpi}

Since it is of central interest in the Euclidean approach to field theory, we consider a replacement of the Feynman path integral with a \emph{statistical path integral} as it appears for example in classical statistical field theory. That is, we replace the path integral formula~(\ref{eq:amplgvac}) with a statistical integral defined as
\begin{equation*}
 \sigma^F_{M}\big(W^\pol\big)=\int_{K_M} \xD\phi\, F(\phi) \exp(-S(\phi)).
\end{equation*}
It is straightforward to accommodate for this change throughout the framework presented in this work. (Of course, the straightforward interpretation of the constructed mathematical objects in~terms of a quantum theory is then lost.) We shall not develop this in detail, but focus on a few essential structures. We start with the analogue of formula~(\ref{eq:piwobs}) for Weyl observables. Consider a linear observable $D$, and its negative exponentiation $G=\exp(-D)$ and let $\eta\in A_M^D\cap L_{\partial M}^\pol$. Thus,
\begin{equation}
 \sigma_{M}^{G}\big(W^\pol\big)=\exp\bigg({-}\frac{1}{2} D(\eta)\bigg).
 \label{eq:pixobse}
\end{equation}
Replacing $G$ with the usual Weyl observable $F=\exp(\im D)$ yields
\begin{equation*}
 \sigma_{M}^{F}\big(W^\pol\big)=\exp\bigg(\frac{1}{2} D(\eta)\bigg).
\end{equation*}
Note that to arrive at this expression we have to replace on the right-hand side of~(\ref{eq:pixobse}) $D$ with~$-\im D$ as well as $\eta$ with $-\im\eta$. We can read off from this that for arbitrary observables the relation between the Feynman path integral and the statistical path integral can be expressed as follows. Let $F$ be an observable and define the observable $\tilde{F}$ as
\begin{equation}
 \tilde{F}(\phi)\defeq F\big(\sqrt{\im}\,\phi\big).
 \label{eq:insobs}
\end{equation}
(For definiteness we choose to define the root of $\im$ here as $\sqrt{\im}\defeq \exp(\im\pi/4)$.) Then,
\begin{equation}
 \rho_M^F\big(W^\pol\big)=\sigma_M^{\tilde{F}}\big(W^\pol\big).
 \label{eq:relfs}
\end{equation}
In particular, the relations~(\ref{eq:weylrel}) of the Weyl algebra of slice observables are modified as follows,
\begin{gather}
 G\qp^{\rm S} F=\exp\big(\omega_{\Sigma}(\xi_E,\xi_D)\big) G\cdot F.
 \label{eq:weylrels}
\end{gather}

Proceeding to a ``quantization'' by analogy with the quantum theory, the role of the inner product~(\ref{eq:stdipc}) in $L_\Sigma^\bC$ is taken by the expression
\begin{equation}
 (\phi,\eta)^{\rm S}_{\Sigma}\defeq 4\omega_{\Sigma}\big(\overline{\phi},\eta\big).
 \label{eq:statipc}
\end{equation}
This is clear from the modification that the vacuum correlation function~(\ref{eq:vevweylf}) of a Weyl obser\-vable undergoes by transition to the statistical path integral. We get
\begin{equation*}
 v^{\rm S}_{\Sigma}(F)\defeq \sigma_{\hat{\Sigma}}^F\big(W^\pol\big)
 =\exp\bigg({-}\frac{1}{4} \{\xi,\xi\}^{\rm S}_{\Sigma}\bigg).
\end{equation*}
Here we define
\begin{equation*}
 \{\xi,\xi\}^{\rm S}_{\Sigma}\defeq - \im \{\xi,\xi\}_{\Sigma}=4\omega\big(x^\cpol,y^\pol\big).
\end{equation*}
A ``quantization'' with a construction of Hilbert spaces in the standard K\"ahler sense (Sections~\ref{sec:kquant} and \ref{sec:kobs}) is clearly not possible as the sesquilinear form~(\ref{eq:statipc}) is not even hermitian.
On the other hand, we can write down immediately the analogue of Theorem~\ref{thm:wfac}.
\begin{thm}
 If $F$ is the Weyl observable $\exp(\im D)$, then,
\begin{align*}
 \sigma_M^F\big(K^\pol_{\xi}\big) & = \sigma_M\big(K^\pol_{\xi}\big) F\big({-}\im \xi^\ipol\big) \sigma_M^F\big(W^\pol\big),\qquad\text{with}
 \\
 \sigma_M\big(K^\pol_{\xi}\big) & =\exp\big(\omega_{\partial M}\big(\xi,\xi^\ipol\big)\big),
 \\
 \sigma_M^F\big(W^\pol\big) & =\exp\bigg(\frac{1}{2}D(\eta)\bigg).
\end{align*}
Here $\eta\in A_M^{D,\bC}\cap L_{\partial M}^\pol$.
\end{thm}

With regard to the construction of Hilbert spaces it turns out that we can adapt the construction of Section~\ref{sec:iprod} to the present setting and achieve the analog of a quantization in this way. Indeed, the considerations of Sections~\ref{sec:modreal} and \ref{sec:modstar} carry over completely to the present setting with just one fundamental change. That is, we consider a real structure $\beta$ on the vector space of complexified germs of solutions $L_{\Sigma}^\bC$. In contrast to the real structure $\alpha$ introduced in~Section~\ref{sec:modreal}, $\beta$ does not satisfy the compatibility condition~(\ref{eq:realcomp}), but rather the modified compatibility condition
\begin{equation*}
 \omega(\beta(x),\beta(y))=-\overline{\omega(x,y)}.
\end{equation*}
We call this condition \emph{anti-compatibility} between $\beta$ and $\omega$.
In this way the bilinear form
\begin{equation*}
 (\phi,\eta)^\beta\defeq 4\omega\left(\beta(\phi),\eta\right)
\end{equation*}
is not only non-degenerate and sesquilinear, but also hermitian. Thus, all steps described in~Sections~\ref{sec:modreal} and \ref{sec:modstar} can be performed. In particular, the quantization with the modified GNS construction succeeds with $\beta$ determining the $\sst$-structure of the observable algebra in analogy to~(\ref{eq:alphass}),
\begin{equation*}
 F^\beta(\phi)\defeq \overline{F\left(\beta(\phi)\right)}.
\end{equation*}
This leads to a notion of $\beta$-K\"ahler quantization in analogy to $\alpha$-K\"ahler quantization. In particular, the construction of Hilbert spaces and states can be preformed in accordance with Section~\ref{sec:stdstate} by replacing complex conjugation with the map $\beta$. In contrast to the $\alpha$-case, additionally the inner product $\{\cdot,\cdot\}$ is replaced by $\{\cdot,\cdot\}^{\rm S}$ of~(\ref{eq:statipc}) and equation~(\ref{eq:stdjip}) is replaced by
\begin{equation*}
 \{\phi,\eta\}^{\rm S}_{\Sigma}= -2\im\omega_{\Sigma}\big(\phi,J_{\Sigma} \eta\big)+2\omega_{\Sigma}(\phi,\eta).
\end{equation*}
We denote the Hilbert space constructed in this way by $\cH_{\Sigma}^\beta$. Note that wave functions are defined on $L_{\Sigma}^\beta$, the subspace of $L_{\Sigma}^\bC$ invariant under $\beta_{\Sigma}$. Similarly, coherent states are labeled by~elements of $L_{\Sigma}^\beta$. In order to identify states in this space with the action of slice observables on~the vacuum, the most straightforward consistent choice is to set~(\ref{eq:rencohs}) \big(with $\xi\in L_{\Sigma}^\beta\big)$ and~(\ref{eq:recohs}) as in the genuine quantum case.
We obtain instead of~(\ref{eq:relsoca}) the relation
\begin{equation*}
 K^\pol_\xi \cong \exp\left(\omega_{\Sigma}(\xi,\tau)\right) \ncoh^\beta_{\tau},
\end{equation*}
where $\tau=P^-_\Sigma(\xi)+\beta_{\Sigma}\big(P^-_\Sigma(\xi)\big)$. Otherwise, we proceed analogous to Section~\ref{sec:corrak}. This determines the ``quantization'' completely.

As can be seen from the formulas, there is a close similarity between the $\alpha$-K\"ahler and $\beta$-K\"ahler quantizations. In fact, there is a straightforward one-to-one correspondence given by
\begin{gather}
 \beta=\im\alpha,
 \label{eq:bia}
\end{gather}
which makes the relevant inner products equal,
\begin{equation*}
 (\phi,\eta)^\beta = (\phi,\eta)^\alpha.
\end{equation*}
This means in particular that given polarizations on a hypersurface $\Sigma$ that admit a conventional K\"ahler quantization, the choice $\beta(\xi)\defeq \im \overline{\xi}$ will yield a corresponding ``quantization'' based on~the statistical path integral.

\subsection{Wick rotated Klein--Gordon theory}

We consider in the following a Wick rotation of the Klein--Gordon theory in Minkowski space as follows. We rotate the time coordinate $t$ in the complex plane to an imaginary value. This can be considered as a substitution $t=-\im\tau$, where $\tau$ is subsequently taken to be real. In this way the Klein--Gordon equation
\begin{equation*}
 \bigg(\partial_t^2 -\sum_i \partial_i^2+m^2\bigg)\phi=0
\end{equation*}
turns into the elliptic partial differential equation
\begin{equation*}
 \bigg({-}\partial_\tau^2 -\sum_i \partial_i^2+m^2\bigg)\phi=0.
\end{equation*}
We refer to the theory defined in this way as the \emph{Euclidean theory} and consider it to live in 4-dimensional Euclidean space with coordinates $(\tau,x)$.

We recall the space $L_{0}^\bC$ of germs of complexified solutions of the Klein--Gordon equation on the hypersurface $t=0$ from Section~\ref{sec:kgstates}, with the solutions parametrized in terms of plane waves~(\ref{eq:kgmodes}). Applying the substitution $t=-\im\tau$ we obtain
\begin{equation}
 \phi(\tau,x)=\int\frac{\xd^3 k}{(2\pi)^3 2E}
 \Big(\phi^{\rm a}(k) {\rm e}^{- E \tau + \im k x}+\overline{\phi^{\rm b}(k)} {\rm e}^{E \tau - \im k x}\Big).
 \label{eq:kgemodes}
\end{equation}
This is indeed a parametrization of the complexified solutions of the Euclidean theory near the $\tau=0$ hypersurface in Euclidean space. We denote this space by $L_{0}^{{\rm E},\bC}$. Via the substitution $t=-\im\tau$ we have a canonical identification of the spaces $L_{0}^\bC$ and $L_{0}^{{\rm E},\bC}$. Note, however, that while the elements of $L_{0}^\bC$ may be viewed as \emph{global} solutions in~Klein--Gordon theory, the elements of $L_{0}^{{\rm E},\bC}$ do not extend to sensible global solutions in Euclidean space due to their exponential behavior in the Euclidean time direction $\tau$. Moreover, the subspaces of real solutions $L_0\subseteq L_{0}^\bC$ in~Klein--Gordon theory and $L_0^{{\rm E}}\subseteq L_0^{{\rm E},\bC}$ in the Euclidean theory are different. Whereas the reality condition in~Klein--Gordon theory is $\phi^{\rm a}(k)=\phi^{\rm b}(k)$, the reality condition in the Euclidean theory is given by, $\phi^{\rm a}(k)=\overline{\phi^{\rm a}(-k)}$ and $\phi^{\rm b}(k)=\overline{\phi^{\rm b}(-k)}$. Expressed in terms of Euclidean solutions the reality condition of the Klein--Gordon theory takes the form
\begin{equation*}
 \overline{\phi(\tau,x)}=\phi(-\tau,x).
\end{equation*}

An action for the Euclidean theory can also be obtained by Wick rotation from the Klein--Gordon theory, therefore inducing a symplectic form $\omega_0^{\rm E}$ on $L_0^{\rm E}$,
\begin{align*}
 \omega_0^{\rm E}(\phi_1,\phi_2) & =-\frac{1}{2}\int\xd^3 x\,
 \big(\phi_2(\tau,x) (\partial_\tau \phi_1)(\tau,x) - \phi_1(\tau,x)(\partial_\tau\phi_2)(\tau,x)\big)
 \\
 & =-\frac{1}{2}\int\frac{\xd^3 k}{(2\pi)^3 2E}
 \Big(\phi_2^{\rm a}(k)\overline{\phi_1^{\rm b}(k)}-\phi_1^{\rm a}(k)\overline{\phi_2^{\rm b}(k)}\Big).
\end{align*}
We are using here the analogous convention for the orientation of the hypersurface as in the case of the symplectic form~(\ref{eq:sfkgm}) of the Klein--Gordon theory. That is, we consider the $\tau=0$ hypersurface carrying the orientation of the past boundary of a region to the (here Euclidean) future $\tau>0$. Comparing with the symplectic form $\omega_0$,~(\ref{eq:sfkgm}) of the Klein--Gordon theory via the canonical identification of $L_{0}^{\bC}$ with $L_{0}^{{\rm E},\bC}$ we see
\begin{equation}
 \omega_0=-\im\, \omega_0^{\rm E}.
 \label{eq:relsympl}
\end{equation}
Note in particular that the notion of (complex) Lagrangian subspace is the same for both symplectic forms.

Viewing the Euclidean theory as a classical theory in its own right we can consider its quantization along the lines of the presented framework. The first step is the determination of the vacuum in the sense of~\cite{CoOe:vaclag}. As manifest in the expansion~(\ref{eq:kgemodes}), solutions behave exponentially in the (Euclidean) time direction $\tau$. Thus, the vacuum is simply given by imposing decaying boundary conditions, encoded in terms of (complexifications of) real Lagrangian subspaces of~$L_0^{{\rm E},\bC}$ which we shall denote $L_0^{{\rm E},+}$ and $L_0^{{\rm E},-}$ respectively. They are
\begin{equation*}
 L_0^{{\rm E},+}=\big\{\phi\in L_0^{{\rm E},\bC}\colon \phi^{\rm a}(k)=0\, \forall k\big\},
 \qquad
 L_0^{{\rm E},-}=\big\{\phi\in L_0^{{\rm E},\bC}\colon \phi^{\rm b}(k)=0\, \forall k\big\}.
\end{equation*}
Under the canonical identification of $L_{0}^{\bC}$ with $L_{0}^{{\rm E},\bC}$ these subspaces are identical to the corresponding Lagrangian subspaces of the Klein--Gordon theory, see~(\ref{eq:kgpvac}). However, while they are positive-definite Lagrangian subspaces with respect to the real structure $L_0\subseteq L_0^\bC$ and symplectic form $\omega_0$, they are instead (complexifications of) real Lagrangian subspaces with respect to the real structure $L_0^{{\rm E}}\subseteq L_0^{{\rm E},\bC}$ and symplectic form $\omega_0^{{\rm E}}$. In particular, a traditional K\"ahler quantization to construct a Hilbert space (as recalled in~Section~\ref{sec:kquant}) cannot be performed.

We proceed to consider the vacuum correlation functions, comparing the Euclidean theory to the Klein--Gordon theory. Since the correlation functions for polynomial observables are determined by those of quadratic observables (recall Section~\ref{sec:vcgenobs}) it is sufficient to consider the latter. Moreover, since the theories are related by Wick rotation in spacetime we limit ourselves to observables that correspond to field values at spacetime points. That is, we consider the Feynman propagator and its analog of the Euclidean theory, compare equation~(\ref{eq:fpropsrc}). While it is well known that they are essentially related by Wick rotation, we briefly review in the following how this relation arises from the present perspective. To this end recall the derivation of the 2-point function in Klein--Gordon theory using slice observables, compare Section~\ref{sec:propslice}.
We~can proceed for the Euclidean theory in exactly the same manner, where time ordering is now in terms of the Euclidean time, with $\tau_2>\tau_1$. The slice observables yielding with~(\ref{eq:sobseval}) and~(\ref{eq:lsobsdef}) field evaluation at the point $(\tau,x)$ are now parametrized by the solutions
\begin{equation*}
 \big(\xi^{\rm E}_i\big)^{\rm a}(k)={\rm e}^{E \tau_i - \im k x_i},\qquad
 \big(\xi^{\rm E}_i\big)^{\rm b}(k)=-{\rm e}^{- E \tau_i - \im k x_i}.
\end{equation*}
With the identification of $L^{{\rm E},\bC}$ with $L^\bC$ via the substitution $t=-\im\tau$ we obtain the relation $\xi_i=\im \xi_i^{\rm E}$ for the solutions determining the observables between the Klein--Gordon and Euclidean theory. Due to the definition~(\ref{eq:lsobsdef}) this relation follows in fact immediately from the relation~(\ref{eq:relsympl}) between the symplectic forms. Since for the Euclidean theory formula~(\ref{eq:twopointord}) is equally valid, adapting notation, we obtain for the relation between the 2-point functions
\begin{equation}
 \rho_M^{{\rm E}, D_1 D_2}\big(W^\pol\big)
 = 2 \im \omega^{{\rm E}}\big(\xi^{{\rm E},-}_1,\xi^{{\rm E},+}_2\big)
 = 2 \omega(\xi_1^-,\xi_2^+)
 = -\im \rho_M^{D_1 D_2}\big(W^\pol\big).
 \label{eq:rel2p}
\end{equation}
That is, the 2-point functions (and propagators) are indeed related by analytical continuation, up to a relative factor of $\im$.
For completeness we also write the explicit form of the 2-point functions and propagator of the Euclidean theory,
\begin{equation*}
 \rho_M^{{\rm E}, D_1 D_2}\big(W^\pol\big) = -\im G_F^{\rm E}((\tau_1,x_1),(\tau_2,x_2)) = -\im \int\frac{\xd^3 k}{(2\pi)^3 2E}\, {\rm e}^{ E(\tau_1-\tau_2)- \im k(x_1-x_2)}.
\end{equation*}
From the relation~(\ref{eq:rel2p}) between the 2-point functions we can infer the relation between arbitrary vacuum correlation functions. With the definition~(\ref{eq:insobs}) we can write this as
\begin{equation}
 \rho_M^{F}\big(W^\pol\big)=\rho_M^{{\rm E},\tilde{F}}\big(W^\pol\big).
 \label{eq:relkge}
\end{equation}
Recall that the observables $F$ between the Euclidean theory and Klein--Gordon theory are to be identified via the Wick rotation. We do not distinguish them notationally.

We proceed to consider the construction of Hilbert spaces for the Euclidean Theory. Since the polarizations $L_0^{{\rm E},\pm}\in L_0^\bC$ determining the vacuum are real, K\"ahler quantization cannot be applied. Instead, we apply an $\alpha$-K\"ahler quantization through a positive-definite reflection map~$\gamma_0$ based on a position/momentum decomposition $L_0^{{\rm E}}=N_0^{{\rm E}}\oplus M_0^{{\rm E}}$ (Section~\ref{sec:realpol}). Here we have
\begin{equation*}
 M_0^{\rm E}=\big\{\phi\in L_0^{\rm E}\colon \phi^{\rm a}(k)=-\overline{\phi^{\rm b}(-k)}\,\forall k\big\}, \qquad
 N_0^{\rm E}=\big\{\phi\in L_0^{\rm E}\colon \phi^{\rm a}(k)=\overline{\phi^{\rm b}(-k)}\,\forall k\big\}.
\end{equation*}
Under the identification of $L_0^{\rm E}$ and $L_0$ this is the same as for the Klein--Gordon theory at $t=0$, compare~(\ref{eq:kgmndec}). Similarly, $\gamma_0$ takes the same form as for the Klein--Gordon theory (at $t=0$),
\begin{equation*}
 (\gamma_0(\phi))^{\rm a}(k)=\overline{\phi^{\rm b}(-k)},\qquad
 (\gamma_0(\phi))^{\rm b}(k)=\overline{\phi^{\rm a}(-k)},
\end{equation*}
compare~(\ref{eq:kggamma}).
In spacetime terms this is a reflection at the $\tau=0$ hypersurface, analogous to~(\ref{eq:kgreflect}),
\begin{equation*}
 (\gamma_0(\phi))(\tau,x)=\phi(-\tau,x).
\end{equation*}
It is also clear that $\gamma_0$ interchanges the polarization $L_0^{{\rm E}\pm}\subseteq L_0^{{\rm E},\bC}$ determining the past and future vacuum. What is more, we easily verify that $\gamma_0$ satisfies condition~(\ref{eq:pdref}). That is,~$\gamma_0$ is a positive-definite reflection map. It therefore defines an $\alpha$-K\"ahler quantization with $\alpha_0\colon L_0^{{\rm E},\bC}\to L_0^{{\rm E},\bC}$ given by equation~(\ref{eq:alpharef}), i.e.,
\begin{equation}
 (\alpha_0(\phi))^{\rm a}(k)=-\im \phi^{\rm b}(k),\qquad
 (\alpha_0(\phi))^{\rm b}(k)=\im\phi^{\rm a}(k).
 \label{eq:alphae}
\end{equation}
In spacetime terms the map $\alpha_0$ takes the form
\begin{equation*}
 (\alpha_0(\phi))(\tau,x)=-\im \overline{\phi(-\tau,x)}.
\end{equation*}

We return now to the perspective of the Euclidean approach to quantum field theory, where a Wick rotation is used to relate the quantum theory to a theory in Euclidean space. The latter has properties that make it mathematically more accessible than the original theory, such as the decay properties on correlation functions. In the Euclidean theory considered above this is manifest in the boundary conditions encoded in the Lagrangian subspaces $L_0^{{\rm E},\pm}\in L_0^\bC$. Moreover, in the Euclidean approach the theory in Euclidean space is not treated as a quantum theory in its own right, but rather as a theory with a behavior consistent with a statistical path integral. This was precisely the subject of Section~\ref{sec:statpi}. Combining its results with the previous treatment of the Euclidean theory, we are able to relate the latter viewed as defined via a statistical path integral to the ordinary Klein--Gordon quantum field theory. Up to conventions of signs and orientations, this relation reproduces the one that has been of interest in the literature on Euclidean methods and puts it into a more general context. In particular, we shall see that the prescription for the construction of Hilbert spaces for the Euclidean statistical theory using \emph{reflection positivity}~\cite{OsSc:axeucl} is precisely an example of a $\beta$-K\"ahler quantization. Indeed, reflection positivity served as an inspiration for the development of the $\alpha$-K\"ahler quantization as laid out in~Section~\ref{sec:iprod}.

Combining~(\ref{eq:relfs}) and~(\ref{eq:relkge}) we obtain with $\check{F}(\phi)\defeq F(\im \phi)$
\begin{equation*}
 \rho_M^{F}\big(W^\pol\big)=\sigma_M^{{\rm E},\check{F}}\big(W^\pol\big).
\end{equation*}
Recall that the observables $F$ between the Klein--Gordon and Euclidean statistical theories are related moreover by Wick rotation. Note that when considering observables of even degree the factor of $\im$ just gives rise to a sign. In particular, for the 2-point functions we have
\begin{equation*}
 \rho_M^{D_1 D_2}\big(W^\pol\big)=-\sigma_M^{{\rm E},D_1 D_2}\big(W^\pol\big).
\end{equation*}
Indeed, the factor of $\im$ and the signs are related to our conventions. For example, changing the sign of $\omega_0^{\rm E}$ would make these disappear and recover the usual relation between $n$-point functions just being related by Wick rotation. In order to construct a Hilbert space for the statistical Euclidean theory we follow the steps laid out in~Section~\ref{sec:statpi}. The most important ingredient is the anti-compatible real structure $\beta_0$ on $L_0^{{\rm E},\bC}$ exchanging future and past boundary conditions. However, we have already constructed a corresponding compatible real structure $\alpha_0$ for the quantum theory above, see~(\ref{eq:alphae}). Thus, with the relation~(\ref{eq:bia}) we obtain $\beta_0$ given by
\begin{equation*}
 (\beta_0(\phi))^{\rm a}(k)= \phi^{\rm b}(k),\qquad
 (\beta_0(\phi))^{\rm b}(k)=\phi^{\rm a}(k).
\end{equation*}
In spacetime terms the map $\beta_0$ takes the form
\begin{equation}
 (\beta_0(\phi))(\tau,x)=\overline{\phi(-\tau,x)}.
 \label{eq:betaest}
\end{equation}
We recognize this precisely as corresponding to the ordinary complex conjugation of the Klein--Gordon theory under the identification of $L_0^{{\rm E},\bC}$ with $L_0^{\bC}$.

The role of $\beta_0$ is to provide the $\sst$-structure of the algebra of slice observables which now carries the product~(\ref{eq:weylrels}). In particular, the product $F^\beta \qp F$ for a slice observable $F$ is a \emph{positive} element of the algebra and its (statistical) vacuum correlation function satisfies by construction
\begin{gather}
 v_0^{\rm S}\big(F^\beta \qp F\big)\ge 0.
 \label{eq:svcorrpos}
\end{gather}
As we have seen in~Section~\ref{sec:propslice} and in the present section the 2-point (and similarly $n$-point) correlation functions can be constructed by (time-)ordered composition of slice observables. In~this way, we can carry over the identity~(\ref{eq:svcorrpos}) to spacetime observables that encode (possibly smeared) $n$-point functions. More specifically, let the spacetime observable $F$ be defined in terms of finitely many smearing functions $f_k$ with compact support,
\begin{equation*}
 F(\phi)\defeq \sum_{k=1}^n \int \xd\tau_1 \cdots \xd\tau_n \xd x_1 \cdots \xd x_n f_k(\tau_1,x_1,\ldots,\tau_n,x_n) \phi(\tau_1,x_1)\cdots \phi(\tau_n,x_n).
\end{equation*}
Here we require temporal ordering in the sense that $f_k(\tau_1,x_1,\ldots,\tau_n,x_n)$ vanishes, except if \mbox{$\tau_1< \cdots < \tau_n$}. Recalling the spacetime form~(\ref{eq:betaest}) of $\beta_0$ it is now easy to recognize that equation~(\ref{eq:svcorrpos}) is precisely the \emph{reflection positivity} condition $E2$ of Osterwalder and Schrader~\cite{OsSc:axeucl}.

\section{Application to Minkowski and Rindler space}
\label{sec:minkrind}

In the present section we review two topics involving a change of vacuum for the Klein--Gordon theory in flat spacetime. The first concerns the question of splitting the standard state space of the equal-time hyperplane in Minkowski space into two pieces, corresponding to half-hyperplanes. The second concerns the Minkowski vacuum as seen from an embedded Rindler space. Both topics involve the embedding of one or two Rindler wedges into Minkowski space. The purpose of the inclusion of these topics here is on the one hand to show that and how standard results of quantum field theory are correctly recovered. On the other hand it is to add new insights and perspectives that go beyond the traditional methods.

\subsection{Splitting the Minkowski vacuum}
\label{sec:splitmink}

Consider the standard equal-time hyperplane $\Sigma$ in Minkowski space at $t=0$. Cut this hyperplane into two half-hyperplanes $\Sigma_{\rm L}$, $\Sigma_{\rm R}$ along the $x_1$-direction with $x_1<0$ and $x_1>0$ respectively. In other words, space at $t=0$ is divided into two half-spaces. It has long been known that the Hilbert space $\cH$ of states associated to $\Sigma$ does not allow for a tensor product decomposition $\cH=\cH_{\rm L}\tens \cH_{\rm R}$ in such a way that $\cH_{\rm L}$ encodes states on $\Sigma_{\rm L}$ and $\cH_{\rm R}$ states on $\Sigma_{\rm R}$. For example, such a decomposition would contradict the Reeh--Schlieder theorem~\cite{ReSc:unitequiv}. Recall the discussion at the end of Section~\ref{sec:intro_gbqftobs} of the Introduction.

On the other hand, reaping the full potential of GBQFT in the context of ordinary QFT would require being able to describe amplitudes and correlators in compact pieces of Minkowski space (such as hypercubes or diamonds) together with providing the powerful identities that allow to compose them. In terms of the axioms of Appendix~\ref{sec:qobsaxioms}, a necessary ingredient would be the splitting of Hilbert spaces of states on hypersurfaces into localized pieces. The previously mentioned splitting of $\cH$ in terms of the decomposition of $\Sigma$ into $\Sigma_{\rm L}$ and $\Sigma_{\rm R}$ would be but the simplest example of this. Its failure appears to doom efforts to make GBQFT work effectively in ordinary QFT, beyond the restriction to very particular hypersurfaces and regions. Note that less realistic (e.g., Euclidean or 2-dimensional) versions of QFT are not necessarily affected~\cite{Oe:2dqym}. The answer proposed in the present paper is to abandon the axioms of Appendix~\ref{sec:qobsaxioms} and replace them with the framework developed principally in~Section~\ref{sec:genquant}, supplementing with (generalized, see Section~\ref{sec:iprod}) Hilbert spaces only were necessary or convenient.

It is nevertheless instructive to consider the example of the decomposition of $\Sigma$ into $\Sigma_{\rm L}$ and~$\Sigma_{\rm R}$ and the properties of the Minkowski vacuum with respect to this decomposition. Since the equations of motions are local, i.e., they are partial differential equations, the space of germs of~solutions (i.e., the space of initial data) on $\Sigma$ simply decomposes as a direct sum of germs of~solutions on $\Sigma_{\rm L}$ and $\Sigma_{\rm R}$. That is, $L_{\Sigma}=L_{\rm L}\oplus L_{\rm R}$. Similarly, the symplectic form, originating from a hypersurface integral and involving differentials, is local. That is, $\omega_{\Sigma}=\omega_{\rm L} + \omega_{\rm R}$, with~$\omega_{\rm L}$ and $\omega_{\rm R}$ being the symplectic forms on $\Sigma_{\rm L}$ and $\Sigma_{\rm R}$ respectively. Recall that the Minkowski vacuum is determined by the Lagrangian subspace $L_{\Sigma}^+\subseteq L_{\Sigma}^\bC$ of negative energy solutions in the past,~(\ref{eq:kgpvac}), and the conjugate subspace $L_{\Sigma}^-=\overline{L_{\Sigma}^+}$ of positive-energy solutions in the future. Now the simple fact is that $L_{\Sigma}^+$ (or $L_{\Sigma}^-$) is not local with respect to the decomposition of~$\Sigma$ into~$\Sigma_{\rm L}$ and $\Sigma_{\rm R}$. That is, there do not exist individual subspaces $L_{\rm L}^{(+)}\subseteq L_{\rm L}^{\bC}$ and $L_{\rm R}^{(+)}\subseteq L_{\rm R}^{\bC}$ such that $L_{\Sigma}^+=L_{\rm L}^{(+)}\oplus L_{\rm R}^{(+)}$. Equivalently, the complex structure $J_{\Sigma}$ is not a differential operator on $L_{\Sigma}$, but merely a pseudo-differential operator and thus non-local. If the subspaces $L_{\rm L}^{(+)}$ and $L_{\rm R}^{(+)}$ existed, that would give us upon quantization precisely a decomposition of the type $\cH=\cH_{\rm L}\tens \cH_{\rm R}$. We might say that the Minkowski vacuum is ``indivisible'' or ``non-local''. This is in fact the typical situation for vacua in quantum field theory which is not surprising given that they are usually encoding asymptotic boundary conditions.

The parametrization of $L_{\Sigma}^\bC$ in terms of plane wave modes as in~(\ref{eq:kgmodes}) is not convenient in order to make manifest the locality of $L_{\Sigma}^\bC$ and of $\omega_{\Sigma}$. To this end we switch in a first step to the \emph{Minkowski Bessel modes} which are complex eigenmodes of the Lorentz boost operator~\cite{Ger:minkbessel}. For simplicity, we consider (1+1)-dimensional Minkowski space. (Additional spatial dimensions are easily incorporated.) We denote the modes by $B_\kappa$ and $\overline{B_\kappa}$ (complex conjugates), where $\kappa$ is the eigenvalue of the boost generator. $\kappa$ takes values on the real line. The modes $B_\kappa$ admit a~representation in terms of an integral over plane waves parametrized by the rapidity $q$,
\begin{equation}
 B_\kappa(t,x)=\int_{-\infty}^{\infty}\frac{\xd q}{2\pi}\, \exp(-\im\, m (t \cosh q- x \sinh q)-\im \kappa q).
 \label{eq:minkbessel}
\end{equation}
We thus parametrize elements $\phi\in L_{\Sigma}^\bC$ via coefficient functions $\phi^{\rm c}$ and $\phi^{\rm d}$ on $(-\infty,\infty)$, as follows,
\begin{gather}
 \phi(t,x)=\int_{-\infty}^\infty\xd\kappa\, \Big(\phi^{\rm c}(\kappa) B_\kappa(t,x) + \overline{\phi^{\rm d}(\kappa) B_\kappa(t,x)}\Big).
 \label{eq:kgmbmodes}
\end{gather}
From the integral representation in terms of plane waves we can read off immediately that the modes $B_\kappa$ are positive-energy modes, while the modes $\overline{B_\kappa}$ are negative-energy modes. Since the Lagrangian subspace determining the Minkowski vacuum (to the past) consists precisely of the negative-energy solutions, we have, in analogy to~(\ref{eq:kgpvac}),
\begin{equation}
 L^+_{\Sigma}=\big\{\phi\in L_{\Sigma}^\bC \colon \phi^{\rm c}(\kappa)=0\,\forall \kappa\big\}.
 \label{eq:kgpvacmb}
\end{equation}
We note that the symplectic form~(\ref{eq:sfsft}) in the present parametrization~(\ref{eq:kgmbmodes}) results in
\begin{equation}
 \omega_{\Sigma}(\phi_1,\phi_2)=\im \int_{-\infty}^{\infty}\xd \kappa\,
 \Big(\phi^{\rm c}_2(\kappa) \overline{\phi^{\rm d}_1(\kappa)}
 - \overline{\phi^{\rm d}_2(\kappa)} \phi^{\rm c}_1(\kappa)\Big).
 \label{eq:symfmb}
\end{equation}
The inner product~(\ref{eq:stdipc}) is thus
\begin{equation}
 (\phi_1,\phi_2)_{\Sigma}= 4 \int_{-\infty}^{\infty}\xd \kappa\,
 \Big(\phi^{\rm d}_1(\kappa) \overline{\phi^{\rm d}_2(\kappa)}
 - \overline{\phi^{\rm c}_1(\kappa)} \phi^{\rm c}_2(\kappa)\Big).
 \label{eq:ipmb}
\end{equation}
From these expressions it is again manifest that $L^+_{\Sigma}\subseteq L^{\bC}_{\Sigma}$ is a positive-definite Lagrangian subspace, this time in terms of the Minkowski Bessel modes.

\begin{figure}
 \centering
 \begin{tikzpicture}[scale=1]
\draw[thick] (-3,3) -- (3,-3);
\draw[thick] (-3,-3) -- (3,3);
\draw[blue,very thick](-3.4,0) -- (3.4,0);
\node at (-1.5,-0.35) {$\Sigma_L$};
\node at (1.5,-0.35) {$\Sigma_R$};
\node at (0,2.15) {F};
\node at (0,-2.15) {P};
\node at (2.5,.75) {R};
\node at (-2.5,.75) {L};
\end{tikzpicture}
 \caption{Minkowski space divided into four wedge regions: Future (\text{F}), Past (\text{P}), Left ({\rm L}), Right (\text{R}). Drawn in blue is the $t=0$ hyperplane with left and right halves, $\Sigma_{\rm L}$ and $\Sigma_{\rm R}$ respectively.}
 \label{fig:minkpflr}
\end{figure}
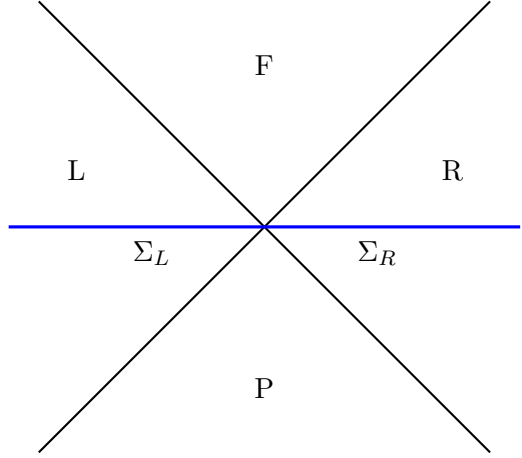

We proceed to introduce another set of modes, arising as simple linear combinations of the Minkowski Bessel modes. We define for $\kappa>0$
\begin{gather}
 B^{\rm L}_\kappa(t,x) = \frac{1}{\sqrt{2\sinh(\pi\kappa)}}
 \Big({\rm e}^{\pi\kappa/2} \overline{B_{-\kappa}(t,x)}
 - {\rm e}^{-\pi\kappa/2} B_{\kappa}(t,x)\Big), \nonumber
 \\[.5ex]
 B^{\rm R}_\kappa(t,x) = \frac{1}{\sqrt{2\sinh(\pi\kappa)}}
 \Big({\rm e}^{\pi\kappa/2} B_{\kappa}(t,x)
 - {\rm e}^{-\pi\kappa/2} \overline{B_{-\kappa}(t,x)}\Big).
 \label{eq:rindboost}
\end{gather}
Evidently, these modes, together with their complex conjugates and restricted to $\kappa>0$ only, provide a parametrization of $L_{\Sigma}^\bC$ equivalent to
that in terms of the Minkowski Bessel modes (with~$\kappa\in\R$). However, in contrast to the latter, they exhibit remarkable localization properties. To~discuss this, we divide Minkowski space into four wedge regions, $\text{F}$, $\text{P}$, ${\rm L}$ and $\text{R}$, see Figure~\ref{fig:minkpflr}. These are the causal future of the origin, its causal past and points spacelike separated from it to the left and to the right respectively. The regions {\rm L} and \text{R} are also called the left and the right \emph{Rindler wedge}, respectively, as they are isometric to \emph{Rindler space}~\cite{Rin:kruskalsuniframe}. The modes~$B^{\rm R}_\kappa$ vanish identically in the region {\rm L}, while the modes $B^{\rm L}_\kappa$ vanish identically in the region \text{R}. In~particular, as germs in $L_{\Sigma}^\bC=L_{\rm L}^\bC\oplus L_{\rm R}^\bC$ the modes $B^{\rm R}_\kappa$ vanish in the component $L_{\rm R}^\bC$ while the modes $B^{\rm L}_\kappa$ vanish in the component $L_{\rm L}^\bC$. That is, we have at hand a set of modes that makes manifest the decomposition of the space of germs of solutions on $\Sigma$ into a left and right half. We will refer to these modes as \emph{Unruh modes} as their relevant properties were first discussed by Unruh~\cite{Unr:bhevap}.
By~inspection of the relations~(\ref{eq:rindboost}) we can strengthen our earlier statement on the ``indivisibility'' of the Minkowski vacuum. Viewing $L_{L}^\bC$ and $L_{R}^\bC$ as subspaces of $L_{\Sigma}^\bC$ as suggested by writing $L_{\Sigma}^\bC=L_{\rm L}^\bC\oplus L_{\rm R}^\bC$, we can now affirm that $L_{\Sigma}^{\pm}\cap L_{\rm L/R}^\bC=\{0\}$. That is, there are no non-zero germs of solutions that simultaneously vanish either on $\Sigma_{\rm L}$ or $\Sigma_{\rm R}$ and are also either purely of positive or negative energy.

We introduce the following parametrization of $\phi\in L_{\Sigma}^{\bC}$ in terms of the Unruh modes, using coefficient functions $\phi^{\rm L,a}$, $\phi^{\rm L,b}$, $\phi^{\text{R,a}}$, $\phi^{\text{R,b}}$ on $[0,\infty)$,
\begin{gather}
 \phi(t,x)=\int_0^\infty\xd\kappa\, \Big(\phi^{\rm L,a}(\kappa) B^{\rm L}_\kappa(t,x) + \overline{\phi^{\rm L,b}(\kappa) B^{\rm L}_\kappa(t,x)}
 + \phi^{\text{R,a}}(\kappa) B^{\rm R}_\kappa(t,x) \nonumber
 \\ \hphantom{\phi(t,x)=\int_0^\infty\xd\kappa\, \Big(}
 {}+ \overline{\phi^{\text{R,b}}(\kappa) B^{\rm R}_\kappa(t,x)} \Big).
 \label{eq:kgumodes}
\end{gather}
This is related to our previous parametrization~(\ref{eq:kgmbmodes}), with $\kappa>0$ as\footnote{There is a subtlety here connected to the fact that the conversion coefficients become singular in the limit $\kappa\to 0$. To really obtain an equivalence between the relevant Hilbert spaces built on the Minkowski Bessel modes~(\ref{eq:minkbessel}) and the Unruh modes~(\ref{eq:rindboost}) respectively, this has to be taken into account appropriately when defining the measure in the integral expansion~(\ref{eq:kgumodes}) near $\kappa=0$. Since we wish to keep our presentation simple, we refer the interested reader to the appropriate literature~\cite{NFKMB:repcombunruh}.}
\begin{gather}
 \phi^{\rm c}(\kappa) =\frac{1}{\sqrt{2\sinh(\pi\kappa)}}
 \Big(\phi^{\text{R,a}}(\kappa) {\rm e}^{\pi\kappa/2}
 -\phi^{\rm L,a}(\kappa) {\rm e}^{-\pi\kappa/2}\Big), \nonumber
 \\[.5ex]
 \overline{\phi^{\rm d}(\kappa)} =\frac{1}{\sqrt{2\sinh(\pi\kappa)}}
 \Big(\overline{\phi^{\text{R,b}}(\kappa)} {\rm e}^{\pi\kappa/2}
 -\overline{\phi^{\rm L,b}(\kappa)} {\rm e}^{-\pi\kappa/2}\Big), \nonumber
 \\[.5ex]
 \phi^{\rm c}(-\kappa) =\frac{1}{\sqrt{2\sinh(\pi\kappa)}}
 \Big(\overline{\phi^{\rm L,b}(\kappa)} {\rm e}^{\pi\kappa/2}
 -\overline{\phi^{\text{R,b}}(\kappa)} {\rm e}^{-\pi\kappa/2}\Big), \nonumber
 \\[.5ex]
 \overline{\phi^{\rm d}(-\kappa)} =\frac{1}{\sqrt{2\sinh(\pi\kappa)}}
 \Big(\phi^{\rm L,a}(\kappa) {\rm e}^{\pi\kappa/2}
 -\phi^{\text{R,a}}(\kappa) {\rm e}^{-\pi\kappa/2}\Big).
 \label{eq:bmuparam}
\end{gather}
{\samepage
For the symplectic form~(\ref{eq:symfmb}) we obtain the manifest decomposition, $\omega_{\Sigma} =\omega_{\rm L}+\omega_{\rm R}$, with
\begin{gather}
 \omega_{\rm L}(\phi_1,\phi_2) = \im \int_{0}^{\infty}\xd \kappa\,
 \Big(\phi^{\rm L,a}_1(\kappa) \overline{\phi^{\rm L,b}_2(\kappa)}
 - \overline{\phi^{\rm L,b}_1(\kappa)} \phi^{\rm L,a}_2(\kappa)\Big),\nonumber
 \\
 \omega_{\rm R}(\phi_1,\phi_2) = \im \int_{0}^{\infty}\xd \kappa\,
 \Big(\phi^{\text{R,a}}_2(\kappa) \overline{\phi^{\text{R,b}}_1(\kappa)}
 - \overline{\phi^{\text{R,b}}_2(\kappa)} \phi^{\text{R,a}}_1(\kappa)\Big).
 \label{eq:symfum}
\end{gather}}
The inner product~(\ref{eq:ipmb}), decomposed as $(\cdot,\cdot)_{\Sigma}=(\cdot,\cdot)_{\rm L}+(\cdot,\cdot)_{\rm R}$, takes the form
\begin{align}
 (\phi_1,\phi_2)_{\rm L} & =
 4 \int_{0}^{\infty}\xd \kappa\,
 \Big(\overline{\phi^{\rm L,a}_1(\kappa)} \phi^{\rm L,a}_2(\kappa)
 - \phi^{\rm L,b}_1(\kappa) \overline{\phi^{\rm L,b}_2(\kappa)}\Big),\nonumber
 \\
 (\phi_1,\phi_2)_{\rm R} & =
 4 \int_{0}^{\infty}\xd \kappa\,
 \Big(\overline{\phi^{\text{R,b}}_2(\kappa)} \phi^{\text{R,b}}_1(\kappa)
 - \phi^{\text{R,a}}_2(\kappa) \overline{\phi^{\text{R,a}}_1(\kappa)}\Big).
 \label{eq:ipum}
\end{align}

\begin{figure}
 \centering
 \begin{tikzpicture}[scale=1]
\draw[->] (0,-2) -- (0,2) node [left] {$t$};
\draw[->] (-1,0) -- (4,0) node [right] {$x, \eta=0$};
\draw (0,0) -- (4,1.7)node [right] {$\eta=\eta_2>0$};
\draw (0,0) -- (4,-1.7)node [right] {$\eta=\eta_1<0$};
\draw plot[domain=-1.23:1.23] ({1.3*cosh(\x)}, {1.3*sinh(\x)});
\draw plot[domain=-0.634:0.634] ({3*cosh(\x)}, {3*sinh(\x)});
\draw (2,-2) --(0,0) node [midway, below, sloped] {$\rho=0,\eta=-\infty$};
\draw (0,0) -- (2,2) node [midway, above, sloped] {$\rho=0,\eta=\infty$};
\node at (2.6,-2.2) {$\rho_1$};
\node at (3.8,-2.2) {$\rho_2$};
\end{tikzpicture}
 \caption{Coordinates $(\eta,\rho)$ in the right Rindler wedge in Minkowski space.}
 \label{fig:minkrindco}
\end{figure}
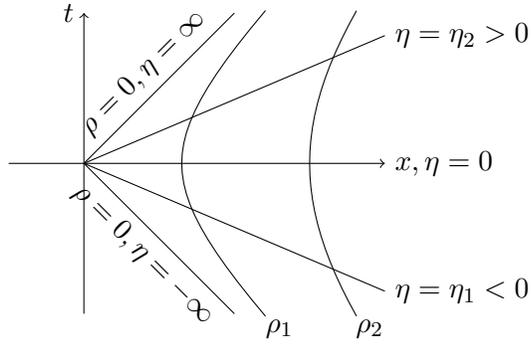

For comparison, we proceed to consider a different vacuum on the hypersurface $\Sigma$ that does split into local component vacua on $\Sigma_{\rm L}$ and $\Sigma_{\rm R}$. For Minkowski space this was first described by Fulling~\cite{Ful:nonuniquequant}. This \emph{Fulling--Rindler vacuum} arises by considering the left and right Rindler wedges as spacetimes in their own right, namely copies of Rindler space. Convenient coor\-di\-na\-tes~$(\eta,\rho)$ for Rindler space as identified with the right wedge \text{R} can be introduced as follows~\cite{Rin:kruskalsuniframe}, see Figure~\ref{fig:minkrindco},
\begin{equation*}
 t=\rho\sinh\eta,\qquad x=\rho\cosh\eta.
\end{equation*}
The coordinate $\eta\in\R$ plays the role of a time coordinate while $\rho>0$ is a spatial coordinate. The Unruh modes restricted to the right Rindler wedge recover the \emph{Fulling modes}~\cite{Ful:nonuniquequant}. For~$\rho>0$ we~have
\begin{equation*}
 B_\kappa^{\rm R}(\rho,\eta)= \frac{\sqrt{2\sinh(\kappa\pi)}}{\pi} K_{\im\kappa}(m\rho) {\rm e}^{-\im\kappa\eta}.
\end{equation*}
Here, $K_{\im\kappa}$ denotes the modified Bessel function of the second kind of order $\im\kappa$. The Fulling--Rindler vacuum arises by taking the time coordinate $\eta$ as the natural evolution parameter. We can then read off from the formula above that the modes $B_\kappa^{\rm R}$ are the respective positive-energy modes. Indeed, from~(\ref{eq:symfum}) and~(\ref{eq:ipum}) we can read off immediately that the corresponding subspace
\begin{equation}
 L_{\rm R}^+=\big\{\phi\in L_{\rm R}^{\bC}\colon \phi^\text{R,a}(\kappa)=0\,\forall\kappa\big\}
 \label{eq:kgrfvacr}
\end{equation}
is a positive-definite Lagrangian subspace~\cite{CoOe:vaclag}. The corresponding vacuum on $\Sigma_{\rm L}$ in the left Rindler wedge is given by
\begin{equation*}
 L_{\rm L}^+=\big\{\phi\in L_{\rm R}^{\bC}\colon \phi^\text{L,b}(\kappa)=0\,\forall\kappa\big\}.
\end{equation*}

Denote the Hilbert spaces of states obtained by K\"ahler quantizing the vacua corresponding to $L_{\rm L}^+$, $L_{\rm R}^+$ and $L_{\Sigma}^{+'}\defeq L_{\rm L}^+ \oplus L_{\rm R}^+$ by $\cH_{\rm L}$, $\cH_{\rm R}$, and $\cH_{\Sigma}'$ respectively. Then, we canonically have $\cH_{\Sigma}'=\cH_{\rm L} \tens \cH_{\rm R}$. In other words, if we choose two copies of the Fulling--Rindler vacuum instead of the Minkowski vacuum, the Hilbert space of states associated to the hyperplane $\Sigma$ does decompose into Hilbert spaces associated to the half-hyperplanes $\Sigma_{\rm L}$ and $\Sigma_{\rm R}$. (In the algebraic framework, where states are functionals on algebras of observables, the corresponding statement was worked out in detail by Kay via a double-wedge algebra construction~\cite{Kay:dwedgessmink}.) One may wonder whether the present example could be generalized to a procedure for constructing vacua localized on a larger class of hypersurfaces with boundaries. It seems a key ingredient could be for solutions in the would-be Lagrangian subspace to approximate the behavior of Unruh modes near the boundary with respect to the light-sheets emanating from it. However, these considerations are far outside the scope of the present paper.

Supposing we had chosen vacua localized on hypersurfaces with boundaries in order to allow for gluing in the sense of the old axioms (Appendix~\ref{sec:qobsaxioms}), then we would still need to compare to the physical vacua that typically do not localize along boundaries. For the simple example at hand, the comparison between the Minkowski vacuum encoded by $L_{\Sigma}$ and the double Fulling--Rindler vacuum encoded by $L_{\Sigma}'$ was already performed by Unruh~\cite{Unr:bhevap}. More precisely, we are interested in the pseudo-state $Y$ that represents the Minkowski vacuum in the Hilbert space~$\cH_{\Sigma}'$. As discussed in~Section~\ref{sec:wfvac}, the general form of its holomorphic wave function is given by equation~(\ref{eq:chvacwf}). Here, we obtain for $\xi\in L_{\Sigma}$
\begin{equation}
 Y(\xi)=\exp\bigg(2\int_0^\infty\xd\kappa\, {\rm e}^{-\pi\kappa}\,
 \xi^{\rm L,a}(\kappa) \overline{\xi^{\text{R,b}}(\kappa)}\bigg).
 \label{eq:minktodfrwv}
\end{equation}
This makes explicit how the Minkowski vacuum entangles the data on the two half-hyperpla\-nes~$\Sigma_{\rm L}$ and $\Sigma_{\rm R}$. On the other hand, the wave function is not square-integrable, corresponding to the lack of unitary equivalence of the Hilbert spaces $\cH_{\Sigma}$ and $\cH_{\Sigma}'$ as representations of the algebra of slice observables on $\Sigma$.
In order to facilitate a particle interpretation we may rewrite the wave function in a factorized form (recall equation~(\ref{eq:mswf})),
\begin{equation}
 Y(\xi)=\exp\bigg(
 \frac{1}{2}\int_0^\infty\xd\kappa\, {\rm e}^{-\pi\kappa}\,
 \big\{\tilde{B}_\kappa^{\rm L},\xi\big\}_{\Sigma} \big\{\tilde{B}_\kappa^{\rm R},\xi\big\}_{\Sigma}\bigg).
 \label{eq:minktodfrwvp}
\end{equation}
Here, $\tilde{B}_\kappa^{\rm L}$ and $\tilde{B}_\kappa^{\rm R}$ are the real parts of $B_\kappa^{\rm L}$ and $B_\kappa^{\rm R}$ respectively. We can roughly interpret this as a sea of particles, consisting of pairs. A pair consists of one particle localized in the left half-space, the other localized in the right half-space, with matching quantum number $\kappa$. Moreover, the amplitude for the pairs is exponentially suppressed with $\kappa$.
From the previous expression we can also read off a representation in terms of creation operators on the double Fulling--Rindler vacuum,
\begin{equation}
 Y=\exp\bigg(\int_0^\infty\xd\kappa\, {\rm e}^{-\pi\kappa}\,
 a^\dagger_{{\rm L},\kappa} a^\dagger_{\text{R},\kappa}\bigg)
 \coh_{{\rm L},0}\tens\coh_{\text{R},0}.
 \label{eq:minktodfro}
\end{equation}
Here, we use the notation $a^\dagger_{{\rm L},\kappa}=a^\dagger_{\tilde{B}_\kappa^{\rm L}}$ and $a^\dagger_{\text{R},\kappa}=a^\dagger_{\tilde{B}_\kappa^{\rm R}}$. We recover (a more precise version of) Unruh's original formula~\cite{Unr:bhevap}. However, as already mentioned, this formula does not define a~normalizable state in the Hilbert space $\cH_{\Sigma}'$ and therefore needs to be interpreted with caution. For mathematical clarity, the wave function representations~(\ref{eq:minktodfrwv}) and~(\ref{eq:minktodfrwvp}) are preferable.

Given a slice observable $F$ on $L_{\Sigma}$ we can evaluate its expectation value in the Minkowski vacuum in the conventional way by using its action on the Hilbert space $\cH_{\Sigma}$ (see Section~\ref{sec:kobs}),
\begin{equation*}
 \langle F \rangle = \big\langle \coh_0, \hat{F} \coh_0\big\rangle_{\Sigma}
 = \tr_{\cH_\Sigma}\big(\hat{F}\, \tilde{\coh}_0 \big)\qquad\text{with}\quad
 \tilde{\coh}_0\defeq |\coh_0\rangle\langle\coh_0|.
\end{equation*}
On the other hand we might use the Hilbert space $\cH_{\Sigma}'$ instead with the pseudo-state $Y$,
\begin{equation*}
 \langle F \rangle = \big\langle Y, \hat{F} Y\big\rangle_{\Sigma}'
 = \tr_{\cH_\Sigma'}\big(\hat{F}\, \tilde{Y}\big)\qquad\text{with}\quad
 \tilde{Y}\defeq|Y\rangle\langle Y|.
\end{equation*}
Let us now suppose that the slice observable $F$ is localized on the right half-hyperplane $\Sigma_{\rm R}$. That is, $F$ depends only on the subspace $L_{\rm R}\subseteq L_{\Sigma}$. We may then separate the trace over $\cH_{\Sigma}'$ into a trace over $\cH_{\rm R}$ and one over $\cH_{\rm L}$ and perform the latter independently of $F$.
\begin{equation*}
 \langle F \rangle
 = \tr_{\cH_{\rm R}}\big(\hat{F}\, \tilde{Y}_{\rm R}\big)\qquad\text{with}\quad
 \tilde{Y}_{\rm R}\defeq \tr_{\cH_{\rm L}}\big(\tilde{Y}\big).
\end{equation*}
In this sense the (non-normalizable) mixed pseudo-state $\tilde{Y}_{\rm R}$ can be thought of as encoding the Minkowski vacuum on the right half-hyperplane $\Sigma_{\rm R}$, i.e, on $\cH_{\rm R}$. In the literature this is sometimes expressed by saying that the Minkowski vacuum ``is'' a mixed state in the (right) Rindler wedge. With~(\ref{eq:comrelkip}),~(\ref{eq:relaip}),~(\ref{eq:ipum}) we obtain the commutation relations
\begin{equation*}
 \big[a_{{\rm L},\kappa},a^\dagger_{{\rm L},\kappa'}\big]
 =\big\{\tilde{B}^{\rm L}_{\kappa'},\tilde{B}^{\rm L}_{\kappa}\big\}_{\rm L}
 =\frac{1}{4}\big(B^{\rm L}_{\kappa'},B^{\rm L}_{\kappa}\big)_{\rm L}
 =\delta(\kappa-\kappa').
\end{equation*}
With these we obtain $\tilde{Y}_{\rm R}$ from~(\ref{eq:minktodfro}),
\begin{gather*}
\begin{split}
& \tilde{Y}_{\rm R} =\sum_{n=0}^\infty \frac{1}{n!}\int\xd\kappa_1\cdots\xd\kappa_n\, {\rm e}^{-2\pi(\kappa_1+\cdots +\kappa_n)}
 a^\dagger_{\text{R},\kappa_1}\cdots a^\dagger_{\text{R},\kappa_n} \tilde{\coh}_{\text{R},0}\,
 a_{\text{R},\kappa_1}\cdots a_{\text{R},\kappa_n} \nonumber
 \\
&\hphantom{\tilde{Y}_{\rm R}}{} =\exp\bigg(\int_0^\infty\xd\kappa\, {\rm e}^{-2\pi\kappa} A_{\text{R},\kappa} \bigg) \tilde{\coh}_{\text{R},0}.\end{split}\label{eq:minktofrmix}
\end{gather*}
Here we have defined the super-operator
\begin{equation*}
 A_{\text{R},\kappa}(\sigma)=a^\dagger_{\text{R},\kappa}\sigma\, a_{\text{R},\kappa}.
\end{equation*}

\subsection{Minkowski vacuum in Rindler space}

In the present section we consider what the Minkowski vacuum as a boundary condition looks like in Rindler space as identified with the right wedge $\text{R}$, recall Figure~\ref{fig:minkpflr}. This question was first successfully addressed by Rätzel and one of the authors in the language of states and wave functions encoding vacuum change~\cite{CoRa:unruh}. We complement their analysis here with a first-principles approach.

\begin{figure}
 \centering
 \begin{tikzpicture}[scale=1]
\draw[white,fill=gray!20] (-3,-3) rectangle (3,3);
\draw[white,fill=white] (0,0) -- (3,2.2) -- (3,-2.2) -- cycle;
\draw[dashed] (-3,0) -- (0,0);
\draw[blue,thick,-] (0,0) -- (3,3);
\draw[blue,thick,-] (0,0) -- (3,-3);
\draw[very thick] (0,0) -- (3,2.2) node [right] {$\eta_2$};
\draw[very thick] (0,0) -- (3,-2.2) node [right] {$\eta_1$};
\node at (1.7,0) {$M$};
\node at (2.5,1.4) {$\Sigma_2$};
\node at (2.5,-1.4) {$\Sigma_1$};
\node at (-1,1) {$X$};
\end{tikzpicture}
 \caption{Minkowski space with embedded right Rindler wedge (delimited by blue lines). Inside the Rindler wedge a region $M$ is marked between the initial hypersurface $\Sigma_1$ at Rindler time $\eta_1$ and the final one $\Sigma_2$ at Rindler time $\eta_2$.}
 \label{fig:rindreginmink}
\end{figure}
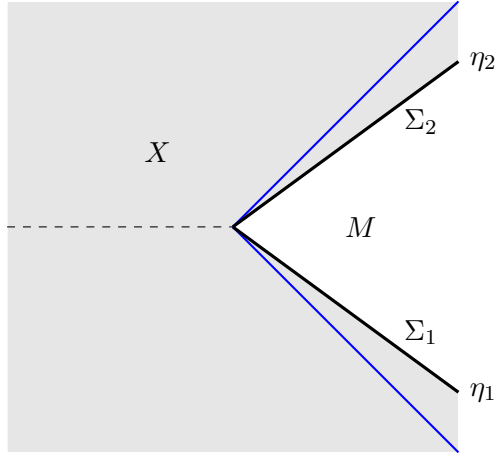

Choose initial and final Rindler times $\eta_1<\eta_2$. Consider the region $M$ enclosed between the corresponding initial and final Rindler-time hypersurfaces $\Sigma_1$ and $\Sigma_2$, see Figure~\ref{fig:rindreginmink}. (We~orient the hypersurfaces as boundaries of regions to the Rindler-future.) In order to evaluate the vacuum correlation functions of observables in $M$ we need to specify the relevant vacuum. Within the present framework this specification is in terms of a Lagrangian subspace $L_{\partial M}^+\subseteq L_{\partial M}^{\bC}=L_{\Sigma_1}^{\bC}\oplus L_{\overline{\Sigma_2}}^{\bC}$ of the complexified space of germs of solutions on the boundary $\partial M$ of $M$. The vacuum encodes boundary conditions in the region~$X$, exterior to $M$. Within Rindler space the region~$X$ restricts to $X\cap R$ and decomposes into two components, one to the past of $\Sigma_1$ and one to the future of $\Sigma_2$. Correspondingly, the boundary conditions for the standard Fulling--Rindler vacuum are independent in the two components. Consequently, the corresponding Lagrangian subspace decomposes as, $L_{\partial M}^+=L_{\Sigma_1}^+ \oplus L_{\overline{\Sigma_2}}^+$, where $L_{\Sigma_1}^+\subseteq L_{\Sigma_1}^{\bC}$ and $L_{\overline{\Sigma_2}}^+\subseteq L_{\overline{\Sigma_2}}^{\bC}$ are individually Lagrangian subspaces. Indeed,
$L_{\Sigma_1}^+$ and $L_{\overline{\Sigma_2}}^+=\overline{L_{\Sigma_2}^+}$ are just (ordinary respectively conjugate) copies of $L_{\rm R}^+$, see~(\ref{eq:kgrfvacr}).
Moreover, all these Lagrangian subspaces are positive-definite, admitting a standard K\"ahler quantization. Furthermore, the future and past vacua are complex conjugate, allowing to speak of just ``the'' vacuum. This is precisely analogous to the standard situation for a~time-interval region in Minkowski space.

In contrast, we are interested here in the standard Minkowski vacuum as seen from the right Rindler wedge. We thus have to work out the Lagrangian subspace $L_{\partial M}^{\mathcal{M}}\subseteq L_{\partial M}^{\bC}$ on $\partial M$ induced from the Minkowski boundary conditions in $X$. To this end we parametrize the solutions in $X$ in terms of the Unruh modes~(\ref{eq:rindboost}) as follows, using coefficient functions $\phi^{\rm L,a}$, $\phi^{\rm L,b}$, $\phi^{1,{\rm a}}$, $\phi^{1,{\rm b}}$, $\phi^{2,{\rm a}}$, $\phi^{2,{\rm b}}$ on $[0,\infty)$,
\begin{gather*}
 \phi(t,x)=\int_0^\infty\xd\kappa\, \Big(\phi^{\rm L,a}(\kappa) B^{\rm L}_\kappa(t,x) + \overline{\phi^{\rm L,b}(\kappa) B^{\rm L}_\kappa(t,x)}+ \phi^{1,{\rm a}}(\kappa) B^{1}_\kappa(t,x)
 + \overline{\phi^{1,{\rm b}}(\kappa) B^1_\kappa(t,x)}
 \\ \hphantom{\phi(t,x)=\int_0^\infty\xd\kappa\, \Big(}
{} + \phi^{2,{\rm a}}(\kappa) B^{2}_\kappa(t,x) + \overline{\phi^{2,{\rm b}}(\kappa) B^{2}_\kappa(t,x)}
 \Big).
\end{gather*}
The solutions $B^1_\kappa$ and $B^2_\kappa$ are suitably cut-off versions of the solutions $B^{\rm R}_{\kappa}$. This is to be understood as follows. Recall that the solutions $B^{\rm R}_{\kappa}$ vanish in the left wedge ${\rm L}$ (compare Figure~\ref{fig:minkpflr}). Therefore, their support within $X$ has two connected components, one is the region $\text{P}$ together with the past of $\eta_1$ in the right wedge $\text{R}$. The other is the region $\text{F}$ together with the future of $\eta_2$ in the right wedge $\text{R}$. $B^1_\kappa$ and $B^2_\kappa$ are versions of $B^{\rm R}_{\kappa}$ with their support restricted to the respective component. Near the hypersurface $\partial M$ the solutions $B^{\rm L}_{\kappa}$ vanish so that they disappear from the parametrization of $L_{\partial M}^{\bC}$. However, they do occur in $X$ and are essential for the asymptotic boundary conditions determining the Minkowski vacuum. The latter are expressible as past and future boundary conditions. To fix them it is sufficient to consider the behavior of~solutions in the regions $\text{P}$ and $\text{F}$ respectively. In the region $\text{P}$ the modes $B^2_\kappa$ are not present and the vacuum condition is given by the subspace~(\ref{eq:kgpvacmb}) in terms of the parametrization~(\ref{eq:kgmbmodes}). With the reparametrization identities~(\ref{eq:bmuparam}) this yields the conditions
\begin{equation}
 \phi^{1,{\rm a}}(\kappa) {\rm e}^{\pi\kappa/2}
 =\phi^{\rm L,a}(\kappa) {\rm e}^{-\pi\kappa/2},\qquad
 \overline{\phi^{\rm L,b}(\kappa)} {\rm e}^{\pi\kappa/2}
 =\overline{\phi^{1,{\rm b}}(\kappa)} {\rm e}^{-\pi\kappa/2}.
 \label{eq:mvp}
\end{equation}
Similarly, implementing the future boundary condition, which is conjugate to~(\ref{eq:kgpvacmb}), in $\text{F}$, yields
\begin{equation}
 \overline{\phi^{2,{\rm b}}(\kappa)} {\rm e}^{\pi\kappa/2}
 =\overline{\phi^{\rm L,b}(\kappa)} {\rm e}^{-\pi\kappa/2} \qquad
 \phi^{\rm L,a}(\kappa) {\rm e}^{\pi\kappa/2}
 =\phi^{2,{\rm a}}(\kappa) {\rm e}^{-\pi\kappa/2}.
 \label{eq:mvf}
\end{equation}
We may now eliminate the coefficients corresponding to the solutions $B^{\rm L}_{\kappa}$ from the system~(\ref{eq:mvp}), (\ref{eq:mvf}) to arrive at the equations determining the subspace $L_{\partial M}^{\mathcal{M}}\subseteq L_{\partial M}^{\bC}$ encoding the Minkowski vacuum on $\partial M$,
\begin{equation*}
 L_{\partial M}^{\mathcal{M}}=\Big\{\phi\in L_{\partial M}^{\bC}\colon
 \phi^{1,{\rm a}}(\kappa) {\rm e}^{\pi\kappa}
 =\phi^{2,{\rm a}}(\kappa) {\rm e}^{-\pi\kappa},\;
 \overline{\phi^{2,{\rm b}}(\kappa)} {\rm e}^{\pi\kappa}
 =\overline{\phi^{1,{\rm b}}(\kappa)} {\rm e}^{-\pi\kappa}\; \forall\kappa \Big\}.
\end{equation*}
We can see immediately that the Minkowski vacuum correlates the modes between the initial and final hypersurfaces, as expected. In stark contrast to the Fulling--Rindler vacuum it does not decompose into independent vacua (and thus subspaces) for each of the two hypersurfaces.

We proceed to consider the pseudo-state representing the Minkowski vacuum on the Hilbert space $\cH_{\partial M}=\cH_{\Sigma_1}\tens\cH_{\overline{\Sigma_2}}$ of states on $\partial M$ based on the Fulling--Rindler vacuum. The corresponding holomorphic wave function $Y_{\partial M}$ takes the form~(\ref{eq:chvacwf}),
\begin{equation*}
 Y_{\partial M}(\xi)=\exp\bigg(2\int_0^\infty\xd\kappa\, {\rm e}^{-2\pi\kappa}\,
 \xi^{2,{\rm a}}(\kappa) \overline{\xi^{1,{\rm b}}(\kappa)}\bigg).
\end{equation*}
This makes explicit how the correlations between past and future boundary conditions inherent in the Minkowski vacuum as seen from Rindler space translate to an entangled state. Here, crucially, the entanglement is not between different subsystems (at the same time), but between the past and future of the same system.

This wave function was obtained first (in a slightly different form) in~\cite{CoRa:unruh}. There it was shown (via what is here formula~(\ref{eq:kcorrgs})) that it yields for spacetime observables localized in $M$ precisely the same vacuum correlation functions as those arising from the standard quantization in Minkowski space. Here, we have derived the wave function from first principles, guaranteeing by construction the coincidence of correlation functions. The ansatz for the wave function~$Y_{\partial M}$ in~\cite{CoRa:unruh} was based on the relation to the mixed pseudo-state $\tilde{Y}_{\rm R}$ (formula~(\ref{eq:minktofrmix})). Indeed, for the special case of slice observables on $\Sigma_{\rm R}$, both by construction yield the same vacuum correlation functions. However, conceptually they are very different objects. For example, the mixed pseudo-state $\tilde{Y}_{\rm R}$ cannot be used to directly evaluate spacetime observables in $M$. What is more, in contrast to the pseudo-state $Y_{\partial M}$, the existence of $\tilde{Y}_{\rm R}$ hinges on the time-reversal symmetry of the vacuum. (That is, the Lagrangian subspaces for past and future vacuum are complex conjugates.)

We rewrite the wave function of $Y_{\partial M}$ to clarify its particle interpretation
\begin{equation*}
 Y_{\partial M}(\xi)=\exp\bigg(
 \frac{1}{2}\int_0^\infty\xd\kappa\, {\rm e}^{-2\pi\kappa}\,
 \big\{\tilde{B}_\kappa^{1},\xi\big\}_{\partial M} \big\{\tilde{B}_\kappa^{2},\xi\big\}_{\partial M}\bigg).
\end{equation*}
In terms of creation operators we may write
\begin{equation*}
 Y_{\partial M}=\exp\bigg(\int_0^\infty\xd\kappa\, {\rm e}^{-2\pi\kappa}\,
 a^{\dagger}_{1,\kappa} a^{\dagger}_{2,\kappa}\bigg) \coh_{\Sigma_1,0}\tens\coh_{\overline{\Sigma_2},0}.
\end{equation*}
Here, $a^{\dagger}_{1,\kappa}\defeq a^{\dagger}_{\tilde{B}_\kappa^1}$ and $a^{\dagger}_{2,\kappa}\defeq a^{\dagger}_{\tilde{B}_\kappa^2}$. We may interpret this formula as saying that the Minkowski vacuum appears as a sea of particle pairs. Each pair consists of one initial particle and one final particle, with matching quantum numbers $\kappa$. The contribution of pairs is exponentially suppressed with increasing $\kappa$.

\section{Discussion and outlook}

\looseness=-1 The present work represents an important step in the program of \emph{general boundary quantum field theory $($GBQFT$)$}~\cite{Oe:gbqft}. This program is based on the premise that the physics in a spacetime region can be described by recourse exclusively to quantities associated to the region and its (infinitesimal) boundary. What is more, the physics in a region can be recovered completely by combining (through an operation of \emph{composition}) the physics of subregions into which this region may have been divided.
In previous work~\cite{Oe:feynobs}, the description of physics on hypersurfaces (through Hilbert spaces of states) and in spacetime regions (through amplitudes and correlation functions) was limited by the requirement that the vacuum on the hypersurface or outside of the region in question had to be a K\"ahler vacuum. However, as was shown by the authors in~\cite{CoOe:vaclag}, the~vacuum of standard QFTs on hypersurfaces that are not spacelike is in general not of K\"ahler type.

In the present work, the limitation to K\"ahler vacua in the quantization of field theory, present in virtually all of the literature on curved spacetime QFT, is removed. Crucially, the generalization beyond K\"ahler vacua succeeds not only at the level of individual hypersurfaces or regions, but also at the level of composition (Sections~\ref{sec:compobs} and \ref{sec:corrak}) and thus for the framework as a whole. In the following, we aim to shed additional light on some of the results and their context.

An important ingredient in achieving the generalization beyond K\"ahler vacua is a shift in focus from (Hilbert space) states to observables. In this respect we are bringing GBQFT closer to \emph{algebraic quantum field theory $($AQFT$)$}~\cite{Haa:lqp}. In this version of GBQFT (developed in~Section~\ref{sec:genquant}) no Hilbert spaces are constructed. Rather, the role that states usually play is taken over by a special type of observable, the \emph{slice observables} (developed in~Sections~\ref{sec:sliceobs} and \ref{sec:kobs}). As states would be, these are localized on hypersurfaces. Given a K\"ahler vacuum, there is moreover a~direct correspondence between slice observables and ordinary Hilbert space states that can be mediated via the GNS construction, see Section~\ref{sec:kobs}. Crucial for this is a notion of algebra structure for slice observables. In~GBQFT there is in general no notion of product between observables in a given spacetime region. Composition between observables is only possible by composing underlying spacetime regions. However, the composition of two slice regions of the same type yields a slice region, which, moreover, is again of the same type. This induces the algebra structure on the space of slice observables associated to a given hypersurface and thus slice region, see Section~\ref{sec:qsobs}. At the same time this provides a bridge to the algebra structure present in AQFT for the observables in a given spacetime region. On the AQFT side an important ingredient in this respect is the \emph{time slice axiom}~\cite{Haa:lqp}. This implies that there is an isomorphism between the space, call it $A(M)$, of observables in a spacetime region $M$, and the space $A(\Sigma)$ of observables in an infinitesimal neighborhood of a hypersurface $\Sigma$ in $M$ that is Cauchy for $M$. From the AQFT perspective the algebra structures of $A(M)$ and $A(\Sigma)$ are both fundamental and the isomorphism is an isomorphism of algebras. From the GBQFT perspective on the other hand only $A(\Sigma)$, being essentially the space $\qsoa_{\Sigma}$ of slice observables on $\Sigma$, carries an algebra structure. The algebra structure on $A(M)$ appears thus as inherited from that of $A(\Sigma)$ via the isomorphism (here of vector spaces). It should be emphasized at this point that the space $A(M)$ considered in~AQFT is not the same as the space of observables in the spacetime region $M$ in GBQFT. Rather, $A(M)$ arises as a quotient of the latter. Two Weyl observables give rise to the same object in~$A(M)$ if their associated affine spaces of solutions (compare Section~\ref{sec:piwobs}) coincide outside of~$M$. This is also the key ingredient for constructing the mentioned isomorphism between~$A(M)$ and~$A(\Sigma)$ from a~GBQFT perspective. The relevant equation is equation~(\ref{eq:obssol}). There are some caveats, and we leave the details to the reader.

As explained, the methods developed in the present work should finally allow to provide a truly local description of realistic QFTs (recall Figure~\ref{fig:st-decomp}). As laid out in~Section~\ref{sec:genquant}, such a description can be achieved purely based on slice observables, without any Hilbert space of states. However, a Hilbert space of states, rather than merely an algebra of slice observables is for certain purposes desirable or possibly even necessary. For example, the standard notion of particle makes reference to a Hilbert space. Remarkably, as we have laid out in~Section~\ref{sec:iprod}, with just one additional ingredient (a real structure), a Hilbert space of states can be constructed even for hypersurfaces where the polarization encoding the vacuum is not K\"ahler, and where therefore standard quantization prescriptions fail. What is more, as shown in~Section~\ref{sec:corrak} (and proven in Appendix~\ref{sec:compip}) a TQFT-style gluing prescription extends to these Hilbert spaces raising the prospect of quantizing realistic field theories to obtain a local quantum field theory that would satisfy the full-blown GBQFT axioms (Appendix~\ref{sec:qobsaxioms}). For example, in this way it should be possible to extend existing GBQFT quantizations of Klein--Gordon theory that are limited to a narrow class of infinitely-extended regions~\cite{Col:desitterpaper,CoOe:smatrixgbf,DoOe:complexads,Oe:timelike,Oe:kgtl,Oe:holomorphic} to much more general regi\-ons, crucially including compact ones of arbitrary size, thus realizing full locality. Results on massless Klein--Gordon theory in Lorentzian spacetimes~\cite{CaMn:waverel} indicate that the classical axi\-oms (Appendix~\ref{sec:caxioms}) required for a successful quantization in the present sense are satisfied quite generically at least in 1+1 dimensions. Also, it should be possible to carry over the present methods to affine theories~\cite{Oe:affine} and theories with (at least abelian) gauge symmetries~\cite{DiOe:qabym}, achieving quantizations that implement full locality in physically realistic settings.

The quantization problem of constructing Hilbert spaces of states associated to hypersurfaces is a~generic part of TQFT-type approaches to QFT. Usually, this has been addressed by producing (explicitly or implicitly) a~K\"ahler polarization and quantizing in the standard way (as in~Section~\ref{sec:stdstate}). In hindsight~\cite{CoOe:vaclag} it is clear that this may lead to a~``wrong'', i.e., unphysical quantization, if the K\"ahler polarization used does not correspond to the physical vacuum. This problem occurs in particular for timelike hypersurfaces in spacetimes with Lorentzian signature, and generically in spacetimes with Euclidean signature, where the physical vacuum does not correspond to any K\"ahler polarization on the hypersurface. An example for an ad~hoc choice (motivated by analytic continuation) of a~K\"ahler polarization on timelike hypersurfaces in Klein--Gordon theory can be found in~\cite{Oe:holomorphic}. An example of addressing the problem by classifying K\"ahler polarizations (in terms of complex structures) can be found in~\cite{DoOe:complexads}. For field theory in spacetime with Euclidean signature, a~rather general prescription for constructing a~K\"ahler polarization has been given by Graeme Segal \cite[equation~(3.0.5)]{Seg:qft2notes}. This has been used for example in~\cite{Dia:dtonopab,DiOe:qabym} in terms of the Dirichlet-to-Neumann map. From the present point of view, a~``physically correct'' construction of Hilbert spaces in these cases should be based instead on the $\alpha$-K\"ahler quantization prescription as exhibited in~Section~\ref{sec:iprod}. So,~are the previously mentioned K\"ahler quantization prescriptions indeed simply wrong? Fortunately, the answer is ``no'', in the following sense. As shown in~Section~\ref{sec:retrel}, we can pull back an $\alpha$-K\"ahler polarization to an ordinary K\"ahler polarization via a~suitable identification map between the real and the $\alpha$-subspace of the complexified phase space on the hypersurface. This also brings into correspondence the associated quantization prescriptions, including Hilbert spaces, amplitudes and correlation functions. The bad news is that this alters in general the symplectic form and also the observables. So,~quantization with this ordinary K\"ahler structure would still be wrong if we used it with the standard symplectic form on phase space and with standard observables. However, there is a~particular class of real polarizations that admit a~positive-definite reflection map which in turn yields an $\alpha$-K\"ahler polarization, see Section~\ref{sec:realpol}. What is special about this class is that it admits a~pull-back to an ordinary K\"ahler polarization which preserves the symplectic form. A~K\"ahler quantization based on this pulled back polarization and the standard symplectic form might thus be considered ``correct'', at least if one is not to include observables. Even more specifically, a~positive-definite reflection map yielding such an $\alpha$-K\"ahler polarization might be obtained as a~rather specific map based on a~position-momenta decomposition of the phase space, see equation~(\ref{eq:nmreflect}). It~turns out that all the mentioned examples can be obtained precisely in this way. In~particular, the Segal polarization is precisely the pull-back of the $\alpha$-K\"ahler polarization induced by the positive-definite reflection map~(\ref{eq:nmreflect}), where ``positions'' correspond to field values and ``momenta'' correspond to their normal derivatives at the hypersurface. In~this sense, our methods and results provide an a~posteriori justification for the Segal polarization and certain other quantization prescriptions in TQFT-style approaches to QFT. At the same time, they clarify how the integration of observables into such prescriptions would have to be performed.\looseness=-1

As laid out in detail in~Section~\ref{sec:refpos}, the $\alpha$-K\"ahler quantization prescription allowing for the construction of Hilbert spaces of states for non-K\"ahler polarizations was inspired by the notion of \emph{reflection positivity} in the \emph{Euclidean approach} to QFT. Indeed, the content of Section~\ref{sec:iprod} may be seen as a providing a vast generalization of the concept of reflection positivity. This might foster hope that our methods in turn contribute to progress on the long-standing problem of generalizing the Euclidean approach to curved spacetime.

\appendix
\section{Axioms for classical linear field theory}
\label{sec:caxioms}

We provide here the axiomatization of classical linear field theory underlying the GBQFT quantization scheme. The present version is taken from~\cite{CoOe:vaclag}, but with notation adapted for the present paper. It generalizes the version provided in~\cite{Oe:holomorphic}, but without including choices of complex structures. Also, it does not include classical observables. For an axiomatization of the latter, which we assume to hold throughout the present work, see~\cite{Oe:feynobs}. The following axiomatic system relies on encoding spacetime in terms of \emph{regions} and \emph{hypersurfaces}, forming a \emph{spacetime system}. For a complete definition of the latter see \cite[Section~2.1]{Oe:feynobs}.
\begin{itemize}\itemsep=0pt
\item[{(C1)}] Associated to each hypersurface $\Sigma$ is a real vector space $L_{\Sigma}$. $L_{\Sigma}$ is equipped with a~non-degenerate symplectic form $\omega_{\Sigma}$.
\item[{(C2)}] Associated to each hypersurface $\Sigma$ there is an (implicit) linear involution $L_\Sigma\to L_{\overline\Sigma}$, such that $\omega_{\overline{\Sigma}}=-\omega_{\Sigma}$.
\item[{(C3)}] Suppose the hypersurface $\Sigma$ decomposes into a union of hypersurfaces $\Sigma=\Sigma_1\cup\cdots\cup\Sigma_n$. Then, there is an (implicit) isomorphism $L_{\Sigma_1}\oplus\cdots\oplus L_{\Sigma_n}\to L_\Sigma$. The isomorphism preserves the symplectic form.
\item[{(C4)}] Associated to each region $M$ is a complex vector space $L_M^\bC$.
\item[{(C5)}] Associated to each region $M$ there is a complex linear map $r_M\colon L_M^\bC\to L_{\partial M}^\bC$. The image $r_M\big(L_M^\bC\big)$ is a Lagrangian subspace of $L_{\partial M}^\bC$.
\item[{(C6)}] Let $M_1$ and $M_2$ be regions and $M= M_1\sqcup M_2$ be their disjoint union. Then $L_M^\bC$ is the direct sum $L_{M}^\bC=L_{M_1}^\bC\oplus L_{M_2}^\bC$. Moreover, $r_M=r_{M_1} + r_{M_2}$.
\item[{(C7)}] Let $M$ be a region with its boundary decomposing as a union $\partial M=\Sigma_1\cup\Sigma\cup \overline{\Sigma'}$, where $\Sigma'$ is a copy of $\Sigma$. Let $M_1$ denote the gluing of $M$ to itself along $\Sigma$, $\overline{\Sigma'}$ and suppose that~$M_1$ is a region. Then, there is an injective complex linear map $r_{M;\Sigma,\overline{\Sigma'}}\colon L_{M_1}^\bC\toi L_{M}^\bC$ such that
\begin{equation*}
 L_{M_1}^\bC\toi L_{M}^\bC\rightrightarrows L_\Sigma^\bC
\end{equation*}
is an exact sequence. Here the arrows on the right-hand side are compositions of the map $r_M$ with the complexified projections of $L_{\partial M}$ to $L_\Sigma$ and $L_{\overline{\Sigma'}}$ respectively (the latter identified with $L_\Sigma$). Moreover, the following diagram commutes, where the bottom arrow is the projection:
\begin{gather}
\begin{split}
&\xymatrix{
 L_{M_1}^\bC \ar[rr]^{r_{M;\Sigma,\overline{\Sigma'}}} \ar[d]_{r_{M_1}} & & L_{M}^\bC \ar[d]^{r_{M}}\\
 L_{\partial M_1}^\bC & & L_{\partial M}^\bC. \ar[ll]}
\end{split}
\label{eq:solext}
\end{gather}
\end{itemize}

\section{Axioms for quantum field theory with observables}
\label{sec:qobsaxioms}

A GBQFT with observables is a model satisfying the axioms listed in the following~\cite{Oe:holomorphic, Oe:feynobs}. Concerning the notions of \emph{region} and \emph{hypersurface}, see the comments at the beginning of~Appen\-dix~\ref{sec:caxioms}.

\begin{itemize}\itemsep=0pt
\setlength{\leftskip}{0.47cm}
\item[{(T1)}] Associated to each hypersurface $\Sigma$ is a complex
 separable Hilbert space $\cH_\Sigma$, called the \emph{state space} of
 $\Sigma$. We denote its inner product by
 $\langle\cdot,\cdot\rangle_\Sigma$. If $\Sigma$ is the empty set, $\cH_{\Sigma}=\bC$.
\item[{(T1b)}] Associated to each hypersurface $\Sigma$ is a conjugate linear
 isometry $\iota_\Sigma\colon \cH_\Sigma\to\cH_{\overline{\Sigma}}$. This map
 is an involution in the sense that $\iota_{\overline{\Sigma}}\circ\iota_\Sigma$
 is the identity on $\cH_\Sigma$.
\item[{(T2)}] Given a hypersurface $\Sigma$ decomposing into a
 union of hypersurfaces $\Sigma=\Sigma_1\cup\cdots\cup\Sigma_n$, there is an isometric isomorphism of Hilbert spaces $\tau_{\Sigma_1,\dots,\Sigma_n;\Sigma}\colon \cH_{\Sigma_1}\tens\cdots\tens\cH_{\Sigma_n}\to\cH_\Sigma$. This is required to be associative in the obvious way.
\item[{(T2b)}] The involution $\iota$ is compatible with the above decomposition. That is,
 \begin{equation*}
 \tau_{\overline{\Sigma}_1,\dots,\overline{\Sigma}_n;\overline{\Sigma}}
 \circ(\iota_{\Sigma_1}\tens\cdots\tens\iota_{\Sigma_n})
 =\iota_\Sigma\circ\tau_{\Sigma_1,\dots,\Sigma_n;\Sigma}.
 \end{equation*}
\item[{(TO4)}] Associated to each region $M$ there is a complex vector space $\corr_M$ of complex linear maps from a dense subspace $\cH_{\partial M}^\ds\subseteq \cH_{\partial M}$ to the complex numbers. These are called \emph{correlation functions}. There is a special element $\rho_M\in\corr_M$, called the \emph{amplitude} map.
\item[{(T3x)}] Let $\Sigma$ be a hypersurface. The boundary $\partial\hat{\Sigma}$ of the associated slice region $\hat{\Sigma}$ decomposes into the union $\partial\hat{\Sigma}=\overline{\Sigma}\cup\Sigma'$, where $\Sigma'$ denotes a second copy of $\Sigma$. Then, $\tau_{\overline{\Sigma},\Sigma';\partial\hat{\Sigma}}\big(\cH_{\overline{\Sigma}}\tens\cH_{\Sigma'}\big) \subseteq\cH_{\partial\hat{\Sigma}}^\ds$. Moreover, $\rho_{\hat{\Sigma}}\circ\tau_{\overline{\Sigma},\Sigma';\partial\hat{\Sigma}}$ restricts to a bilinear pairing $(\cdot,\cdot)_\Sigma\colon \cH_{\overline{\Sigma}}\times\cH_{\Sigma'}\to\bC$ such that $\langle\cdot,\cdot\rangle_\Sigma=(\iota_\Sigma(\cdot),\cdot)_\Sigma$.
\item[{(TO5a)}] Let $M_1$ and $M_2$ be regions and $M= M_1\sqcup M_2$ be their disjoint union. Then $\partial M=\partial M_1\sqcup \partial M_2$ is also a disjoint union and $\tau_{\partial M_1,\partial M_2;\partial M}(\cH_{\partial M_1}^\ds\tens \cH_{\partial M_2}^\ds)\subseteq \cH_{\partial M}^\ds$. Moreover, there is a bilinear map $\qcomp\colon \corr_{M_1}\times\corr_{M_2}\to \corr_M$ such that for all $\psi_1\in\cH_{\partial M_1}^\ds$, $\psi_2\in\cH_{\partial M_2}^\ds$ and~$\rho_{M_1}^{F_1}\in \corr_{M_1}$, $\rho_{M_2}^{F_2}\in \corr_{M_2}$,
\begin{equation*}
 \big(\rho_{M_1}^{F_1}\qcomp\rho_{M_2}^{F_2}\big)\circ\tau_{\partial M_1,\partial M_2;\partial M}(\psi_1\tens\psi_2)= \rho_{M_1}^{F_1}(\psi_1)\rho_{M_2}^{F_2}(\psi_2).
\end{equation*}
What is more, $\rho_M=\rho_{M_1}\qcomp\rho_{M_2}$.
\item[{(TO5b)}] Let $M$ be a region with its boundary decomposing as a union $\partial M=\Sigma_1\cup\Sigma\cup \overline{\Sigma'}$, where~$\Sigma'$ is a copy of $\Sigma$. Let~$M_1$ denote the gluing of $M$ with itself along $\Sigma,\overline{\Sigma'}$ and suppose that~$M_1$ is a region. Then, $\tau_{\Sigma_1,\Sigma,\overline{\Sigma'};\partial M}(\psi\tens\xi\tens\iota_\Sigma(\xi))\in\cH_{\partial M}^\ds$ for all $\psi\in\cH_{\partial M_1}^\ds$ and $\xi\in\cH_\Sigma$. Moreover, there is a linear map $\qcomp\colon \corr_{M}\to\corr_{M_1}$ such that for any ON-basis $\{\zeta_k\}_{k\in I}$ of $\cH_\Sigma$, we have for all $\rho_M^F\in\corr_M$ and $\psi\in\cH_{\partial M_1}^\ds$
\begin{equation*}
 \qcomp\left(\rho_{M}^F\right)(\psi)\cdot c\big(M;\Sigma,\overline{\Sigma'}\big)
 =\sum_{k\in I}\rho_M^F\circ\tau_{\Sigma_1,\Sigma,\overline{\Sigma'};\partial M}\left(\psi\tens\zeta_k\tens\iota_\Sigma(\zeta_k)\right).
\end{equation*}
Here, $c\big(M;\Sigma,\overline{\Sigma'}\big)\in\bC\setminus\{0\}$ is called the \emph{gluing anomaly factor} and depends only on the geometric data. What is more, $\qcomp(\rho_M)=\rho_{M_1}$.
\end{itemize}

\section{Composition via inner product}
\label{sec:compip}

In this section we provide the proof for Theorem~\ref{thm:acompo}. Recall the context of Section~\ref{sec:corrak}.

An important ingredient in the proof of Theorem~\ref{thm:stdcompo} is \cite[Lemma~4.1]{Oe:feynobs}. It tuns out there is an analog in the present setting. Again, this is a useful auxiliary result for showing composition.

\begin{lem}
 Let $\xi\in A_M^{D,\bC}$ arbitrary and $\tau\in L_{\partial M}^\bC$. Then,
 \begin{equation*}
 \rho_M^F\big(K^{\pol}_{\xi+\tau}\big) = \rho_M\big(K^\pol_{\tau}\big)\exp\bigg(\frac{\im}{2}D(\xi)+\im\, \omega_{\partial M}(\tau,\xi)\bigg).
 \end{equation*}
\end{lem}
\begin{proof}
 Note $\xi=\eta+\xi^\ipol$ with $\eta\in A_M^{D,\bC}\cap L_{\partial M}^\pol$ and the previous conventions:
 \begin{align*}
 \rho_M^F(K^\pol_{\xi+\tau}) & = \exp\bigg(\frac{\im}{2}D(\eta)\bigg)
 F\big(\xi^\ipol+\tau^\ipol\big)
 \exp\big(\im\,\omega_{\partial M}\big(\xi+\tau,\xi^\ipol+\tau^\ipol\big)\big)
 \\
 & = \exp\bigg(\frac{\im}{2}D(\eta)+2\im\, \omega_{\partial M}(\xi^\ipol+\tau^\ipol,\eta)
 -\im\,\omega_{\partial M}\big(\xi^\ipol+\tau^\ipol,\eta+\tau\big)\bigg)
 \\
 & = \exp\bigg(\frac{\im}{2}D(\eta)+\im\, \omega_{\partial M}(\xi^\ipol+\tau^\ipol,\eta)
 -\im\,\omega_{\partial M}\big(\xi^\ipol+\tau^\ipol,\tau\big)\bigg)
 \\
 & = \exp\bigg(\frac{\im}{2}D(\xi)+\im\, \omega_{\partial M}(\tau,\eta)
 -\im\,\omega_{\partial M}\big(\xi^\ipol+\tau^\ipol,\tau\big)\bigg)
 \\
 & = \exp\bigg(\frac{\im}{2}D(\xi)+\im\, \omega_{\partial M}(\tau,\xi)
 -\im\,\omega_{\partial M}\big(\tau^\ipol,\tau\big)\bigg)
 \\
 & = \exp\bigg(\frac{\im}{2}D(\xi)+\im\, \omega_{\partial M}(\tau,\xi)\bigg)
 \rho_M\big(K^\pol_{\tau}\big).\tag*{\qed}
\end{align*}
\renewcommand{\qed}{}
\end{proof}

\begin{lem}
 \label{lem:iphol}
 Let $L$ be a real symplectic vector space with symplectic form $\omega$. Let $L^\bC$ be the complexification and denote the complex bilinear extension of $\omega$ also by $\omega$. Let $L^\pol$ and $L^\cpol$ be complementary Lagrangian subspaces of $L^\bC$. For $\xi\in L^\bC$ we write $\xi=\xi^\pol+\xi^\cpol$, where $\xi^\pol\in L^\pol$ and $\xi^\cpol\in L^\cpol$. Let $\eta,\tau\in L^\bC$. Then, the following functions $L^\bC\to\bC$ are holomorphic:
 \begin{equation*}
 \xi\mapsto \omega\Big((\eta+\xi)^\pol,(\tau)^\cpol\Big),\qquad
 \xi\mapsto \omega\Big((\eta+\xi)^\pol,(\tau+\xi)^\cpol\Big).
 \end{equation*}
\end{lem}
We leave the simple proof to the reader.

Recall Definition~\ref{dfn:aadmiss} of admissibility.

\begin{thm}
 Suppose that the gluing is admissible. Let $\phi\in L_{\partial M_1}^\bC$ and $\tau\in L_{\Sigma}^\bC$. Then,
 \begin{equation}
 \label{eq:cohglueampl}
 \rho_{M_1}\big(K^{\pol_1}_{\phi}\big) \cdot c\big(M;\Sigma,\overline{\Sigma'}\big) =\int_{\hat{L}_{\Sigma}^\alpha} \rho_M\big(K^\pol_{(\phi,\xi+\tau,\xi+\tau)}\big)\,
 \exp\bigg(\frac{1}{2} g_{\Sigma}^{\alpha}(\xi,\xi)\bigg)\, \xd\nu(\xi).
 \end{equation}
\end{thm}
\begin{proof}
 Let $\eta\in L_{\Sigma}^\bC$ and decompose $\eta=\eta^{\alpha}+\im \eta^{-\alpha}$ uniquely with $\eta^{\alpha},\eta^{-\alpha}\in L_{\Sigma}^\alpha$.
 For $\xi\in L_{\Sigma}^{\alpha}$ consider the function $\bC\to\bC$
 \begin{align*}
 t&\mapsto \rho_M\big(K^\pol_{(\phi,\xi+\eta^\alpha+t \eta^{-\alpha}, \xi+\eta^\alpha+ t \eta^{-\alpha})}\big)
 \\
 &= \exp\big(\im\, \omega_{\partial M}\big((\phi,
 \xi+\eta^\alpha+ t \eta^{-\alpha},\xi+\eta^\alpha\!+ t \eta^{-\alpha})^\epol, (\phi,
 \xi+\eta^\alpha+ t \eta^{-\alpha},\xi+\eta^\alpha\!+ t \eta^{-\alpha})^\ipol\big)\big).
 \end{align*}
 By Lemma~\ref{lem:iphol} this function is holomorphic in $t$. By Proposition~3.11 of~\cite{Oe:holomorphic} we have for $t\in\R$ the equality
 \begin{gather*}
 \int_{\hat{L}_{\Sigma}^\alpha} \rho_M\big(K^\pol_{(\phi,\xi,\xi)}\big)\,\exp\bigg(\frac{1}{2} g_{\Sigma}^{\alpha}(\xi,\xi)\bigg)\, \xd\nu(\xi)
 \\ \qquad
 {}= \int_{\hat{L}_{\Sigma}^\alpha} \rho_M\big(K^\pol_{(\phi,\xi+\eta^\alpha + t \eta^{-\alpha},\xi+\eta^\alpha+ t \eta^{-\alpha})}\big)\,\exp\bigg(\frac{1}{2} g_{\Sigma}^{\alpha}(\xi,\xi)\bigg)\, \xd\nu(\xi).
 \end{gather*}
 In particular, the integral on the right-hand side is constant for $t\in\R$. On the other hand, the integrand is holomorphic in $t$, so the integral is also holomorphic in $t$. But constancy for $t\in\R$ then implies constancy for $t\in\bC$. In particular, choosing $t=\im$ implies the equality
 \begin{gather*}
 \int_{\hat{L}_{\Sigma}^\alpha} \rho_M\big(K^\pol_{(\phi,\xi,\xi)}\big)\,\exp\bigg(\frac{1}{2} g_{\Sigma}^{\alpha}(\xi,\xi)\bigg)\, \xd\nu(\xi)
 \\ \qquad
 {}= \int_{\hat{L}_{\Sigma}^\alpha} \rho_M\big(K^\pol_{(\phi,\xi+\eta,\xi+\eta)}\big)\,\exp\bigg(\frac{1}{2} g_{\Sigma}^{\alpha}(\xi,\xi)\bigg)\, \xd\nu(\xi).
 \end{gather*}

 Decompose $\phi=\phi^\ipol+\phi^\epol$ uniquely with $\phi^\ipol\in L_{\partial M_1}^\ipol=L_{\tilde{M}_1}^\bC$ and $\phi^\epol\in L_{\partial M_1}^\epol$. By Axiom~(C7) of Appendix~\ref{sec:caxioms} there exists $\phi_{\Sigma}^\ipol\in L_{\Sigma}^\bC$ such that $(\phi^\ipol,\phi_{\Sigma}^\ipol,\phi_{\Sigma}^\ipol)\in L_{\partial M}^\ipol=L_{\tilde{M}}^\bC$. Applying the corresponding assumption to the ``exterior'' of $M$ and $M_1$ (compare Section~\ref{sec:corrak}), there exists $\phi_{\Sigma}^\epol\in L_{\Sigma}^\bC$ such that $(\phi^\epol,\phi_{\Sigma}^\epol,\phi_{\Sigma}^\epol)\in L_{\partial M}^\epol$. The result just obtained implies
 \begin{gather*}
 \int_{\hat{L}_{\Sigma}^\alpha} \rho_M\big(K^\pol_{(\phi,\xi+\tau,\xi+\tau)}\big)\,
 \exp\bigg(\frac{1}{2} g_{\Sigma}^{\alpha}(\xi,\xi)\bigg)\, \xd\nu(\xi) \nonumber
 \\ \qquad
 {}= \int_{\hat{L}_{\Sigma}^\alpha} \rho_M\big(K^\pol_{\left(\phi,\xi+\phi_{\Sigma}^\ipol+\phi_{\Sigma}^\epol,\xi +\phi_{\Sigma}^\ipol+\phi_{\Sigma}^\epol\right)}\big)\,
 \exp\bigg(\frac{1}{2} g_{\Sigma}^{\alpha}(\xi,\xi)\bigg)\, \xd\nu(\xi).
 \end{gather*}
 We can rewrite part of the integrand as follows
 \begin{gather*}
 \rho_M\big(K^\pol_{\left(\phi,\xi+\phi_{\Sigma}^\ipol+\phi_{\Sigma}^\epol,\xi +\phi_{\Sigma}^\ipol+\phi_{\Sigma}^\epol\right)}\big)
 \\ \qquad
 {}= \exp\big(\im\,\omega_{\partial M}\big(\big(\phi^\ipol+\phi^\epol,\xi+\phi_{\Sigma}^\ipol +\phi_{\Sigma}^\epol,\xi+\phi_{\Sigma}^\ipol+\phi_{\Sigma}^\epol\big)^\epol,
 \\ \qquad \hphantom{=\exp\big(\im\,\omega_{\partial M}\big(}
 \big(\phi^\ipol+\phi^\epol,\xi+\phi_{\Sigma}^\ipol +\phi_{\Sigma}^\epol,\xi +\phi_{\Sigma}^\ipol+\phi_{\Sigma}^\epol\big)^\ipol \big)\big)
 \\ \qquad
 {} = \exp\big(\im\,\omega_{\partial M}\big(\big(\phi^\epol,\xi+\phi_{\Sigma}^\epol,\xi+\phi_{\Sigma}^\epol\big)^\epol,
 \big(\phi^\ipol,\xi+\phi_{\Sigma}^\ipol,\xi+\phi_{\Sigma}^\ipol\big)^\ipol
 \big)\big)
 \\ \qquad
 {} = \exp\big(\im\,\omega_{\partial M}\big(\big(0,\xi,\xi\big)^\epol,
 \big(0,\xi,\xi\big)^\ipol
 \big)+\im\,\omega_{\partial M}\big(\big(\phi^\epol,\phi_{\Sigma}^\epol,\phi_{\Sigma}^\epol\big),
 \big(\phi^\ipol,\phi_{\Sigma}^\ipol,\phi_{\Sigma}^\ipol\big)
 \big)
 \\ \qquad \hphantom{=\exp\big(}
 {} +\im\,\omega_{\partial M}\big(\big(\phi^\epol,\phi_{\Sigma}^\epol,\phi_{\Sigma}^\epol\big),
 \big(0,\xi,\xi\big)\big)
 +\im\,\omega_{\partial M}\big(\big(0,\xi,\xi\big),
 \big(\phi^\ipol,\phi_{\Sigma}^\ipol,\phi_{\Sigma}^\ipol\big)
 \big)\big)
 \\ \qquad
 {} = \exp\big(\im\,\omega_{\partial M}\big(\big(0,\xi,\xi\big)^\epol\!,
 \big(0,\xi,\xi\big)^\ipol
 \big)\big) \exp\!\big(\im\,\omega_{\partial M_1}\big(\phi^\epol,\phi^\ipol\big)\big)
 \! =\! \rho_M\big(K^\pol_{(0,\xi,\xi)}\big) \rho_{M_1}\big(K^\pol_{\phi}\big).
 \end{gather*}
 Combining results and using the definition~(\ref{eq:defglanom}) of the gluing anomaly factor, we obtain the desired equality~(\ref{eq:cohglueampl}).
\end{proof}

We are now ready to provide the proof of Theorem~\ref{thm:acompo}.
\begin{proof}[Proof of Theorem~\ref{thm:acompo}]
 Let $\eta\in A^{D_1,\bC}_{M_1}$ arbitrary. By Diagram~(60) of Axiom~(C7) of \cite[Section~4.6]{Oe:feynobs} there exists $\eta_\Sigma\in L_\Sigma^\bC$ such that $(\eta,\eta_{\Sigma},\eta_{\Sigma})\in A^{D,\bC}_M$. By the same axiom we have $D((\eta,\eta_{\Sigma},\eta_{\Sigma}))=D_1(\eta)$:\footnote{Compare the proof of Proposition~4.2 in~\cite{Oe:feynobs} for more detailed explanations.}
\begin{gather*}
\int_{\hat{L}_{\Sigma}^\alpha} \rho_M^F\big(K^\pol_{(\phi,\xi,\xi)}\big)
\exp\bigg(\frac{1}{2} g_{\Sigma}^{\alpha}(\xi,\xi)\bigg)\, \xd\nu(\xi)
\\ \qquad
{} = \int_{\hat{L}_{\Sigma}^\alpha} \rho_M\big(K^\pol_{(\phi-\eta,\xi-\eta_{\Sigma},\xi-\eta_{\Sigma})}\big)
\exp\bigg(\frac{1}{2} g_{\Sigma}^{\alpha}(\xi,\xi)\bigg)
\\ \qquad\phantom{=}
{}\times \exp\bigg(\frac{1}{2}D((\eta,\eta_{\Sigma},\eta_{\Sigma}))+\im\, \omega_{\partial M}((\phi-\eta,\xi-\eta_{\Sigma},\xi-\eta_{\Sigma}),(\eta,\eta_{\Sigma},\eta_{\Sigma}))\bigg)\, \xd\nu(\xi)
\\ \qquad
{} = \int_{\hat{L}_{\Sigma}^\alpha} \! \rho_M\big(K^\pol_{(\phi-\eta,\xi-\eta_{\Sigma},\xi-\eta_{\Sigma})}\big)
\exp\bigg(\frac{1}{2} g_{\Sigma}^{\alpha}(\xi,\xi)\bigg)
\exp\bigg(\frac{1}{2}D_1(\eta)+\im\, \omega_{\partial M_1}(\phi-\eta,\eta)\bigg)\, \xd\nu(\xi)
\\ \qquad
{} = \rho_{M_1}\big(K^{\pol_1}_{\phi-\eta}\big) \exp\bigg(\frac{1}{2}D_1(\eta)+\im\, \omega_{\partial M_1}(\phi-\eta,\eta)\bigg)\, c(M;\Sigma,\overline{'})
\\ \qquad
{} = \rho_{M_1}^{F_1}\big(K^{\pol_1}_{\phi}\big)\, c\big(M;\Sigma,\overline{\Sigma'}\big).
\tag*{\qed}
\end{gather*}
\renewcommand{\qed}{}
\end{proof}

\subsection*{Acknowledgements}

The authors would like to thank the anonymous referees for contributing to improving the clarity of exposition of the present article. This work was partially supported by CONACYT project grant 259258 and UNAM-PAPIIT project grant IA-106418. This publication was made possible through the support of the ID\# 61466 grant from the John Templeton Foundation, as part of the ``Quantum Information Structure of Spacetime (QISS)'' Project (qiss.fr). The opinions expressed in this publication are those of the authors and do not necessarily reflect the views of the John Templeton Foundation.

\LastPageEnding

\end{document}